\documentclass[aoas]{imsart}

\RequirePackage{amsthm,amsmath,amsfonts,amssymb}
\RequirePackage[authoryear]{natbib}
\RequirePackage[colorlinks,citecolor=blue,urlcolor=blue]{hyperref}
\RequirePackage{graphicx}

\RequirePackage{bm} 
\RequirePackage{url}
\usepackage{multirow}
\usepackage{graphicx}
\usepackage{subfigure}
\usepackage{makecell}
\usepackage{booktabs}
\usepackage{array}
\usepackage{url}
\usepackage{algorithm}
\usepackage{algorithmic}
\usepackage{dsfont}

\usepackage{enumerate}
\usepackage{graphicx}
\usepackage{mathrsfs}
\usepackage{color}

\startlocaldefs
\theoremstyle{plain}

\newtheorem{theorem}{Theorem}[section]
\newtheorem{lemma}{Lemma}
\newtheorem{remark}{Remark}
\newtheorem{condition}{Condition}
\newtheorem{proposition}{Proposition}
\newtheorem{assumption}{Assumption}
\newtheorem{corollary}{Corollary}

\theoremstyle{plain}

\let\hat\widehat
\let\tilde\widetilde

\newcommand{\ba}{\bm{a}}
\newcommand{\bb}{\bm{b}}

\newcommand{\be}{\bm{e}}
\newcommand{\bbf}{\bm{f}}

\newcommand{\bn}{\bm{n}}

\newcommand{\br}{\bm{r}}

\newcommand{\bu}{\bm{u}}
\newcommand{\bv}{\bm{v}}
\newcommand{\bw}{\bm{w}}
\newcommand{\bx}{\bm{x}}

\newcommand{\bz}{\bm{z}}

\newcommand{\Ab}{\mathbf{A}}
\newcommand{\Bb}{\mathbf{B}}

\newcommand{\Db}{\mathbf{D}}
\newcommand{\Eb}{\mathbf{E}}

\newcommand{\Gb}{\mathbf{G}}
\newcommand{\Hb}{\mathbf{H}}
\newcommand{\Ib}{\mathbf{I}}

\newcommand{\Lb}{\mathbf{L}}
\newcommand{\Mb}{\mathbf{M}}

\newcommand{\Pb}{\mathbf{P}}

\newcommand{\Rb}{\mathbf{R}}

\newcommand{\Tb}{\mathbf{T}}
\newcommand{\Ub}{\mathbf{U}}
\newcommand{\Vb}{\mathbf{V}}

\newcommand{\Xb}{\mathbf{X}}
\newcommand{\Yb}{\mathbf{Y}}
\newcommand{\Zb}{\mathbf{Z}}

\newcommand{\bB}{\bm{B}}

\newcommand{\bD}{\bm{D}}

\newcommand{\bG}{\bm{G}}

\newcommand{\bI}{\bm{I}}

\newcommand{\bS}{\bm{S}}

\newcommand{\bU}{\bm{U}}
\newcommand{\bV}{\bm{V}}

\newcommand{\bX}{\bm{X}}
\newcommand{\bY}{\bm{Y}}
\newcommand{\bZ}{\bm{Z}}

\newcommand{\cB}{\mathcal{B}}

\newcommand{\cH}{\mathcal{H}}

\newcommand{\cL}{\mathcal{L}}
\newcommand{\cM}{\mathcal{M}}
\newcommand{\cN}{\mathcal{N}}

\newcommand{\cQ}{\mathcal{Q}}
\newcommand{\cR}{\mathcal{R}}
\newcommand{\cS}{{\mathcal{S}}}

\newcommand{\cU}{\mathcal{U}}
\newcommand{\cV}{\mathcal{V}}

\newcommand{\EE}{\mathbb{E}}

\newcommand{\PP}{\mathbb{P}}

\newcommand{\RR}{\mathbb{R}}

\newcommand{\balpha}{\bm{\alpha}}
\newcommand{\bbeta}{\bm{\beta}}
\newcommand{\bgamma}{\bm{\gamma}}
\newcommand{\bdelta}{\bm{\delta}}

\newcommand{\bvarepsilon}{\bm{\varepsilon}}

\newcommand{\btheta}{\bm{\theta}}

\newcommand{\blambda}{\bm{\lambda}}

\newcommand{\bnu}{\bm{\nu}}

\newcommand{\bupsilon}{\bm{\upsilon}}
\newcommand{\bphi}{\bm{\phi}}

\newcommand{\bpsi}{\bm{\psi}}
\newcommand{\bomega}{\bm{\omega}}

\newcommand{\bGamma}{\bm{\Gamma}}
\newcommand{\bDelta}{\bm{\Delta}}
\newcommand{\bTheta}{\bm{\Theta}}
\newcommand{\bLambda}{\bm{\Lambda}}

\newcommand{\bPi}{\bm{\Pi}}
\newcommand{\bSigma}{\bm{\Sigma}}
\newcommand{\bUpsilon}{\bm{\Upsilon}}
\newcommand{\bPhi}{\bm{\Phi}}
\newcommand{\bPsi}{\bm{\Psi}}
\newcommand{\bOmega}{\bm{\Omega}}

\newcommand{\argmin}{\mathop{\mathrm{argmin}}}

\newcommand{\sign}{\mathop{\mathrm{sign}}}

\DeclareMathOperator{\Cov}{\rm Cov}

\newcommand*{\zero}{{\bm 0}}
\newcommand*{\one}{{\bm 1}}

\newcommand{\diag}{{\rm diag}}
\def\T{{ \intercal }}

\usepackage{tikz}
\usetikzlibrary{positioning}

\endlocaldefs

\begin{document}

\begin{frontmatter}
\title{Statistical Inference for Covariate-Adjusted and Interpretable Generalized Latent Factor Model
with Application to Testing Fairness
}
\runtitle{Inference for Covariate-Adjusted Generalized Factor Model}

\begin{aug}

\author[A]
{\fnms{Jing}~\snm{Ouyang}\ead[label=e1]{jingoy@hku.hk}},
\author[B]{\fnms{Chengyu}~\snm{Cui}\ead[label=e2]{chyc@umich.edu}}\footnote{Ouyang and Cui contribute equally to this work},
\author[B]{\fnms{Kean Ming}~\snm{Tan}\ead[label=e3]{keanming@umich.edu}}
\and
\author[B]{\fnms{Gongjun}~\snm{Xu}\ead[label=e4]{gongjun@umich.edu}}
\address[A]{Faculty of Business and Economics,
       University of Hong Kong \printead[presep={ ,\ }]{e1}}

\address[B]{Department of Statistics,
University of Michigan\printead[presep={,\ }]{e2,e3,e4}}
\end{aug}

\begin{abstract}
Latent variable models are popularly used to measure latent embedding factors from large-scale assessment data. Beyond understanding these latent factors, the covariate effect on responses controlling for latent factors is also of great scientific interest and has wide applications, such as evaluating the fairness of educational testing, where the covariate effect reflects whether a test question is biased toward certain individual characteristics (e.g., gender and race), taking into account their latent abilities. However, the large sample sizes and high dimensional responses pose challenges to developing efficient methods and drawing valid inferences. Moreover, to accommodate the commonly encountered discrete responses, generalized latent factor models are often assumed, adding further complexity. To address these challenges, we consider a covariate-adjusted generalized factor model and develop novel and interpretable conditions to address the identifiability issue. Based on the identifiability conditions, we propose a joint maximum likelihood estimation method and establish estimation consistency and asymptotic normality results for the covariate effects. Furthermore, we derive estimation and inference results for latent factors and the factor loadings.  We illustrate the finite sample performance of the proposed method through extensive numerical studies and an educational assessment dataset from the Programme for International Student Assessment (PISA). 
\end{abstract}


\begin{keyword}
\kwd{Generalized latent factor model}
\kwd{Identifiability}
\kwd{Large-scale assessment}
\end{keyword}

\end{frontmatter}

\section{Introduction}
\label{sec:intro}

Latent factors, often referred to as hidden factors, play an increasingly important role in modern statistics and machine learning to analyze large-scale complex measurement data and find wide-ranging applications across various scientific fields, including educational assessments~\citep{reckase2009, hambleton2013item}, macroeconomics forecasting~\citep{stock2002forecasting, lam2011estimation}, and biomedical diagnosis~\citep{carvalho2008high, frichot2013testing}.
For instance, in educational testing and social sciences, latent factors are used to model unobservable traits of respondents, such as skills, personality, and attitudes~\citep{von2008toefl, reckase2009}; in biology and genomics, latent factors are used to capture underlying genetic factors, gene expression patterns, or hidden biological mechanisms~\citep{carvalho2008high, frichot2013testing}.  
To uncover the latent factors and analyze large-scale complex data, various latent factor models
have been developed and extensively investigated in the existing literature~\citep{bai2012statistical, fan2013large, chen2023statistical, wang2022maximum}.
In addition to measuring the latent factors, the observed covariates and the covariate effects conditional on the latent factors hold significant scientific interpretations in many applications~\citep{reboussin2008locally, park2018explanatory}.

\subsection{Application and Motivation}
One important application is testing fairness, which receives increasing attention in the fields of education, psychology, and social sciences~\citep{candell1988iterative, belzak2020improving, chen2023dif}. In educational assessments, testing fairness, or measurement invariance, implies that groups from diverse backgrounds have the same probability of endorsing the test items, controlling for individual proficiency levels~\citep{millsap2012statistical}. 
Testing fairness is not only of scientific interest for psychometricians and statisticians but also attracts wide public awareness~\citep{lawsuit1984}.  
In the era of rapid technological advancement, large-scale international educational assessments are becoming increasingly prevalent. 
One example is the Programme for International Student Assessment (PISA), which is a large-scale worldwide assessment to evaluate the academic performance of 15-year-old students across many countries and economies~\citep{oecd2019}. The testing programme is conducted every four years and in the 2018 cycle, over 600,000 students from 79 countries and economies participated in this assessment, representing a population of approximately 31 million 15-year-olds~\citep{pisatechnicalreport2018}. PISA 2018 adopted the computer-based assessment mode and the assessment mainly evaluated students' proficiency in mathematics, reading, and science domains~\citep{oecd2019}. 
In addition to the testing questions, background questionnaires were administered to collect information including student demographics (e.g., gender and race), family background, and school characteristics~\citep{pisatechnicalreport2018}. 
By collecting both response data and background information, PISA aims to evaluate not only students' latent abilities but also the fairness of the assessment~\citep{pisatechnicalreport2018}. 
 To evaluate testing fairness in large-scale assessments, statisticians and psychometricians have been developing modern and computationally efficient methodologies for interpreting the effects of demographic covariates (e.g., gender and race) on the responses to testing items.

 However, the discrete nature of item responses, the growing sample size, and the large number of test items in modern educational assessments pose great statistical challenges for the estimation and inference of covariate effects and latent factors. 
For instance, in educational and psychological measurement, such a testing fairness issue (measurement invariance) is typically assessed by differential item functioning (DIF) analysis of item response data, which aims to detect DIF items whose response distributions depend not only on the latent factors but also on respondents’ covariates (such as group membership). 
The DIF problem has been extensively studied, yet many existing methods rely on domain knowledge to pre-specify DIF-free items, commonly referred to as anchor items. Traditional anchor-based methods can be classified into two categories: one line of work does not assume an item response theory (IRT) model~\citep{mantel1959statistical, dorans1986demonstrating, zwick2000using,  may2006multilevel, soares2009integrated, frick2015rasch}, while another line is developed within the IRT framework~\citep{thissen1988use, oort1998simulation, steenkamp1998assessing, tay2016item, cao2017monte}. Compared to non-IRT models, IRT-based methods rely on an explicitly specified IRT model prescribed prior to the study, which enables a clearer formulation of the DIF problem. 
Within the IRT-based framework,  DIF methods are often developed under the multiple indicators, multiple causes (MIMIC) IRT model~\citep{joreskog1975estimation}, a structural equation modeling framework for both continuous~\citep{zellner1970estimation, goldberger1972structural} and categorical item responses~\citep{muthen1985method, muthen1991instructionally}. To address DIF issues, the MIMIC model is constructed by two components including a measurement model, which models the dependence of item responses on latent abilities and observed covariates, and a structural model, which models the conditional distribution of latent abilities given observed covariates~\citep{muthen1989latent}. 



  Despite its widespread use, the anchor-based method has raised a major concern that misspecifying anchor items can lead to biased estimation and invalid inference~\citep{thissen1988use, tay2016item}. 
To address this limitation, researchers have proposed item purification methods that iteratively select anchor items using stepwise selection models~\citep{candell1988iterative, fidalgo2000effects, kopf2015framework}, as well as anchor alignment methods based on inequality criterion~\citep{asparouhov2014multiple, strobl2021anchor}. Recently, more non-anchor-based approaches have been developed, including tree-based methods~\citep{tutz2016item, bollmann2018item}, regularized estimation procedures~\citep{bauer2020simplifying, belzak2020improving, wang2023using}, and other empirical methods~\citep{yuan2021differential}. For example, \cite{bechger2015statistical} introduced a method based on item pair functioning that does not require pre-specified anchor items; \cite{bauer2020simplifying} and~\cite{belzak2020improving} introduced LASSO-type regularized methods for model selection and parameter estimation. However, regularized estimation approaches often involve intensive computation due to the need to solve multiple regularized maximum likelihood estimation problems and perform parameter tuning. Moreover, the existing non-anchor-based methods often do not provide valid statistical inference guarantees for testing the covariate effects. It remains an open problem to perform statistical inference on the covariate effects and the latent factors.





\subsection{Our Contribution}
Motivated by these applications, we study the statistical estimation and inference for a general family of covariate-adjusted generalized latent factor models, which include the popular factor models for binary, count, continuous, and mixed-type data that commonly occur in educational and psychological assessments.
Despite recent progress in the factor analysis literature, most existing studies focus on estimation and inference under linear factor models \citep{stock2002forecasting, bai2012statistical, fan2013large}
and covariate-adjusted linear factor models~\citep{leek2008general, wang2017confounder, gerard2020empirical, bing2023inference}. 
The techniques employed in linear factor model settings are not applicable here due to the nonlinearity inherent in the general models under consideration. 
The nonlinear nature of the setting has led to a significant research gap in the literature, and the inference theory of generalized factor analysis has remained largely underexplored. 
Recently, several researchers have also investigated the parameter estimation and inference for generalized linear factor models~\citep{chen2019joint, wang2022maximum, chen2023statistical}. 
However, they either focus only on the overall consistency properties of the estimation or do not incorporate covariates into the models. 
In generalized factor analysis incorporating covariates, \cite{chen2021nonlinear} considered homogeneous covariate effects in a nonlinear panel data model, which captures a homogeneous effect invariant across all items. However, such homogeneous effect makes it not applicable to the motivating  applications such as educational and psychological assessments, where item-specific covariate effects are of primary interest. 
More recently \cite{du2023simultaneous} considered a generalized linear factor model with covariates in single-cell omics and studied its inference theory, using latent factors as surrogate variables to control for unmeasured confounding. However, they imposed relatively stringent assumptions on the sparsity of covariate effects and the dimension of covariates, and their theoretical results also rely on data-splitting. Moreover, \cite{chen2021nonlinear} and \cite{du2023simultaneous} focused primarily on statistical inference on the covariate effects, while that on factors and loadings was unexplored, which is of both theoretical and practical interest.

Establishing inference results for covariate effects and latent factors simultaneously under generalized latent factor models remains an open and challenging problem.
 A key difficulty lies in the inherent identifiability issues, which arise not only between factors and loadings but also between the latent factors and covariate effects, steming from the possible  correlations between the latent factors and covariates.   
To overcome these issues,  we develop a novel framework for performing statistical inference on all model parameters and latent factors under a general family of covariate-adjusted generalized factor models. 
We propose new identifiability conditions for identifying the model parameters, and further incorporate these conditions into the development of new likelihood-based statistical inference theory. In particular, our identifiability conditions are practically interpretable and easily satisfied in common educational and psychological assessments.  
Under these identifiability conditions, we develop new techniques to address the aforementioned theoretical challenges and obtain estimation consistency and asymptotic normality results for  covariate effects as well as the latent factors and factor loadings.


The rest of the paper is organized as follows. In Section~\ref{sec:model setup}, we introduce the model setup of the covariate-adjusted generalized factor model. Section~\ref{sec:method} discusses the associated identifiability issues and further presents the proposed identifiability conditions and estimation method. Section~\ref{sec:theory} establishes the theoretical properties for not only the covariate effects but also the latent factors and factor loadings. In Section~\ref{sec:simulation}, we perform extensive numerical studies to illustrate the performance of the proposed estimation method and the validity of the theoretical results. In Section~\ref{sec:data application}, we analyze an educational testing dataset from Programme for International Student Assessment (PISA) and identify test items that may lead to potential bias among different test-takers.
We provide some potential future directions in Section~\ref{sec:discussion}. The proofs for the theoretical results presented in the paper, along with additional simulation results, are included in the Supplementary Material.   
\\

\noindent Notation: For any integer $N$, let $[N] = \{1, \dots, N\}$. 
For any set $\cS$, let $\# \cS$ be its cardinality.
For any vector $\br = (r_1, \dots, r_l)^{\T}$, let $\| \br\|_0  = \# (\{j: r_j \neq 0 \})$,  $\|\br \|_{\infty}= \max_{j = 1, \ldots, l} |r_j|$, and $\|\br\|_q = (\sum_{j=1}^l |r_j|^q)^{1/q}$ for $q \geq 1$.
 We define $\bm{1}_x^{(y)}$ as the $y$-dimensional vector with $x$-th entry to be 1 and all other entries to be 0.
For any symmetric matrix $\Mb$, let $\lambda_{\min}(\Mb)$ and $\lambda_{\max}(\Mb)$ be the smallest and largest eigenvalues of $\Mb$, respectively. 
For any matrix $\Ab = (a_{ij})_{n\times l}$, let $\|\Ab\|_{\infty,1} = \max_{j=1,\ldots, l} $ $\sum_{i = 1}^n |a_{ij} |$ be the maximum absolute column sum, $\| \Ab\|_{1,\infty} = \max_{i=1,\ldots, n} \sum_{j=1}^l |a_{ij}| $ be the maximum of the absolute row sum, $\| \Ab\|_{\max} = \max_{i=1, \dots, n;j=1, \dots,l} |a_{ij}|$ be the maximum of the absolute matrix entry,
$\| \Ab\|_{F} = (\sum_{i=1}^n \sum_{j=1}^l |a_{ij}|^2)^{1/2}$ be the Frobenius norm of $\Ab$, and $\|\Ab\|=\sqrt{\lambda_{\max }\left(\Ab^{\T} \Ab\right)}$ be the spectral norm of $\Ab$.  Let $\| \cdot \|_{\varphi_1}$ be sub-exponential norm. 
 Define $\Ab_v = \text{vec}(\Ab) \in \RR^{nl}$ to indicate the vectorized form of matrix $\Ab\in \RR^{n\times l}$. Finally, we denote $\otimes $ as the Kronecker product. For ease of reading, we also provide a summary of the key concepts and notations in Table~S1 in the Supplementary Material.

\section{Model Setup}
\label{sec:model setup}
Consider $n$ independent subjects with $q$ measured responses and $p_*$ observed covariates.  
For the $i$th subject, let $\bY_i \in \RR^q$ be a $q$-dimensional vector of responses corresponding to $q$ measurement items and   $\bX_i^c \in \RR^{p_*}$ be a $p_*$-dimensional vector of observed covariates. Moreover, let $\bU_i$ be a $K$-dimensional vector of latent factors representing the unobservable traits such as skills and personalities, where we assume $K$ is specified as in many educational assessments.  
 We assume that the $q$-dimensional responses $\bY_i$ are conditionally independent, given $\bX_i^c$ and $\bU_i$.  
 Specifically, for $i \in [n]$ and $j \in [q]$, we model the $j$th response for the $i$th subject, $Y_{ij}$, by the following  conditional distribution:
\begin{equation}
\label{eq:modelsetup}
	Y_{ij} \sim p_{ij}( y \mid w_{ij}), \qquad \mathrm{where}~~w_{ij} = \beta_{j0} + \bgamma_j^{\intercal}\bU_i  + \bbeta_{jc}^{\intercal}\bX_i^c. 
\end{equation}
 Here $p_{ij}( y | w_{ij})$ denotes the probability density or mass function of the response $Y_{ij}$. We assume that the parametric form of $p_{ij}( y | w_{ij})$ is known and thus this conditional probability function is determined by $w_{ij}$. 
The $\beta_{j0} \in \RR$ is the intercept parameter, $\bbeta_{jc} =(\beta_{j1},\ldots,\beta_{jp_*})^{\intercal} \in \RR^{p_*}$ are the coefficient parameters for the observed covariates, and $\bgamma_j = (\gamma_{j1},\ldots,\gamma_{jK})^{\intercal}\in \RR^K$ are the factor loadings. 
For ease of presentation, we write $\bbeta_j= (\beta_{j0},\bbeta_{jc}^{\T})^{\T}$ as an assembled vector of intercept and coefficients and define $\bX_i = (1, (\bX_i^{c})^{\T})^{\T}$ with dimension $p = p_* + 1$, which gives an equivalent form, $w_{ij} =  \bgamma_j^{\intercal}\bU_i  + \bbeta_{j}^{\intercal}\bX_i. $ 

In this paper, we consider a general and flexible modeling framework by allowing different types of $p_{ij}$ functions to model diverse response data in wide-ranging applications, such as binary item response data in educational and psychological assessments \citep{mellenbergh1994generalized, reckase2009} and mixed types of data in educational and macroeconomic applications \citep{rijmen2003nonlinear,wang2022maximum}; see also Remark~\ref{remark1}.
 A schematic diagram of the proposed model setup is presented in Figure~\ref{fig:model setup}.
\begin{figure}[!htp]
    \centering
    \includegraphics[width=0.65\textwidth]{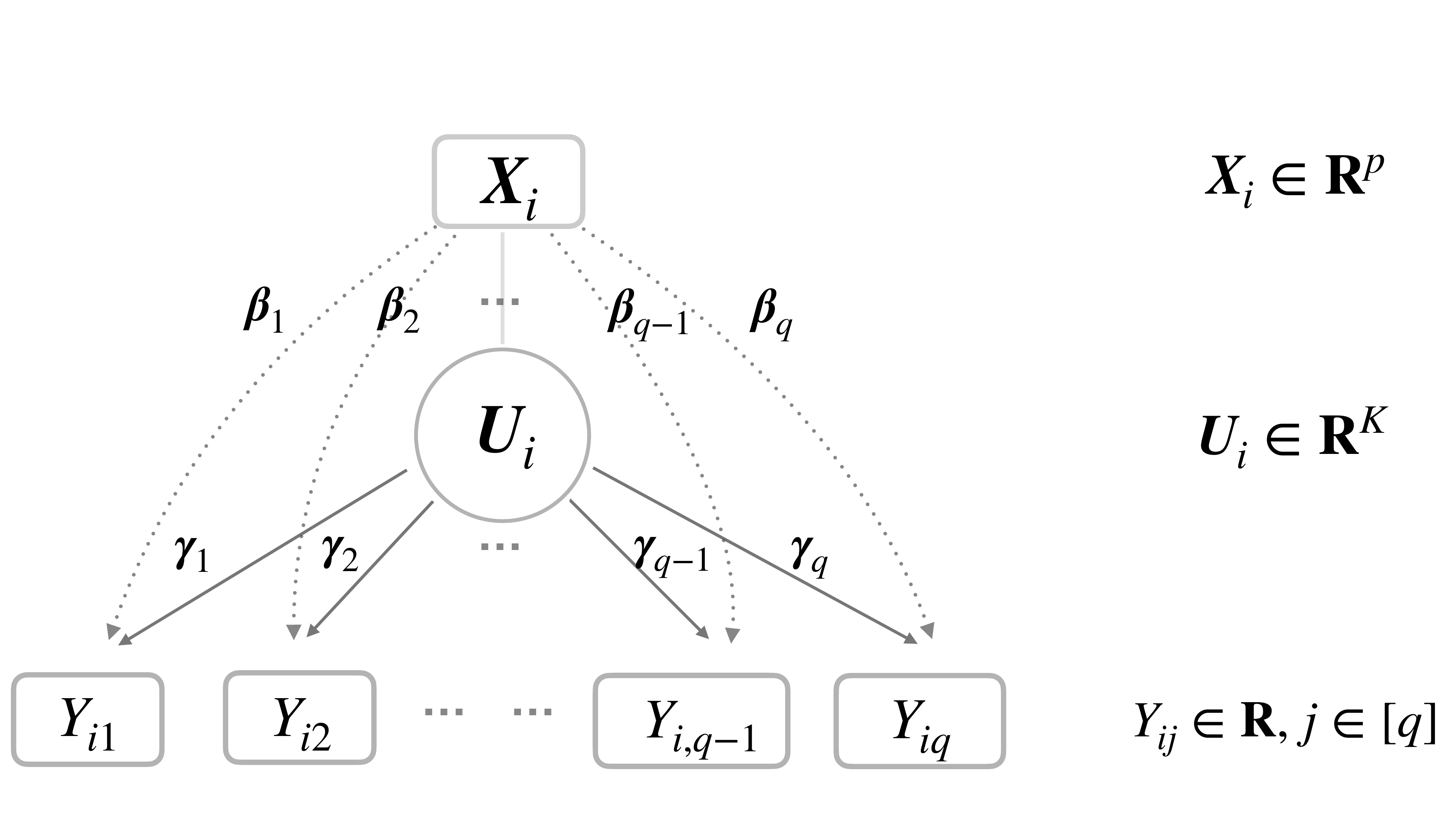}
    \caption{A schematic diagram of the proposed model in~\eqref{eq:modelsetup}. The subscript $i$ indicates the $i$th subject, out of $n$ independent subjects. The response variable $Y_{ij}$ can be discrete or continuous. The observed covariates $\bX_i \in \RR^p$ and the latent factors $\bU_i \in \RR^K$ can be arbitrarily correlated.
    }
    \label{fig:model setup}
\end{figure}

The considered covariate-adjusted generalized factor model in~\eqref{eq:modelsetup} is motivated by applications in testing fairness.
 In the context of educational assessment, the subject's responses to question items are dependent on latent factors $\bU_i$ such as students' abilities and skills, and are potentially affected by observed covariates $\bX_i^c$ such as age, gender, and race, among others~\citep{collins2009}. 
  The intercept $\beta_{j0}$ 
  is often interpreted as the difficulty level of item $j$ and is referred to as the difficulty parameter in psychometrics~\citep{hambleton2013item, reckase2009}.
 The capability of item $j$ to further differentiate individuals based on their latent abilities is captured by $\bgamma_j = (\gamma_{j1},\ldots,\gamma_{jK})^{\intercal}$, which are also referred to as discrimination parameters~\citep{hambleton2013item, reckase2009}. 
 The effects of observed covariates $\bX_i^c$ on subject's response to the $j$th question $Y_{ij}$, conditioned on latent abilities $\bU_i$, are captured by $\bbeta_{jc} =(\beta_{j1},\ldots,\beta_{jp_*})^{\intercal}$, which are referred to as differential item functioning (DIF) effects in psychometrics~\citep{holland2012differential}.
Here a DIF item with nozero $\bbeta_{jc}$ has a response distribution that depends not only on the   latent factors  of interest but also respondents’ covariates (such as group membership). 
This setting gives rise to the fairness problem of validating whether the response probabilities to the measurements differ across different genders, races, or countries of origin while holding their abilities and skills at the same level. 
    To place our considered model within the MIMIC framework, the measurement model is given in equation (1).
Moreover, our model allows for a flexible structural relationship between the latent factors $\bU_i$ and observed covariates $\bX_i^c$, generalizing beyond the linear form or Guassian distribution of $\bU_i$ conditioned on $\bX_i^c$, commonly assumed in traditional MIMIC models.

Given the observed data from $n$ independent subjects, we are interested in studying the relationships between $\bY_i$ and $\bX_i^c$ after adjusting for the latent factors $\bU_i$ in~\eqref{eq:modelsetup}. 
Specifically, our goal is to test the statistical hypothesis $H_0: \beta_{js} = 0$ versus $H_a: \beta_{js} \neq 0$ for $s \in [p_*]$, where $\beta_{js}$ is 
the regression coefficient for the $s$th covariate and the $j$th response, after adjusting for the latent factor $\bU_i$. 
In many applications, the latent factors and factor loadings also carry important scientific interpretations such as students' abilities and test items' characteristics. This also motivates us to perform statistical inferences on the parameters $\beta_{j0}$, $\bgamma_j$, and~$\bU_i$.  

\vspace{-0.1in}

\begin{remark}\label{remark1}
 The proposed model setup~\eqref{eq:modelsetup} is general and flexible as various functions $p_{ij}$'s could be used to model diverse types of response data in wide-ranging applications.
For instance, in educational assessments, logistic factor model~\citep{reckase2009} with 
$p_{ij} (y \mid w_{ij}) = \exp(w_{ij}y)/\{1+\exp(w_{ij})\},  y \in \{0,1\},$
and probit factor model~\citep{birnbaum1968some} with $p_{ij}  (y \mid w_{ij}) = \{\Phi (w_{ij})\}^y \{1-\Phi (w_{ij})\}^{1-y}, y \in \{0,1\},$
where $\Phi(\cdot)$ is the cumulative density function of standard normal distribution,
are widely used to model the binary responses, indicating correct or incorrect answers to the test items.
 Such types of models are often referred to as item response theory models \citep{reckase2009}.
In economics and finances, linear factor models with
$
p_{ij} (y\mid w_{ij}) \propto \exp \{ - (y - w_{ij})^2/(2\sigma^2)\}
$, where $y \in \RR$ and $\sigma^2$ is the variance parameter,
are commonly used to model continuous responses, such as GDP, interest rate, and consumer index~\citep{bai2003inferential, bai2012statistical, stock2016dynamic}. Moreover, 
different types of functions
$p_{ij}$'s can be used to model various types of item responses. Therefore, mixed types of data, which are common in educational measurements~\citep{rijmen2003nonlinear},
can also be analyzed by our proposed model.
\end{remark}

\vspace{-0.2in}

\begin{remark} 
   In our motivating applications, the covariate effect is item-specific. For example, in psychometrics, DIF effects often vary across items. 
   Different from the existing panel data model in \cite{chen2021nonlinear}, which treats covariate effect as a unified parameter $\bbeta$, our model allows $\bbeta_j$ to vary with $j\in[q]$. Therefore, our approach is able to capture heterogeneous covariate effects across different items, addressing the limitations of these existing models that are only applicable to measure homogeneous covariate effects and therefore not suitable for our applications.
      In addition to testing fairness, the considered model finds wide-ranging applications in the real world. For instance, in genomics, the gene expression status may depend on unmeasured confounders or latent biological factors  and also be associated with the variables of interest including medical treatment, disease status, and gender~\citep{wang2017confounder,ouyang2023,du2023simultaneous}.
The covariate-adjusted general factor model helps to investigate the effects of the variables of interest on gene expressions, controlling for the latent factors~\citep{du2023simultaneous}.  This setting is also applicable to other scenarios, such as brain imaging, where the activity of a brain region may depend on measurable spatial distances 
and latent structures due to unmodeled factors~\citep{leek2008general}.
\end{remark}

\vspace{-0.05in}



To analyze large-scale measurement data, we aim to develop inference theory for quantifying uncertainty in the estimation.
Motivated by recent work in high-dimensional factor analysis, we treat the latent factors as fixed parameters and apply a joint maximum likelihood method for estimation~\citep{bai2003inferential,fan2013large,chen2020structured}. 
Specifically, we let the collection of the item responses from $n$ independent subjects be $\Yb = (\bY_1, \dots, \bY_n)_{n\times q}^{\intercal}$ and the design matrix of observed covariates to be $\Xb = (\bX_1, \dots, \bX_n)_{n \times p}^{\intercal}$. For model parameters, the discrimination parameters for all $q$ items are denoted as 
$\bGamma = (\bgamma_1, \dots, \bgamma_q)_{q\times K}^{\intercal}$, while the intercepts and the covariate effects for all $q$ items are denoted as $\Bb = (\bbeta_1, \dots, \bbeta_q)_{q \times p}^{\intercal}$. The latent factors from all $n$ subjects are $\Ub = (\bU_1, \dots, \bU_n)_{n \times K}^{\intercal}$.
Then the joint log-likelihood function can be written as 
\begin{align}
     L( \Yb \mid \bGamma, \Ub, \Bb, \Xb) = \frac{1}{nq}\sum_{i=1}^n \sum_{j=1}^q l_{ij}(\beta_{j0} + \bgamma_j^{\intercal}\bU_i  + \bbeta_{jc}^{\intercal}\bX_i^c), \label{eq:log likelihood}
\end{align}
where the function $l_{ij} (w_{ij}) = \log p_{ij} (Y_{ij} | w_{ij}) $ is the individual log-likelihood function with $w_{ij} = \beta_{j0} + \bgamma_j^{\intercal}\bU_i  + \bbeta_{jc}^{\intercal}\bX_i^c$. 
We aim to obtain $(\hat{\bGamma}, \hat{\Ub}, \hat{\Bb}) $ from maximizing the joint likelihood function $L( \Yb \mid \bGamma, \Ub, \Bb, \Xb)$.

\begin{remark}

    The considered modeling approach has also been studied in the context of generalized linear latent variable models (GLLVMs) by~\cite{moustaki2000generalized, skrondal2004generalized}. 
Existing work often treats latent factors as random effects and adopts a marginal maximum likelihood estimation method~\citep{moustaki2003general, moustaki2004factor, chen2023dif, wallin2024dif}. Such models can be fitted through R package \texttt{gllvm}~\citep{niku2019gllvm}, stata program \texttt{GLLAMM}~\citep{rabe2004gllamm}, \texttt{Mplus}~\citep{muthen2002latent}, \texttt{LatentGold}~\citep{vermunt2005latent}, and many other softwares~\citep{korhonen2024review}. Despite the importance of this marginal maximum likelihood method, challenges arise in computational efficiency and theoretical development, especially for large-scale response data as considered in this work. 
\end{remark}

While the estimators can be computed efficiently by maximizing the joint likelihood function through an alternating minimization algorithm~\citep{collins2001generalization,chen2019joint},  challenges emerge for performing statistical inference on the model parameters. 
  One challenge concerns the model identifiability.
Without additional constraints, the covariate effects are not identifiable due to the incorporation of covariates and their potential dependence on latent factors. 
Ensuring that the model is statistically identifiable is the fundamental prerequisite for achieving model reliability and making valid inferences~\citep{allman2009,gu2018partial}.
Moreover,  we emphasize that our problem setting is fundamentally different from the existing work \cite{wang2022maximum,chen2021nonlinear}, and due to the correlation between the covariates and latent factors,  the identifiability requirements are different and existing techniques cannot be applied.
 Another challenge arises from the nonlinearity of our proposed model. 
In the existing literature, most studies focus on the statistical inference for our proposed setting in the context of linear models 
\citep{bai2012statistical, fan2013large, wang2017confounder}.
On the other hand, settings with general log-likelihood function $l_{ij} (w_{ij})$, including covariate-adjusted logistic and probit factor models, are less investigated. 
Common techniques for linear models 
are not applicable to the considered generalized model setting.



Motivated by these challenges, we propose interpretable and practical identifiability conditions in Section~\ref{sec:identifiability}. We then incorporate these conditions into the joint-likelihood-based estimation method in Section~\ref{sec:jml}. Furthermore, we introduce a novel inference framework for performing statistical inference on $\bbeta_j$, $\bgamma_j$, and $\bU_i$ in Section~\ref{sec:theory}.

\section{Method}
	\label{sec:method} 
 
 \subsection{Model Identifiability}
 \label{sec:identifiability}

{ 
Identifiability issues commonly occur in latent variable models~\citep{allman2009,bai2012statistical,xu2017}. 
By definition, the model parameters are identifiable if and only if for any response $\Yb$, there does not exist $(\bGamma, \Ub, \Bb) \neq (\tilde{\bGamma}, \tilde{\Ub}, \tilde{\Bb})$ such that $ P( \Yb \mid \bGamma, \Ub, \Bb, \Xb) = P(\Yb \mid \tilde{\bGamma}, \tilde{\Ub}, \tilde{\Bb}, \Xb)$, where $P( \Yb \mid \bGamma, \Ub, \Bb, \Xb)$ is the joint probability distribution of responses and under the model setup in~\eqref{eq:modelsetup}.
The proposed model in~\eqref{eq:modelsetup} has two major identifiability issues. 
The first issue is that the proposed model remains unchanged after certain linear transformations of both $\Bb$ and $\Ub$.
 Specifically,  for any $(\bGamma, \Ub, \Bb)$ and any transformation matrix $\Ab\in \RR^{K\times p}$, there exist $\tilde{\bGamma} = \bGamma$, $\tilde{\Ub} =  \Ub + \Xb \Ab^{\T}$, and $\tilde{\Bb} = \Bb -  \bGamma \Ab$ such that $P( \Yb \mid \bGamma, \Ub, \Bb, \Xb) = P(\Yb \mid \tilde{\bGamma}, \tilde{\Ub}, \tilde{\Bb}, \Xb)$. 
This identifiability issue causes the covariate effects together with the intercepts, represented by $\Bb$, and the latent factors, denoted by $\Ub$, to be unidentifiable. 
 The second issue is that the model is invariant after an invertible transformation of both $\Ub$ and $\bGamma$ as in the linear factor models~\citep{bai2012statistical, fan2013large}. 
 Specifically, for any $({\bGamma}, {\Ub}, {\Bb})$ and any invertible matrix $\Gb\in \RR^{K\times K}$, there exist $\bar{\bGamma} = \bGamma (\Gb^{\T})^{-1}$, $\bar{\Ub} = {\Ub} \Gb$, and $\bar{\Bb} = {\Bb}$ such that $P( \Yb \mid {\bGamma}, {\Ub}, {\Bb}, \Xb) = P(\Yb \mid \bar{\bGamma}, \bar{\Ub}, \bar{\Bb}, \Xb)$. This causes the latent factors $\Ub$ and factor loadings $\bGamma$ to be unidentifiable. }
Considering the two issues together, for any set of model parameters $(\bGamma, \Ub, \Bb)$,  there exist transformed parameters 
\begin{align}
    \tilde{\bGamma} = \bGamma (\Gb^{\T})^{-1}, \quad \tilde{\Ub} =  (\Ub + \Xb \Ab^{\T})\Gb, \quad \tilde{\Bb} = \Bb-\bGamma \Ab, \label{eq:id transformation}
\end{align} with $\Ab\in \RR^{K\times p} $ and invertible $\Gb\in \RR^{K\times K}$,
such that
$
    \tilde{\Ub} \tilde  \bGamma^{\T} + \Xb \tilde \Bb^{\T}  =  (\Ub + \Xb \Ab^{\T})\Gb \Gb^{-1} \bGamma^{\T} +   \Xb(\Bb-\bGamma \Ab)^{\T}  = \Ub \bGamma^{\T} + \Xb \Bb^{\T},
$
or equivalently,
$P (\Yb \mid \bGamma, \Ub, \Bb, \Xb) = P (\Yb \mid \tilde{\bGamma}, \tilde{\Ub}, \tilde{\Bb}, \Xb),$
which leads to nonidentifiability of the model parameters. 
In the rest of this subsection, we propose identifiability conditions to address these two identifiability issues caused by the transformation matrices  $\Ab$ and  $\Gb$.


\vspace{0.2in}

 \noindent {\bf Identifiability Conditions.} 
 Motivated by the different scientific interpretations of the intercept $\beta_{j0}$ and covariate effects $\bbeta_{jc}$, we propose different identifiability conditions for them. For instance, in educational and psychological assessment, the intercept $\beta_{j0}$ is commonly interpreted to as the difficulty parameter of a test item, while $\bbeta_{jc}$ represents the effects of observed covariates, namely DIF effects, on the response to item $j$~\citep{reckase2009, holland2012differential}. 
 Practically it is common to assume that the biased DIF items are relatively few, leading to a sparse structure of the $\bbeta_{jc}$'s, while such an assumption is not applied to the intercept $\beta_{j0}$'s.
Following the existing literature \citep{reckase2009}, we thus propose a centering condition on $\Ub$ to ensure the identifiability of the intercept $\beta_{j0}$ for all items $j \in [q]$.
On the other hand, to identify the covariate effects $\bbeta_{jc}$, we follow the commonly used assumption in the literature to assume the covariate effects $\bbeta_{jc}$ for all items $j \in [q]$ to be relatively sparse
\citep[e.g.,][]{candell1988iterative, belzak2020improving}. 
In particular, motivated by~\cite{chen2023dif}, we propose the following practical minimal $\ell_1$ condition applicable to the considered general family of models. To better present the identifiability conditions, we write $\Ab = (\ba_0, \ba_1, \dots, \ba_{p_*}) \in \RR^{K\times p}$ and define $\Ab_c = (\ba_1, \dots, \ba_{p_*}) \in \RR^{K \times p_*}$ as the part applied to the covariate effects.
\begin{condition}
\label{cond:ID1}
(i) $\sum_{i=1}^n\bU_i = \bm{0}_K$.
(ii)
$ \sum_{j=1}^q \| \bbeta_{jc} \|_1 < \sum_{j=1}^q \| \bbeta_{jc} - \Ab_{c}^{\T} \bgamma_j\|_1$ for any $\Ab_c\neq \zero$. 
\end{condition}

Condition~\ref{cond:ID1}(i) assumes that the latent abilities $\Ub$ are centered to ensure the identifiability of the intercepts $\bbeta_{j0}$'s, a common assumption in the item response theory literature~\citep{reckase2009}. 
The minimal $\ell_1$ condition,  Condition~\ref{cond:ID1}(ii), is satisfied when the majority of the items are unbiased (i.e., having zero covariate effects), which is commonly the case in educational and psychological assessment~\citep{holland2012differential, chen2023dif}. In particular, the following Proposition \ref{prop:cond1 hold} present a sufficient and necessary condition for Condition~\ref{cond:ID1}(ii) to hold.

\begin{proposition}
\label{prop:cond1 hold}
    Condition~\ref{cond:ID1}(ii) holds if and only if for any $\bv \in\RR^{K}$ with $\|\bv\|=1$, 
    \begin{equation}
        \sum_{j=1}^q\big|\bv^\T\bgamma_j\big|I{(\beta_{js}= 0)}>\sum_{j=1}^q\sign(\beta_{js})\bv^\T\bgamma_j I{(\beta_{js}\neq 0)},\quad\forall s\in[p_*]. \label{eq:prop1}
    \end{equation}
\end{proposition}

\begin{remark}
  Proposition~\ref{prop:cond1 hold} implies that Condition~\ref{cond:ID1}(ii) holds when $\{j: \beta_{js} \neq 0\}$ is separated into $\{j: \beta_{js} > 0\}$ and $\{j: \beta_{js} < 0\}$ in a balanced way.
 With diversified signs of $\beta_{js}$, Proposition~\ref{prop:cond1 hold} holds when a considerable proportion of test items have no covariate effect (i.e., $\beta_{js} = 0$). 
For example, when $\bgamma_j = m \bm{1}_{K}^{(k)}$ with $m >0$, 
Condition~\ref{cond:ID1}(ii) holds if and only if $\sum_{j=1}^q |m| \{ - I(\beta_{js}/m > 0) + I(\beta_{js}/m \leq  0)\} > 0$ and $\sum_{j=1}^q |m| \{- I(\beta_{js}/m \geq 0) + I(\beta_{js}/m <  0)\} <0$. With slightly more than $q/2$ items correspond to $\beta_{js} = 0$, Condition~\ref{cond:ID1}(ii) holds. Moreover, if $\#\{j: \beta_{js} > 0\}$ and $\#\{j: \beta_{js} < 0\}$ are comparable, then Condition~\ref{cond:ID1}(ii) holds even when less than $q/2$ items correspond to $\beta_{js} = 0$ and more than $q/2$ items correspond to $\beta_{js} \neq 0$. 
Though assuming a ``sparse'' structure, our assumption here differs from existing high-dimensional literature. 
In high-dimensional regression,
the covariate coefficient when regressing the dependent variable on high-dimensional covariates, is often assumed to be sparse, with the proportion of the non-zero covariate coefficients asymptotically approaching zero. In our setting, Condition~\ref{cond:ID1}(ii) allows for relatively dense settings where the proportion of items with non-zero covariate effects is some positive constant.
\end{remark}


 Besides the covariate effects, the factor loading $\bGamma$ and latent factor $\Ub$ hold important scientific interpretations in many applications. Their estimation and inference significantly contribute to our understanding of generalized latent factor models. 
To perform simultaneous estimation and inference on $\bGamma$ and $\Ub$, we consider the following identifiability conditions to address the second identifiability issue. 
 \begin{condition}
 \label{cond:ID2}
  (i) $\Ub^{\T} \Ub$ is diagonal with distinct and nonzero elements. (ii) $\bGamma^{\T} \bGamma$ is diagonal with distinct and nonzero elements. (iii) $n^{-1} \Ub^{\T} \Ub=q^{-1}\bGamma^{\intercal} \bGamma$.
 \end{condition}

 Condition~\ref{cond:ID2} is a set of commonly used identifiability conditions in the factor analysis literature~\citep{bai2003inferential, bai2012statistical, wang2022maximum}. 
This condition addresses the identifiability issue related to $\Gb$ and brings practical and theoretical benefits to parameter estimation methods and uncertainty quantification.
 It is worth mentioning that this condition may be replaced by other identifiability conditions in the factor analysis literature \citep{bai2012statistical, cui2025identifiability}, and the proposed estimation method and theoretical results in the subsequent sections still apply, up to certain   transformation.  

 In the following proposition, we show that the proposed Conditions~\ref{cond:ID1} and \ref{cond:ID2} ensure the uniqueness of the transformation matrices  $\Ab$ and  $\Gb$ (up to signed permutation).
 

 \begin{proposition}\label{prop_id}
     If $(\bGamma,\Ub,\Bb)$ and  $(\tilde\bGamma,\tilde\Ub, \tilde\Bb)$ in~\eqref{eq:id transformation} both satisfy Condition~\ref{cond:ID1}, we have  $\Ab = \zero_{K\times p}$ and thus $\Bb=\tilde\Bb$.
     Addtionally, if $(\bGamma,\Ub,\Bb)$ and  $(\tilde\bGamma,\tilde\Ub, \tilde\Bb)$  also satisfy Condition~\ref{cond:ID2}, then $\Gb$ must only be a signed permutation matrix and thus 
     $\bGamma=\tilde\bGamma$ and $\Ub=\tilde\Ub$ up to signed column permutation. 
 \end{proposition}
Note that the identifiability issue related to the signed column permutation of $\bGamma$ and $\Ub$ is trivial and does not affect the interpretation of the factor model since the relationships between responses and factors remain unchanged. 
 In factor analysis, it is common to assume that the loading matrix estimator has the same column signs as the true matrix \citep[e.g.,][]{bai2012statistical}, which we also follow in this work.
 
  We further show in the following proposition that, when $(\bGamma,\Ub,\Bb)$ satisfy Conditions~\ref{cond:ID1} and~\ref{cond:ID2}, for any $(\tilde\bGamma,\tilde\Ub, \tilde\Bb) $ defined in \eqref{eq:id transformation} at which Conditions~\ref{cond:ID1} and~\ref{cond:ID2} may not necessarily hold, they can be transformed back to $(\bGamma,\Ub,\Bb)$ through transformation matrices depending only on $(\tilde\bGamma,\tilde\Ub, \tilde\Bb) $. 
 \begin{proposition}
     Suppose $(\bGamma,\Ub,\Bb)$ satisfy Conditions~\ref{cond:ID1} and~\ref{cond:ID2}. For any $\tilde{\bphi} = (\tilde\bGamma,\tilde\Ub, \tilde\Bb) $ defined in \eqref{eq:id transformation}, there exist linear transformation matrices $\tilde\Ab$ and $\tilde\Gb$, depending only on $\tilde\bphi$, such that $(\tilde\bGamma(\tilde\Gb^{\T})^{-1},(\tilde\Ub +\Xb\tilde\Ab^\T)\tilde\Gb, \tilde\Bb -\tilde\bGamma\tilde\Ab) = (\bGamma, \Ub,\Bb)$, where the equivalences for $\bGamma$ and $\Ub$ are up to signed column permutation. Specifically, we have $\tilde\Ab = (\tilde\ba_0,\tilde\Ab_c)$, where $\tilde\Ab_c = \arg\min_{\Ab\in\RR^{K\times p_*}}\sum_{j=1}^q\|\tilde\bbeta_{jc}-\Ab^\T\tilde\bgamma_j\|_1$ and $\tilde\ba_0 = -n^{-1}\sum_{i=1}^n\tilde\bU_i $. We have $\tilde\Gb = (q^{-1}\tilde\bGamma^{\T} \tilde\bGamma)^{1/2}$ $ \tilde\cV \tilde\cU^{-1/4}$ where $\tilde\cU = \diag(\tilde\varrho_1, \dots, \tilde\varrho_K)$ is a diagonal matrix that contains the $K$ eigenvalues of $(nq)^{-1} (\tilde\bGamma^{\T} \tilde\bGamma )^{1/2} (\tilde\Ub + \Xb \tilde\Ab^{\T})^{\T} (\tilde\Ub + \Xb \tilde\Ab^{\T})$ $(\tilde\bGamma^{\T} \tilde\bGamma )^{1/2}$ and $\tilde\cV$ is a matrix that contains its corresponding eigenvectors. \label{prop_id_trans}
 \end{proposition}

An important implication of this proposition is that for model estimation under Conditions~\ref{cond:ID1} and~\ref{cond:ID2}, given any estimators $(\tilde\bGamma,\tilde\Ub, \tilde\Bb)$ as in \eqref{eq:id transformation}, which may not necessarily satisfy Conditions~\ref{cond:ID1} and~\ref{cond:ID2},   we can still transform $(\tilde\bGamma,\tilde\Ub, \tilde\Bb)$ to obtain the target estimators satisfying Conditions~\ref{cond:ID1} and~\ref{cond:ID2}.
This motivates the development of a transformation-based estimation method for the model parameters, whose details are presented in Section~\ref{sec:jml}.
 

 \subsection{Joint Maximum Likelihood Estimation}
	\label{sec:jml}


In this section, we introduce the likelihood-based estimation method for the covariate effect $\Bb$, the latent factors $\Ub$, and factor loadings $\bGamma$ simultaneously. For notational convenience, we write  $\bphi^* = (\bGamma^*, \Ub^*, \Bb^*)$ as the true parameters and assume they satisfy Conditions~\ref{cond:ID1} and~\ref{cond:ID2}. Inspired by Proposition~\ref{prop_id_trans}, we propose to estimate the parameters by first obtaining an unconstrained maximum likelihood estimator (MLE) and then constructing transformations in accordance with Conditions~\ref{cond:ID1} and~\ref{cond:ID2}, which would be more computationally efficient compared to directly solving the constrained MLE.
This transformation-based approach has been commonly used for factor models without covariates~\citep{bai2003inferential,lam2011estimation,chen2019joint,wang2022maximum}. Specifically, we first derive the unconstrained MLE $\hat{\bphi} = ( \hat{\bGamma}, \hat{\Ub}, \hat{\Bb})$~by
\begin{equation}
	\hat{\bphi} = \underset{\bphi \in \cB(D)}{\argmin}~-L(\Yb \mid  \bphi, \Xb), \label{eq:MLE}
\end{equation}
where the parameter space $\cB(D)$ is given as $ \cB(D)=\{\bphi:\|\bphi\|_{\max}\le D,\;\max_{i,j}|w_{ij}|\le D \}$ for some large constant $D$. This constant $D$ is empirically chosen to ensure that the estimation remains within the bounds required by the regularity conditions given in Section~\ref{sec:theory}, while also covering the true parameter values inside $\cB(D)$. To solve~\eqref{eq:MLE}, we employ an alternating minimization algorithm. 
For steps $t = 1,2,\ldots$, we iterate the following until convergence: 
\begin{align}
     \hat{\bGamma}^{(t)}, \hat{\Bb}^{(t)} &= \argmin_{\substack{ \|\bGamma\|_{\max}\le D,\|\Bb\|_{\max}\le D, \|\Bb\Xb^\T\|_{\max}\le D}} - L(\Yb \mid  \bGamma, \hat\Ub^{(t-1)},  \Bb, \Xb), \nonumber \\
    \hat{\Ub}^{(t)} &= \argmin_{\|\Ub\|_{\max}\le D} - L(\Yb \mid   \hat\bGamma^{(t)}, \Ub, \hat\Bb^{(t)}, \Xb). \nonumber
\end{align}
  Empirically, the algorithm ends when the quantity 
$\|(\hat{\bGamma}^{(t)}(\hat{\Ub}^{(t)})^\T) + \hat{\Bb}^{(t)}\Xb^\T) - $ $(\hat{\bGamma}^{(t-1)} $ $(\hat{\Ub}^{(t-1)})^\T$ $+\hat{\Bb}^{(t-1)}\Xb^\T)\|_F$ is less than some pre-specified tolerance value for convergence, which is invariant to the transformations in \eqref{eq:id transformation}. The alternating minimization scheme ensures a non-increasing sequence of the objective function $-L(\Yb|\bphi,\Xb)$ at each iteration, i.e., $-L(\Yb| \hat\bGamma^{(t)}, \hat{\Ub}^{(t)}, \hat\Bb^{(t)}, \Xb)\ge -L(\Yb| \hat\bGamma^{(t)}, \hat{\Ub}^{(t-1)}, \hat\Bb^{(t)}, \Xb)\ge -L(\Yb| \hat\bGamma^{(t-1)}, \hat{\Ub}^{(t-1)}, $ $ \hat\Bb^{(t-1)}, \Xb)$. This monotonic decrease guarantees that the algorithm converges to a local minimum of the objective function. To search the global maximum, we follow a common practice of initializing the algorithm multiple times with random starting values and selecting the solution that achieves the highest likelihood among these local optima~\citep{collins2001generalization,wang2022maximum}. The alternating minimization algorithm is computationally efficient, as the two iterative steps only involve fitting generalized linear models, which can be easily implemented.


Note that the unconstrained MLE $\hat{\bphi}$ is not unique and is identifiable up to the transformations shown in \eqref{eq:id transformation}. With the estimator $\hat{\bphi}$, following the construction of  $\tilde{\Ab}$ and $\tilde\Gb$ in Proposition~\ref{prop_id_trans}, we compute the corresponding transformation matrices $\hat\Ab$ and  $\hat\Gb$ and obtain the final estimators as follows.
 Specifically, we have 
$\hat\Ab = (\hat\ba_0,\hat\Ab_c)$, with
$\hat\Ab_c$ obtained by minimizing the $\ell_1$-norm 
 \begin{equation}
     \hat{\Ab}_c = \argmin_{\Ab_c\in\RR^{K\times p_*}} \sum_{j=1}^q \| \hat{\bbeta}_{jc} - \Ab_c^{\T}\hat{\bgamma}_j \|_1,\label{eq: estimate_A}
 \end{equation}
and $\hat{\ba}_0 = - n^{-1} \sum_{i=1}^n \hat{\bU}_i$. Here $\hat\bbeta_{jc}$ is the estimator of $\bbeta_{jc}$ given in \eqref{eq:MLE} and similarly $\hat\bgamma_j$ is the estimator of $\bgamma_j$ obtained in \eqref{eq:MLE}. The optimization problem~\eqref{eq: estimate_A} is convex and can be efficiently solved in practice using the \texttt{optim} function in R.

 Given $\hat{\Ab}$,   we then 
 construct an estimator of ${\Bb}^*$ as
\begin{equation}\label{B-est}
     \hat{\Bb}^*  = \hat{\Bb} - \hat{\bGamma} \hat{\Ab}.
\end{equation}
Following Proposition~\ref{prop_id_trans}, we take $\hat\Gb = (q^{-1}\hat\bGamma^{\T} \hat\bGamma)^{1/2}$ $ \hat\cV \hat\cU^{-1/4}$, where we compute the singular value decomposition of $(nq)^{-1} (\hat{\bGamma}^{\T} \hat{\bGamma} )^{1/2} $ $ (\hat{\Ub}+ \Xb \hat{\Ab}^{\T})^{\T} (\hat{\Ub} + \Xb \hat{\Ab}^{\T})$ $(\hat{\bGamma}^{\T} \hat{\bGamma} )^{1/2}$ and let $\hat\cU = \diag(\hat\varrho_1, \dots, \hat\varrho_K)$ be a diagonal matrix that contains the eigenvalues and $\hat\cV$ be a matrix that contains its corresponding eigenvectors. 
 Given $\hat{\Ab}$ and $\hat\Gb$, we obtain our estimators for ${\bGamma}^*$ and ${\Ub}^*$ as 
\begin{align}\label{GammaU-est}
        \hat{\bGamma}^* = \hat{\bGamma} (\hat{\Gb}^{\T})^{-1}~\mbox{ and }~
        \hat{\Ub}^* = (\hat{\Ub} + \Xb \hat{\Ab}^{\T}) \hat{\Gb}. 
\end{align}

The statistical guarantees for our proposed estimators $\hat{\bphi}^* = (\hat{\bGamma}^*, \hat{\Ub}^*, \hat{\Bb}^*)$ in \eqref{B-est} and \eqref{GammaU-est}
are given in Section~\ref{sec:theory}. 
We show that under certain regularity conditions, our proposed estimators are unique and consistently estimate the true parameters. We also provide uncertainty quantification for our estimators. Specifically, in 
Theorem~\ref{thm:asymptotic normality post-transformation beta} of Section~\ref{sec:theory}, we establish the asymptotic normality result for $\hat{\bbeta}_j^*$, which allows us to make inference on the covariate effects $\bbeta_{j}^*$. 
Moreover, as the latent factors $\bU_i^*$ and factor loadings $\bgamma_j^*$ often have important interpretations in domain sciences, we are also interested in the inference on parameters $\bU_i^*$ and $\bgamma_j^*$. In Theorem~\ref{thm:asymptotic normality post-transformation beta}, we derive the asymptotic distributions for estimators $\hat{\bU}_i^*$ and $\hat{\bgamma}_j^*$, providing inference results for parameters $\bU_i^*$ and $\bgamma_j^*$.

\section{Theoretical Results}
\label{sec:theory}

We propose a novel framework to establish the estimation consistency and asymptotic normality for the proposed joint-likelihood-based estimators $\hat{\bphi}^* = (\hat{\bGamma}^*, \hat{\Ub}^*, \hat{\Bb}^*)$ in Section~\ref{sec:method}.

To establish the theoretical results for $\hat{\bphi}^*$, we impose   several regularity assumptions. The detailed assumptions (Assumptions 1--5) and their discussions are presented in the Supplementary Material due to the space limitation,  which are briefly summarized as follows. 
In Assumption~1, we assume compact parameter spaces of $\Ub^*$ and $\bGamma^*$ and necessary regularity conditions on $\Xb$. We assume certain smoothness conditions for the log-likelihood function $l_{ij}(w_{ij})$ in Assumption~2. 
These assumptions are mild regularity conditions and commonly imposed for generalized latent factor models~\citep[e.g.,][]{bai2012statistical,wang2022maximum}. 
In Assumption~3, we impose a scaling condition that implies $p=o(n^{1/2}\wedge q)$ up to a small order term. This condition helps achieve a key theoretical property that the derivative of the log-likelihood function evaluated at the MLE equals zero with high probability, avoiding irregular boundary cases. 
Furthermore, Assumption~4 assumes the existence of the asymptotic covariance matrices $\bPhi_{jz}^*$, as the limit of $-n^{-1} \sum_{i=1}^n \EE l_{ij}^{\prime\prime}(w_{ij}^*) \bZ_i^* (\bZ_i^{*})^{\T}$ with $\bZ_i^* = ({\bU_i^*}^\T,\bX_i^\T)^\T$, and $\bPhi_{i\gamma}^*$, as the  limit of $-q^{-1} \sum_{j=1}^q \EE l_{ij}^{\prime\prime}(w_{ij}^*) \bgamma_j^{*} (\bgamma_j^{*})^\T$, which  
are used to derive the asymptotic covariance matrices for the estimators $(\hat{\bGamma}^*, \hat{\Ub}^*, \hat{\Bb}^*)$. Here $w_{ij}^* = (\bgamma^*)^\T\bU_i^* + (\bbeta_j^*)^\T\bX_i$.
Assumption~5 relates to the dependence between covariates $\Xb$ and factors $\Ub^*$.
For the details of these technical assumptions, we refer the reader to Section~A of the  Supplementary Material.

We consider the  asymptotic regime with $n, q \rightarrow \infty$ while allowing diverging $p$.
  We now present our convergence rate results in the following theorem.

\begin{theorem} 
    \label{prop:ave convergence post-transformation}
	Suppose the true parameters $\bphi^* = (\bGamma^*, \Ub^*, \Bb^*)$ satisfy Conditions~\ref{cond:ID1} and~\ref{cond:ID2}. Under Assumptions~1--5, as $n, q \rightarrow \infty$, we have
 \begin{equation}
       q^{-1} \|\hat{\Bb}^* - \Bb^* \|_F^2  = O_p\left( {\frac{p^2\log qp}{n}} + {\frac{p\log n}{q}}  \right). \label{eq:betahat star rate}
 \end{equation}
 If we further assume that  $p^{3/2}(nq)^{\epsilon+3/\xi}(p^{1/2} n^{-1/2} +q^{-1/2})=o(1)$  for a sufficiently small $\epsilon >0$ and large constant $\xi>0$ specified in Assumption~2, then as $n, q \rightarrow \infty$, we have
 \begin{align}
        n^{-1}\| \hat{\Ub}^* - \Ub^*\|_F^2 & =  O_p\left( {\frac{p\log qp}{n}} + {\frac{\log n}{q}}  \right); \label{eq:Uhat star rate} \\
          q^{-1} \|\hat{\bGamma}^* - \bGamma^* \|_F^2 &=O_p\left( {\frac{p\log qp}{n}} + {\frac{\log n}{q}}  \right).\label{eq:Gammahat star rate}
 \end{align}
\end{theorem}
\begin{remark}
 
 In Theorem~\ref{prop:ave convergence post-transformation}, we consider the double asymptotic regime with $n, q \rightarrow \infty$ for estimation consistency of joint MLE. As widely recognized in the literature,  the theoretical properties of the joint MLE differ fundamentally when $q$ is fixed versus when $q$ grows to infinity. When $q$ is fixed, the joint MLE is statistically inconsistent due to the simultaneous growth of the number of parameters and the number of observations. This phenomenon is well known as the Neyman-Scott Paradox~\citep{neyman1948consistent}.
 A simple example in the context of educational testing is that, if a student responds to only one question (i.e.,  $q = 1$), it is impossible to consistently estimate the student’s latent ability (latent factor). Theorem~\ref{prop:ave convergence post-transformation} suggests that when $p$ is fixed, the consistency is achieved when $n, q\to\infty$ with $q \gg \log n$ and $n\gg\log q$, which is a mild condition and consistent with the literature~\citep[e.g.,][]{haberman1977maximum, chen2023statistical}.

\end{remark}

\begin{remark} \label{remark6}
 One major challenge in establishing the estimation consistency for $\hat{\bphi}^*$ arises from handling the unrestricted dependence structure between $\Ub^*$ and $\Xb$ under the proposed identifiablity constraints. 
To illustrate this, if we consider the ideal case where the columns of $\Ub^*$ and $\Xb$ are orthogonal, i.e., $(\Ub^{*})^{\T} \Xb = \bm{0}_{K\times p}$, then we can achieve sharper consistency rates with less stringent assumptions. Specifically, with Assumptions~1--3 only, we can obtain the same convergence rates for ideal estimators  $\hat\Ub^*$ and $\hat\bGamma^*$ as in~\eqref{eq:Uhat star rate} and~\eqref{eq:Gammahat star rate}, respectively. Moreover, with Assumptions~1--3, the average convergence rate for the consistent estimator of $\Bb^*$ is $ O_p( n^{-1}p\log qp + q^{-1}\log n)$, which is tighter than~\eqref{eq:betahat star rate} by a factor of $p$. 
Moreover, due to the correlation between $\Ub^*$ and $\Xb$,  we emphasize that our problem setting is fundamentally different from those without covariates \cite{wang2022maximum}, where the orthogonality of the factors is assumed and plays a key role in the technical proofs. As a result, existing techniques cannot be applied to our setting under the proposed identifiablity conditions.
\end{remark}

\begin{remark}
  Theorem~\ref{prop:ave convergence post-transformation} presents the average convergence rates of $\hat{\bphi}^*$. 
Consider an oracle case with $\Ub^*$ and $\bGamma^*$ known, the estimation of $\Bb^*$ reduces to an $M$-estimation problem. For $M$-estimators under general parametric models, 
it can be shown that the optimal convergence rates in squared $\ell_2$-norm is $O_p(p/n)$ under the scaling condition $p (\log p)^3 /n \rightarrow 0$~\citep{he2000parameters}. 
In terms of our average convergence rate on $\hat{\Bb}^*$, the first term in \eqref{eq:betahat star rate}, $n^{-1} p^2\log (qp)$, approximately matches the convergence rate $O_p(p/n)$ up to a relatively small order term of $p\log (qp)$. The second term in \eqref{eq:betahat star rate}, $q^{-1}p\log n$, is mainly due to the estimation error for the latent factor ${\Ub}^*$. 
In educational applications, it is common to assume the number of subjects $n$ is much larger than the number of items $q$.
Under such a practical setting with $n \gg q$ and $p$ relatively small, the term $q^{-1}\log n$ in \eqref{eq:Uhat star rate} dominates in the derived convergence rate of $\hat{\Ub}^*$, which matches with the optimal convergence rate $O_p(q^{-1})$ for factor models without covariates~\citep{bai2012statistical,wang2022maximum}  up to a small order term. 
The additional scaling condition  $p^{3/2}(nq)^{\epsilon+3/\xi}(p^{1/2} n^{-1/2} +q^{-1/2})=o(1)$ in Theorem \ref{prop:ave convergence post-transformation} is used to handle the challenges related to the invertible matrix $\Gb$ affecting the theoretical properties of $\hat{\Ub}^*$ and $\hat{\bGamma}^*$. It is needed for establishing the consistency of $\hat{\Ub}^*$ and $\hat{\bGamma}^*$ but not for that of $\hat{\Bb}^*$. With sufficiently large $\xi$ and small $\epsilon$, this assumption is approximately $p = o(n^{1/4} \wedge q^{1/3})$ up to a small term.
\end{remark}

 With estimation consistency results established, we next derive the asymptotic normal distributions for the estimators, which enable us to perform statistical inference on the true parameters.
 To better present the inference results, we introduce the following notations. For simplicity, we
let $\bZ_i^* = \big((\bU_i^*)^\T,\bX_i^\T\big)^\T$. As assumed in Assumption~4, for any $j \in [q]$, we have $-n^{-1} \sum_{i=1}^n \EE l_{ij}^{\prime\prime}(w_{ij}^*)   \bZ_i^* (\bZ_i^{*})^{\T}  \overset{}{\rightarrow}  \bPhi_{jz}^* $ in Frobenius norm  with $\bPhi_{jz}^*$ positive definite, and for any $i \in [n]$, $-q^{-1} \sum_{j=1}^q \EE l_{ij}^{\prime\prime}(w_{ij}^*) \bgamma_j^{*} (\bgamma_j^{*})^\T \overset{}{\rightarrow} \bPhi_{i\gamma}^*$ for some positive definite matrix $\bPhi_{i\gamma}^*$.
Next, we define the transformation matrix $\bar\Ab^{\ddagger} = \bSigma_{ux}^*(\bSigma_{x})^{-1}$ where $\bSigma_{ux}^* = \lim_{n\rightarrow \infty}$ $ n^{-1} \sum_{i=1}^n $ $\bU_i^* \bX_i^{\T}$ and $\bSigma_{x} = \lim_{n\rightarrow \infty }n^{-1} \sum_{i=1}^n\bX_{i} \bX_{i}^{\T}$. The transformation matrix $\bar\Gb^{\ddagger}$ is defined as the limit of $\Gb^{\ddagger} = (q^{-1}(\bGamma^{*})^{\T} \bGamma^*)^{1/2}$ $ \cV^* (\cU^{*})^{-1/4}$, where $\cU^* = \diag(\varrho_1^*, \dots, \varrho_K^*)$ with diagonal elements being the $K$ eigenvalues of 
$(nq)^{-1} ((\bGamma^{*})^{\T} \bGamma^* )^{1/2}   (\Ub^{*})^{\T}  (\Ib_{n} - \Pb_x) \Ub^*  ((\bGamma^{*})^{\T} \bGamma^* )^{1/2} $
with $\Pb_x = \Xb (\Xb^{\T} \Xb)^{-1}\Xb^{\T}$ and $\cV^*$ denotes the matrix containing corresponding eigenvectors.

 \begin{theorem}[Asymptotic Normality]
\label{thm:asymptotic normality post-transformation beta} 
 Suppose the true parameters $\bphi^* = (\bGamma^*, \Ub^*, \Bb^*)$ satisfy  Conditions~\ref{cond:ID1} and~\ref{cond:ID2}. Under Assumptions~1--5, 
we have the asymptotic distributions  
as follows. 
If $n, q \rightarrow \infty$ and $p^{3/2}\sqrt{n}(nq)^{3/\xi}(n^{-1}{p\log qp}+q^{-1}{\log n})\to 0$, for any $j \in [q]$ and $\ba \in \RR^p$ with $\|\ba\|_2 = 1$, 
  \begin{equation}        
                  \sqrt{n} \ba^{\T} (\bSigma_{\beta,j}^*)^{-1/2}(\hat\bbeta_j^*-\bbeta_j^*)  \overset{d}{\to} \cN(0,1), \label{eq:asymp beta}
        \end{equation} 
      where $\bSigma_{\beta, j}^* = \bar\Tb_\beta(\bPhi_{jz}^*)^{-1}\bar\Tb^\T_\beta$ with $\bar\Tb_\beta = \big(-(\bar\Ab^{\ddagger})^{\T}\bar\Gb^{\ddagger}(\bar\Gb^{\ddagger})^{-\T},\; \Ib_p + (\bar\Ab^{\ddagger})^{\T}\bar\Gb^{\ddagger}\bar\Ab^{\ddagger}\big) $
      and for any $j\in[q]$,
        \begin{equation}
       \sqrt{n}(\bSigma_{\gamma, j}^*)^{-1/2}(\hat{\bgamma}_j^*-\bgamma_j^*)    \overset{d}{\to} \cN(\zero, \Ib_K),   \label{eq:asymp gamma}
\end{equation}
where $\bSigma_{\gamma, j}^* = \bar\Tb_\gamma(\bPhi_{jz}^*)^{-1}\bar\Tb^\T_\gamma$ with $\bar\Tb_\gamma = \big (-\bar\Gb^{\ddagger}(\bar\Gb^{\ddagger})^{-\T},\; \bar\Gb^{\ddagger}\bar\Ab^{\ddagger}\big) $. 
Furthermore, for any $i \in [n]$, if $q= O(n)$, $n, q \rightarrow \infty$, and $p^{3/2}\sqrt{q}(nq)^{3/\xi}(n^{-1}{p\log qp}+q^{-1}{\log n})\to 0$, we have
 \begin{equation}
      \sqrt{q}(\bSigma_{u, i}^*)^{-1/2}(\hat{\bU}_i^*-\bU_i^*)  \overset{d}{\to} \cN(\zero, \Ib_K ), \label{eq:asymp U}
 \end{equation}
  where $\bSigma_{u,i}^* = (\bPhi_{i\gamma}^*)^{-1}$. 
    \end{theorem}

    The asymptotic covariance matrices in Theorem~\ref{thm:asymptotic normality post-transformation beta}  can be consistently estimated. Due to space limitations, we defer the construction of the estimators $\hat\bSigma_{\beta, j}^*$,  $\hat\bSigma_{\gamma, j}^*$, and $\hat\bSigma_{u, i}^*$ to the Section C.4 of Supplementary Material. 

\begin{corollary}
 
\label{prop:consistent estimator cov}
Suppose the true parameters $\bphi^* = (\bGamma^*, \Ub^*, \Bb^*)$ satisfy  Conditions~\ref{cond:ID1} and~\ref{cond:ID2}.  Under Assumptions~1--5, the asymptotic normality results
~\eqref{eq:asymp beta}--\eqref{eq:asymp U} hold when $\bSigma_{\beta, j}^*$, $\bSigma_{\gamma, j}^*$, and $\bSigma_{u, i}^*$ are substituted with their estimators
$\hat\bSigma_{\beta, j}^*$ in (A8), $\hat\bSigma_{\gamma, j}^*$ in~(A9), and $\hat\bSigma_{u, i}^*$ in~(A10) of the Supplementary Material. 
   
\end{corollary}

Theorem~\ref{thm:asymptotic normality post-transformation beta} provides the asymptotic distributions for all estimators $\hat{\bbeta}_j^*$'s, $\hat{\bgamma}_j^*$'s, and $\hat{\bU}_i^*$'s. 
In particular, with the asymptotic distributions for the covariate effects and the estimators for asymptotic covariance matrices, we can perform hypothesis testing on any sub-vector of $\bbeta_{j}^*$, 
such as testing on a single entry or the entire vector. 
These testing problems find wide practical applications in educational assessments and psychological measurements, where the practitioners are often interested in investigating whether the test items are biased for particular sets of covariates or even across all covariates. 
Specifically, as introduced in Section~\ref{sec:model setup}, we aim to test the covariate effect from the $s$th covariate to the $j$th response, i.e., perform hypothesis testing on single entry $\beta_{js}^*$ for $j \in [q]$ and $s \in [p_*]$. We reject the null hypothesis $\beta_{js}^* = 0$ at significance level $\alpha$ if $| \sqrt{n} (\hat{\sigma}_{\beta,js}^*)^{-1} \hat{\beta}_{js}^* | >\Phi^{-1} ( 1- \alpha/2)$, where $(\hat{\sigma}_{\beta,js}^*)^2$ is the $(s+1)$-th diagonal entry in $\hat{\bSigma}_{\beta, j}^*$. 
Empirically, with these inference results, we conduct simulation studies in Section~\ref{sec:simulation} and real data analysis 
in Section~\ref{sec:data application}.


For the asymptotic normality of $\hat{\bbeta}_j^*$, the condition $ p^{3/2}\sqrt{n}(nq)^{3/\xi}(n^{-1}p\log qp+q^{-1}\log n )$ $\to 0$ together with Assumption~3 gives $ p = o\{n^{1/5} \wedge ({q^2}/n)^{1/3}\}$ up to a small order term, and further implies $n \ll q^2$, which is consistent with established conditions in the existing factor analysis literature~\citep{bai2012statistical, wang2022maximum}.
For the asymptotic normality of $\hat{\bU}_i^*$, the additional condition that $q = O(n)$ is a reasonable assumption in educational applications where the number of items $q$ is much smaller than the number of subjects $n$.
In this case, the scaling conditions imply $p = o\{ q^{1/3} \wedge (n^2/q)^{1/5}\}$ up to a small order term. 
Similarly for the asymptotic normality of $\hat{\bgamma}_j^*$, 
the proposed conditions give $p = o\{n^{1/5} \wedge (q^2/n)^{1/3}\}$ up to a small order term. 

\begin{remark}\label{remark-6+}
   Similar to the discussion in Remark \ref{remark6}, the challenges arising from the unrestricted dependence between $\Ub^*$ and $\Xb$ also affect the derivation of the asymptotic distributions for the proposed estimators. If we consider the ideal case with $(\Ub^{*})^{\T} \Xb = \bm{0}_{K\times p}$,  we can establish the asymptotic normality for all individual estimators under Assumptions~1--4 only and weaker scaling conditions. 
 Specifically, when $(\Ub^{*})^{\T} \Xb = \bm{0}_{K\times p}$, 
 the scaling condition becomes $p\sqrt{n}(nq)^{3/\xi}(n^{-1}p\log qp+q^{-1}\log n )\to 0$ for deriving asymptotic normality of $\hat\bbeta_j^*$ and $\hat\bgamma_j^*$, which is milder than that for~\eqref{eq:asymp beta} and~\eqref{eq:asymp gamma}.
\end{remark}

\begin{remark} As illustrated in Remarks \ref{remark6} and \ref{remark-6+},
compared with existing work on generalized latent factor models without covariates \citep{wang2022maximum},
incorporating covariates into generalized latent factor models substantially expands the parameter space and requires additional constraints to handle the dependence between $\Ub^*$ and $\Xb$ and the indeterminacy issue from the transformation matrix $\Ab$. As a result, the theoretical analysis of such constrained MLE is more complex and challenging compared to that of existing generalized latent factor models. For instance, a different and more sophisticated analytical framework is developed to establish the average consistency of estimators, address the local convexity, and derive the bounds for the Hessian matrix, which is non-trivial as the Hessian matrix significantly increases in scale with a diverging number of covariates that are correlated with the latent factors. In addition to the theoretical insights provided by the asymptotic distribution results, these findings also hold significant practical value. Beyond assessing test fairness using the inferential results for $\hat\bbeta_{js}$, our framework enables further downstream analyses involving the latent factors under DIF settings. These results lay important groundwork for future methodological advancements and applied research.

\end{remark}

\section{Simulation Study}
\label{sec:simulation}
In this section, we study the finite-sample performance of the proposed joint-likelihood-based estimator. 
We focus on the logistic latent factor model in~\eqref{eq:modelsetup} with $p_{ij} (y \mid w_{ij}) = \exp(w_{ij}y)/\{1+\exp(w_{ij})\}$, where $w_{ij} = (\bgamma_j^*)^{\T} \bU_i^* + (\bbeta_j^*)^{\T} \bX_i$.  
The logistic latent factor model is commonly used in the context of educational assessment and is also referred to as the item response theory model \citep{mellenbergh1994generalized,hambleton2013item}. 
We apply the proposed method to estimate $\Bb^*$ and 
perform statistical inference on testing the null hypothesis $\beta_{js}^* = 0$.  

We start with presenting the data generating process. 
We set the number of subjects $n \in \{300, 500,  1000, 1500, 2000\}$, the number of items $q \in \{100, 300, 500\}$, the covariate dimension $p_* \in \{5, 10, 30\}$, and the factor dimension $K = 2$, respectively. 
We jointly generate $\bX_{i}^c$ and $\bU_{i}^*$ from $ \cN(\bm{0}_{K+p^*}, \bSigma)$ where $\bSigma\in\RR^{(K+p^*)\times (K+p^*)}$ with the $(i,j)$-th entry being $\bSigma_{ij} = \tau^{|i - j|}$ for $\tau \in \{0, 0.2, 0.5, 0.7\}$. 
In addition, we set the loading matrix $\bGamma_{[,k]}^* = \bm{1}_k^{(K)} \otimes \bv_k$, where  $\otimes $ is the Kronecker product and $\bv_k$ is a $(q/K)$-dimensional vector with each entry generated independently and identically from Unif$[0.5, 1.5]$.
For the covariate effects $\Bb^*$, we set the intercept terms to equal $\beta_{j0}^* = 0$. 
For the remaining entries in $\Bb^*$, we consider the following two settings: (1) sparse setting: $\beta_{js}^* = \rho$ for $s = 1, \dots, p_*$ and $j = 5s-4, \dots, 5s$ and other $\beta_{js}^*$ are set to zero; (2) dense setting: $\beta_{js}^* = \rho$ for $s = 1, \dots, p_*$ and $j =  R_s q/5 + 1, \dots, (R_s + 1) q/5$ with  $R_s = s - 5 \lfloor{s/5}\rfloor$, and other $\beta_{js}^*$ are set to zero. 
Here, the signal strength is set as $\rho \in \{0.3, 0.5\}$. Intuitively, in the sparse setting, we set 5 items to be biased for each covariate whereas in the dense setting, 20$\%$ of items are biased items for each covariate.

\begin{figure}[htbp]
\centering    
   \subfigure{
        \includegraphics[width=1.9in]{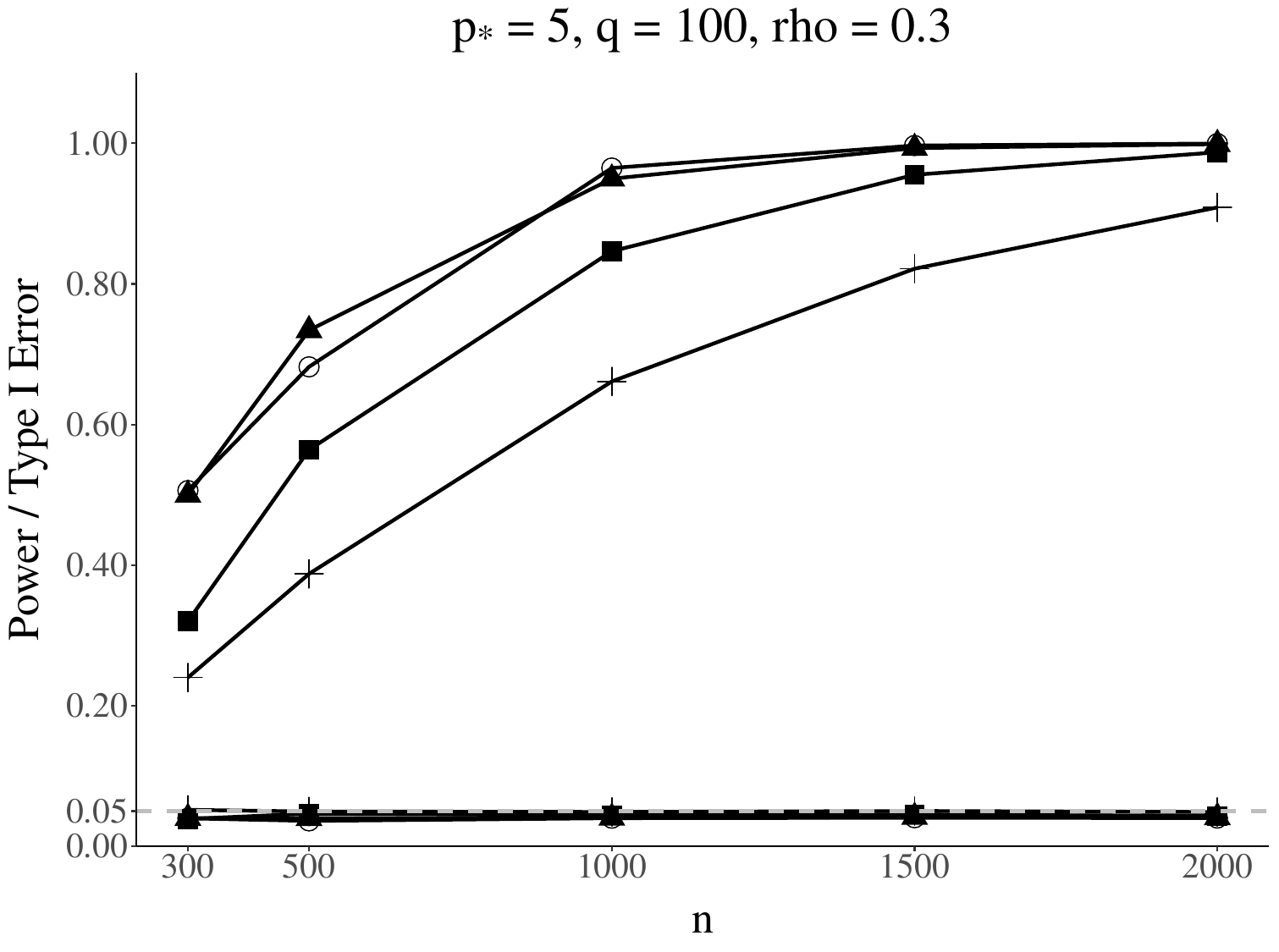}}
         \subfigure{
        \includegraphics[width=1.9in]{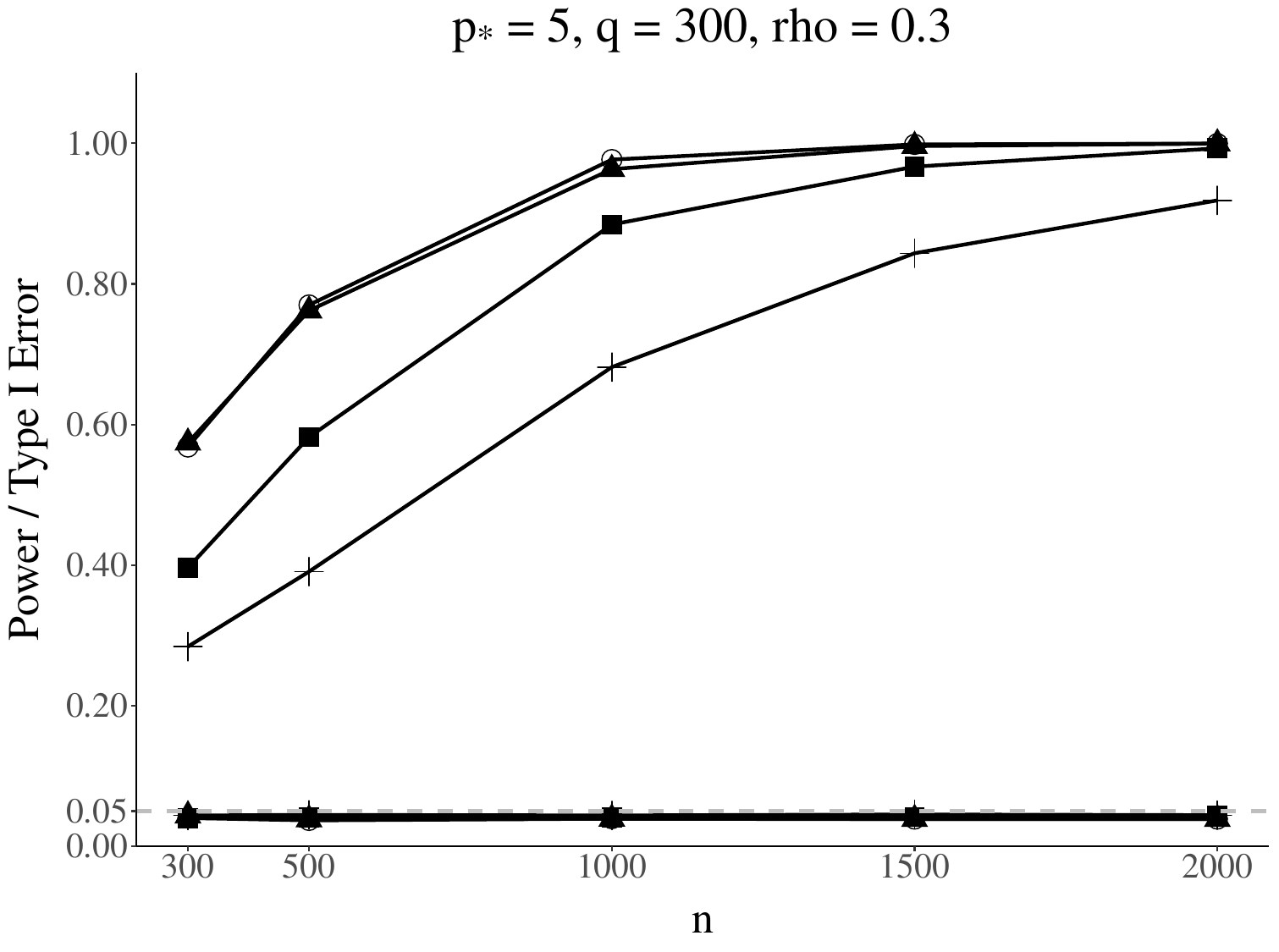}}
         \subfigure{
        \includegraphics[width=1.9in]{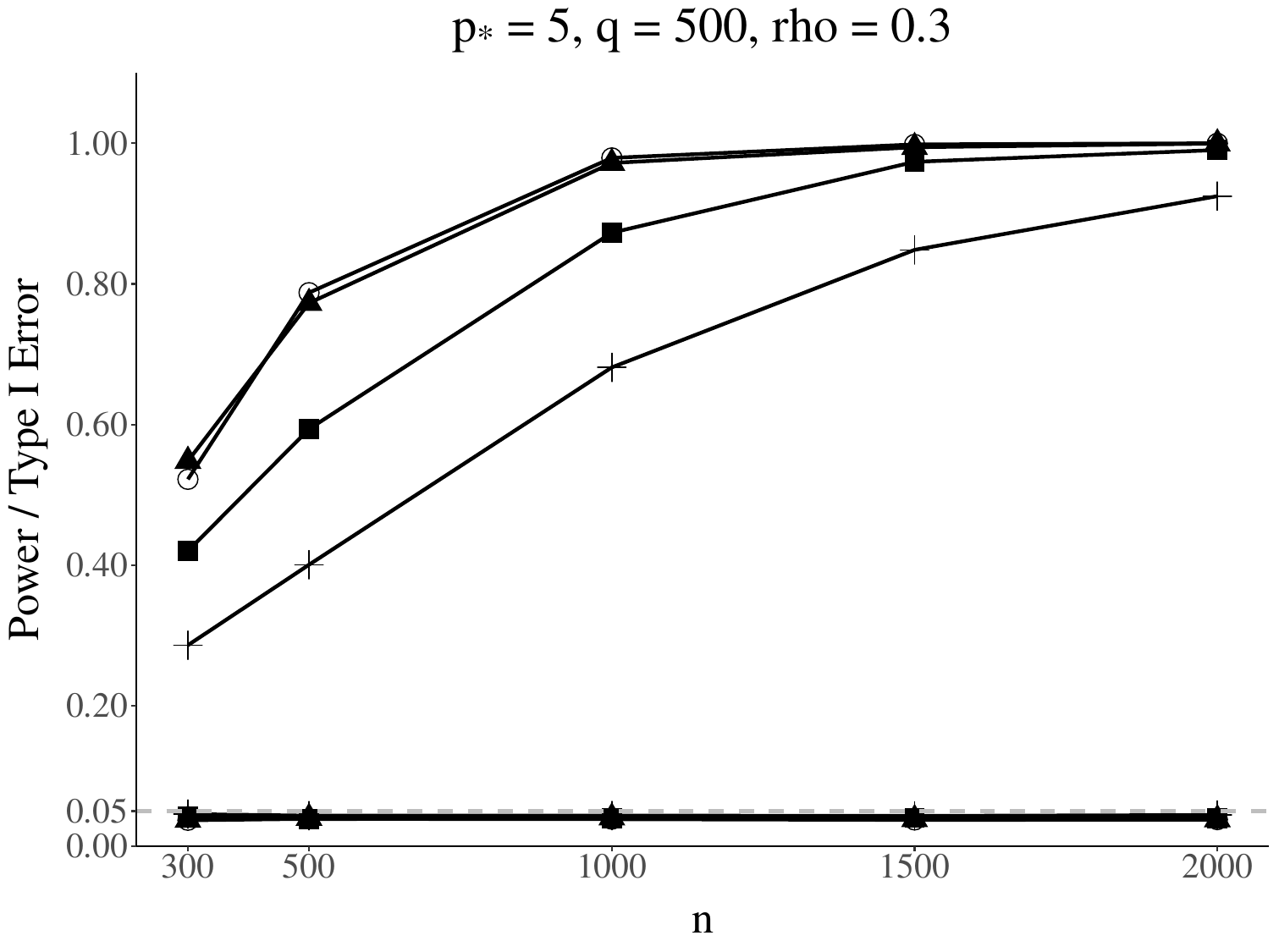}}
        \\
        \subfigure{
        \includegraphics[width=1.9in]{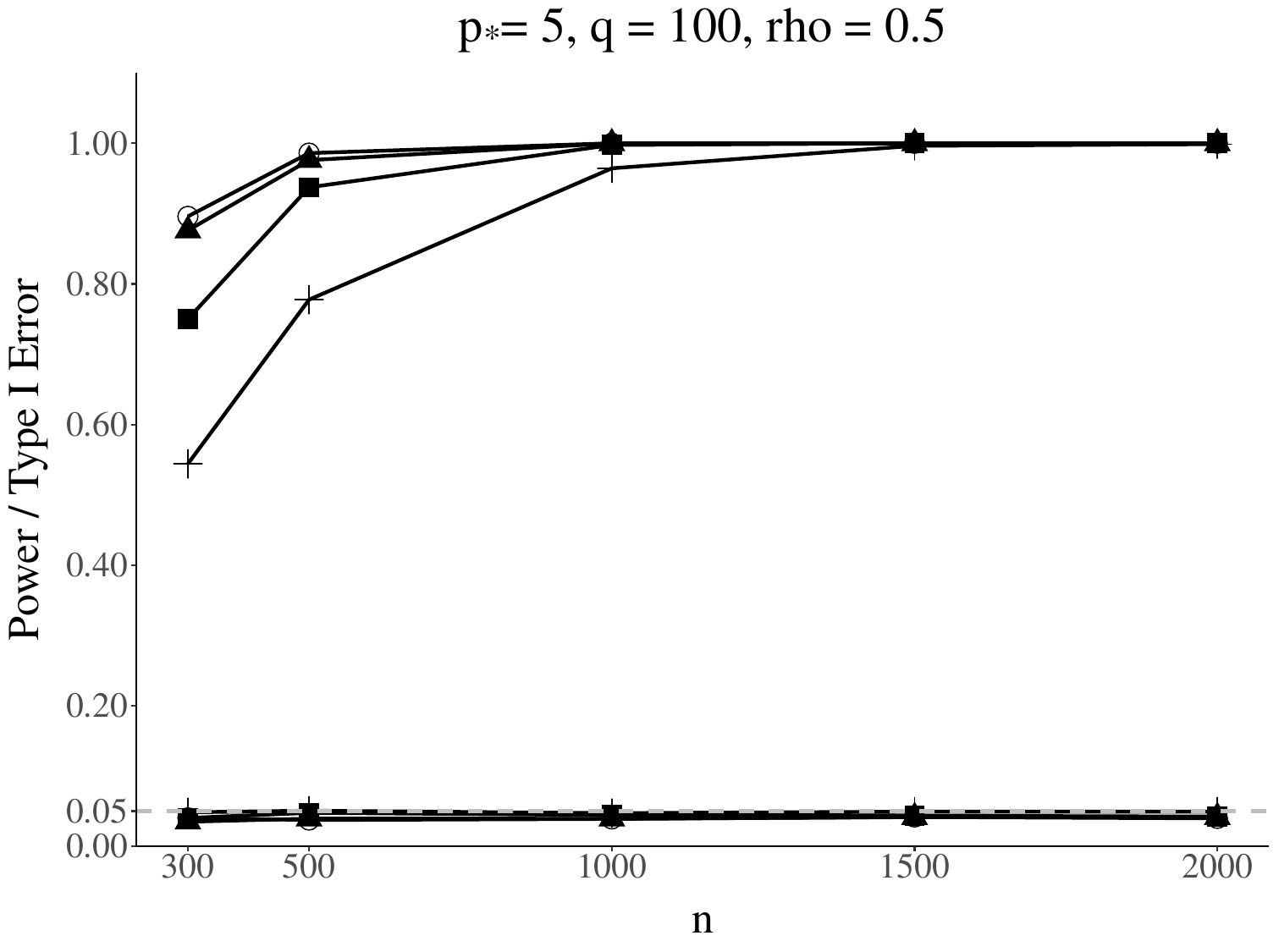}}%
          \subfigure{
        \includegraphics[width=1.9in]{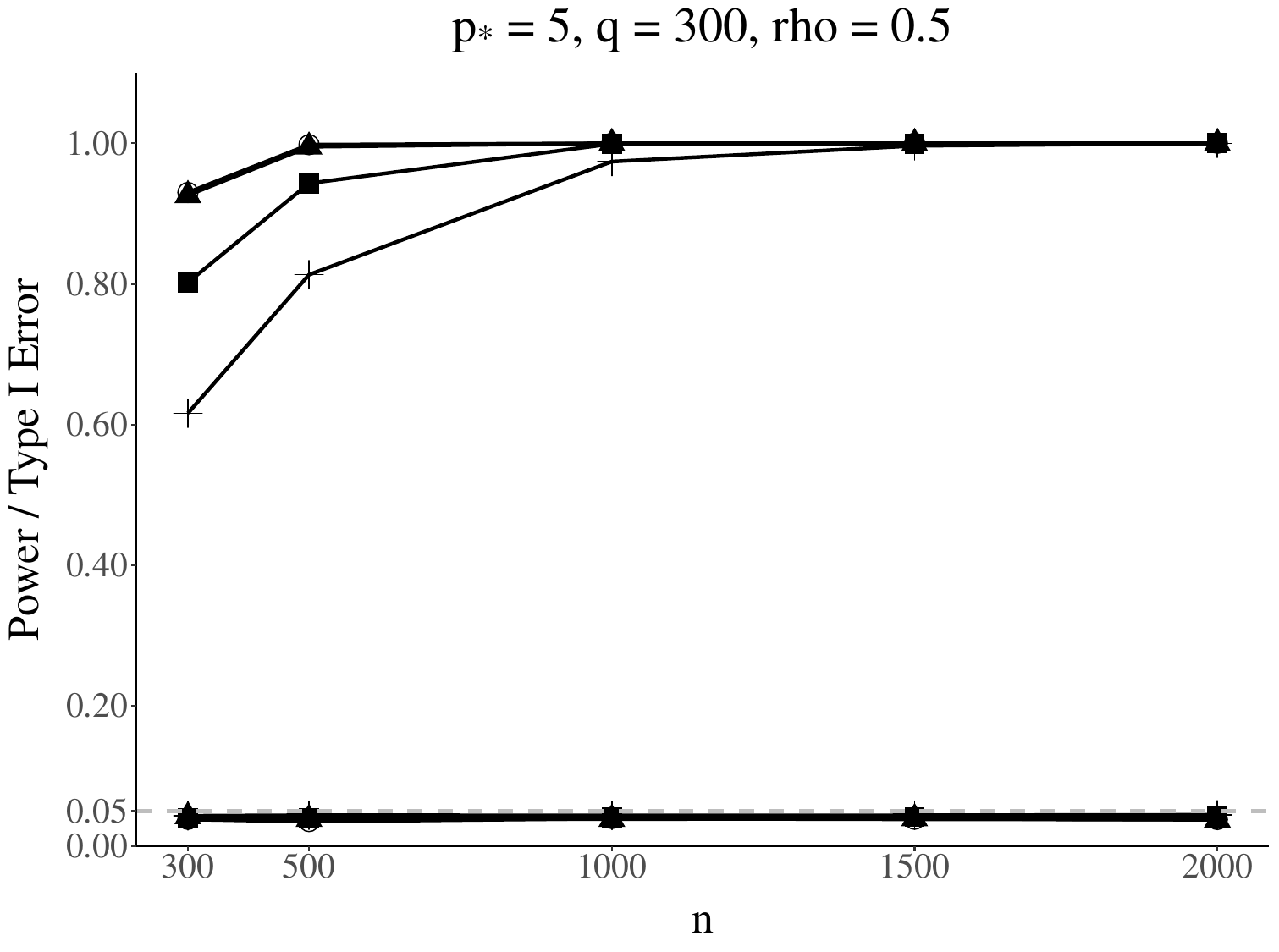}}%
          \subfigure{
        \includegraphics[width=1.9in]{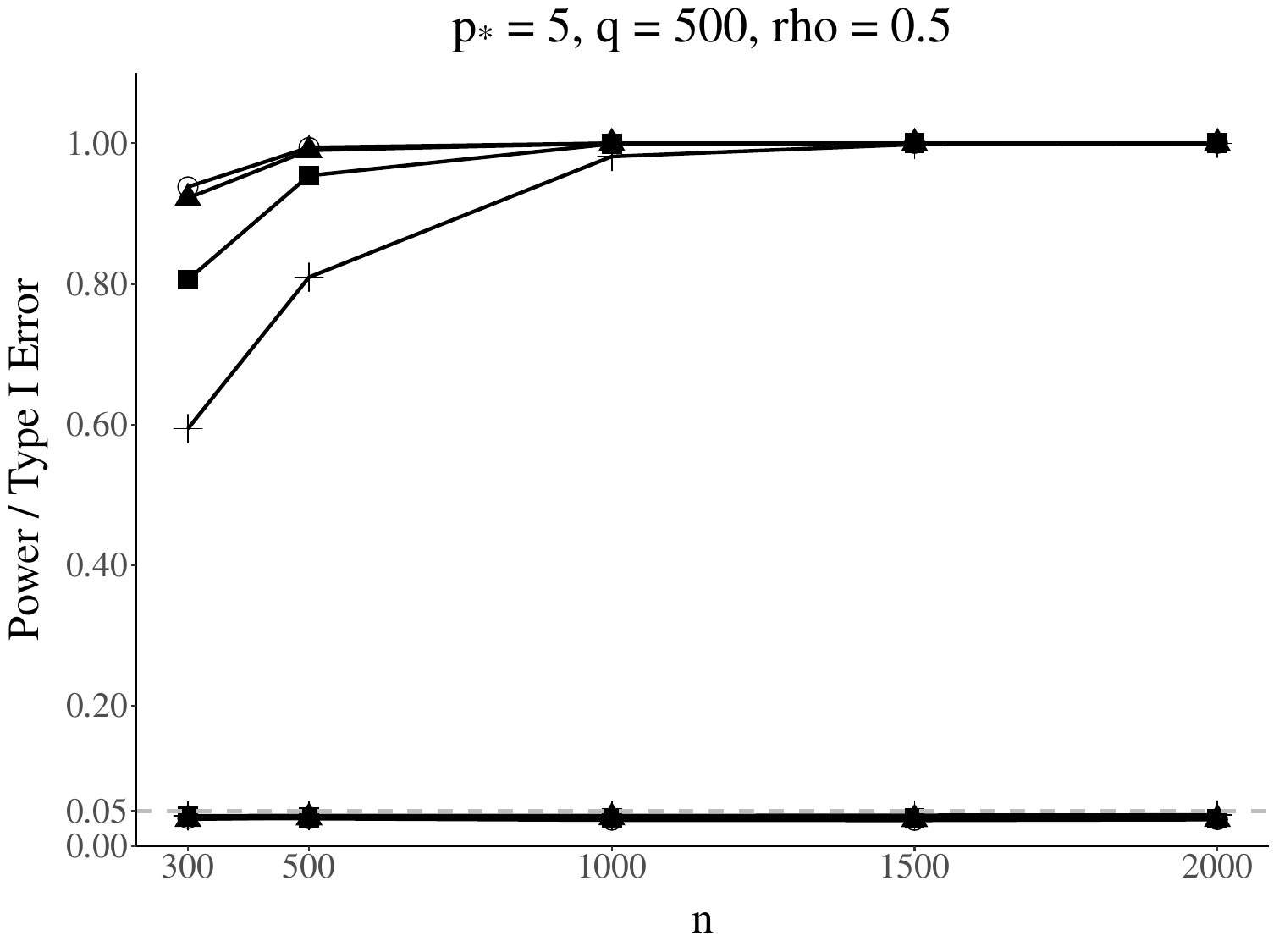}}
   
    \caption{Powers and type I errors under sparse setting at $p_*=5$. Circles (\protect\includegraphics[height=0.8em]{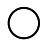}) denote correlation parameter $\tau = 0$. Triangles (\protect\includegraphics[height=0.8em]{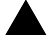})  represent the case $\tau = 0.2$. Squares (\protect\includegraphics[height=0.8em]{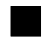}) indicate $\tau = 0.5$. Crosses (\protect\includegraphics[height=1em]{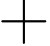}) represent the $\tau = 0.7$.}
    \label{fig:non-anchor-correlated p5}
\end{figure}

\begin{figure}[htbp]
\centering    
   \subfigure{
        \includegraphics[width=1.9in]{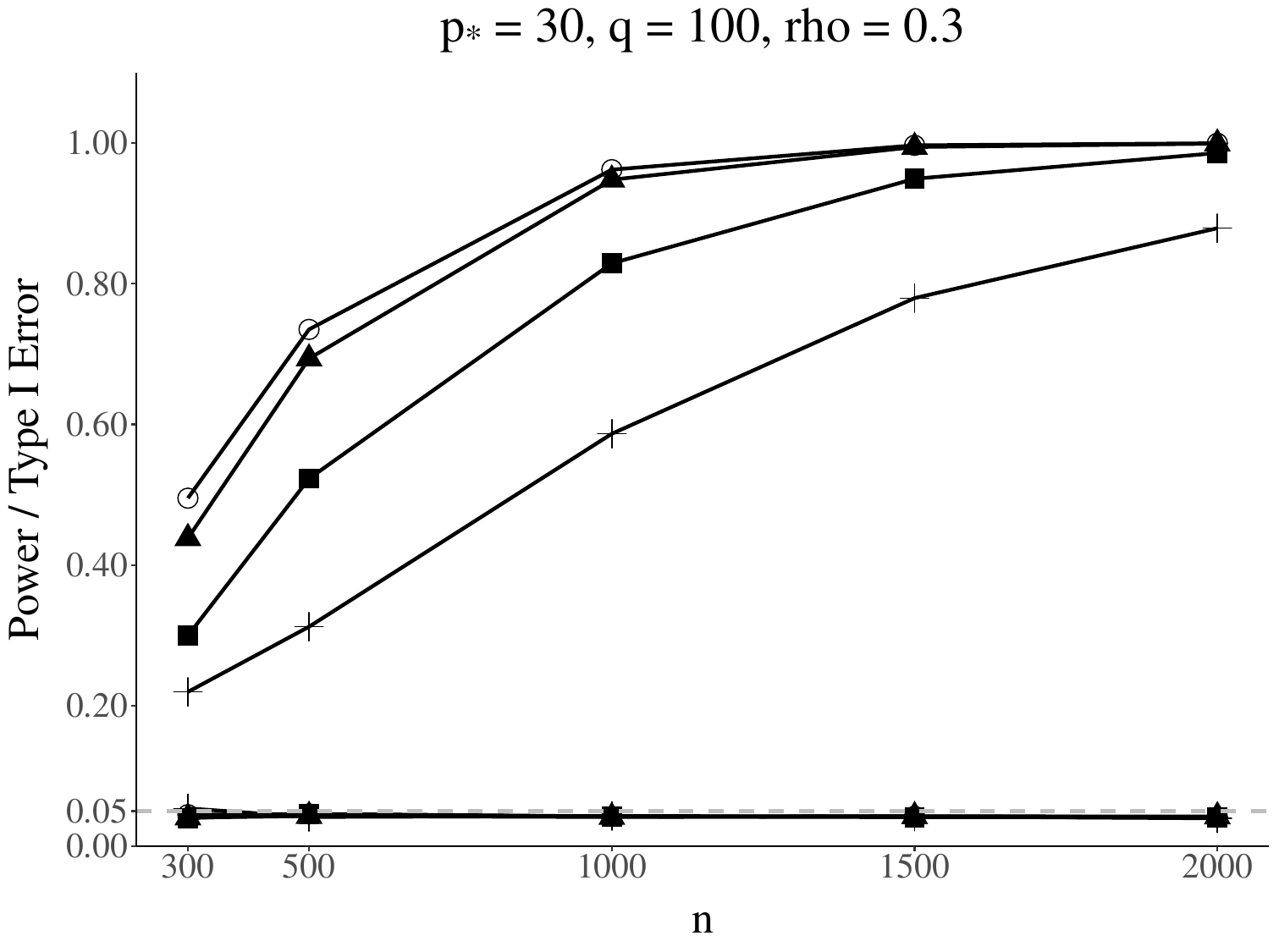}}
        \subfigure{
        \includegraphics[width=1.9in]{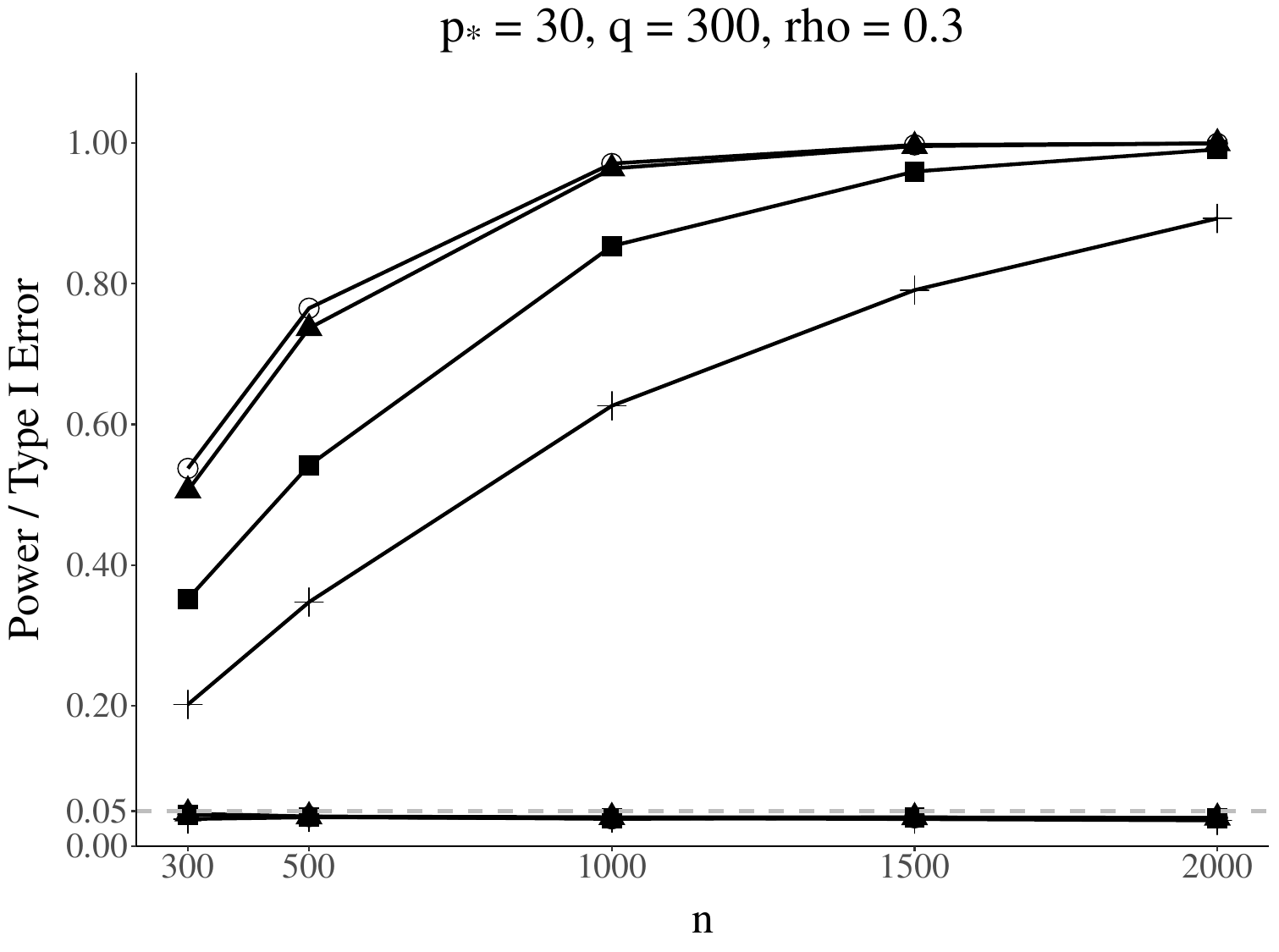}}
        \subfigure{
        \includegraphics[width=1.9in]{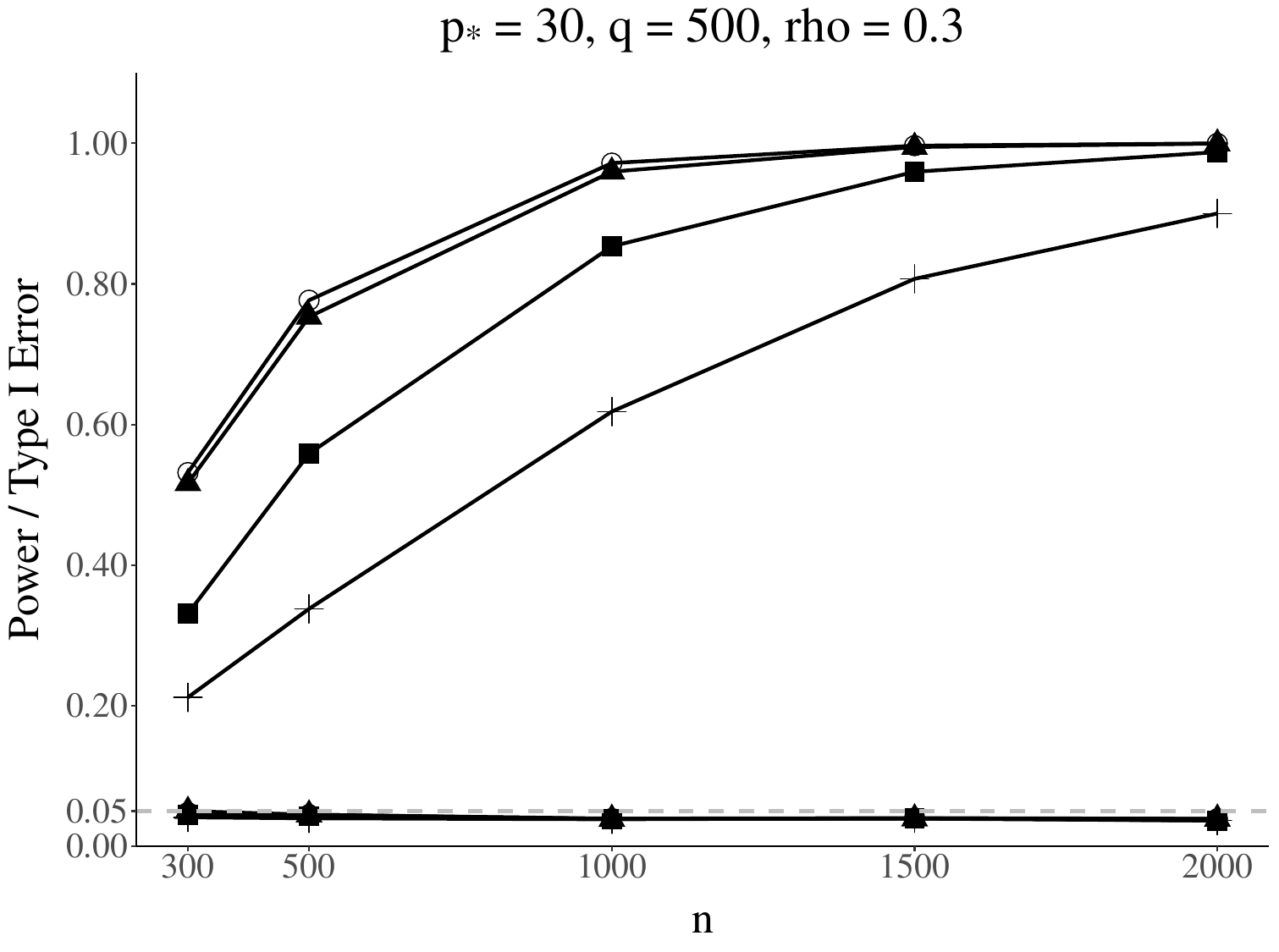}}
        \\
          \subfigure{
        \includegraphics[width=1.9in]{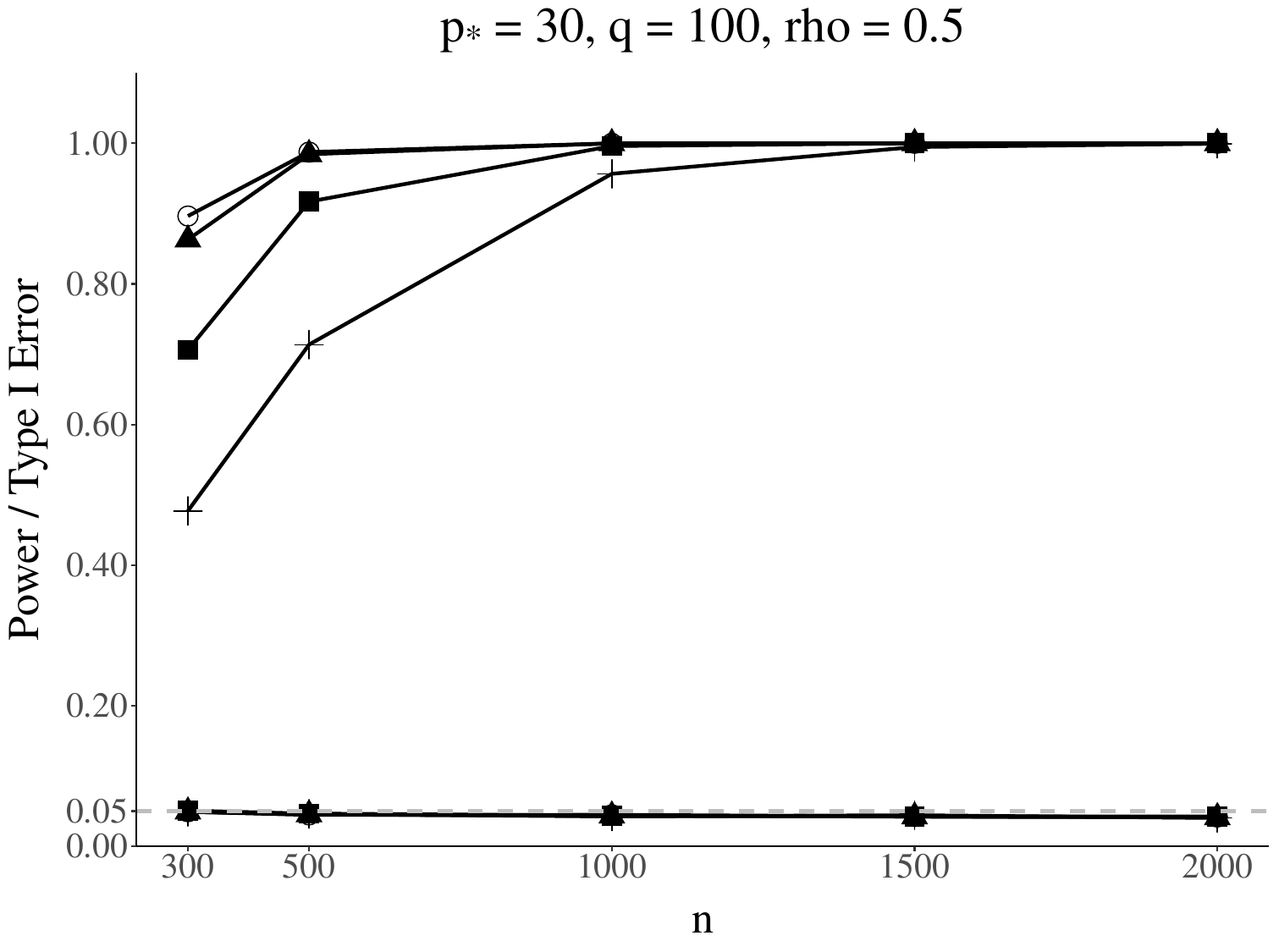}}%
          \subfigure{
        \includegraphics[width=1.9in]{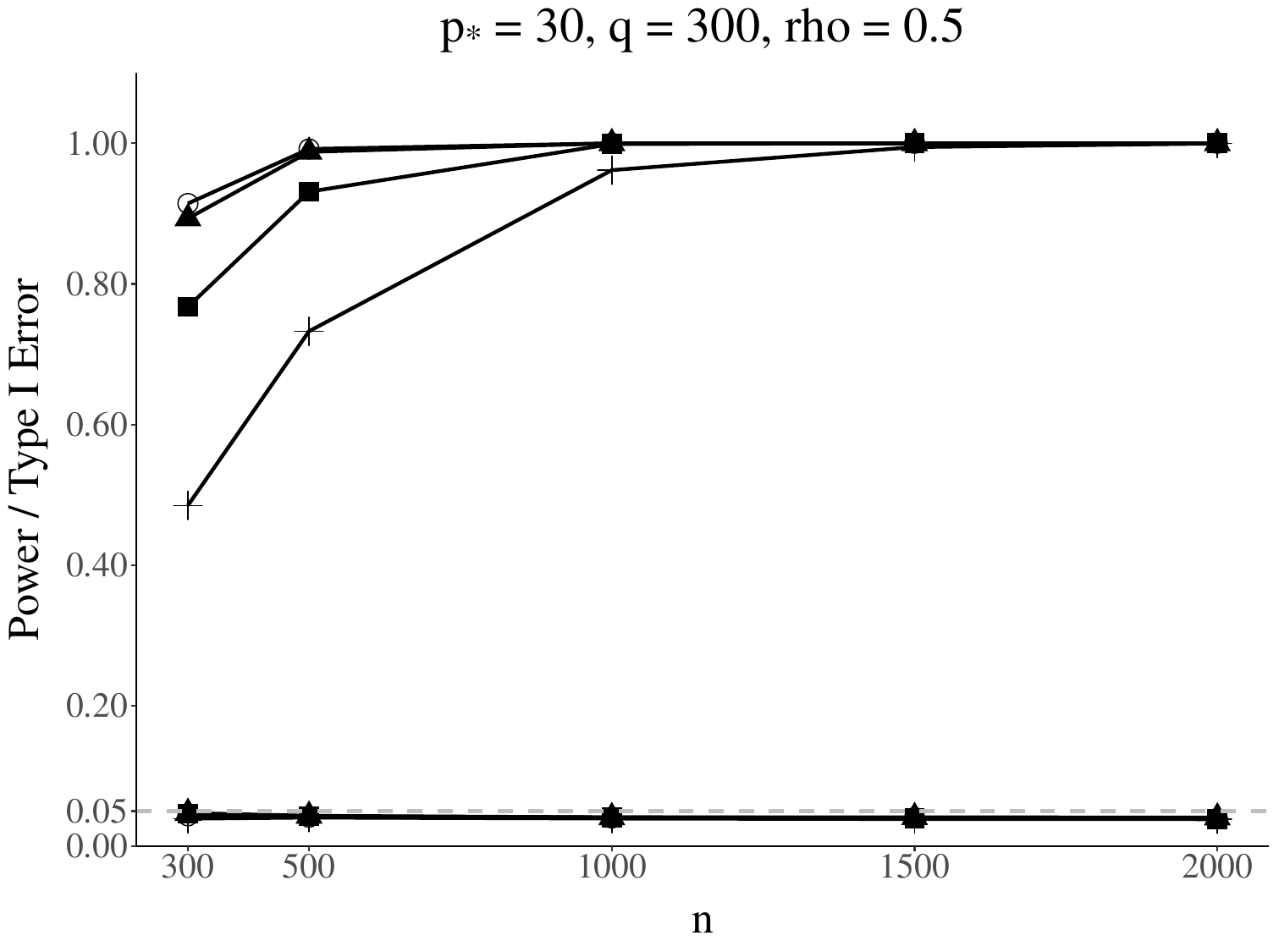}}%
          \subfigure{
        \includegraphics[width=1.9in]{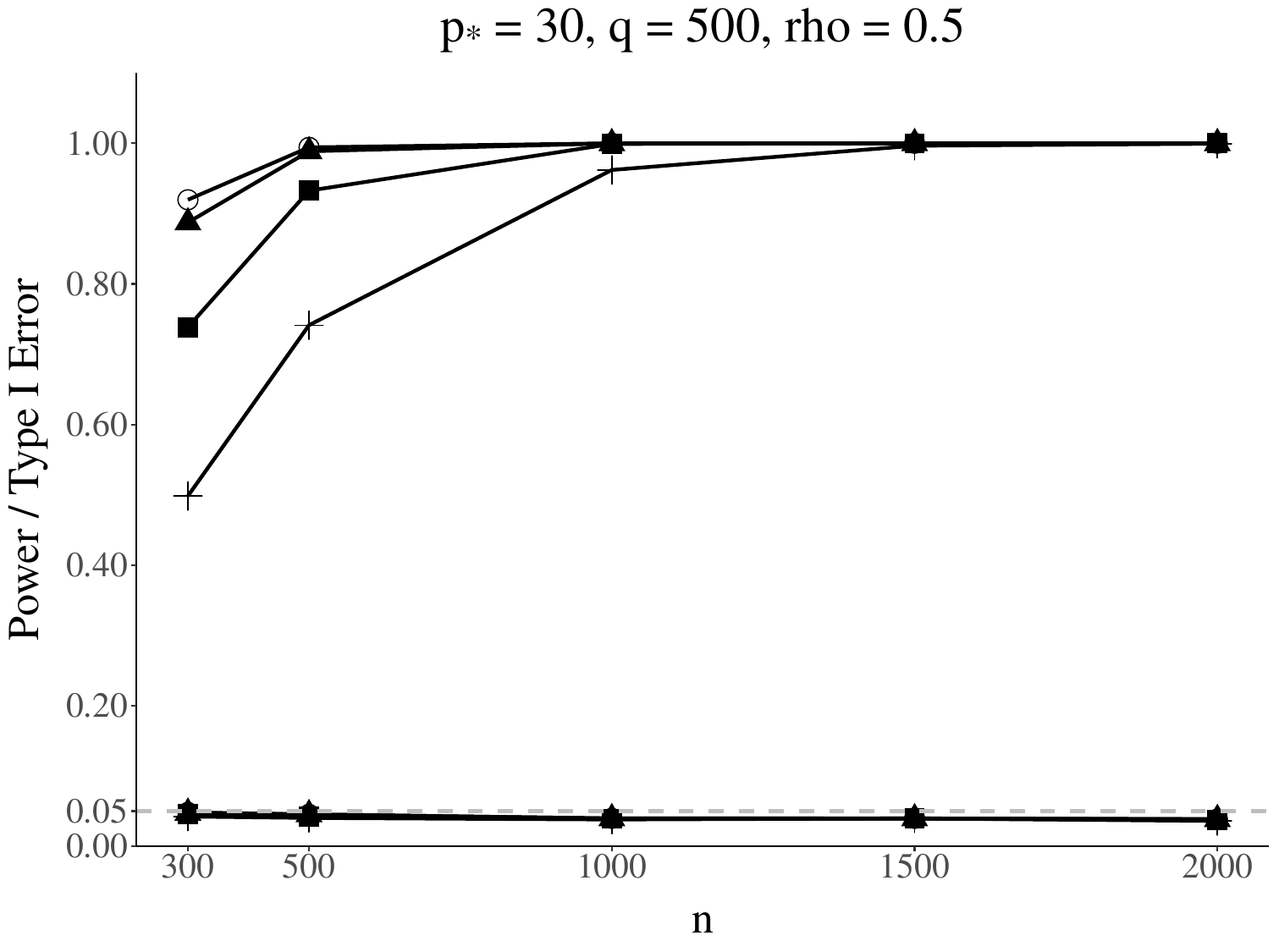}}
    \caption{Powers and type I errors under sparse setting at $p_*=30$. Circles (\protect\includegraphics[height=0.8em]{legend/new_rho0.png}) denote correlation parameter $\tau = 0$. Triangles (\protect\includegraphics[height=0.8em]{legend/rho0.2.png})  represent the case $\tau = 0.2$. Squares (\protect\includegraphics[height=0.8em]{legend/rho0.5.png}) indicate $\tau = 0.5$. Crosses (\protect\includegraphics[height=1em]{legend/rho0.7.png}) represent the $\tau = 0.7$.}
    \label{fig:non-anchor-correlated p30}
\end{figure}

\begin{figure}[htbp]
\centering    
   \subfigure{
        \includegraphics[width=1.9in]{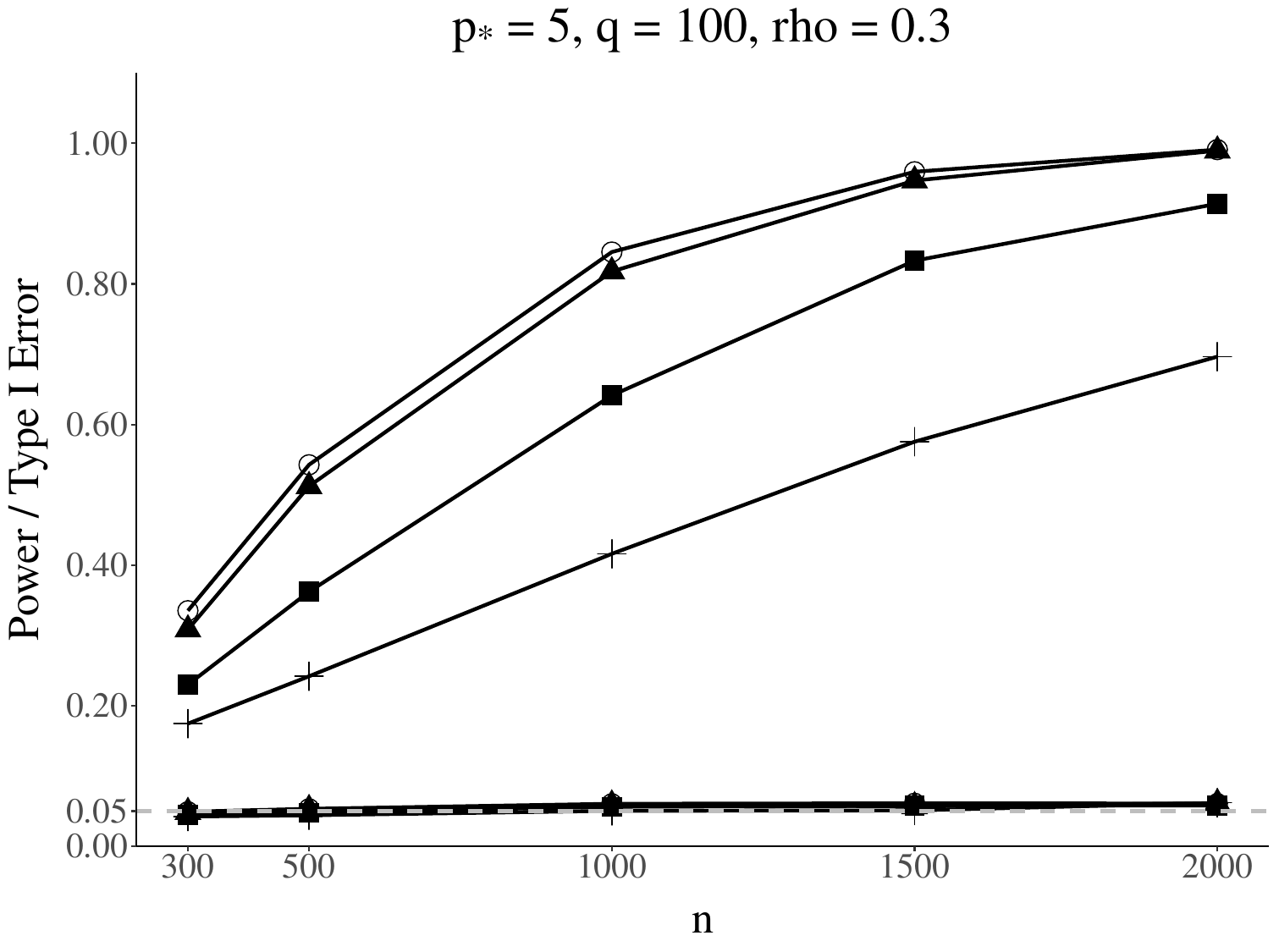}}
         \subfigure{
        \includegraphics[width=1.9in]{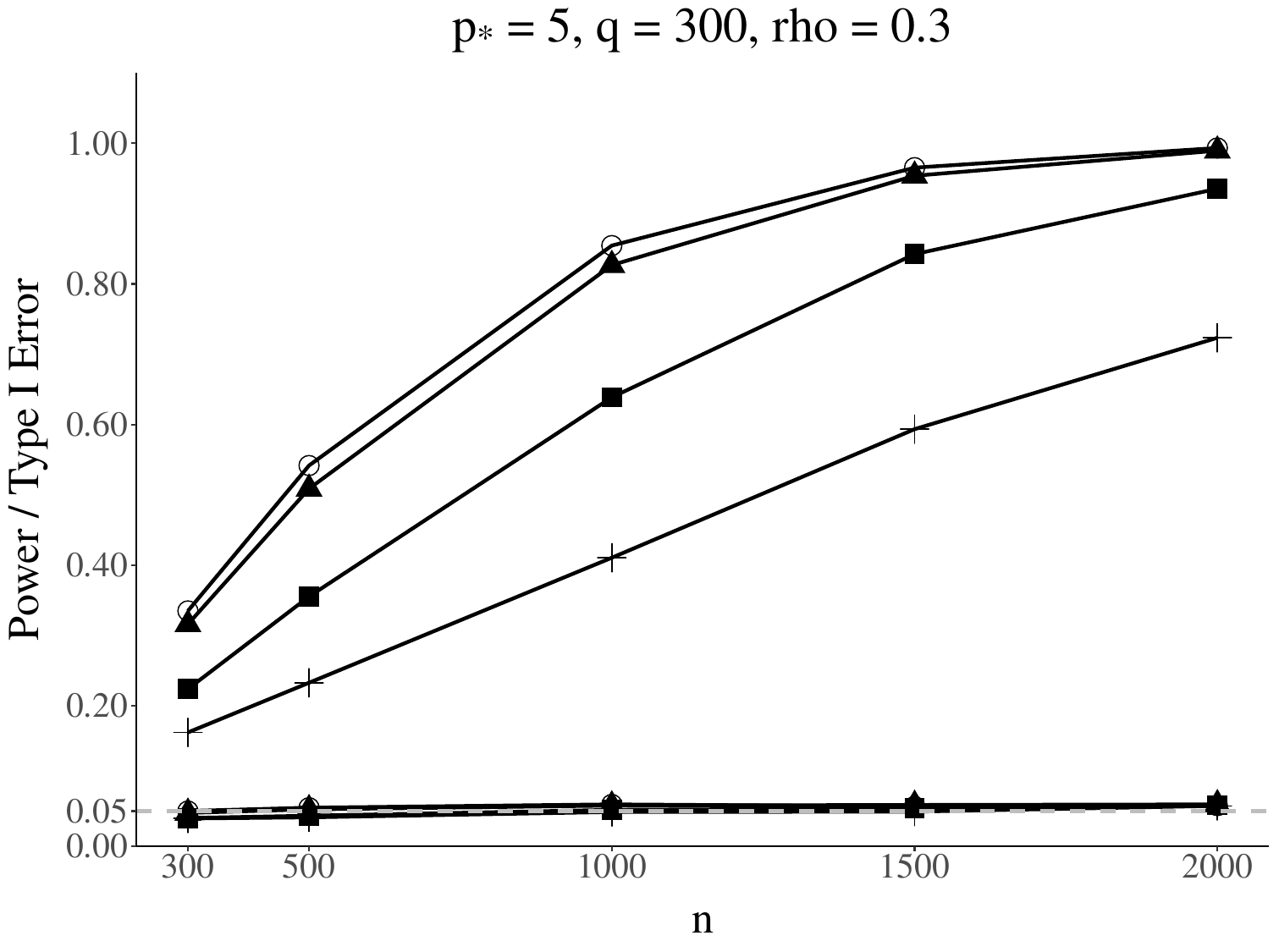}}
        \subfigure{
        \includegraphics[width=1.9in]{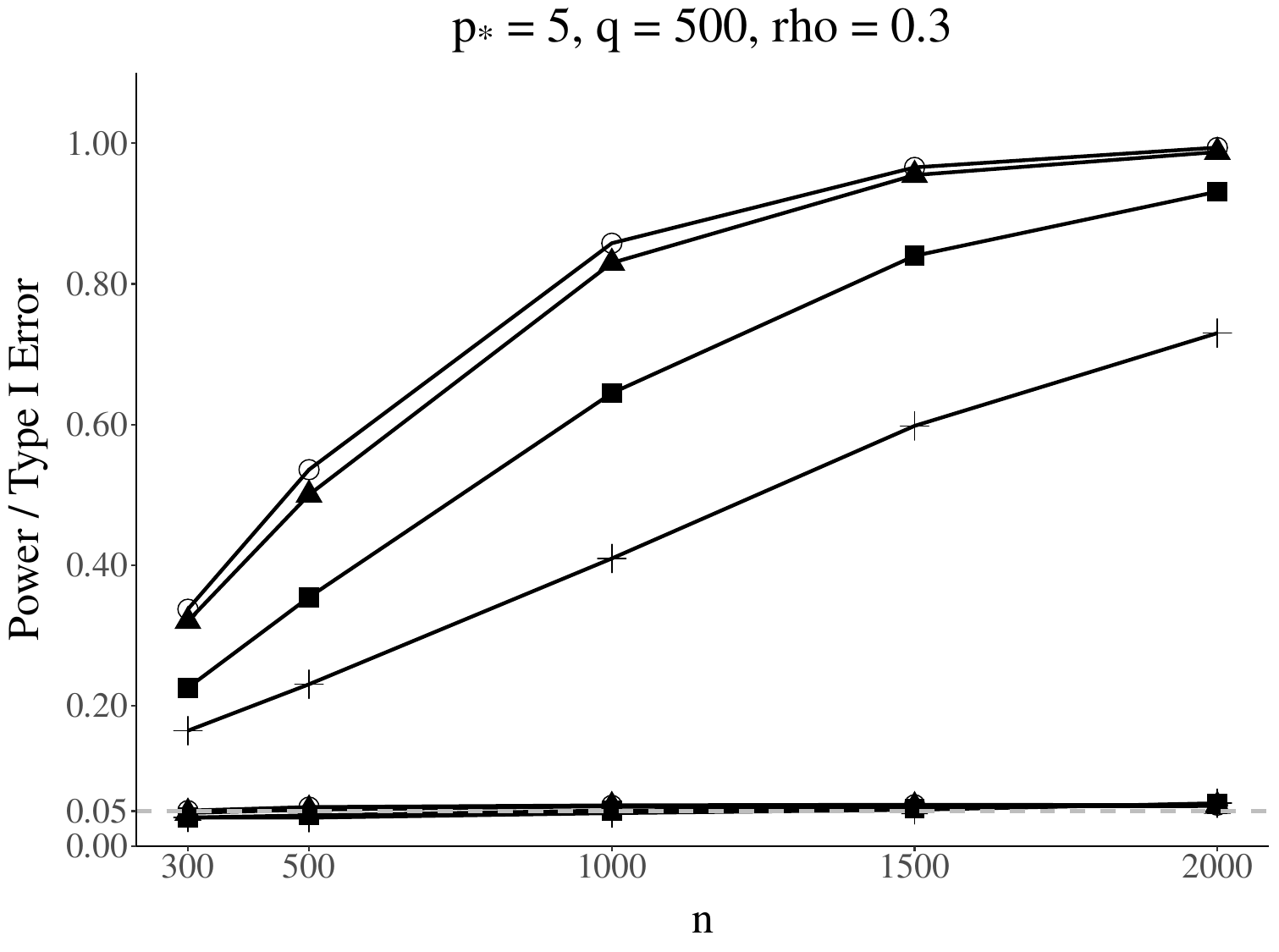}}
        \\
          \subfigure{
        \includegraphics[width=1.9in]{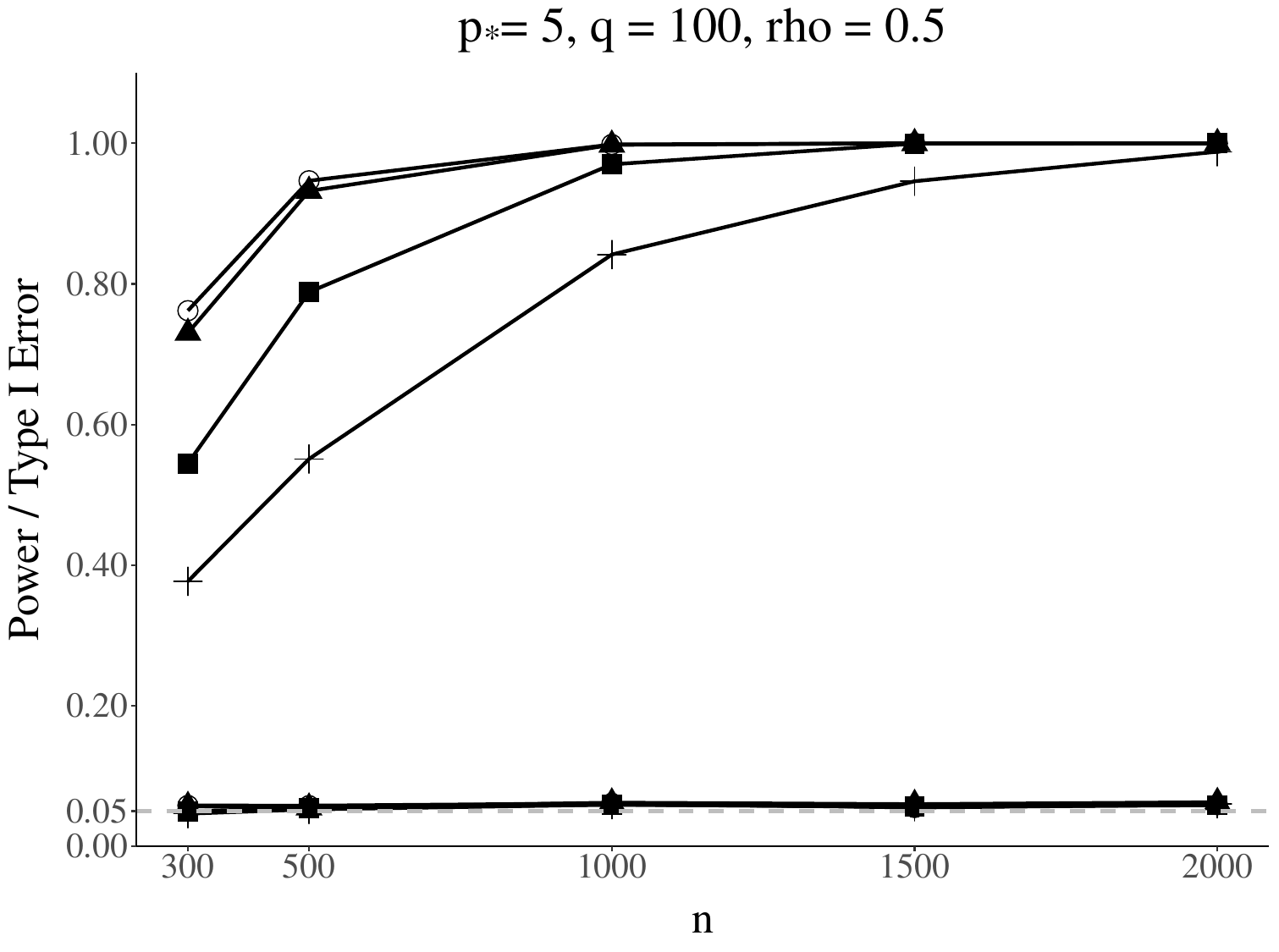}}%
          \subfigure{
        \includegraphics[width=1.9in]{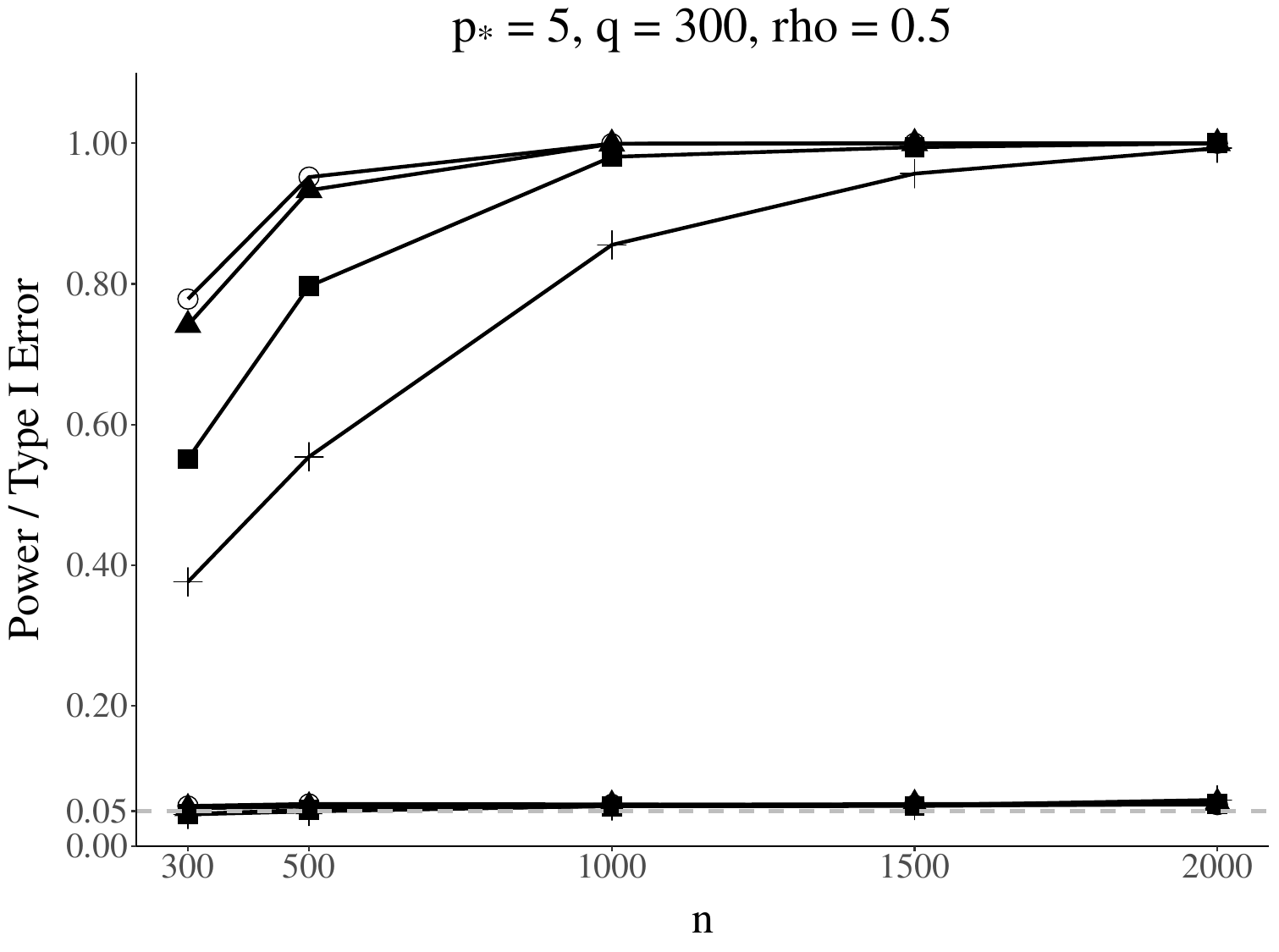}}%
        \subfigure{
        \includegraphics[width=1.9in]{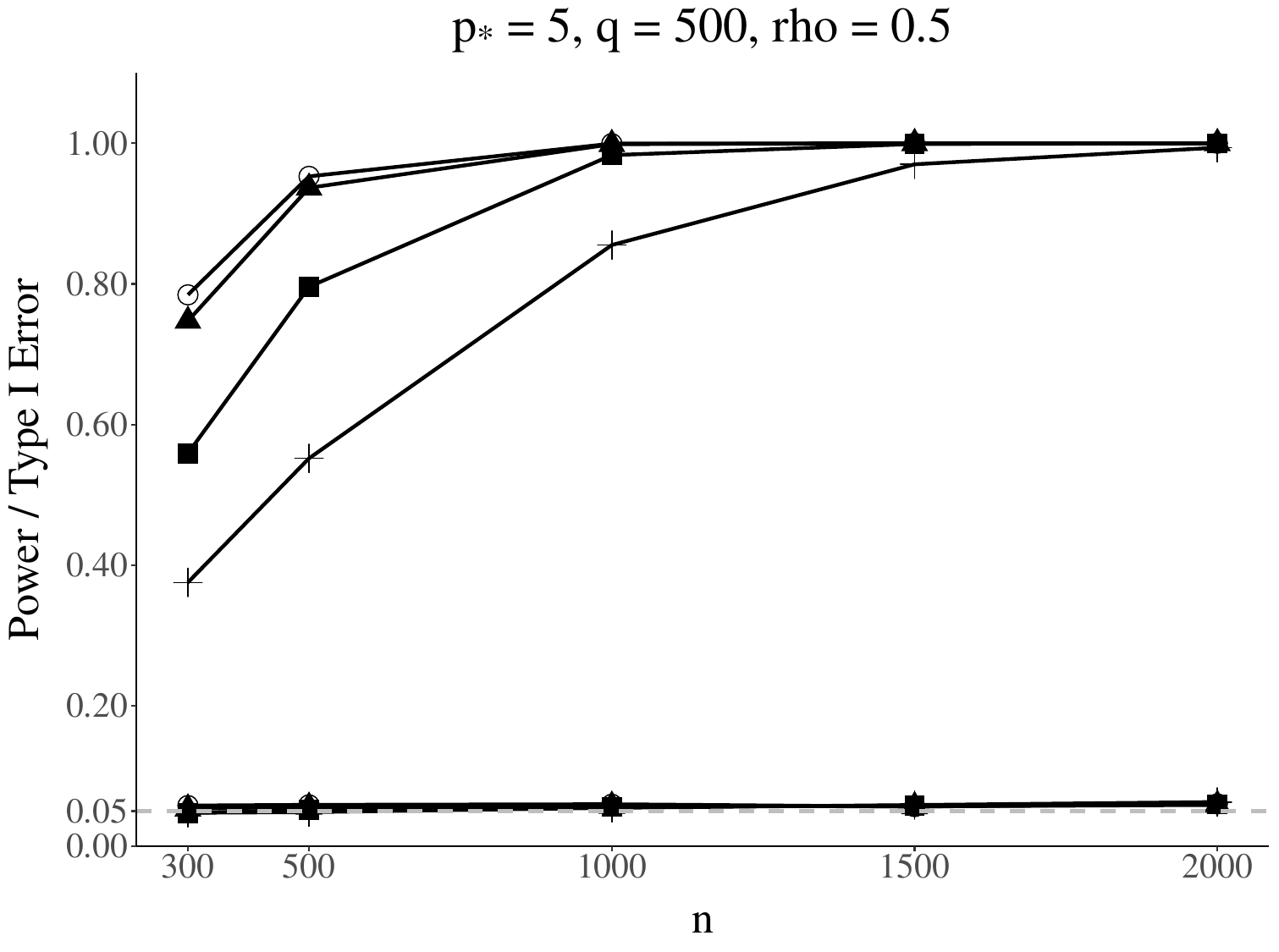}}%
    \caption{Powers and type I errors under dense setting at $p_*=5$. Circles (\protect\includegraphics[height=0.8em]{legend/new_rho0.png}) denote correlation parameter $\tau = 0$. Triangles (\protect\includegraphics[height=0.8em]{legend/rho0.2.png})  represent the case $\tau = 0.2$. Squares (\protect\includegraphics[height=0.8em]{legend/rho0.5.png}) indicate $\tau = 0.5$. Crosses (\protect\includegraphics[height=1em]{legend/rho0.7.png}) represent the $\tau = 0.7$.}
    \label{fig:non-anchor-correlated p5 dense}
\end{figure}

\begin{figure}[htbp]
\centering    
   \subfigure{
        \includegraphics[width=1.9in]{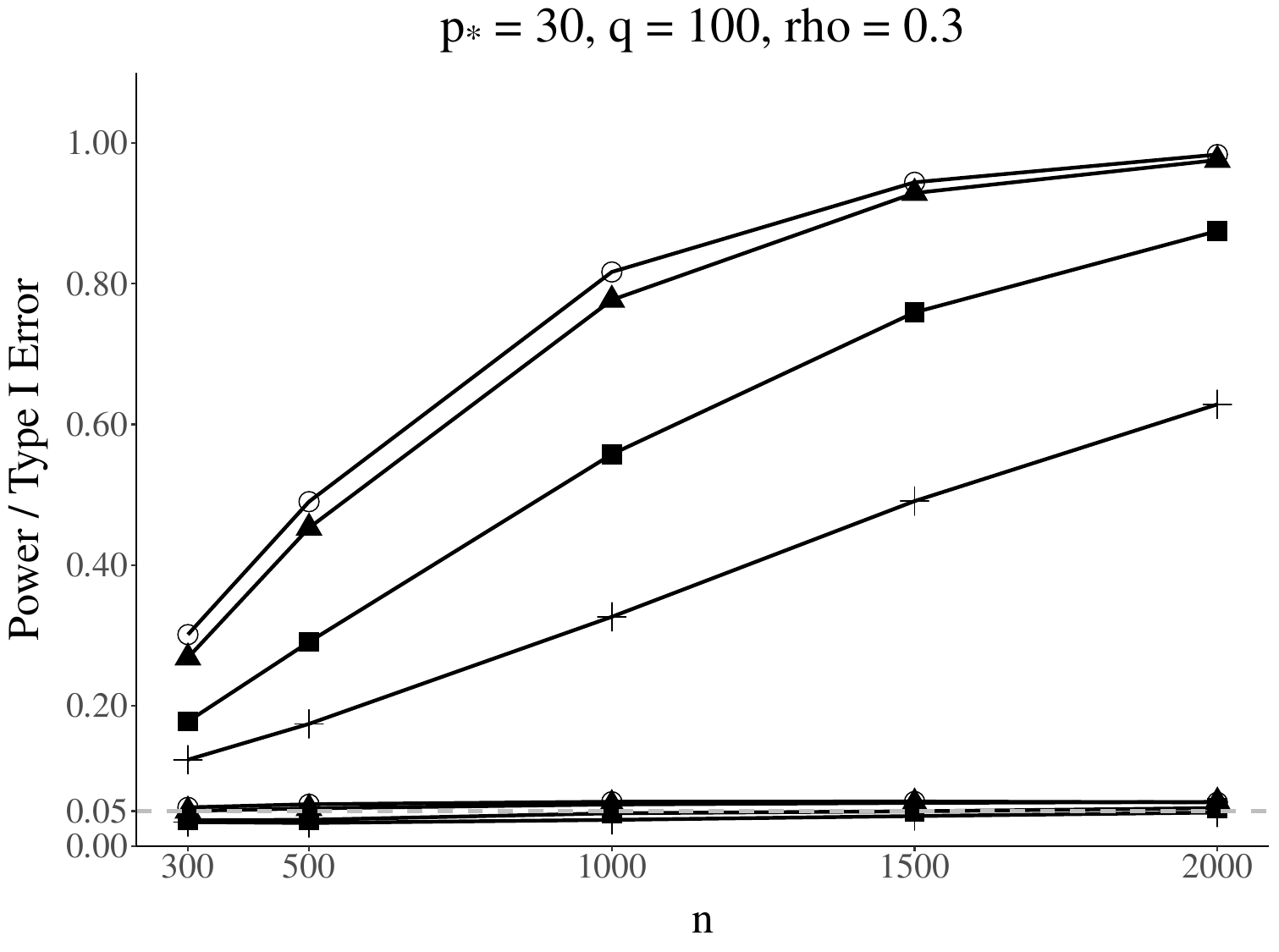}}
         \subfigure{
        \includegraphics[width=1.9in]{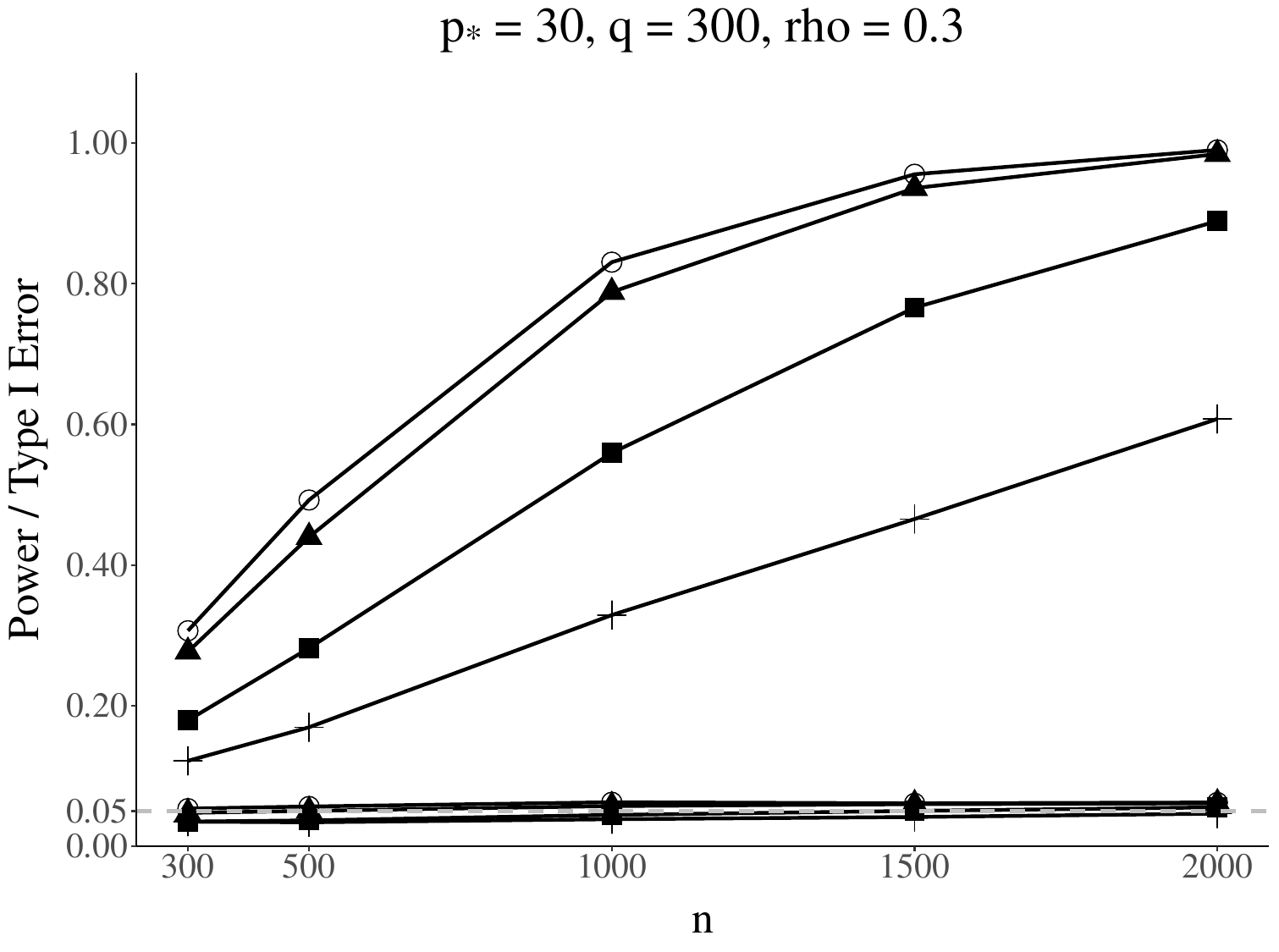}}
         \subfigure{
        \includegraphics[width=1.9in]{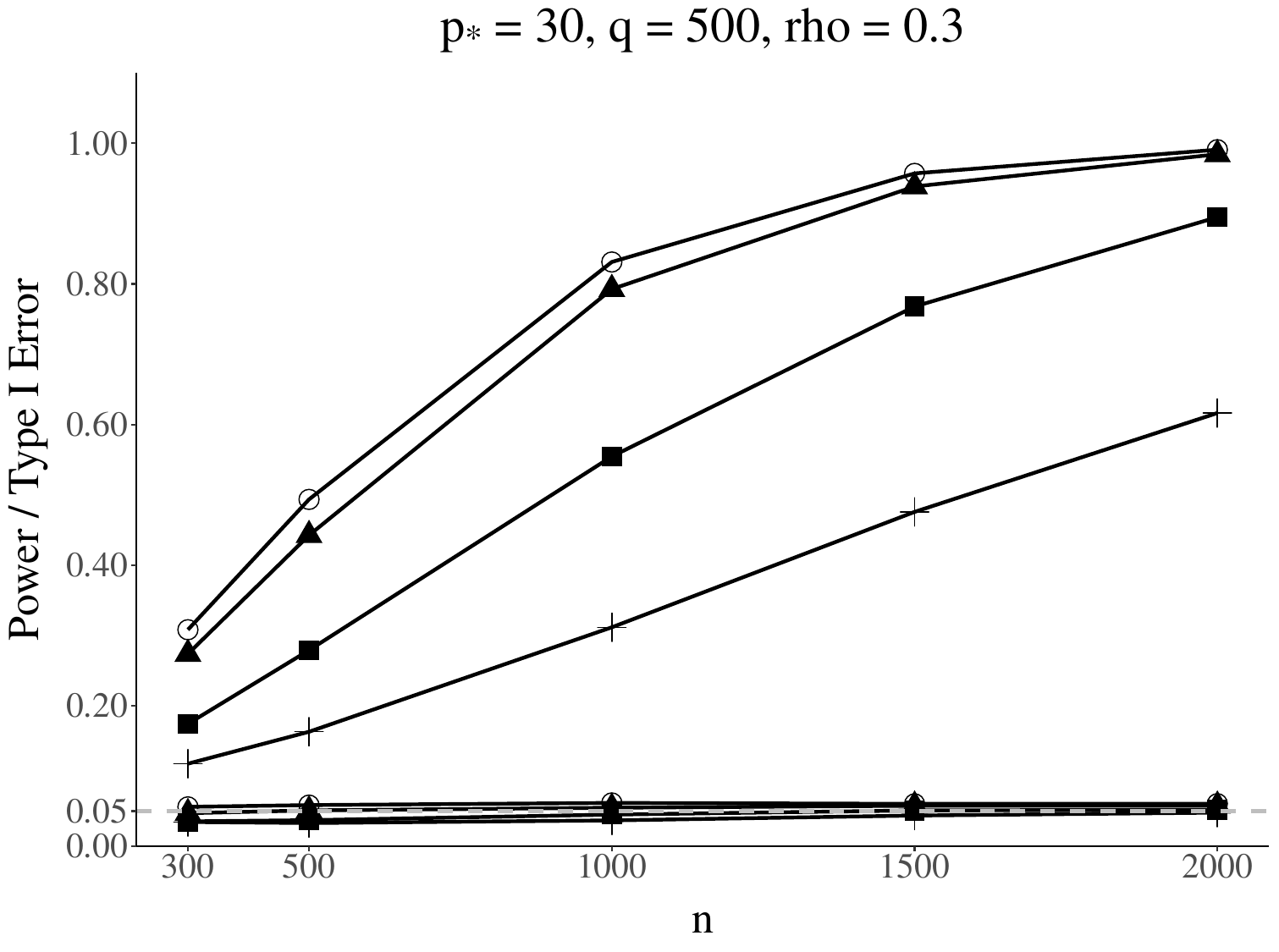}}
        \\
          \subfigure{
        \includegraphics[width=1.9in]{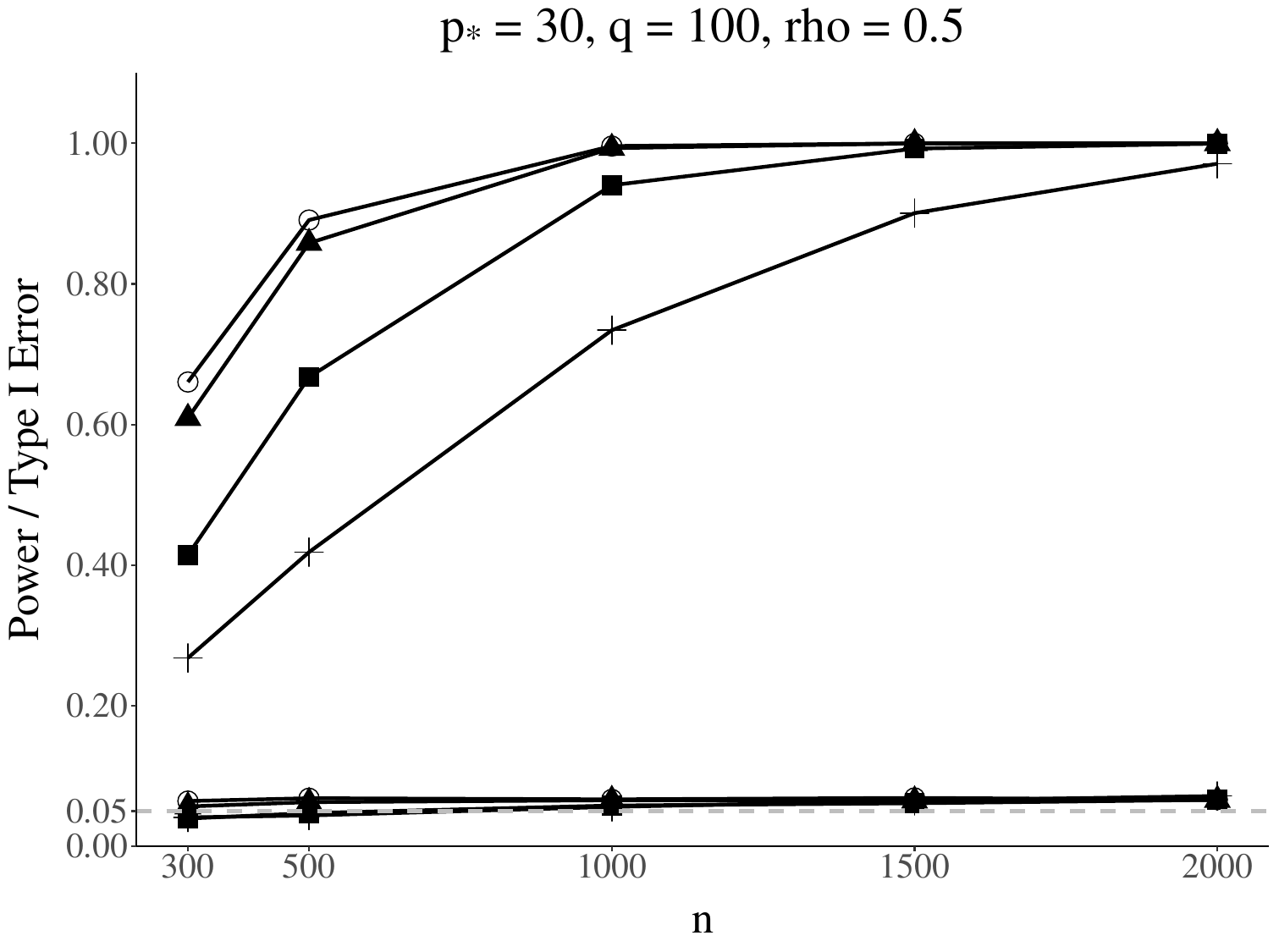}}%
          \subfigure{
        \includegraphics[width=1.9in]{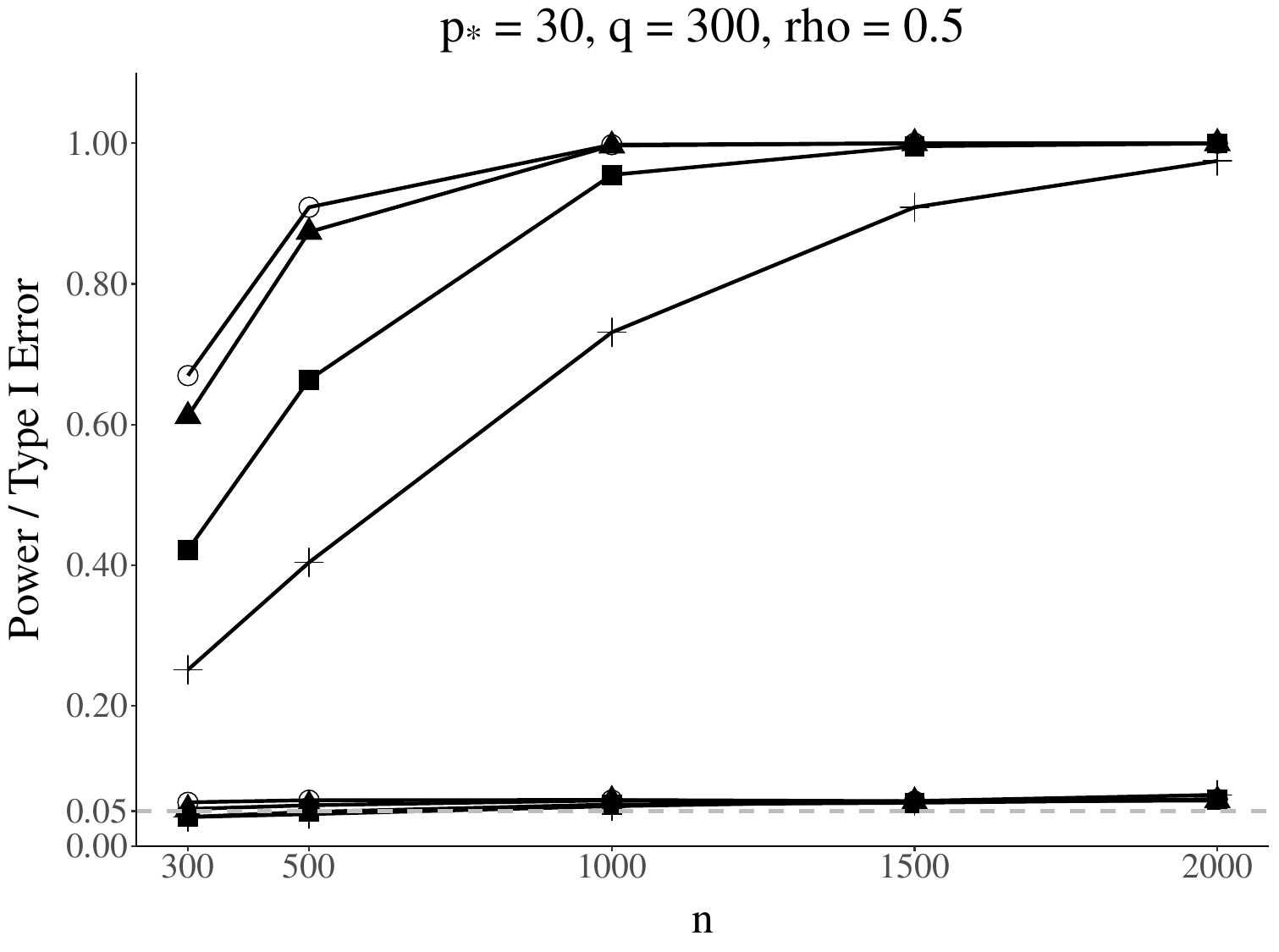}}%
          \subfigure{
        \includegraphics[width=1.9in]{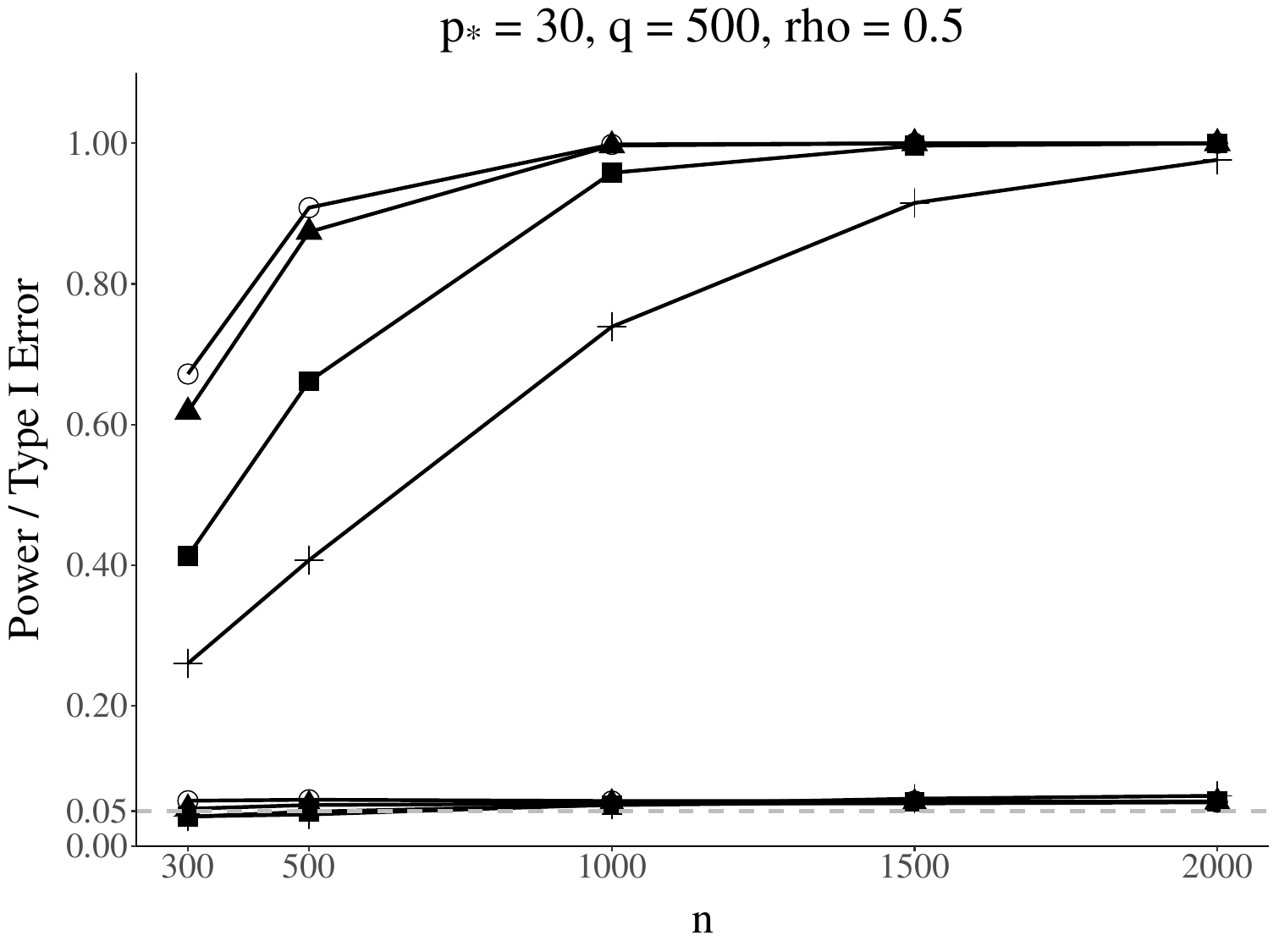}}
    \caption{Powers and type I errors under dense setting at $p_*=30$. Circles (\protect\includegraphics[height=0.8em]{legend/new_rho0.png}) denote correlation parameter $\tau = 0$. Triangles (\protect\includegraphics[height=0.8em]{legend/rho0.2.png})  represent the case $\tau = 0.2$. Squares (\protect\includegraphics[height=0.8em]{legend/rho0.5.png}) indicate $\tau = 0.5$. Crosses (\protect\includegraphics[height=1em]{legend/rho0.7.png}) represent the $\tau = 0.7$.}
    \label{fig:non-anchor-correlated p30 dense}
\end{figure}

For better empirical stability, after reaching convergence in the proposed  alternating minimization algorithm and transforming the obtained MLEs into ones that satisfy Conditions~\ref{cond:ID1} and~\ref{cond:ID2}, we repeat another round of maximization and transformation. 
We take the significance level at $5\%$ and calculate the averaged type I error based on each of the entries $\beta_{js}^* = 0$ and the averaged power for each of non-zero entries, over 100 replications.
We also include the empirical coverage probability of $\bU_i^*$ as an evaluation metric for inferential results of latent factors.
Specifically, we construct confidence intervals for each $U_{ik}$ for $i \in [n]$ and $k \in [K]$ and calculate the empirical coverage probabilities of these intervals on true parameter values $U_{ik}^*$ over 100 replications.
The averaged hypothesis testing results are presented in Figures~\ref{fig:non-anchor-correlated p5}--\ref{fig:non-anchor-correlated p30 dense} for $p_* = 5$ and $p_* = 30$, across different settings. 
Additional numerical results are provided in Section~F of the Supplementary Material.

From Figures~\ref{fig:non-anchor-correlated p5}--\ref{fig:non-anchor-correlated p30 dense}, we observe that the type I errors are well controlled at the significance level $5\%$, which is consistent with the asymptotic properties of $\hat{\Bb}^*$ in Theorem~\ref{thm:asymptotic normality post-transformation beta}. 
Moreover, the power increases to one as the sample size $n$ increases across all of the settings we consider. 
Comparing the upper row $(\rho = 0.3)$ to the bottom row $(\rho = 0.5)$ in Figures~3--6, we see that the power increases as we increase the signal strength $\rho$. 
Comparing the plots in Figures~\ref{fig:non-anchor-correlated p5}--\ref{fig:non-anchor-correlated p30} to the corresponding plots in Figures~\ref{fig:non-anchor-correlated p5 dense}--\ref{fig:non-anchor-correlated p30 dense}, we see that the powers under the sparse setting (Figures~\ref{fig:non-anchor-correlated p5}--\ref{fig:non-anchor-correlated p30}) are generally higher than that of the dense setting (Figures~\ref{fig:non-anchor-correlated p5 dense}--\ref{fig:non-anchor-correlated p30 dense}).  
Nonetheless, our proposed method is generally stable under both sparse and dense settings. 
In addition, we observe similar results when we increase the
covariate dimension $p_*$ from $p_*=5$ (Figures~\ref{fig:non-anchor-correlated p5} and~\ref{fig:non-anchor-correlated p5 dense}) to $p_* = 30$ (Figures~\ref{fig:non-anchor-correlated p30} and~\ref{fig:non-anchor-correlated p30 dense}).
We refer the reader to the Supplementary Material for additional numerical results for $p_*=10$.   
Moreover, we observe similar results when we increase the test length $q$ from $q=100$ (left panel) to $q=500$ (right panel) in 
Figures~\ref{fig:non-anchor-correlated p5}--\ref{fig:non-anchor-correlated p30 dense}. 
In terms of the correlation between $\Xb$ and $\Ub^*$, we observe that while the power converges to one as we increase the sample size, the power decreases as the correlation $\tau$ increases.

\section{Data Application}
\label{sec:data application}
We apply our proposed method to analyze the Programme for International Student Assessment (PISA) 2018 data\footnote{The data can be downloaded from: https://www.oecd.org/pisa/data/2018database/}. In this study, we focus on PISA 2018 data from Taipei. 
The observed responses are binary, indicating whether students' responses to the test items are correct, and we use the popular item response theory model with the logit link \citep[i.e., logistic latent factor model;][]{reckase2009}. Due to the block design nature of the large-scale assessment, each student was only assigned to a subset of the test items, and for the Taipei data, $86\%$ response matrix is unobserved. Note that this missingness can be considered as conditionally independent of the responses given the students' characteristics. 
 Our theoretical results can be extended to accommodate missing data. Under commonly studied missing patterns, such as when the missingness status indicators are independently and identically distributed Bernoulli random variables~\citep{davenport20141}, follow non-uniform distributions~\citep{cai2013max}, or follow a flexible missing-entry scheme that generalizes beyond random sampling scheme~\citep{chen2023statistical}, our results can be easily extended to the joint maximum likelihood estimation in the presence of missing data. 
Specifically, we can modify the 
joint log-likelihood function in~\eqref{eq:log likelihood} into $L^{obs}(\Yb \mid \bGamma, \Ub, \Bb,\Xb) = \sum_{i=1}^n \sum_{j \in \cQ_i} l_{ij}(\bgamma_j^{\intercal}\bU_i + \bbeta_j^{\intercal}\bX_i)$,
where $\cQ_i$ defines the set of questions to which the responses from student $i$ are observed. It can then be verified that our consistency results in Theorem 1 still hold under some regularity conditions and asymptotic normality in Theorem 2 and can also be established, with appropriate modifications to parameter definitions. Our real data analysis is carried out with certain adjustments to accommodate missing data. Detailed discussion can be found in Section G of the Supplementary Materials.



In this study, we include gender and 8 variables for school strata as covariates $(p_*=9)$.  The detailed description of school strata variables is provided in Supplementary Materials. These variables record whether the school is public, in a rural place, etc. In data preprocessing, we focus on data from Taipei as a representative sample. Items with binary scores are selected, while those with more than two response categories are excluded. We also retain only students who answered at least 10 questions. This screening process produces a final sample of $n=6063$ students and $q = 194$ items. 
In practice, the number of latent factors is typically pre-specified based on domain expertise or prior analysis~\citep{brown2012confirmatory, chen2025item}. Specifically, PISA 2018 evaluates student performance via three domains: mathematics, science, and reading. Accordingly, following existing literature~\citep{schleicher2019pisa, pisatechnicalreport2018}, we pre-specify the number of latent factors to be $K = 3$ to align with the three underlying abilities targeted in these domains. For the application to PISA 2018 dataset, the choice of $K=3$ carries the scientific meaning to directly correspond to each of latent abilities in this assessment.

We apply the proposed method to estimate the effects of gender and school strata variables on students' responses. We obtain the estimators of the gender effect for each PISA question and construct the corresponding $95\%$ confidence intervals.
The constructed $95\%$ confidence intervals for the gender coefficients are presented in Figure~\ref{fig:TAP gender bias}. There are 7 questions highlighted in dark as their estimated gender effect is statistically significant after the Bonferroni correction.
Among the reading items, there are two significant items and the corresponding confidence intervals are below zero, indicating that these questions are biased towards female test-takers, conditioning on the students' latent abilities.
Most of the confidence intervals corresponding to the biased items in the math and science sections are above zero, indicating that
these questions are biased towards male test-takers.
In social science research, it is documented that female students typically score better than male students during reading tests, while male students often outperform female students during math and science tests~\citep{quinn2015science, Balart2019females}. Our results indicate that there may exist potential measurement biases resulting in such an observed  gender gap in educational testing. Our proposed method offers a useful tool to identify such biased test items, thereby contributing to enhancing testing fairness by providing practitioners with valuable information for item calibration.

\begin{figure}[h]
    \centering
    \includegraphics[scale=0.56]
    {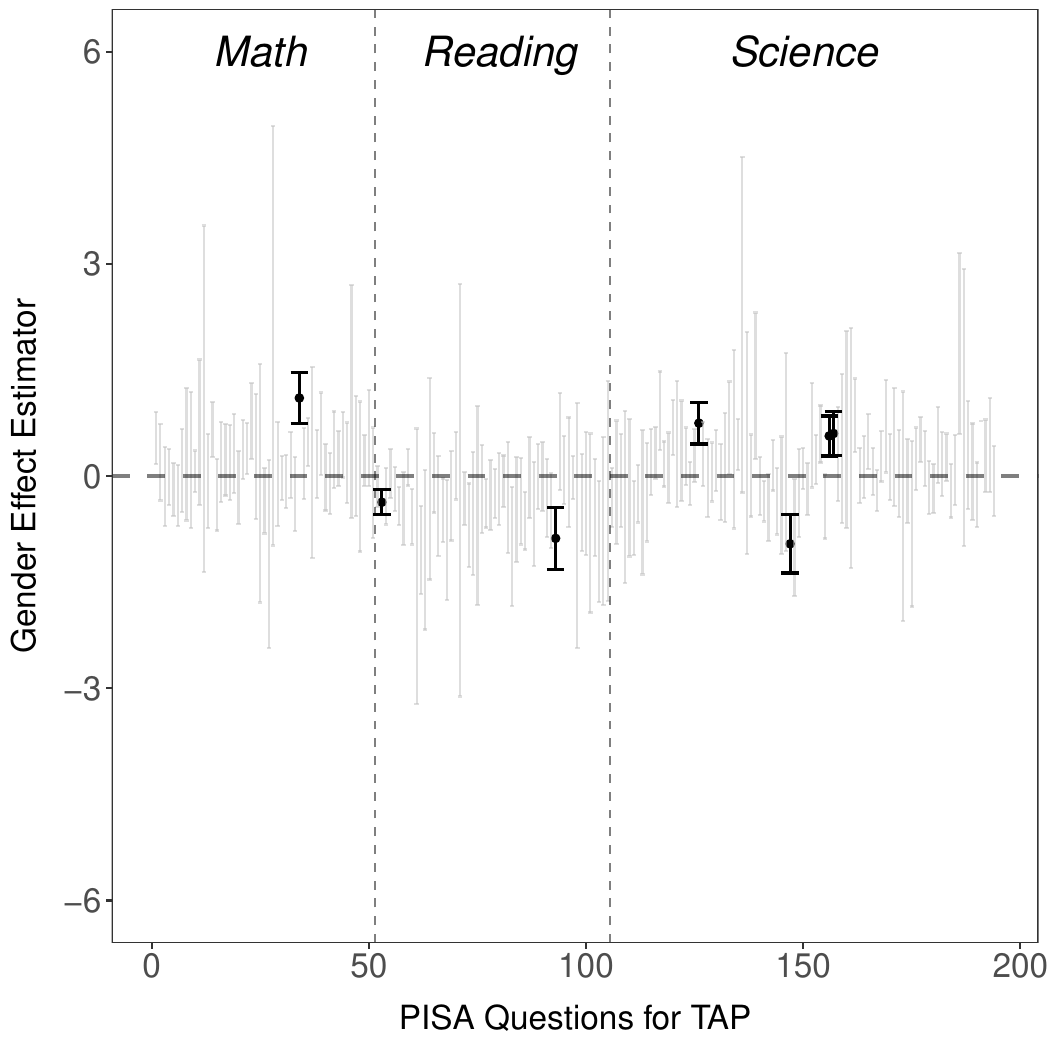}
    \caption{Confidence intervals for the effect of gender covariate on each PISA question using Taipei data. Dark intervals correspond to confidence intervals for questions with significant gender bias after Bonferroni correction. (For illustration purposes, we omit the confidence intervals with the upper bounds exceeding 6 and the lower bounds below -6 in this figure).}
    \label{fig:TAP gender bias}
\end{figure}


 \begin{table}[!h]
\begin{center}
\begin{tabular}{c|c|c|c|c}
\hline
Item code &  Item Title & Female ($\%$)  & Male ($\%$) & p-value \\  \hline
\multicolumn{5}{c}{\em Mathematics}   \\  \hline
 CM915Q01S & Carbon Tax &  55.93 &   62.05 & <1$\times 10^{-8}\; (+)$  \\  \hline

\multicolumn{5}{c}{\em Reading}   \\  \hline
  CR424Q03S &Fair trade&  65.88 &  58.16 & 3.96$\times 10^{-5}\; (-)$ \\  \hline
  CR466Q06S &Work right3&  91.91 &  86.02 & 8.23$\times 10^{-5}\; (-)$ \\  \hline

\multicolumn{5}{c}{\em Science}   \\  \hline
CS626Q01S &Sounds in Marine Habitats& 31.54 & 47.11 & 6.10$\times 10^{-7}\; (+)$
\\  \hline 
CS602Q04S &Urban Heat Island Effect& 79.53 & 72.93 & 5.49$\times 10^{-6}\; (-)$
\\  \hline 
CS527Q03S &Extinction of dinosours2& 59.14 & 70.81 & 8.38$\times 10^{-5}\; (+)$
\\  \hline 
CS527Q04S& Extinction of Dinosours3& 36.19 & 50.18 & 1.28$\times 10^{-4}\; (+)$
\\  \hline 
\end{tabular} 
\caption{Proportion of full credit in females and males to significant items of PISA2018 in Taipei. $(+)$ and $(-)$ denote the items with positively and negatively estimated gender effects, respectively.}
\label{table:gender gap}
\end{center}
\end{table}

 To further illustrate the estimation results, Table~\ref{table:gender gap} lists the $p$-values for testing the gender effect for each of the identified 7 significant questions, along with
 the proportions of female and male test-takers who answered each question correctly.
We can see that the signs of the estimated gender effect by our proposed method align with the disparities in the reported proportions between females and males.
For example, 
the estimated gender effect corresponding to the item ``CM915Q01S Carbon Tax'' is positive with a $p$-value of $<1\times 10^{-8}$, implying that this question is statistically significantly biased towards male test-takers. This is consistent with the observation that in Table~\ref{table:gender gap}, $62.05\%$ of male students correctly answered this question, which exceeds the proportion of females,  $55.93\%$.

 Based on the estimation and inference results of the individual effects $\beta_{js}$, we can also conduct group-wise testing of an overall covariate effect within a certain group of interest. For example, to test whether there is a group-wise gender bias on mathematics items, we consider the null hypothesis $H_0: \beta_{js}^* = 0, \forall j \in \cQ_{math}$. Thanks to the asymptotic results, we can use a chi-square type test statistics: $\sum_{j\in \cQ_{math}} \{(\hat{\beta}_{js}^*
            -0)/{s.e.(\hat{\beta}_{js}^*)}\}^2 = 117.32$, yielding a p-value of $2.46\times 10^{-7} < 0.05$, indicating a significant overall gender bias in the math category of items.
       Similarly, the group-wise tests are significant for testing the overall gender effect on all reading items and all science items with test statistics $151.05$ (p-value $=3.41\times 10^{-9}$) and $237.32$ (p-value $<1\times 10^{-10}$), respectively. 
These results are consistent with the individual tests, which indicate that for each category, there is at least one item that is gender-biased.       

Besides gender effects, we estimate the effects of school strata on the students' response and present part of the
 point and interval estimation results in the left panel of Figure~\ref{fig:TAP stratum1 bias} for illustrative purpose, and leave the complete results in Section G of Supplementary Materials. All the detected biased questions are from math and science sections, with 9 questions for significant effects of whether attending public school and 3 questions for whether residing in rural areas.
  To further investigate the importance of controlling for the latent ability factors, we compare results from our proposed method with the latent factors, to the results from directly regressing responses on covariates without latent factors.
From the right panel of Figure~\ref{fig:TAP stratum1 bias}, we see that without conditioning on the latent factors, there are much more items detected for the covariate of whether the school is public. On the other hand, there are no biased items detected if we only apply generalized linear regression to estimate the effect of the covariate of whether the school is in rural areas.


\begin{figure}[h!]
    \centering
    \subfigure{
        \includegraphics[width=2in]{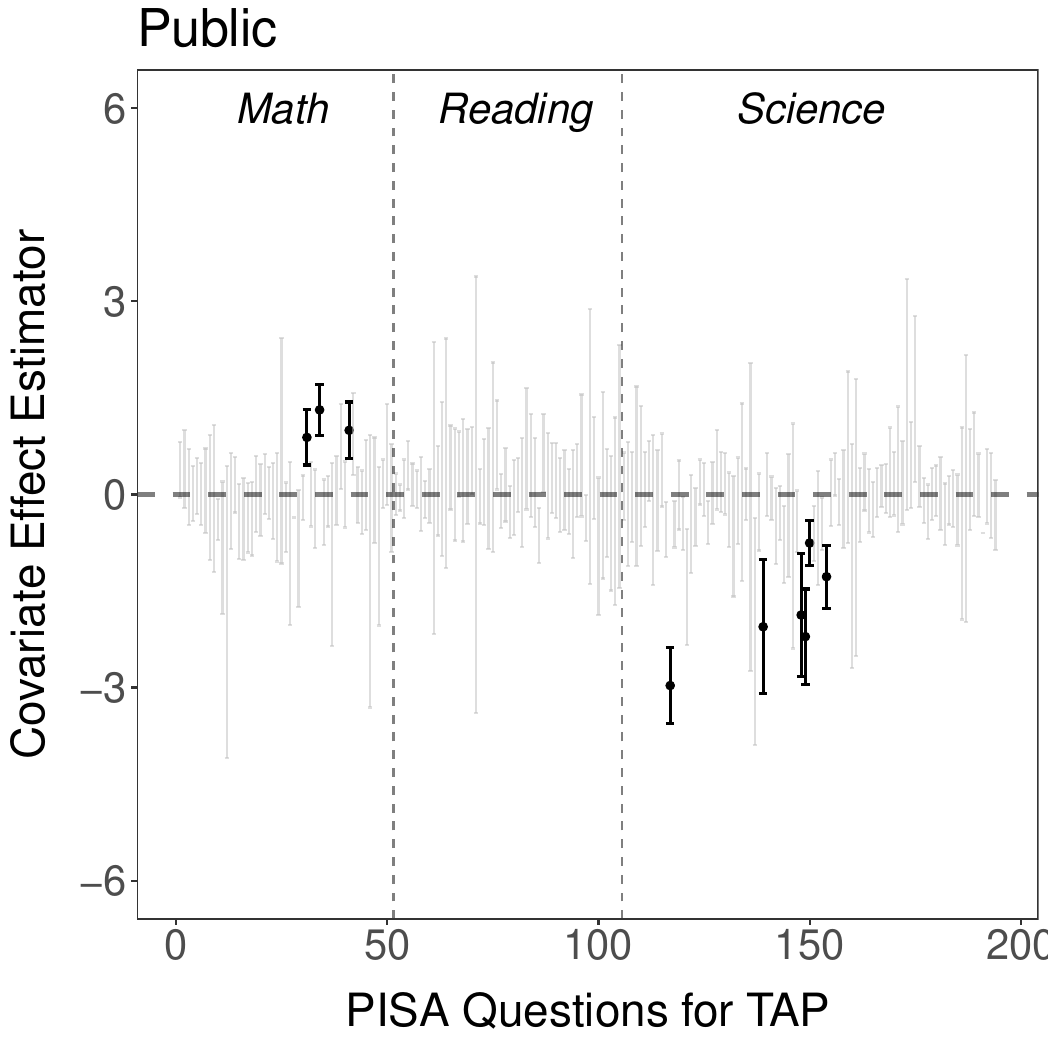}}
           \qquad 
           \subfigure{
        \includegraphics[width=2.2in]{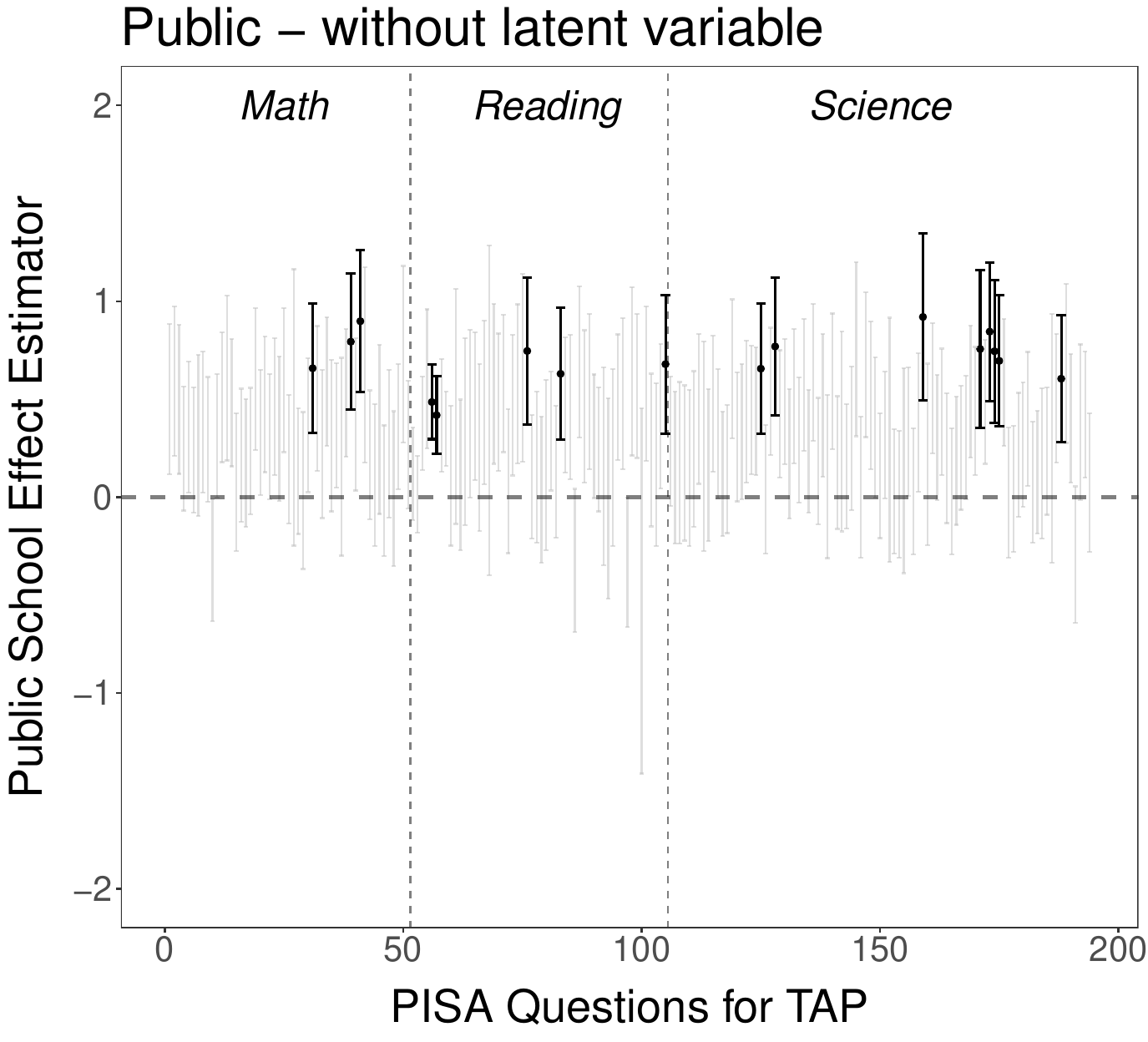}}
\\
    \subfigure{
        \includegraphics[width=2in]{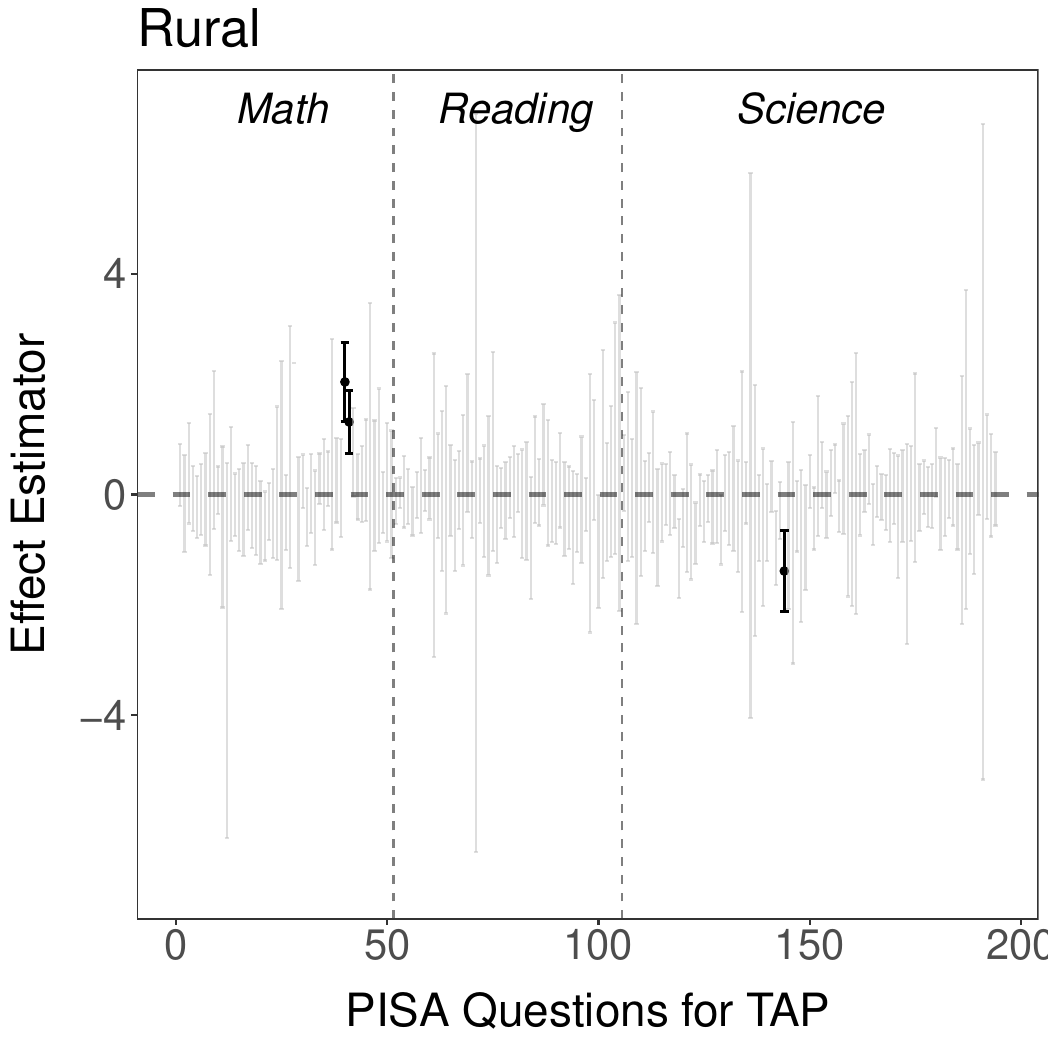}}
        \qquad  
        \subfigure{
        \includegraphics[width=2.2in]{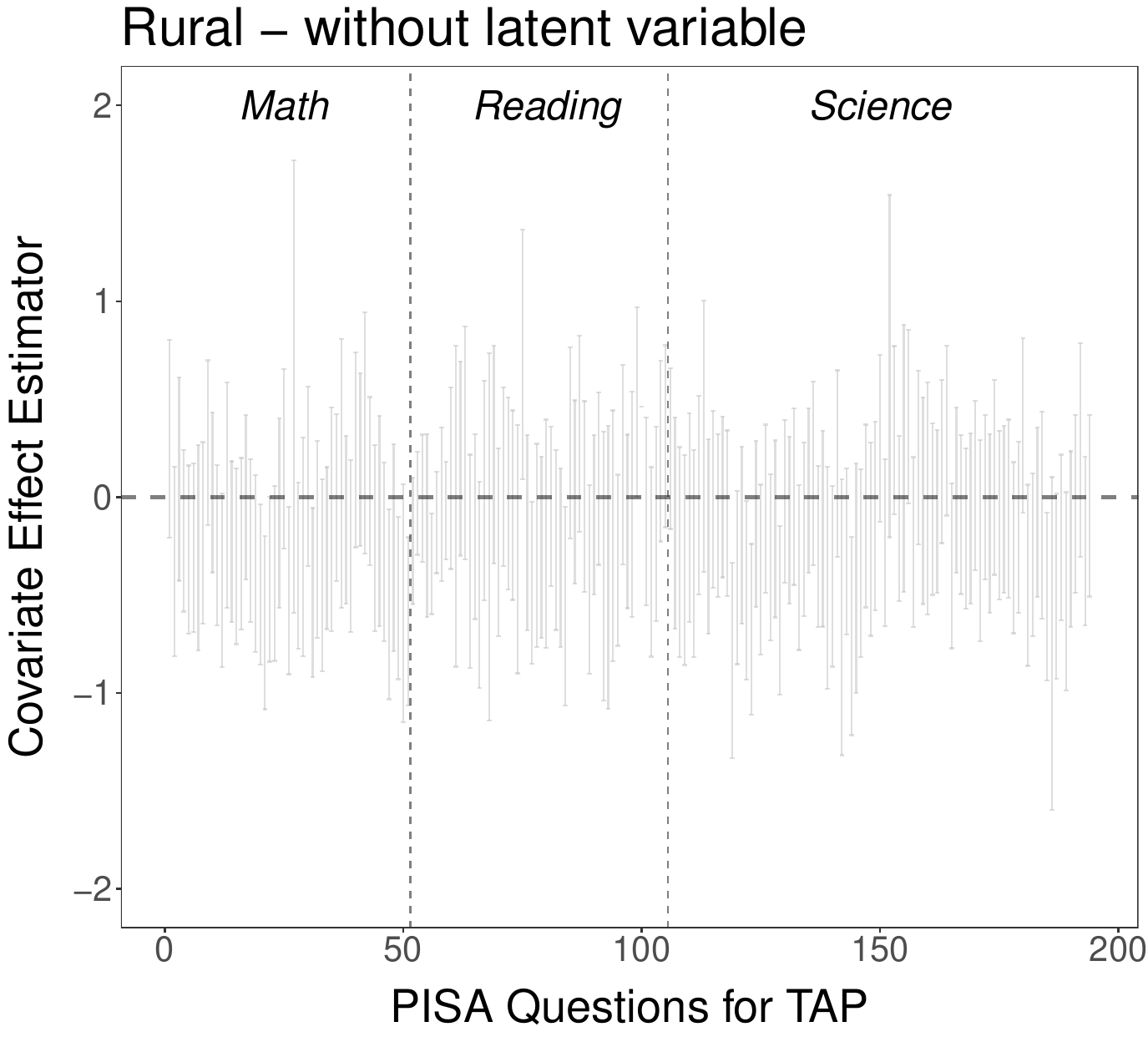}}
        \\
    \caption{Confidence intervals for the effect of part of school stratum covariate on each PISA question. Dark intervals correspond to confidence intervals for questions with significant school stratum bias after Bonferroni correction. 
    }
    \label{fig:TAP stratum1 bias}
\end{figure}

\section{Discussion}
\label{sec:discussion}

In this work, we study the covariate-adjusted generalized factor model that has wide interdisciplinary applications such as educational assessments and psychological measurements. 
In particular, new identifiability issues arise due to the incorporation of covariates in the model setup. To address the issues and identify the model parameters, we propose novel and interpretable conditions that are crucial for developing the estimation approach and inference results. With model identifiability guaranteed, we consider a joint-likelihood-based estimation method for model parameters.
Theoretically, we obtain the estimation consistency and asymptotic normality for not only the covariate effects but also latent factors and factor loadings. 

In our work, the covariate-adjusted generalized factor model is used within a confirmatory factor analysis framework, where the number of factors is usually pre-specified based on theoretical considerations or prior research~\citep{brown2012confirmatory, chen2025item}.
Specifically, in our data application to PISA 2018, the number of latent factors is pre-specified as $K = 3$, corresponding to the three latent abilities assessed by PISA 2018: mathematics, science, and reading~\citep{schleicher2019pisa, pisatechnicalreport2018}. In other datasets where prior knowledge about the factor structure may be limited, exploratory data analytic tools may be applied to determine the number of latent factors, including scree plots~\citep{cattell1966scree}, parallel analysis~\citep{horn1965}, the eigenvalue ratio method~\citep{kaiser1960application, lam2012factor}, and information-based criteria such as AIC~\citep{akaike1987factor}, BIC~\citep{schwarz1978estimating}, and other likelihood-based information criteria~\citep{chen2022determining}. 

 There are several future directions motivated by the proposed method. 
 In this manuscript, we focus on the case in which $p$ grows at a slower rate than the number of subjects $n$ and the number of items $q$, a common setting in educational assessments.  
It is interesting to further develop estimation and inference results under the high-dimensional setting in which $p$ is larger than $n$ and $q$. 
Moreover, in this manuscript, we assume that the dimension of the latent factors $K$ is fixed and known. 
One possible generalization is to allow $K$ to grow with $n$ and $q$.  
Intuitively, an increasing latent dimension $K$ makes the identifiability and inference issues more challenging due to the increasing degree of freedom of the transformation matrix. 
With the theoretical results in this work, another interesting related problem is to further develop simultaneous inference on group-wise covariate coefficients over all $q$ items, which we leave for future investigation. 
Finally, the model setup considered in this manuscript corresponds to the uniform DIF setting in psychometrics, where the effects of observed covariates on item responses are invariant across the latent factors~\citep{holland2012differential}. Besides the uniform DIF setting, the non-uniform DIF setting, where the covariate effects may vary across the latent factors, also enjoys wide applications in educational assessments~\citep{holland2012differential, wang2023using}. For instance, the advantage of certain demographic groups on exam questions can change along the continuum of the assessed latent skills and abilities~\citep{wang2023using}.
In the future, it is also an interesting direction to extend our considered model setup to accommodate the non-uniform DIF setting and develop the estimation approach and theoretical results under such generalized models. Moreover, in this paper, we focus on joint maximum likelihood estimation because of its computational efficiency. Nonetheless, it would also be interesting to explore marginal marginal maximum likelihood estimation, which treats latent factors as random effects different from joint MLE, where latent factors are treated as fixed effects.  Statistical inference for marginal MLE in the double asymptotic regime $n, q \rightarrow \infty$ remains an open yet important problem and we leave this interesting problem for future investigation.

\section{Significance Statement}\label{sec_sigf_state}


This paper addresses an important applied problem of detecting biased items in large-scale educational assessments. A test item is considered biased if students with the same level of latent ability (e.g., reading or math proficiency) but with different individual characteristics (e.g., gender, race, or other demographics) have different response functions to it. We develop a novel statistical framework to detect biased items in large-scale assessments.  The framework is developed under a general setting that accommodates different types of responses and enables efficient estimation in large-scale setups. We establish statistical guarantees for the estimated model parameters and provide uncertainty quantification for the detection of biased items. In addition, our results enable simultaneous statistical inference for both latent abilities and item parameters.
Empirically, we apply the proposed method to Programme for International Student Assessment (PISA) 2018, a large-scale educational assessment dataset, and demonstrate that our method detects biased test items and provides valid inference results. These findings can assist researchers and practitioners in identifying and revising biased items, and support policymakers in designing fair assessments and promoting equity in education systems.
\\


\begin{acks}[Acknowledgments]
The authors would like to thank the four anonymous referees, an Associate Editor, and the Editor for their constructive comments that improved the
quality of this paper.
\end{acks}

\begin{funding}
Ouyang is partially supported by Hong Kong Early Career Scheme (ECS) Grant \#27308125 and Seed Fund for Basic Research \#2401102046.
Xu is partially supported by NSF SES-2150601 and SES-1846747.
\end{funding}


\begin{center}
{SUPPLEMENTARY MATERIAL}
\end{center}

The Supplementary Materials contain the proofs of the theorems and additional numerical study results. The R code for the simulation studies and data application is available at \url{https://github.com/jingoystat/Covariate_Adjusted_Generalized_Factor_Analysis/tree/main}.


\bibliographystyle{imsart-nameyear} 
\bibliography{arxiv}       


\newpage 
\begin{appendix}
\renewcommand{\theequation}{A\arabic{equation}}
\renewcommand{\thefigure}{S\arabic{figure}} 
\renewcommand{\thetable}{S\arabic{table}} 

\setcounter{page}{1}    

\begin{center}
    {\Large \bf Supplementary Material for ``Statistical Inference for Covariate-Adjusted and Interpretable Generalized Latent Factor Model with Application to Testing Fairness''}

\end{center}




\vspace{0.2in}


This Supplementary Material provides proofs of the theoretical results in the main text and additional simulation results.
It is organized as follows. Section~\ref{sec:reg assumptions} presents the detailed technical assumptions for establishing the theoretical results. 
Section~\ref{sec:supp prove mainprop} proves the propositions in the main text. 
Section~\ref{sec:prove main} contains the proofs of the main theorems and Corollary~1. 
Sections~\ref{sec:proveadd_prop} presents the proofs of lemmas used in Section~\ref{sec:prove main}. Section~\ref{sec:prove_other} provides the proofs of other technical lemmas. Finally, Section~\ref{sec:a_simula} includes additional simulation results.

We begin by introducing notations and expressions used throughout the subsequent theoretical proofs. For any integer $N$, let $[N]=\{1,\cdots,N\}$. For any vector $\br = (r_1, \dots, r_l)^{\T}$, let $\| \br\|_0  = \text{card} (\{j: r_j \neq 0 \})$,  $\|\br \|_{\infty}= \max_{j = 1, \ldots, l} |r_j|$, and $\|\br\|_q = (\sum_{j=1}^l |r_j|^q)^{1/q}$ for $q \geq 1$.  We define $\bm{1}_x^{(y)}$ as the $y$-dimensional vector with $x$-th entry to be 1 and all other entries to be 0.
For any symmetric matrix $\Mb$, let $\lambda_{\min}(\Mb)$ and $\lambda_{\max}(\Mb)$ be the smallest and largest eigenvalues of $\Mb$. 
For any matrix $\Ab = (a_{ij})_{n\times l}$, let $\|\Ab\|_{1} = \max_{j=1,\ldots, l} $ $\sum_{i = 1}^n |a_{ij} |$ be the maximum absolute column sum, $\| \Ab\|_{\infty} = \max_{i=1,\ldots, n} \sum_{j=1}^l |a_{ij}| $ be the maximum of the absolute row sum, and $\| \Ab\|_{\max} = \max_{i,j} |a_{ij}|$ be the maximum of the absolute matrix entry. Let $\Ab_v = \text{vec}(\Ab) \in \RR^{nl}$ to indicate the vectorized form of matrix $\Ab\in \RR^{n\times l}$.
For subsets $S_1 \subseteq [n]$ and $S_2\subseteq [l]$, we denote $[\Ab]_{[S_1,S_2]}$ the sub-matrix of $\Ab$ with entries in rows indexed by $S_1$ and columns indexed by $S_2$. When $S_1=[n]$, we omit $S_1$ and write $[\Ab]_{[,S_2]}$. Similarly, we simplify the notation to $[\Ab]_{
[S_1,]}$ when $S_2=[l]$.   The abbreviation ``$w.h.p.$'' stands for ``with high probability approaching 1''.
For notation simplicity, we define the following expressions frequently used in the proofs: $\delta_{nq} = \min \{n^{1/2}, q^{1/2}\}$, $\epsilon_{nq}=(nq)^{\epsilon}$ for sufficiently small $\epsilon>0$, and 
\begin{equation*}
    \zeta_{nq,p}= \left(\sqrt{\frac{p\log qp}{n}} + \sqrt{\frac{\log n}{q}}  \right)^{-1}. 
\end{equation*}
To aid the reader's understanding, we summarize the key notations used throughout the main text in Table~\ref{tab:notation}.

\renewcommand{\thetable}{S\arabic{table}} %

\begin{table}[ht]
\centering
\caption{Table of Notations}
\begin{tabular}{ll}
\toprule
\textbf{Notation} & \textbf{Description} \\
\midrule
$\{Y_{ij}\}_{n\times q}$& Item responses
from $n$ subjects to $q$ items.\\
$\bX_i = (1,(\bX_i^c)^\T)^\T$& Covariates for subject $i$.\\
$\bbeta_{j} = (\beta_{j0},\bbeta_{jc}^\T)^\T$ & Assembled vector of intercept and coefficients for item $j$.\\
$\Xb = (\bX_1,\dots,\bX_n)_{n\times p}^\T$&  Design matrix of observed covariates associated with the $n$ subjects.\\
$\Bb = (\bbeta_1,\dots,\bbeta_q)^\T_{q\times p}$& Matrix for intercepts and the covariate
effects.\\
$\Ub=(\bU_1,\dots,\bU_n)^\T_{n\times K}$& Factor matrix.\\
$\bGamma = (\bgamma_1,\dots,\bgamma_q)^\T_{q\times K}$&Loading matrix.\\
$\bphi = (\bGamma,\Ub,\Bb)$ & Model parameters for (1). \\
$\bphi^* = (\bGamma^*,\Ub^*,\Bb^*)$ & True parameters for (1). \\
$\hat\bphi^* = (\hat\bGamma^*,\hat\Ub^*,\hat\Bb^*)$ & Joint maximum likelihood estimation of $\bphi^*$ given in Section~3.2. \\
$\bZ_i^* = ({\bU_i^*}^\T,\bX_i^\T)^\T$&Combined vector of subject $i$'s unobserved latent factor \\&and observed covariates in Assumption~\ref{assumption:asymptotic normality}. \\
$\Gb^{\ddagger}$, $\Ab^{\ddagger}$&Auxiliary transformation matrices defined two paragraph prior to Assumption~\ref{assumption:A consistency}.\\
$\bSigma_{\beta,j}^*$, $\bSigma_{\gamma,j}^*$, $\bSigma_{u,i}^*$ &  Asymptotic variance-covariance matrices in Theorem~4.2.\\
\bottomrule
\end{tabular}
\label{tab:notation}
\end{table}



\section{Regularity Assumptions}

\label{sec:reg assumptions}

 We denote $\bphi^* = (\bGamma^*, \Ub^*, \Bb^*)$ as the true parameters and introduce the regularity assumptions as follows. 

\begin{assumption}
	\label{assumption: psd covariance}
There exist constants $M >0$ and $\kappa  >0$ such that:

(i) There exist a positive definite $\bSigma_u^*$ such that  $n^{-1} (\Ub^{*})^{\T} \Ub^*  {\rightarrow}  \bSigma_u^* $ as $n\rightarrow\infty$.  For $i \in [n]$,  $ \| \bU_i^*\|_{2} \le M$.

 (ii) There exist a positive definite $\bSigma_\gamma^*$ such that $q^{-1} (\bGamma^{*})^{\T} \bGamma^*  {\rightarrow} \bSigma_\gamma^*$ as $q\rightarrow\infty$. For $j \in [q]$,  $\|\bgamma_j^*\|_{2} \le M$. 

 (iii) 
 $\bSigma_x = \lim_{n\rightarrow \infty} n^{-1} \sum_{i=1}^n\bX_{i} \bX_{i}^{\T}$ exists and satisfy $ 1/ \kappa^2 \le \lambda_{\min} (\bSigma_x)  \le \lambda_{\max} (\bSigma_x) \le \kappa^{ 2}$. $\max_{i\in[n]}\|\bX_i\|_{\infty} \le M$;  $\max_{i\in[n],j\in[q]}\big|(\bbeta_j^*)^\T\bX_i\big|\le M$.

(iv)  There exist $\bSigma_{ux}^*$ such that $n^{-1} \sum_{i=1}^n \bU_i^* \bX_i^{\T} {\rightarrow} \bSigma_{ux}^*  $ as $n \rightarrow \infty$ with $\|\bSigma_{ux}^*\bSigma_{x}^{-1}\|_{\infty}\le M$. The eigenvalues of $(\bSigma_{u}^*-\bSigma_{ux}^*\bSigma_{x}^{-1}(\bSigma_{ux}^{*})^{\T})\bSigma_{\gamma}^*$ are distinct and nonzero.

	\end{assumption}

		Assumptions~\ref{assumption: psd covariance} is commonly used in the factor analysis literature. In particular, 
  Assumptions~\ref{assumption: psd covariance}(i)--(ii) correspond to Assumptions A-B in~\cite{bai2003inferential} under linear factor models, ensuring the compactness of the parameter space on $\Ub^*$ and $\bGamma^*$. 
  Under nonlinear factor models, such conditions on compact parameter space are also commonly assumed~\citep{wang2022maximum, chen2023statistical}. Assumptions~\ref{assumption: psd covariance}(iii) and~\ref{assumption: psd covariance}(iv) are standard regularity conditions for the nonlinear setting that is needed to establish the concentration of the gradient and estimation error for the model parameters when $p$ diverges. 
 In this work, we treat the latent factors $\bU_i$s as fixed model parameters and the covariates $\bX_i$ as realizations of random variables representing each subject's covariates. Our theoretical framework is established based on the probability density function of $Y_{ij}$ given $\bX_i$ and the derived asymptotic properties of the estimators are also developed conditionally on $\bX_i$. Instead of making distributional assumptions on the random variables, we make assumptions on the sample moments of $\bX_i$.


	\begin{assumption}
	\label{assumption:smoothness}
For any $i \in [n]$ and $j \in [q]$, assume that 
$l_{ij} (\cdot)$ is three times differentiable, and we denote the first, second, and third order derivatives of $l_{ij}(w_{ij})$ with respect to $w_{ij}$ as $l_{ij}^{\prime} (w_{ij}), l_{ij}^{\prime\prime} (w_{ij})$, and $l_{ij}^{\prime\prime\prime} (w_{ij})$, respectively. There exist $M > 0$ and $\xi \ge 4$ such that  $\EE(|l_{ij}^{\prime} (w_{ij}^*) |^{\xi}) \le M$ and 
 $|l_{ij}^{\prime} (w_{ij}^*)|$ is sub-exponential with $\|l_{ij}^{\prime} (w_{ij}^*)\|_{\varphi_1} \le M$. Furthermore, we assume $\EE\{l_{ij}^{\prime} (w_{ij}^*)\} = 0$. 
 Within a compact space of $w_{ij}$, we have $b_L \le -l_{ij}^{\prime\prime} (w_{ij}) \le b_U$ and $|l_{ij}^{\prime\prime\prime}(w_{ij})| \le b_U$ for $b_U > b_L > 0 $. 

   \end{assumption}

Assumption~\ref{assumption:smoothness} assumes smoothness on the log-likelihood function $l_{ij}(w_{ij})$. 
In particular, it assumes sub-exponential distributions and finite fourth-moments of the first order derivatives $l^{\prime}_{ij} (w_{ij}^*)$. 
For commonly used linear or nonlinear factor models, the assumption is not restrictive and can be satisfied with a large $\xi$. 
For instance, consider the logistic
model with $l_{ij}^{\prime}(w_{ij}) = Y_{ij} - \exp(w_{ij})/\{1+\exp(w_{ij})\}$, we have $|l_{ij}^{\prime}(w_{ij})|\le 1$ and $\xi$ can be taken as $\infty$.
The boundedness conditions for $l_{ij}^{\prime\prime}(w_{ij})$ and $l_{ij}^{\prime\prime\prime}(w_{ij})$ 
are necessary to guarantee the convexity of the joint likelihood function. 
In a special case of linear factor models, 
$l_{ij}^{\prime\prime}(w_{ij})$ is a constant and the boundedness conditions naturally hold.
For popular nonlinear models such as logistic factor models, probit factor models, and Poisson factor models, the boundedness of $l_{ij}^{\prime\prime}(w_{ij})$ and $l_{ij}^{\prime\prime\prime}(w_{ij})$ can also be easily verified.

\begin{assumption}\label{assumption: Scaling}
    For $\xi$  specified in Assumption 2 and a sufficiently small $\epsilon >0$, 
we assume that as $n,q,p \rightarrow \infty$,
    \begin{equation}
    \frac{p}{\sqrt{n \wedge(pq)}}   (nq)^{\epsilon+3/\xi}\to 0.
    \end{equation}
\end{assumption}

Assumption~\ref{assumption: Scaling} is imposed to ensure that the derivative of the likelihood function equals zero at the maximum likelihood estimator with high probability, a key property in the theoretical analysis. 
In particular, we need the estimation errors of all model parameters to converge to 0 uniformly with high probability. Such uniform convergence results require delicate analysis of the convexity of the objective function, for which technically we need Assumption~\ref{assumption: Scaling}. For most of the popularly used generalized factor models, $\xi$ can be taken as any large value as discussed above,  thus
$(nq)^{\epsilon + 3/\xi}$ is of a smaller order of $\sqrt{n \wedge (pq)}$, given a small $\epsilon$.
Specifically, Assumption~\ref{assumption: Scaling} implies $p=o(n^{1/2}\wedge q)$ up to a small order term, an asymptotic regime that is reasonable for many educational assessment data. 
 
\begin{assumption}
\label{assumption:asymptotic normality}
(i) Define $\bZ_i^* = \big((\bU_i^*)^\T,\bX_i^\T\big)^\T$. For any $j \in [q]$, 
$\lim_{n\rightarrow\infty} \| -n^{-1} \sum_{i=1}^n \EE $ $ l_{ij}^{\prime\prime}(w_{ij}^*) \bZ_i^* (\bZ_i^{*})^{\T} - \bPhi_{jz}^* \|_F = 0 $ with $\bPhi_{jz}^*$ positive definite.
(ii) For any $i \in [n]$, $-q^{-1} \sum_{j=1}^q $ $ \EE l_{ij}^{\prime\prime}(w_{ij}^*) \bgamma_j^{*} (\bgamma_j^{*})^\T $ $ {\rightarrow} \bPhi_{i\gamma}^*$ for some positive definite matrix $\bPhi_{i\gamma}^*$.\end{assumption}

 Assumption~\ref{assumption:asymptotic normality} assumes the existence of the asymptotic covariance matrices $\bPhi_{jz}^*$  and $\bPhi_{i\gamma}^*$, which are used to derive the asymptotic covariance matrices for the MLE estimators. 
For popular generalized factor models, this assumption holds under mild conditions. For example, under linear models, $l_{ij}^{\prime\prime}(w_{ij})$ is a constant. Then $\bPhi_{jz}^*$ and $\bPhi_{i\gamma}^*$ naturally exist and are positive definite from Assumption~\ref{assumption: psd covariance}. 
Under logistic and probit models,  $l_{ij}^{\prime\prime}(w_{ij})$ is finite within a compact parameters space, and similar arguments can be applied to show the validity of Assumption~\ref{assumption:asymptotic normality}.

The next assumption relates to handling the dependence between the covariates and latent factors and plays an important role in deriving asymptotic distributions for the proposed estimators. Before we formally introduce the assumption, we first define some notations. Let ${\Gb}^{\ddagger} = (q^{-1}(\bGamma^{*})^{\T} \bGamma^*)^{1/2}$ $ \cV^* (\cU^{*})^{-1/4}$ and $\Ab^{\ddagger} = (\Ub^{*})^{\T}\Xb(\Xb^\T\Xb)^{-1}$, where $\cU^* = \diag(\varrho_1^*, \dots, \varrho_K^*)$ with diagonal elements being the $K$ eigenvalues of 
$(nq)^{-1} ((\bGamma^{*})^{\T} \bGamma^* )^{1/2} (\Ub^{*})^{\T} $ $ (\Ib_{n} - \Pb_x) \Ub^*  ((\bGamma^{*})^{\T} \bGamma^* )^{1/2} $
with $\Pb_x = \Xb (\Xb^{\T} \Xb)^{-1}\Xb^{\T}$ and $\cV^*$ containing the matrix of corresponding eigenvectors. 
We further define $\Ab^0=(\Gb^{\ddagger})^\T\Ab^{\ddagger}= (\ba_0^{0}, \dots, \ba_{p_*}^{0})$. Let $\bar\Ab^{\ddagger} = \bSigma_{ux}^*(\bSigma_{x})^{-1}$ and $\bar\Gb^{\ddagger}$ be the probability limit of $\Gb^{\ddagger}$, as $n,q,p \rightarrow \infty$.

The estimation problem~(6) in the main text is related to the median regression problem with measurement errors. To understand the properties of this estimator, following existing M-estimation literature~\citep{he1996bahadur, he2000parameters}, we define $\psi_{js}^0(\ba)=(\Gb^\ddagger)^{-1}\bgamma_j^*\sign\{\beta_{js}^*+\ba^\T(\Gb^\ddagger)^{-1}\bgamma_j^*\}$ and
    $\chi_{s}(\ba)=\sum_{j=1}^q\psi_{js}^0(\ba)$ for $j \in [q]$ and $s \in [p_*]$. For $s\in[p_*]$, we define a perturbed version of $\psi_{js}^0(\ba)$, denoted as $\psi_{js}(\ba,\bdelta_{js})$, 
as follows:
     \begin{align}
        \psi_{js}(\ba,\bdelta_{js})=(\Gb^\ddagger)^{-1}\Big(\bgamma_j^*+\frac{[\bdelta_{js}]_{[1:K]}}{\sqrt{n}}\Big) \sign\Big\{\beta_{js}^*+\frac{[\bdelta_{js}]_{K+1}}{\sqrt{n}}+\ba^\T(\Gb^\ddagger)^{-1}\big(\bgamma_j^*+\frac{[\bdelta_{js}]_{[1:K]}}{\sqrt{n}}\big)\Big\},
        \label{eq:phi js def}
    \end{align}
    where the perturbation
\begin{align*}
        \bdelta_{js}&=\begin{pmatrix}
            \Ib_K& \bm{0}\\\zero&(\one_s^{(p)})^\T
\end{pmatrix}\Big(-\sum_{i=1}^nl_{ij}^{\prime\prime}(w_{ij}^*)\bZ_i^*(\bZ_i^*)^\T\Big)^{-1}\Big(\sqrt{n}\sum_{i=1}^nl_{ij}^{\prime}(w_{ij}^*)\bZ_i^*\Big),
    \end{align*}  
    follows asymptotically normal distribution by verifying Lindeberg-Feller condition. 
    Define $\hat\chi_s(\ba)=\sum_{j=1}^q\EE\psi_{js}(\ba,\bdelta_{js})$.

\begin{assumption}
    \label{assumption:A consistency}  
    For $\chi_s(\ba)$, we assume that there exists some constant $c>0$ such that $\min_{\ba\neq \zero}|q^{-1}\chi_s(\ba)|> c$ holds for all $s\in[p_*]$. Assume there exists $\ba_{s0}$ for each $s\in[p_*]$ such that $\hat\chi_{s}(\ba_{s0})$ $= 0$ with $p\sqrt{n}\|\balpha_{s0}\|\to0$. In a neighbourhood of $\balpha_{s0}$, $\hat\chi_{s}(\ba)$ has a nonsingular derivative such that $\{q^{-1}\nabla_{\ba}\hat\chi_{s}(\balpha_{s0})\}^{-1}=O(1)$ and $q^{-1}|\nabla_{\ba}\hat\chi_{s}(\ba)-\nabla_{\ba}\hat\chi_{s}(\balpha_{s0})|\le k|\ba-\balpha_{s0}|$. We assume $\iota_{nq,p}:=\max \big\{\|\balpha_{s0}\|,q^{-1}\sum_{j=1}^q\psi_{js}(\ba_{s0},\bdelta_{js})\big\}=o \big((p\sqrt{n})^{-1}\big)$. 
\end{assumption}

Assumption~\ref{assumption:A consistency} is required to address the theoretical difficulties in establishing the consistent estimation for $\Ab^0$, a challenging problem related to median regression with weakly dependent measurement errors. 
In Assumption~\ref{assumption:A consistency}, we treat the minimizer of $|\sum_{j=1}^q\EE \psi_{js}(\ba,\bdelta_{js})|$ as an $M$-estimator and adopt the Bahadur representation results in \cite{he1996bahadur} for the theoretical analysis. 
For an ideal case where $\bdelta_{js}$ are independent and normally distributed with finite variances, which corresponds to the setting in median regression with measurement errors \citep{he2000quantile}, these assumptions can be easily verified.
Assumption~\ref{assumption:A consistency} discusses beyond such an ideal case and covers general settings. 
In addition to independent and Gaussian measurement errors, this condition also accommodates the case when $\bdelta_{js}$ are asymptotically normal and weakly dependent with finite variances, as implied by Lindeberg-Feller central limit theorem and the conditional independence of $Y_{ij}$.

We want to emphasize that Assumption~\ref{assumption:A consistency} allows for both sparse and dense settings of the covariate effects. Consider an example of $K=p=1$ and $\gamma_j=1$ for $j\in [q]$. Suppose $\beta_{js}^*$ is zero for all $j \in [q_1]$ and nonzero otherwise. 
Then this condition is satisfied as long as $\#\{j: \beta_{js}^* > 0\}$ and $\#\{j: \beta_{js}^* < 0\}$ are comparable, even when the sparsity level $q_1$ is small.

\section{Proofs of Propositions 1--3 in Main Text}\label{sec:supp prove mainprop}
\subsection{Proof of Proposition~1}



For $s\in [p_*]$, let $L_s(\ba)=\sum_{j=1}^q|\beta_{js}-\ba^\T\bgamma_j|$. We denote the directional gradient of $L_s$ as 
\begin{equation*}
    \nabla_{\bv}L_s(\ba)=\sum_{j=1}^q|\bv^\T\bgamma_j|I{(\beta_{js}=0)}+\sum_{j=1}^q\text{sign}(\beta_{js})\bv^\T\bgamma_jI{(\beta_{js}\neq 0)}.
\end{equation*}

   \textbf{ `if' part:} suppose~(4)  holds and for some $\Ab_c\neq \zero$ we have $\sum_{j=1}^q\|\bbeta_{jc}\|_1\ge \sum_{j=1}^q\|\bbeta_{jc}-\Ab^\T_c\bgamma_j\|_1$. Then there exists some $s\in [p_*]$ such that $\sum_{j=1}^q|\beta_{js}|\ge \sum_{j=1}^q|\beta_{js}-\ba_s^\T\bgamma_j|$ and $\ba_s\neq \zero$, where $\ba_s$ is the $s$-th column of $\Ab_c$. Then we know by the convexity of $L_s$ that $\nabla_{\ba_s}L_s(\zero)\le0$, which contradicts (4) when taking $\bv$ as $-\ba_s$.

    \textbf{`only if' part:} suppose 
    (4) fails to hold for some $\bv^*$ and $s\in [p_*]$, which implies that $\nabla_{\bv^*}L_s(\zero)\le0$. We claim that there exists some $\varepsilon$ such that $L_s(\bv^*\varepsilon)\le L_s(\zero)$. Note that we can find small $\varepsilon$ such that when on the segment from $\zero$ to $\bv^*\varepsilon$, each term in $L_s(\ba)$ preserves its sign, and thus we always have $\nabla_{\bv^*}L(\ba)\le 0$ when $\ba$ moves from $\zero$ to $\bv^*\varepsilon$.
    Taking the $s$-th column of $\Ab_c$ as $\bv^*\varepsilon$ and the other columns as $\zero$s, then we have $\sum_{j=1}^q\|\bbeta_{jc}\|_1\ge \sum_{j=1}^q\|\bbeta_{jc}-\Ab_c^\T\bgamma_j\|_1$, contradicting Condition 1(ii).

\subsection{Proof of Proposition~2}

We first prove $\Ab = \bm{0}_{K\times p}$ when both $(\bGamma,\Ub,\Bb)$ and $(\tilde\bGamma,\tilde\Ub,\tilde\Bb)$ satisfy Conditions~1. 
By Condition~1(i) we know $\zero_K = n^{-1}\one_n^\T(\Ub+\Xb\Ab^\T)\Gb^\T = n^{-1}\one^\T\Xb\Ab^\T\Gb^\T$ as $\Ub$ is centered. Since the first column of $\Xb$ is all $1$s and other columns have been centered, we know from $\Gb$ is invertible that $\Ab_{[,1]}=\zero_K$. Suppose $\Ab_c = \Ab_{[,2:p]}\neq \zero$. Since  $(\bGamma,\Ub,\Bb)$ satisfies Condition~1(ii), we know
\begin{equation*}
    \sum_{j=1}^q\|\bbeta_{jc}\|_1<\sum_{j=1}^q\|\bbeta_{jc} - \Ab_c^\T\bgamma_j\|_1.
\end{equation*} But $\big(\bGamma\Gb^{-1},(\Ub+\Xb\Ab^\T)\Gb^\T,\Bb-\bGamma\Ab\big)$ also satisfies Condition~1(ii), which implies
\begin{align*}
    \sum_{j=1}^q\big\|[\bbeta_{j} - \Ab^\T\bgamma_j]_{2:p}\|_1=&\,\sum_{j=1}^q\big\|\bbeta_{jc} - \Ab_c^\T\bgamma_j\|_1\\<&\sum_{j=1}^q\big\|\bbeta_{jc} - \Ab_c^\T\bgamma_j -\big(-\Ab_c^\T\Gb\big)\Gb^{-\T}\bgamma_j\|_1\\=&\,\,\sum_{j=1}^q\|\bbeta_{jc}\|_1,
\end{align*}
which leads to a contradiction. Next we show that $\Gb$ can at most be a signed permutation. Since $n^{-1}\Ub^\T\Ub = q^{-1}\bGamma^\T\bGamma$ is diagonal with distinct elements, denoted by $\Db=\text{diag}(d_1,\cdots,d_K)$, we have by $\Ab =\zero$ that
\begin{equation*}
    \Gb\Db\Gb^\T = \Gb^{-\T}\Db\Gb^{-1} =\tilde\Db.
\end{equation*}Then $\Gb^\T\Gb\Db\Gb^\T\Gb=\Db$ with $\Db$ having distinct elements. Then $\Gb^\T\Gb=\Ib_K$. Since $\Db\Gb^\T = \tilde\Db^\T\tilde\Db$, we know that the rows of $\Gb$ are the eigenvectors of $\Db$ by definition. Since $\Db$ is diagonal with distinct elements, the eigenvectors can be uniquely given as $\pm \be_r$ where $\be_r$ is indicator vector. Taking these together, $\Gb$ can only be a signed permutation matrix.

\subsection{Proof of Proposition~3}
Recall that we define $\tilde\Ab = (\tilde\ba_0,\tilde\Ab_c)$ with $\tilde\Ab_c =$ $  \arg\min_{\Ab\in\RR^{K\times p_*}}$ $\sum_{j=1}^q$ $\|\tilde\bbeta_{jc}-\Ab^\T\tilde\bgamma_j\|_1$ and $\tilde\ba_0 = -n^{-1}\sum_{i=1}^n\tilde\bU_i $.
Following the definitions of $\tilde\bbeta_{jc} = \bbeta_{jc}-\Ab_c^\T\bgamma_j $ and $\tilde\bgamma_j = (\Gb)^{-1}\bgamma_j$ as in~(3) of main text, we write
\begin{align*}
    \argmin_{\Ab\in\RR^{K\times p_*}}\sum_{j=1}^q\|\tilde\bbeta_{jc} - \Ab^\T\tilde\bgamma_j\|=\argmin_{\Ab\in\RR^{K\times p_*}}\sum_{j=1}^q\|\bbeta_{jc}-\Ab_c^\T\bgamma_j - \Ab^\T(\Gb)^{-1}\bgamma_j\|.
\end{align*}
Here we see the solution to above problem is $-(\Gb)^\T \Ab_c$, and is unique as $(\bGamma,\Ub,\Bb)$ satisfy Condition~1. Next it is easy to verify that $-n^{-1}\sum_{i=1}^n\tilde\Ub_i = -n^{-1}\sum_{i=1}^n\big(\Gb^\T\bU_i + \Gb^\T\Ab\bX_i \big) = -\Gb^\T\ba_0$ as the first element in $\bX_i$ is $1$. For $\tilde\Gb$, we check that 
\begin{equation*}
    q^{-1}(\tilde\Gb)^{-1}\tilde\bGamma^\T \tilde\bGamma (\tilde\Gb)^{-\T} =\tilde\cU^{1/2} = n^{-1}\tilde\Gb^\T (\tilde\Ub+\Xb\tilde\Ab^\T)^\T(\tilde\Ub+\Xb\tilde\Ab^\T)\tilde\Gb.
\end{equation*}
Since the diagonal elements of $q^{-1}\bGamma^\T\bGamma = n^{-1}\Ub^\T\Ub$ are distinct, the eigenvalues of
\begin{equation*}
    (nq)^{-1} (\tilde\bGamma^{\T} \tilde\bGamma )(\tilde\Ub + \Xb \tilde\Ab^{\T})^{\T} (\tilde\Ub + \Xb \tilde\Ab^{\T})= (nq)^{-1} \Gb^{-1}(\bGamma^{\T} \bGamma )\Ub\Ub^\T\Gb
\end{equation*}
are also distinct. Therefore $\tilde\Gb$ is unique up to a signed permutation.

\section{Proofs of Main Theorems}\label{sec:prove main}
In this section, we start with an overview of the proof strategy of the main results in the main text. Then we provide detailed proofs of Theorems~4.1 and~4.2 in Section~\ref{sec:prove_thm1} and \ref{sec:prove_thm2}, respectively. In Section~\ref{sec:estimate_variance}, we present consistent estimation for the asymptotic covariance matrices defined in Theorem~4.2. Then in Section~\ref{sec:consistent estimator cov} we prove Corollary~1.

 Throughout this subsequent proof, for the statement convenience, we slightly adjust the notation and re-index the covariate effects as $\bbeta_{j}=(\beta_{j1},\cdots,\beta_{jp})^\T$ with the first component $\beta_{j1}$ being the intercept and others $(\beta_{j2}, \dots, \beta_{jp})$ being the covariate effects, namely DIF effects in psychometrics~\citep{holland2012differential}.
Accordingly, we clarify that $\Xb\in\RR^{n\times p}$ is a design matrix with the first column being all ones and that each $\bX_i \in \RR^{p} $ incorporates entry-one in the vector. 
We further define $\bZ_i=(\bU_i^\T,\bX_i^{\T})^\T$, $\boldsymbol{f}_j=(\bgamma_j^\T,\bbeta_j^\T)^\T$, and $w_{ij}=\bgamma_j^\T\bU_i+\bbeta_j^\T\bX_i=\boldsymbol{f}_j^\T\bZ_i$ for $i \in [n]$ and $j \in [q]$. Following the above introduced notations, we let $\bphi = (\bbf_v^\T,\bU_v^\T)^\T$ be the assembled vector of all parameters, where $\bbf_v=(\bbf_1^\T,...,\bbf_q^\T)^\T$ and $\bU_v=(\bU_1^\T,...,\bU_n^\T)^\T$ are the assembled vectors of all $\bbf_j$'s and all $\bU_i$'s, respectively.
The true parameters that satisfy our proposed identifiability conditions~1--2 are denoted as $\bphi^*=\{(\bbf_v^*)^\T,(\bU_v^*)^\T\}^\T$. We let $\mathcal{B}(D)$ as the parameter space where $\max_{i\in[n]}\|\bU_i\|_{\infty} \le D$, $\max_{j\in[q]}\| \bgamma_j\|_{\infty}\le D$, and $\max_{j\in[q]}\|\bbeta_j\|_{\infty} \le D$ for some large $D$ such that the true parameters $\bphi^*$ lie in $\cB(D)$.

\subsection{Proof Strategy and Framework}\label{sec:proof_main_strategy}
In this section, we outline the proof strategy for the theorems in the main text. 
Directly establishing the theoretical properties of the estimator $\hat{\bphi}^*$ is challenging due to the potentially complicated correlations between the latent variables $\bU_i^*$ and the observed covariates $\bX_i$.
To handle this relationship, we employ the following transformed $\Ub^0$ that are orthogonal with $\Xb$, which plays an important role in establishing the theoretical results. In particular, we let $\Ub^0 = \big(\Ub^* -   \bX(\Ab^{\ddagger})^\T\big)\Gb^{\ddagger}$, $\bGamma^{0} =  \bGamma^*(\Gb^{\ddagger})^{-\T}$, and $\Bb^0 = \Bb^* + \bGamma^*\Ab^{\ddagger}$, where $\Ab^\ddagger$ and ${\Gb}^{\ddagger}$ are defined before Assumption \ref{assumption:A consistency}.
We let $\bGamma^0=(\bgamma_1^0,\cdots,\bgamma_q^0)^\T$, $\Ub^0=(\bU_1^0,\cdots,\bU_n^0)^\T$, and $\Bb^0=(\bbeta_1^0,\cdots,\bbeta_q^0)^\T$. 
We also write $\bphi^0=\{(\bbf^0_v)^\T,(\bU_v^0)^\T\}^\T$, where $\bbf_v^0=\{(\bbf_1^0)^\T,\cdots,(\bbf_q^0)^\T\}^\T$ with $\bbf_j^0=\{(\bgamma_j^0)^\T,(\bbeta_j^0)^\T\}^\T$ and $\bU_v^0=\{(\bU_1^0)^\T,\cdots,(\bU_n^0)^\T\}^\T$. Write $w_{ij}^0 = (\bgamma_j^0)^\T\bU_i^0 + (\bbeta_j^0)^\T\bX_i$. It can be verified that $w_{ij}^0 = w_{ij}^*$ for all $i\in[n],j\in[q]$.

These transformed parameters $\bphi^0$ give the same joint likelihood as that of the true parameters $\bphi^*$, which facilitate our theoretical understanding of the joint-likelihood-based estimators. 
In particular, these transformed parameters $\bphi^0$ can be readily shown to satisfy the following identifiability conditions:
\begin{itemize}\item[] Condition $1^{\prime}$: $(\Ub^0)^{\T} \Xb = \bm{0}_{K\times p}$. \item[] Condition $2^{\prime}$: $n^{-1}(\Ub^0)^\T\Ub^0=q^{-1}(\bGamma^0)^\T(\bGamma^0)=$ diagonal with distinct and nonzero elements.\end{itemize} 
Similarly, we can also establish the transformation from  $\bphi^0$ to the true parameters $\bphi^*$.
Specifically, by Proposition~3, we have
$\Bb^*  = \Bb^0 - \bGamma^0 \Ab^0,
        \bGamma^* = \bGamma^0 (\Gb^{0\T})^{-1}$, and $ 
        \Ub^* = (\Ub^0 + \Xb \Ab^{0\T}) \Gb^0$, where $\Gb^0 =  (q^{-1}\bGamma^{0\T} \bGamma^0)^{1/2} \cV^0 (\cU^0)^{-1/4}$ and $\Ab^0=(\zero_K,\Ab^0_c)$ with
 $\Ab^0_c = \argmin_{\Ab\in\RR^{K\times(p-1)}}~q^{-1} \sum_{j=1}^q \|\bbeta_{jc}^0 - \Ab^{\T}\bgamma_j^{0}  \|_{1}$. Here the equivalences for factors and loadings are up to signed column permutation; ${\cV^0}$ and ${\cU}^0$ contain the eigenvalues and eigenvectors of $(nq)^{-1} \big(({\bGamma}^{0})^{\T} {\bGamma}^0 \big)^{1/2}\big(({\Ub}^{0})^{\T} +{\Ab}^{0} \Xb^{\T}\big) \big({\Ub}^{0} + \Xb({\Ab}^{0})^\T \big) $ $\big(({\bGamma}^{0})^{\T} {\bGamma}^0 \big)^{1/2}$, respectively. Note that $\Ab^0 = (\Gb^{\ddagger})^\T\Ab^{\ddagger}$ since that the first column of $\Ab^\ddagger$ is $\zero_K$ and that {$q^{-1} \sum_{j=1}^q\big\|[\bbeta_j^0]_{[2:p]} - \Ab_{[,2:p]}^{\T}\bgamma_j^{0}  \big\|_{1}  = q^{-1} \sum_{j=1}^q\|\bbeta_j^* + \{(\Ab^{\ddagger}_{[,2:p]})^{\T} - \Ab_{[2:p]}^{\T} (\Gb^{\ddagger})^{-1}\} (\bgamma_j^*)^{\T}  \|_{1}$} is minimized when $\Ab_{[,2:p]} = (\Gb^{\ddagger})^\T\Ab^{\ddagger}_{[,2:p]}$, as a result of Condition 1 in the main text. Also, it can be readily verified that $\Gb^0\Gb^\ddagger=\Ib_K$.

\begin{remark}
 We verify an analogous version of Assumption~\ref{assumption: psd covariance} for $\bphi^0$. Define $\bar\Gb^{\ddagger}$ as the probability limit of $\Gb^{\ddagger}$. By Assumption~\ref{assumption: psd covariance}, we know $\lim_{q\to\infty}q^{-1}(\bGamma^0)^\T\bGamma^0=(\bar\Gb^{\ddagger})^{-1}\bSigma^*_{\gamma}(\bar\Gb^{\ddagger})^{-\T}$ is positive definite, which we denote by $\bSigma_{\gamma}^0$, and 
\begin{equation*}
    \lim_{n\to\infty}n^{-1}(\Zb^0)^\T\Zb^0=\begin{pmatrix}
        \bSigma_{u}^*-\bSigma_{ux}^*\bSigma_{x}^{-1}\bSigma_{xu}^*& \zero \\ \zero &\bSigma_x
    \end{pmatrix}
\end{equation*}
is also positive definite, which we denote by $\bSigma^0_z$. Then we also have $1/(\kappa^{\prime})^2\le \lambda_{\min}(\bSigma_z^0)\le \lambda_{\max}(\bSigma_z^0)\le (\kappa^{\prime})^2$ for some $\kappa^{\prime}>0$. We claim that for some $M^{\prime}>0$, $\max_{i\in[n]}\|\bZ_i^0\|_{\infty}\le M^{\prime}$ and $\max_{j\in[q]}\|\bgamma_j^0\|_{\infty}\le M^{\prime}$. 
The boundedness for $\bgamma_{j}^0$ can be easily verified since $\Gb^{\ddagger}$ converges and is of finite dimension. For $\bU_i^0=(\Gb^{\ddagger})^\T(\bU_i^*-\Ab^\ddagger \bX_i)$,  since $\lim_{n\to\infty}\Ab^\ddagger=\bSigma_{ux}^*\bSigma_{x}^{-1}$ with $\|\Ab^\ddagger \bX_i\|_{\infty}\le \|\Ab^\ddagger\|_{\infty}\|\bX_i\|_{\infty}\le M^2$ by Assumption~\ref{assumption: psd covariance}, we can show $\max_{i\in[n]}$ $\|\bU_i^0\|_{\infty} \le M^{\prime}$. Also, we know that $\max_{i\in[n],j\in[q]}\big|w_{ij}^0\big|= \max_{i\in[n],j\in[q]}\big|w_{ij}^*\big|\le M^2+M$. For simplicity we enlarge $M$ and still write $\max_{i\in[n]}\|\bU_i^0\|_{\infty}\le M$, $\max_{j\in[q]}\|\bgamma_j^0\|_{\infty}\le M$ and $\max_{i\in[n],j\in[q]}|w_{ij}^0|\le M$.
\end{remark}

Following the above discussion, we study the asymptotic properties of $\hat{\bphi}^*$ by first examining the properties of the maximum likelihood estimator of ${\bphi}^0$, which technically is more feasible due to the orthogonality of the transformed latent factors $\Ub^0$ and the covariates $\Xb$ (Condition $1^\prime$). Then we use the obtained estimator for  $\bphi^0$ to construct the estimators for the transformations $\Ab^0$ and $\Gb^0$, by which we are able to further establish the asymptotic consistency and normality of $\hat{\bphi}^*$. 
In particular, the maximum likelihood estimator for $\bphi^0$ under the identifiability conditions $1^\prime$--$2^\prime$, denoted as $\hat\bphi^0$, can be obtained as follows
\begin{align}
    	\hat{\bphi}^0  = &\mathop{\arg\min}_{\bphi\in\cB(D)}~-L(\Yb |  \bphi), \label{eq:cmle st 12 prime}
     \\& \text{subject to }
     \Ub^{\T} \Xb = \bm{0}_{K\times p} \nonumber  \\
     &\qquad\text{  and }
     n^{-1}\Ub^\T\Ub=q^{-1}\bGamma^\T\bGamma= \text{diagonal with distinct and nonzero elements.}
     \nonumber
\end{align}
For any maximum likelihood estimator $\hat\bphi=(\hat\Ub,\hat\bGamma,\hat\Bb)$ defined in \eqref{eq:MLE}, let $\hat\Ab = \hat\Ub^\T\Xb(\Xb^\T\Xb)^{-1}$ and $\hat\Gb = (q^{-1}\hat\bGamma^{\T} \hat\bGamma)^{1/2}$ $ \hat\cV \hat\cU^{-1/4}$, where we compute the singular value decomposition of $(nq)^{-1} (\hat{\bGamma}^{\T} \hat{\bGamma} )^{1/2}(\hat\Ub -\Xb\hat\Ab^\T)^\T(\hat\Ub -\Xb\hat\Ab^\T)(\hat{\bGamma}^{\T} \hat{\bGamma} )^{1/2}$ and let $\hat\cU = \diag(\hat\varrho_1, \dots, \hat\varrho_K)$ be a diagonal matrix that contains the eigenvalues and $\hat\cV$ be a matrix that contains its corresponding eigenvectors. Then $\hat\bphi^{\prime}=(\hat\bGamma(\hat\Gb)^{-\T},(\hat\Ub -\Xb\hat\Ab^\T)\hat\Gb,\hat\Bb+\hat\bGamma\hat\Ab)$ satisfy Conditions 1$^{\prime}$ and 2$^{\prime}$. We can select $D$ large enough to include this estimation, and therefore, this is a solution to \eqref{eq:cmle st 12 prime} as it maximizes the joint likelihood function by definition.

With $\hat\bphi^0$ obtained in \eqref{eq:cmle st 12 prime}, we construct estimators for the transformed parameters 
as follows: $\tilde\Bb^*  = \hat\Bb^0 - \hat\bGamma^0 \hat\Ab^0,
        \tilde\bGamma^* = \hat\bGamma^0 (\hat\Gb^{0\T})^{-1}$, and $ 
        \tilde\Ub^* = (\hat\Ub^0 + \Xb \hat\Ab^{0\T}) {\hat\Gb^0}$, where the transformation matrices are given as $\hat\Ab^0=(\zero_K,\hat\Ab_c^0)$ with $\hat\Ab_c^0 = \arg\min_{\Ab}~q^{-1} \sum_{j=1}^q \|\hat{\bbeta}_{jc}^0 -  \Ab^{\T} \hat{\bgamma}_j^{0} \|_{1}$ and $\hat{\Gb}^0 = (q^{-1}(\hat{\bGamma}^0)^{\T} \hat{\bGamma}^0)^{1/2}$ $ \hat{\cV}^0 (\hat{\cU}^0)^{-1/4}$. Here ${\hat\cV^0}$ and ${\hat\cU}^0$ contain the eigenvalues and eigenvectors of $(nq)^{-1} \big(({\hat\bGamma}^{0})^{\T} {\hat\bGamma}^0 \big)^{1/2}\big(({\hat\Ub}^{0})^{\T} +{\hat\Ab}^{0} \Xb^{\T}\big) \big({\hat\Ub}^{0} + \Xb({\hat\Ab}^{0})^\T \big) $ $\big(({\hat\bGamma}^{0})^{\T} {\hat\bGamma}^0 \big)^{1/2}$, respectively. 

By the definition of $\hat\Gb^0$, we know that $\tilde\bphi^* := (\tilde\bGamma^*,\tilde\Ub^*,\tilde\Bb^* )$ satisfy Condition~2. Next we know from the proof of Lemma~\ref{prop_consistency_AG} that $\hat\Ab^0$ is unique w.h.p., which implies that $\tilde\bphi^*$ also satisfy Condition~1 w.h.p.. By Proposition~3, and that $\hat\bphi^0$ is equivalent to $\tilde\bphi^*$ and $\hat\bphi^*$ up to a set of (different) linear transformations, we know that $\tilde\bphi^*$ is asymptotic equivalent to $\hat\bphi^*$ obtained in Section~3.2 of main text.


        
To summarize, our main results for $\hat\bphi^*$ will be established by showing the consistency and asymptotic normality of the estimation $\hat\bphi^0$ for $\bphi^0$ given in \eqref{eq:cmle st 12 prime} together with the consistency of the estimated transformation matrices $\hat\Ab^0$ and $\hat\Gb^0$. 
       This strategy is summarized in the following and illustrated in the subsequent figure.
        \begin{enumerate}
            \item[(I)] Establish the consistency and asymptotic normality for $\hat\bphi^0$ in Lemmas~\ref{prop:average consistency} and \ref{thm:asymptotic normality}.
            \item[(II)] 
            Show the transformation matrices from $\bphi^0$ to $\bphi^*$, $\Ab^0$ and $\Gb^0$, can be consistently estimated by $\hat{\Ab}^0$ and $\hat{\Gb}^0$ in Lemma~\ref{prop_consistency_AG}. 
            The transformed estimators $\tilde\Bb^*$ $\tilde\bGamma^*$ and $\tilde\Ub^*$ are asymptotically equivalent to the target joint maximum likelihood estimators $\hat\Bb^*$, $\hat\bGamma^*$, and $\hat\Ub^*$.
            \item[(III)]  Prove the estimation consistency and asymptotic normality for $\hat\bphi^*$ in Section~\ref{sec:prove_thm2} with the results for $\hat{\bphi}^0$ in (I) and for $\hat{\Ab}^0$ and $\hat{\Gb}^0$ in (II). 
        \end{enumerate}
\begin{center}
\begin{tikzpicture}[scale=2, every node/.style={font=\large}]  
    \node (phihat0) {$\hat{\boldsymbol{\phi}}^0$};
    \node[right=5cm of phihat0] (phihatstar) {$\tilde{\boldsymbol{\phi}}^*$};  
    \node[above=4cm of phihat0] (phi0) {${\boldsymbol{\phi}}^0$};  
        \node[right=2cm of phihatstar](asyhatphistar){$\hat\bphi^*$};
    \node[above=4cm of asyhatphistar] (phistar) {${\boldsymbol{\phi}}^*$};  
    
    \draw[->, line width=1mm] (phihat0) -- (phihatstar) node[midway, xshift = 0.35cm, above] {$\hat{\mathbf{A}}^0, \hat{\mathbf{G}}^0$};  
    \draw[<->, line width=1mm] (phihatstar) -- (asyhatphistar) node[midway, xshift = 0cm, above,font=\small] {asymp.};  
    \draw[->, line width=1mm] (phihat0) -- (phi0) node[midway, left, rotate=90, left=0.5cm, xshift=1cm] {{estimate}};  // Thicker arrows
    \draw[->, line width=1mm] (asyhatphistar) -- (phistar) node[midway, right, rotate=90, left=-0.5cm, xshift=1.5cm] {estimate (goal)};  // Thicker arrows
    \draw[->, line width=1mm] (phi0) -- (phistar) node[midway, xshift = -0.6cm, below] {${\mathbf{A}}^0, {\mathbf{G}}^0$};  // Thicker arrows
\draw[->, bend left=35, looseness=1.2, line width=0.75mm] ([yshift=0.2cm,  xshift=1cm]phihat0.east) to ([yshift=-0.2cm,  xshift=1cm]phi0.east) node[midway, right, yshift = 1cm, xshift=1cm,font=\tiny] {estimate};

\end{tikzpicture}
\end{center}
\begin{remark}
    As discussed previously, Condition 1$^{\prime}$ and 2$^{\prime}$ are easier to work with because for any $(\bGamma,\Ub,\Bb)$, we can find transformation matrices $\Ab$ and $\Gb$ such that $(\bGamma(\Gb^\T)^{-1},(\Ub - \Xb\Ab^\T)\Gb,\Bb+\bGamma\Ab)$ satisfy Condition 1$^{\prime}$ and 2$^{\prime}$.
\end{remark}
\subsection{Proof of Theorem~4.1}\label{sec:prove_thm1}

Based on $\hat\bphi^0$ defined in~\eqref{eq:cmle st 12 prime}, we have constructed the estimators as follows
\begin{align}
        \tilde\Bb^*  = \hat\Bb^0 - \hat\bGamma^0 \hat\Ab^0 ; \quad
        \tilde\bGamma^* = \hat\bGamma^0  (\hat\Gb^{0})^{-\T}; \quad
        \tilde\Ub^* = \{\hat\Ub^0  + \Xb (\hat\Ab^{0})^{\T}\} \hat\Gb^0 ,\label{eq:Uhat to Ustar}
\end{align}
which are asymptotically equivalent to the joint maximum likelihood estimators $\hat\Bb^*$, $\hat\bGamma^*$, and $\hat\Ub^*$. Therefore we have
     \begin{align}
          \hat{\Bb}^* - \Bb^* &= \hat{\Bb}^0 - \Bb^0 - (\hat{\bGamma}^0 \hat{\Ab}^0- \bGamma^0 \Ab^0) \nonumber  \\
          &= \hat{\Bb}^0 - \Bb^0 - (\hat{\bGamma}^0  -\bGamma^0) \hat{\Ab}^0 - \bGamma^0 (\hat{\Ab}^0 -\Ab^0 ). \label{eq:Bstar decompose}
     \end{align}

To continue the derivation of estimation error bound for $\hat{\Bb}^*$, we introduce the following lemmas related to the estimation consistency of $\hat{\bphi}^0$ and transformation matrices $\hat{\Ab}^0$ and $\hat{\Gb}^0$. The detailed proof and discussion for these lemmas are provided in Section~\ref{sec:proveadd_prop}.

\noindent {\bf Lemma~\ref{prop:average consistency}.}~(Average Consistency)~Under Assumptions~\ref{assumption: psd covariance}--\ref{assumption:smoothness} and $p=o(n\wedge q)$, $n^{-1}\|\hat{\Ub}^0 - \Ub^0\|_F^2 = O_p(\zeta_{nq,p}^{-2})$, $q^{-1} \|\hat{\bGamma}^0 - \bGamma^0 \|_F^2 = O_p(\zeta_{nq,p}^{-2})$, and  $q^{-1} \|\hat{\Bb}^0 - \Bb^0 \|_F^2 = O_p(\zeta_{nq,p}^{-2})$.

\noindent {\bf Lemma~\ref{prop:ind consistency}.}~(Individual Consistency)~Under Assumptions~\ref{assumption: psd covariance}--\ref{assumption: Scaling}, we have for any $j\in[q]$
    \begin{equation*}
        \|\hat {\bgamma}_j^0-\bgamma_j^0\|_2=O_p\big(\zeta_{nq,p}^{-1}\big),\;\|\hat{\bbeta}_j^0 -\bbeta_j^0\|_2=O_p\big(\zeta_{nq,p}^{-1}\big),
    \end{equation*}
    and for any $i\in[n]$
    \begin{equation*}
        \|\hat\bU_i^0-\bU_i^0\|_2=O_p({\sqrt{p}}{\zeta_{nq,p}^{-1}}).
    \end{equation*}

\noindent {\bf Lemma~\ref{prop_consistency_AG}.}~(Consistency of transformation matrices)~Under Assumption~\ref{assumption: psd covariance}--\ref{assumption:asymptotic normality}, 
   $\hat\Ab^0$ and $\hat\bG^0$ are consistent estimation of $\Ab^0$ and $\Gb^0:=(\Gb^\ddagger)^{-1}$ such that
    \begin{align*}
    \| \hat{\Ab}^0 - \Ab^0 \|_F =&  O_p\Big(\frac{\sqrt{p}}{\zeta_{nq,p}}\Big),\\\| \hat{\Gb}^0 - \Gb^0\|_F =& O_p \left(  \frac{\sqrt{p}}{\zeta_{nq,p}}\vee\frac{{p}^{3/2}(nq)^{3/\xi}}{\zeta_{nq,p}^{2}}\right).
\end{align*}
If we further assume Assumption~\ref{assumption:A consistency}, we have 
\begin{align*}
    \| \hat{\Ab}^0 - \Ab^0 \|_F &= O_p \Big( \Big(\frac{\sqrt{p}}{\zeta_{nq,p}}\Big)\wedge  \Big(\sqrt{p}\iota_{nq,p}\vee\frac{p^{3/2}}{(nq)^{3/\xi}}\Big) \Big) = : O_p(A_{nq,p});\\
    \| \hat{\Gb}^0 - \Gb^0\|_F &= O_p \Big(p\iota_{nq,p}\vee\frac{{p}^{3/2}(nq)^{3/\xi}}{\zeta_{nq,p}^{2}}\Big) =: O_p(G_{nq,p}).
\end{align*}
In particular, under condition $p^{3/2}\zeta_{nq,p}^{-1} (nq)^{3/\xi}\to0$ and Assumption~\ref{assumption:A consistency} that $p\sqrt{n}\iota_{nq,p}\to 0$, we have $\| \hat{\Ab}^0 - \Ab^0 \|_F,\| \hat{\Gb}^0 - \Gb^0\|_F=o_p\big(  \zeta_{nq,p}^{-1}\big)$.


 Combining~\eqref{eq:Bstar decompose} with consistency results in Lemma~\ref{prop:average consistency} and Lemma~\ref{prop_consistency_AG}, we have
\begin{align*}
    \|  \hat{\Bb}^* - \Bb^* \|_F & \le \|  \hat{\Bb}^0 - \Bb^0 \|_F + \| \hat{\bGamma}^0  -\bGamma^0\|_F ( \|\Ab^0 \|_F +\| \hat{\Ab}^0 -\Ab^0 \|_F) + \| \bGamma^0\|_F \|\hat{\Ab}^0 -\Ab^0\|_F \\
    & =O_p\big( q^{1/2} \zeta_{nq,p}^{-1} + q^{1/2} \zeta_{nq,p}^{-1} ( \sqrt{p} + A_{nq,p})  +q^{1/2} A_{nq,p}\big) \\
     & =O_p\big( q^{1/2}p^{1/2} \zeta_{nq,p}^{-1}\big),
\end{align*}
where the last equality is because $\| \hat{\Ab}^0 - \Ab^0 \|_{F} =O_p(A_{nq,p})= O_p(p^{-1/2}n^{-1})$ under Assumption~\ref{assumption: psd covariance}--\ref{assumption:A consistency}.
Hence we have
    $$q^{-1}  \|  \hat{\Bb}^* - \Bb^* \|_F^2 = O_p\big(p\zeta_{nq,p}^{-2}\big).$$


Under Assumptions~\ref{assumption: psd covariance}--\ref{assumption:A consistency}, and the scaling condition $p^{3/2}(nq)^{\epsilon+3/\xi}(p^{1/2} n^{-1/2} +q^{-1/2})=o(1)$,  we have $\|\hat{\Gb}^0 - \Gb^0\|_F = O_p(G_{nq,p}) = o_p(\zeta_{nq,p}^{-1})$.
 For $\| \hat{\bGamma}^* - \bGamma^*\|_F$, by the definition in \eqref{eq:Uhat to Ustar} and consistency results in Lemmas~\ref{prop:average consistency} and~\ref{prop_consistency_AG}, 
 we have
    \begin{align*}
        q^{-1/2}\| \hat{\bGamma}^* - \bGamma^*\|_F \le & q^{-1/2}\big\| \hat{\bGamma}^0 \big\{ (\hat{\Gb}^0)^{-\T} - (\Gb^0)^{-\T}\big\}\big \|_F + q^{-1/2}\big\| (\hat{\bGamma}^0 - \bGamma^0)  (\Gb^0)^{-\T} \big\|_F \\
        = & O_p\big( \zeta_{nq,p}^{-1}+G_{nq,p}\big)  \\
        = & O_p\big( \zeta_{nq,p}^{-1}\big).
    \end{align*}
For $\|\hat{\Ub}^*-\Ub^*\|$, we first show by Assumption~\ref{assumption: psd covariance} and Lemma~\ref{prop_consistency_AG} that
\begin{align*}
     \big\|\hat\Ab^{0}\bX_i\big\|_2&\le \big\|\Ab^{0}\bX_i\big\|_2+\big\|\Ab^0-\hat\Ab^{0}\big\|_F\big\|\bX_i\big\|_2\\&\le\big\|(\Gb^0)^{-\T}\bU_i^*\big\|_2+\big\|\bU_i^0\big\|_2+\big\|\Ab^0-\hat\Ab^{0}\big\|_F\big\|\bX_i\big\|_2= O_p(1),
\end{align*}
which also implies that $\|\Xb(\hat\Ab^0)^\T\|=O_p(\sqrt n)$.
Then we have
\begin{align*}
    n^{-1/2}\| \hat{\Ub}^* - \Ub^* \|_F \le&
    n^{-1/2} \| (\hat\Ub^0 + \Xb (\hat\Ab^{0})^{\T}) \hat\Gb^0 -  (\Ub^0 + \Xb (\Ab^{0})^{\T}) \Gb^0 \|_F \\
    \le & n^{-1/2}\| \hat\Ub^0 + \Xb (\hat\Ab^{0})^{\T}\|_F \|\hat\Gb^0  - \Gb^0 \|_F \\
    &+n^{-1/2} \| \hat\Ub^0 + \Xb (\hat\Ab^{0})^{\T}-  (\Ub^0 + \Xb (\Ab^{0})^{\T})\|_F \| \Gb^0\|_F \\=&O_p(G_{nq,p}+\zeta_{nq,p}^{-1})\\
        = & O_p\big( \zeta_{nq,p}^{-1}\big).
\end{align*}        
\begin{remark}\label{remark:add_results_phi_star}
 For the individual rate $\hat{\bbeta}_j^* - \bbeta_j^*$, we apply a similar technique and write
    \begin{align*}
        \hat{\bbeta}_j^* - \bbeta_j^* &= \hat{\bbeta}_j^0 - (\hat\Ab^{0})^{\T} \hat{\bgamma}_j^0 - (\bbeta_j^0 - (\Ab^{0})^{\T} \bgamma_j^0) \\
        & = \hat{\bbeta}_j^0 - \bbeta_j^0 -(\hat{\Ab}^0 - \Ab^0)^{\T} \hat{\bgamma}_j^0- (\Ab^{0})^{\T} (\hat{\bgamma}_j^0 - \bgamma_j^0).
    \end{align*}
    Therefore, by Lemmas~\ref{prop:ind consistency}--\ref{prop_consistency_AG}, we show that $
        \|  \hat{\bbeta}_j^* - \bbeta_j^*\|_2 = O_p\big(
{\zeta_{nq,p}^{-1} +A_{nq,p} +p^{1/2}  \zeta_{nq,p}^{-1}} \big)
         =O_p( p^{1/2} \zeta_{nq,p}^{-1})$. 

      Similarly, for the individual rates of $\hat{\bgamma}_j^* - \bgamma_j^*$ and $\hat{\bU}_i^* - \bU_i^*$, we have that $
        \hat{\bgamma}_j^* - \bgamma_j^*$ $ = (\hat{\Gb}^0)^{-1} \hat{\bgamma}_j^0 - (\Gb^0)^{-1} \bgamma_j^0$.
Therefore, by Lemmas~\ref{prop:ind consistency}--\ref{prop_consistency_AG},  \begin{align*}
    \|  \hat{\bgamma}_j^* - \bgamma_j^* \|_2 & \le \|(\hat{\Gb}^0)^{-1} \|_F \| \hat{\bgamma}_j^0 - \bgamma_j^0\|_2 + \| (\hat{\Gb}^0)^{-1}  - (\Gb^0)^{-1}\|_F \|\bgamma_j^0\|_2 \\
    & = O_p\big( \zeta_{nq,p}^{-1}+G_{nq,p}\big)\\
      &  =  O_p\big( \zeta_{nq,p}^{-1}\big).
\end{align*}
Similarly, by $
        \hat{\bU}_i^* - \bU_i^* = (\hat{\Gb}^0)^{\T} (\hat{\bU}_i^0 + \hat{\Ab}^0\bX_i ) - ({\Gb}^0)^{\T} ({\bU}_i^0 + {\Ab}^0\bX_i )$,
we conclude
\begin{align*}
    \| \hat{\bU}_i^* - \bU_i^*\| & \le \| (\hat{\Gb}^0)^{\T} \| \| \hat{\bU}_i^0 + \hat{\Ab}^0\bX_i  -  ({\bU}_i^0 + {\Ab}^0\bX_i )\| + \|\hat{\Gb}^0 - \Gb^{0}\| \|{\bU}_i^0 + {\Ab}^0\bX_i\|\\
     & = O_p\big(\zeta_{nq,p}^{-1}+G_{nq,p}\big)\\
       & =  O_p\big( \zeta_{nq,p}^{-1}\big).
\end{align*}

\end{remark}

\subsection{Proof of Theorem~4.2}\label{sec:prove_thm2}
To prove the asymptotic normality of $\hat{\bbeta}_j^*$, we first write $\hat{\bbeta}_j^* - \bbeta_j^*$ into
 \begin{align*}
       \sqrt{n} (\hat{\bbeta}_j^* - \bbeta_j^*) &=  \sqrt{n}\{ \hat{\bbeta}_j^0 - (\hat\Ab^{0})^{\T} \hat{\bgamma}_j^0 - (\bbeta_j^0 - (\Ab^{0})^{\T} \bgamma_j^0)\} \\
        & =  \sqrt{n} \{ (\hat{\bbeta}_j^0 - \bbeta_j^0) -  (\Ab^{0})^{\T} (\hat{\bgamma}_j^0 - \bgamma_j^0)\} -  \sqrt{n}(\hat{\Ab}^0 - \Ab^0)^{\T} \hat{\bgamma}_j^0.
    \end{align*}
Because $\|\sqrt{n}(\hat{\Ab}^0 - \Ab^0)^{\T} \hat{\bgamma}_j^0\|_{\infty} \le \sqrt{n} \|\hat{\Ab}^0 - \Ab^0\|_{ \infty} \|\hat{\bgamma}_j^0\|_{\infty} \lesssim \sqrt{n}\big(p\iota_{nq,p}+p^{3/2}(nq)^{3/\xi}\zeta_{nq,p}^{-2}\big)\to 0$, we further have
\begin{align*}
     \sqrt{n} (\hat{\bbeta}_j^* - \bbeta_j^*) &= \sqrt{n} \{ (\hat{\bbeta}_j^0 - \bbeta_j^0) -  (\Ab^{0})^{\T} (\hat{\bgamma}_j^0 - \bgamma_j^0)\} + o_p(1).
\end{align*}

Recall the definition $\boldsymbol{f}_j=(\bgamma_j^\T,\bbeta_j^\T)^\T$. The asymptotic distributions for $\hat{\bU}_i^0$ and $\hat{\bbf}_j^0$ are derived in the next lemma, which are crucial in proving the asymptotic normality for $\hat{\bbeta}_j^*$, $\hat{\bgamma}_j^*$, and $\hat{\bU}_i^*$.  To present the asymptotic normality of the estimators $\hat{\bU}_i^0$ and $\hat{\bbf}_j^0$, we define \begin{equation}
     \bPhi_{jz}^0 = \begin{pmatrix}
         (\bar\Gb^{\ddagger})^{\T}&(\bar\Gb^{\ddagger})^{\T}\bar\Ab^{\ddagger}\\\zero&\Ib_p\end{pmatrix}\bPhi_{jz}^*\begin{pmatrix}
         \bar\Gb^{\ddagger}&\zero\\(\bar\Ab^{\ddagger})^\T\bar\Gb^{\ddagger}&\Ib_p\end{pmatrix},\label{eq_definephi_jz}\end{equation} and \begin{equation}
     \bPhi_{i\gamma}^0 =
         (\bar\Gb^{\ddagger})^{-\T}\bPhi_{i\gamma}^*
         (\bar\Gb^{\ddagger})^{-1}.\label{eq_definephi_igamma}\end{equation}
         
\noindent{\bf Lemma~\ref{thm:asymptotic normality}.}~(Asymptotic Normality)
Under Assumptions~\ref{assumption: psd covariance}--\ref{assumption:asymptotic normality}, 
we have the asymptotic distributions for the constrained maximum likelihood estimators $\hat{\bU}_i^0$ as 
\begin{equation*}
    \sqrt{q}(\bPhi_{i\gamma}^0)^{1/2}(\hat{\bU}_i^0-\bU_i^0) \overset{d}{\to} \cN(\zero_K, \Ib_K )\quad \text{ if }   p^{3/2}\sqrt{q}(nq)^{3/\xi}\zeta_{nq,p}^{-2}\to 0\text{, for all }i\in[n]
\end{equation*}
and $\hat{\bbf}_j^0$ as
\begin{equation*}
     \sqrt{n} \ba^{\T} (\bPhi_{jz}^0)^{1/2}(\hat{\bbf}_j^0-\bbf_j^0)\overset{d}{\to} \cN(0,1)  \quad \text{ if }   p\sqrt{n}(nq)^{3/\xi}\zeta_{nq,p}^{-2}\to 0\text{, for all }j\in[q].
\end{equation*}
for any $\ba\in\RR^{K+p}$ with $\|\ba\|_2 = 1$, 
where the asymptotic variance $\ba^\T(\bPhi_{i\gamma}^0)^{-1}\bb$ and $\ba^\T(\bPhi_{jz}^0)^{-1}\bb$
can be consistently estimated by 
    \begin{align*}
\ba^\T\big(\hat\bPhi_{i\gamma}^0\big)^{-1}\bb&=q\ba^\T\big\{\sum_{j=1}^q\hat l_{ij}^{\prime\prime}\hat\bgamma_j^0(\hat\bgamma_j^{0})^{\T}\big\}\big\{\sum_{j=1}^q(\hat l_{ij}^{\prime})^2\hat\bgamma_j^0(\hat\bgamma_j^{0})^{\T}\big\}^{-1}\{\sum_{j=1}^q\hat l_{ij}^{\prime\prime}\hat\bgamma_j^0(\hat\bgamma_j^{0})^{\T}\big\}\bb;\\ 
\ba^\T\big(\hat\bPhi_{jz}^0\big)^{-1}\bb&=n\bb^\T\big\{\sum_{i=1}^n\hat l_{ij}^{\prime\prime}\hat\bZ_i^0(\hat\bZ_i^0)^{\T}\big\}\big\{\sum_{i=1}^n(\hat l_{ij}^{\prime})^2\hat\bZ_i^0(\hat\bZ_i^0)^{\T}\big\}^{-1}\{\sum_{i=1}^n\hat l_{ij}^{\prime\prime}\hat\bZ_i^0(\hat\bZ_i^0)^{\T}\big\}\bb,
\end{align*}
with $\hat\bZ_i^0=((\hat\bU_i^0)^\T,\bX_i^\T)^\T$, $\hat{l}_{ij}^{\prime} = l_{ij}^{\prime} (\hat{\bgamma}_j^{0\intercal}\hat{\bU}_i^0  + \hat{\bbeta}_j^{0\intercal}\bX_i)$ and $\hat l_{ij}^{\prime\prime}= l_{ij}^{\prime\prime} (\hat{\bgamma}_j^{0\intercal}\hat{\bU}_i^0  + \hat{\bbeta}_j^{0\intercal}\bX_i)$ for any $\ba,\bb\in\RR^{K+p}$ with $\|\ba\|=\|\bb\|=1$.

Based on the asymptotic normality results for $\hat{\bbf}_j^0 - \bbf_j^0$ in Lemma~\ref{thm:asymptotic normality}, we have 
\begin{align*}
    \sqrt{n}\ba^{\T} (\bSigma_{\beta, j}^*)^{-1/2} (\hat{\bbeta}_j^* - \bbeta_j^*) \overset d\rightarrow \cN(0,1),
\end{align*}
for any $\ba \in \RR^{p}$ with $\| \ba\| = 1$. Here $\bSigma_\beta^*$ is defined as
\begin{align*}
    \bSigma_{\beta,j}^* =&\, \begin{pmatrix}-(\bar\Ab^0)^\T&\Ib_p\end{pmatrix}\big(\bPhi_{jz}^0\big)^{-1} \begin{pmatrix}-\bar\Ab^0\\\Ib_p\end{pmatrix}\\=&\, \begin{pmatrix}-(\bar\Ab^\ddagger)^\T\bar\Gb^\ddagger&\Ib_p\end{pmatrix}\begin{pmatrix}
         (\bar\Gb^{\ddagger})^{-\T}&-\bar\Ab^{\ddagger}\\\zero&\Ib_p\end{pmatrix}\big(\bPhi_{jz}^*\big)^{-1}\begin{pmatrix}
         (\bar\Gb^{\ddagger})^{-1}&\zero\\-(\bar\Ab^{\ddagger})^\T&\Ib_p\end{pmatrix} \begin{pmatrix}-(\bar\Gb^\ddagger)^\T\bar\Ab^\ddagger\\\Ib_p\end{pmatrix}\\=&\,\bar\Tb_{\beta}\big(\bPhi_{jz}^*\big)^{-1}\bar\Tb_{\beta}^\T.
\end{align*}
Denote  $\hat{l}_{ij}^{\prime} = l_{ij}^{\prime} (\hat{\bgamma}_j^{0\intercal}\hat{\bU}_i^0  + \hat{\bbeta}_j^{0\intercal}\bX_i)$ and $\hat l_{ij}^{\prime\prime}= l_{ij}^{\prime\prime} (\hat{\bgamma}_j^{0\intercal}\hat{\bU}_i^0  + \hat{\bbeta}_j^{0\intercal}\bX_i)$. With the consistent estimator for $\bPhi_{jz}^0$ given in Lemma~\ref{thm:asymptotic normality}, $\bb^\T\bSigma_{\beta,j}^*\ba$ can be consistently estimated by
\begin{align*}
\bb^\T\hat\bSigma_{\beta,j}^*\ba =&\bb^\T\begin{pmatrix}-(\hat\Ab^0)^\T&\Ib_p\end{pmatrix} \Big\{n^{-1}\sum_{i=1}^n\hat l_{ij}^{\prime\prime}\hat{\bZ}_i^0(\hat\bZ_i^{0})^{\T}\Big\}^{-1}\Big\{n^{-1}\sum_{i=1}^n(\hat l_{ij}^{\prime})^2\hat{\bZ}_i^0(\hat\bZ_i^{0})^{\T}\Big\} \\
&\Big\{n^{-1}\sum_{i=1}^n\hat l_{ij}^{\prime\prime}\hat{\bZ}_i^0(\hat\bZ_i^{0})^{\T}\Big\}^{-1}\begin{pmatrix}-\hat\Ab^0\\\Ib_p\end{pmatrix}\ba\\
=&\bb^\T\left[n\Big\{\sum_{i=1}^n\hat l_{ij}^{\prime\prime}\bX_i\bX_i^\T\Big\}^{-1}\Big\{\sum_{i=1}^n(\hat l_{ij}^{\prime})^2\bX_i\bX_i^\T\Big\}\Big\{\sum_{i=1}^n\hat l_{ij}^{\prime\prime}\bX_i\bX_i^\T\Big\}^{-1}\right.\\-&n(\hat\Ab^{0})^{\T}\Big\{\sum_{i=1}^n\hat l_{ij}^{\prime\prime}\hat\bU_i^0(\hat\bU_i^{0})^{\T}\Big\}^{-1}\Big\{\sum_{i=1}^n(\hat l_{ij}^{\prime})^2\hat\bU_i^0\bX_i^\T\Big\}\Big\{\sum_{i=1}^n\hat l_{ij}^{\prime\prime}\bX_i\bX_i^\T\Big\}^{-1}\\-& n\Big\{\sum_{i=1}^n\hat l_{ij}^{\prime\prime}\bX_i\bX_i^\T\Big\}^{-1}\Big\{\sum_{i=1}^n(\hat l_{ij}^{\prime})^2\bX_i(\hat\bU_i^{0})^{\T}\Big\}\Big\{\sum_{i=1}^n\hat l_{ij}^{\prime\prime}\hat\bU_i^0(\hat\bU_i^{0})^{\T}\Big\}^{-1}\hat\Ab^{0}\\ 
 + &n(\hat\Ab^{0})^{\T}\Big\{\sum_{i=1}^n\hat l_{ij}^{\prime\prime}\hat\bU_i^0(\hat\bU_i^{0})^{\T}\Big\}^{-1}\Big\{\sum_{i=1}^n(\hat l_{ij}^{\prime})^2\hat\bU_i^0(\hat\bU_i^{0})^{\T}\Big\}\Big\{\sum_{i=1}^n\hat l_{ij}^{\prime\prime}\hat\bU_i^0(\hat\bU_i^{0})^{\T}\Big\}^{-1}\hat\Ab^{0}\biggr]\ba,
\end{align*}
for any $\ba,\bb\in\RR^{K+p}$ with $\|\ba\|=\|\bb\|=1$.

Furthermore, 
we write
 \begin{align*}
      \sqrt{n} ( \hat{\bgamma}_j^* - \bgamma_j^*) =&  \sqrt{n}(\hat{\Gb}^0)^{-1} \hat{\bgamma}_j^0 -  \sqrt{n}(\Gb^0)^{-1} \bgamma_j^0. 
    \end{align*}
Similarly, under the condition $p\sqrt{n}\iota_{nq,p} = o(1)$ and $p^{3/2}n^{1/2}(nq)^{3/\xi}\zeta_{nq,p}^{-2}=o(1)$, we have $\big\|\sqrt{n}(\hat\Gb^0-\Gb^0)\hat\bgamma_j^0\big\|_{\infty}=o_p(1)$. By the asymptotic property of $\hat{\bbf}_j^0 - \bbf_j^0$ in Lemma~\ref{thm:asymptotic normality}, we show that
\begin{equation*}
    \sqrt{n} (\bSigma_{\gamma,j}^*)^{-1/2} (\hat{\bgamma}_j^* - \bgamma_j^*) \overset d\rightarrow \cN (\bm{0}, \Ib_K),
\end{equation*}
where $\bSigma_{\gamma,j}^* = \bar\Tb_\gamma(\bPhi_{jz}^*)^{-1}\bar\Tb^\T_\gamma$,
and can be consistently estimated by the plug-in estimator:
\begin{align*}
    \hat{\bSigma}_{\gamma, j}^* &= (\hat{\Gb}^{0})^{-1}\Big\{ n\big(\sum_{i=1}^n\hat l_{ij}^{\prime\prime}\hat\bU_i^0(\hat\bU_i^{0})^{\T}\big)^{-1}(\sum_{i=1}^n(\hat l_{ij}^{\prime})^2\hat\bU_i^0(\hat\bU_i^{0})^{\T}\big)(\sum_{i=1}^n\hat l_{ij}^{\prime\prime}\hat\bU_i^0(\hat\bU_i^{0})^{\T}\big)^{-1}\Big\} (\hat{\Gb}^0)^{-\T}.
\end{align*}

Under $p\sqrt{q}\iota_{nq,p}=o(1)$, $q=O(n)$ and $p^{3/2}q^{1/2}(nq)^{3/\xi}\zeta_{nq,p}^{-2}=o(1)$, we have $\|\hat\Ab^0-\Ab^0\|=o_p\big(q^{-1/2}\big)$ and $\|\hat\Gb^0-\Gb^0\|=o_p\big(q^{-1/2}\big)$, which imply
    \begin{align*}
     \sqrt{q}(  \hat{\bU}_i^* - \bU_i^* ) &= \sqrt{q}(\hat{\Gb}^0)^{\T}  (\hat{\bU}_i^0 + \hat{\Ab}^0\bX_i ) - \sqrt{q} ({\Gb}^0)^{\T} ({\bU}_i^0 + {\Ab}^0\bX_i ) \\
        &=\sqrt{q}({\Gb}^0)^{\T} \{\hat{\bU}_i^0 - {\bU}_i^0 + (\hat{\Ab}^0 - \Ab^0) \bX_i\} +\sqrt{q}(\hat\Gb^0-\Gb^0)^\T(\hat\bU_i+\hat\Ab^0\bX_i)\\
        & =\sqrt{q} ({\Gb}^0)^{\T} (\hat{\bU}_i^0 - {\bU}_i^0) + o_p(1).
    \end{align*}
Hence, by Lemma~\ref{thm:asymptotic normality},
\begin{equation*}
    \sqrt{q} (\bSigma_{u,i}^*)^{-1/2} (\hat{\bU}_i^* - \bU_i^*) \overset d\rightarrow \cN (\bm{0}, \Ib_K),
\end{equation*}where $   \bSigma_{u,i}^* =  (\bPhi_{i\gamma}^*)^{-1}$,
and can be estimated by the plug-in estimator:
\begin{align*}
     \hat{\bSigma}_{u,i}^*  =  (\hat{\Gb}^0)^\T\Big\{q\big(\sum_{j=1}^q\hat l_{ij}^{\prime\prime}\hat\bgamma_j^0(\hat\bgamma_j^{0})^{\T}\big)^{-1}\big(\sum_{j=1}^q(\hat l_{ij}^{\prime})^2\hat\bgamma_j^0(\hat\bgamma_j^{0})^{\T}\big)\big(\sum_{j=1}^q\hat l_{ij}^{\prime\prime}\hat\bgamma_j^0(\hat\bgamma_j^{0})^{\T}\big)^{-1}\Big\}\hat\Gb^0.
\end{align*}

\subsection{Consistent Estimation for the Asymptotic Covariance Matrices in Theorem~4.2}\label{sec:estimate_variance}
In this subsection, we present the estimators for the asymptotic covariance matrices in Theorem~4.2 based on the estimator $\hat\bphi^0$.  
For each $j\in[q]$, for any $\ba,\bb\in\RR^{p}$ with $\|\ba\|=\|\bb\|=1$, the scaled asymptotic covariance matrices of $\bb^\T\bSigma_{\beta,j}^*\ba$ is consistently estimated by
       \begin{align}
\bb^\T\hat\bSigma_{\beta,j}^*\ba=&n\bb^\T\Big\{\sum_{i=1}^n\hat l_{ij}^{\prime\prime}\bX_i\bX_i^\T\Big\}^{-1}\Big\{\sum_{i=1}^n(\hat l_{ij}^{\prime})^2\bX_i\bX_i^\T\Big\}\Big\{\sum_{i=1}^n\hat l_{ij}^{\prime\prime}\bX_i\bX_i^\T\Big\}^{-1}\ba\nonumber\\&-n(\hat\Ab^{0}\bb)^{\T}\Big\{\sum_{i=1}^n\hat l_{ij}^{\prime\prime}\hat\bU_i^0(\hat\bU_i^{0})^{\T}\Big\}^{-1}\Big\{\sum_{i=1}^n(\hat l_{ij}^{\prime})^2\hat\bU_i^0\bX_i^\T\Big\}\Big\{\sum_{i=1}^n\hat l_{ij}^{\prime\prime}\bX_i\bX_i^\T\Big\}^{-1}\ba\nonumber\\&-n\bb^\T\Big\{\sum_{i=1}^n\hat l_{ij}^{\prime\prime}\bX_i\bX_i^\T\Big\}^{-1}\Big\{\sum_{i=1}^n(\hat l_{ij}^{\prime})^2\bX_i(\hat\bU_i^{0})^{\T}\Big\}\Big\{\sum_{i=1}^n\hat l_{ij}^{\prime\prime}\hat\bU_i^0(\hat\bU_i^{0})^{\T}\Big\}^{-1}\hat\Ab^{0}\ba\nonumber\\&+n(\hat\Ab^{0}\bb)^{\T}\Big\{\sum_{i=1}^n\hat l_{ij}^{\prime\prime}\hat\bU_i^0(\hat\bU_i^{0})^{\T}\Big\}^{-1}\Big\{\sum_{i=1}^n(\hat l_{ij}^{\prime})^2\hat\bU_i^0(\hat\bU_i^{0})^{\T}\Big\}\Big\{\sum_{i=1}^n\hat l_{ij}^{\prime\prime}\hat\bU_i^0(\hat\bU_i^{0})^{\T}\Big\}^{-1}\hat\Ab^{0}\ba, \label{eq:estimate sigma beta star}
    \end{align}
    and the asymptotic covariance matrices of $\bSigma_{\gamma,j}^*$ can be consistently estimated by
\begin{align}
    \hat{\bSigma}_{\gamma, j}^* &= (\hat{\Gb}^{0})^{-1}\Big\{ n\big(\sum_{i=1}^n\hat l_{ij}^{\prime\prime}\hat\bU_i^0(\hat\bU_i^{0})^{\T}\big)^{-1}(\sum_{i=1}^n(\hat l_{ij}^{\prime})^2\hat\bU_i^0(\hat\bU_i^{0})^{\T}\big)(\sum_{i=1}^n\hat l_{ij}^{\prime\prime}\hat\bU_i^0(\hat\bU_i^{0})^{\T}\big)^{-1}\Big\} (\hat{\Gb}^0)^{-\T}. \label{eq:estimate sigma gamma star}
\end{align}
For each $i\in[n]$, the asymptotic covariance matrices of $\bSigma_{u,i}^*$ can be consistently estimated by
\begin{align}
     \hat{\bSigma}_{u,i}^*  =  (\hat{\Gb}^0)^\T\Big\{q\big(\sum_{j=1}^q\hat l_{ij}^{\prime\prime}\hat\bgamma_j^0(\hat\bgamma_j^{0})^{\T}\big)^{-1}\big(\sum_{j=1}^q(\hat l_{ij}^{\prime})^2\hat\bgamma_j^0(\hat\bgamma_j^{0})^{\T}\big)\big(\sum_{j=1}^q\hat l_{ij}^{\prime\prime}\hat\bgamma_j^0(\hat\bgamma_j^{0})^{\T}\big)^{-1}\Big\}\hat\Gb^0. \label{eq:estimate sigma u star}
\end{align}

Next we prove Corollary~1 that the above estimators are consistent.

\subsection{Proof of Corollary~1}\label{sec:consistent estimator cov}

Recall that we denote $\hat{w}_{ij}^0 = \hat{\bgamma}_j^{0\intercal}\hat{\bU}_i^0  + \hat{\bbeta}_j^{0\intercal}\bX_i$ and $\hat\bZ_i^0=((\hat\bU_i^0)^\T,\bX_i^\T)^\T$. 
Before we prove the estimation consistency of $ \hat{\bSigma}_{\beta, j}^*$, $ \hat{\bSigma}_{\gamma, j}^*$ and $ \hat{\bSigma}_{u, i}^*$, we introduce the following lemma.
The convergence rates of individual estimators $\hat{\bgamma}_j^0$ and $\hat{\bbeta}_j^0$ match with their corresponding average convergence rates in Lemma~\ref{prop:average consistency}. Comparing the individual convergence rates of $\hat{\bU}_i^0$ with the average rates of $\hat{\Ub}^0$,
 the individual rate is higher than the corresponding average rate by a factor of $\sqrt{p}$.

Based on the individual consistency results for $\hat{\bgamma}_j^0$, $\hat{\bbeta}_j^0$, and $\hat{\bU}_i^0$ in Lemma~\ref{prop:ind consistency}, we have $|\hat{w}_{ij}^0 - w_{ij}^0| = O_p(\sqrt{p} \zeta_{nq,p}^{-1})$.
    By the continuity of functions $l_{ij}^{\prime} (\cdot)$ and $l_{ij}^{\prime\prime} (\cdot)$ in Assumption~\ref{assumption:smoothness}, we obtain $l_{ij}^{\prime} (\hat{w}_{ij}^0) \overset{p}{\rightarrow} l_{ij}^{\prime} ({w}_{ij}^0) $ and $l_{ij}^{\prime\prime} (\hat{w}_{ij}^0) \overset{p}{\rightarrow} l_{ij}^{\prime\prime} ({w}_{ij}^0)$. Finally by weak law of large numbers and Assumption~\ref{assumption:asymptotic normality}, $\bb^\T\{-n^{-1}\sum_{i=1}^n\hat l_{ij}^{\prime\prime}\hat\bZ_i^0(\hat\bZ_i^0)^{\T}\}^{-1} \ba\overset{p}{\rightarrow} \bb^\T(\bPhi_{jz}^0)^{-1}\ba$ and $n^{-1}\bb^\T \sum_{i=1}^n(\hat l_{ij}^{\prime})^2\hat\bZ_i^0(\hat\bZ_i^0)^{\T}\ba \overset{p}{\rightarrow} \bb^\T\bPhi_{jz}^0\ba$ for any $\ba,\bb\in\RR^{K+p}$ with $\|\ba\|=\|\bb\|=1$, which further show
\begin{align*}
    \bb^\T\hat\bPhi_{jz}^0\ba&=n\bb^\T\big(\sum_{i=1}^n\hat l_{ij}^{\prime\prime}\hat\bZ_i^0(\hat\bZ_i^0)^{\T}\big)^{-1}\big\{\sum_{i=1}^n(\hat l_{ij}^{\prime})^2\hat\bZ_i^0(\hat\bZ_i^0)^{\T}\big\}(\sum_{i=1}^n\hat l_{ij}^{\prime\prime}\hat\bZ_i^0(\hat\bZ_i^0)^{\T}\big)^{-1}\ba
\end{align*}
is a consistent estimator for $\bb^\T\bPhi_{jz}^0\ba$ for any $\ba,\bb\in\RR^{K+p}$ with $\|\ba\|=\|\bb\|=1$. Since 
\begin{align*}
    \bSigma_{\beta,j}^* = \begin{pmatrix}-(\Ab^0)^\T&\Ib_p\end{pmatrix}(\bPhi_{jz}^0)^{-1} \begin{pmatrix}-\Ab^0\\\Ib_p\end{pmatrix}
\end{align*}
and $\| \hat{\Ab}^0 - \Ab^0 \|_F = O_p(\sqrt{p}\zeta_{nq,p}^{-1})$, we can show
\begin{align*}
    \bb^\T(\hat{\bSigma}_{\beta, j}^*)\ba = &\,\bb^\T\begin{pmatrix}-(\hat\Ab^0)^\T&\Ib_p\end{pmatrix}  \hat\bPhi_{jz}^0\begin{pmatrix}-\hat\Ab^0\\\Ib_p\end{pmatrix}\ba\\= &\,\bb^\T\begin{pmatrix}-(\hat\Ab^0)^\T&\Ib_p\end{pmatrix} \Big\{n^{-1}\sum_{i=1}^n\hat l_{ij}^{\prime\prime}\hat{\bZ}_i^0(\hat\bZ_i^{0})^{\T}\Big\}^{-1}\Big\{n^{-1}\sum_{i=1}^n(\hat l_{ij}^{\prime})^2\hat{\bZ}_i^0(\hat\bZ_i^{0})^{\T}\Big\} \ba\\
&+\bb^\T\Big\{n^{-1}\sum_{i=1}^n\hat l_{ij}^{\prime\prime}\hat{\bZ}_i^0(\hat\bZ_i^{0})^{\T}\Big\}^{-1}\begin{pmatrix}-\hat\Ab^0\\\Ib_p\end{pmatrix}\ba,
\end{align*}
consistently estimate $\bb^\T\bSigma_{\beta,j}^*\ba$ for any $\ba,\bb\in\RR^{K+p}$ with $\|\ba\|=\|\bb\|=1$. Similar arguments can be applied to show the consistency of $ \hat{\bSigma}_{\gamma, j}^*$ and $ \hat{\bSigma}_{u, i}^*$. This completes the proof of Corollary~1.

\section{Proofs of Lemmas in Section~\ref{sec:prove main}}\label{sec:proveadd_prop}

\subsection*{Preliminaries}\label{sec:prelim}

To prove the lemmas used in Section~\ref{sec:prove main},  
we introduce the following equivalent reformulation of \eqref{eq:cmle st 12 prime} to facilitate the theoretical analysis:
\begin{eqnarray}
		\hat{\bphi}^0  &=& \underset{\bphi\in\cB(D)}{\arg\min}~\cL (\Yb | \bphi), \label{eq:MLE 1prime and 2}
	\end{eqnarray}
	where we define $\cL (\Yb | \bphi) = -L(\Yb | \bphi) - P(\bGamma, \Ub)$ as the loss function with
 \begin{align}
  &   L(\Yb|\bphi) = (nq)^{-1}\sum_{i=1}^n \sum_{j=1}^q l_{ij}(   \bgamma_j^{\intercal}\bU_i + \bbeta_j^{\intercal}{\Xb}_i), \nonumber\\
& \begin{array}{ll}
P(\bGamma, \Ub)  =& -{c} \|  \operatorname{diag}(q^{-1}\bGamma^{\intercal} \bGamma - n^{-1}\Ub^{\intercal} \Ub) \|_F^2 /{8} \\
&-{c} \|\operatorname{ndiag}(q^{-1}\bGamma^{\intercal} \bGamma )\|_F^2/2 -{c} \|  \operatorname{ndiag}(n^{-1}\Ub^{\intercal} \Ub) \|_F^2/2 \\
 &- c \| n^{-1}\Ub^\T\Xb \|_{F}^2/2,
\end{array} \nonumber \end{align}
where $0<c<b_L$ with $b_L=\min_{\bphi\in\cB(D)}\big|l^{\prime\prime} (\bgamma_j^\T\bU_i+\bbeta_j^\T\bX_i)\big|$. We know $b_L >0$ by Assumption~\ref{assumption:smoothness}.  Here $\operatorname{diag}\left(q^{-1}\bGamma^{\intercal} \bGamma - n^{-1}\Ub^{\intercal} \Ub\right)$ is the diagonal matrix consisting of the diagonal elements of $q^{-1}\bGamma^{\intercal} \bGamma - n^{-1}\Ub^{\intercal} \Ub$. The $\operatorname{ndiag}\left(q^{-1}\bGamma^{\intercal} \bGamma \right)$ is the upper triangular matrix consisting of the nondiagonal elements of $q^{-1}\bGamma^{\intercal} \bGamma$ and $\operatorname{ndiag}(n^{-1}\Ub^{\intercal} \Ub) $ is defined similarly.

We emphasize that under any choice of $c>0$,
minimizing \eqref{eq:MLE 1prime and 2}  is equivalent to minimizing \eqref{eq:cmle st 12 prime} subject to Condition 1$^{\prime}$ and Condition 2$^{\prime}$. It can be verified that for any $\bphi$, we have the penalty term $P(\bGamma, \Ub) = 0$ if identifiability conditions $1^{\prime}$--$2^{\prime}$ hold and the penalty term $P(\bphi) <0$ otherwise. For any $\Ab\in\RR^{K\times p}$ and invertible $\Gb\in\RR^{K\times K}$, $\bphi(\Gb,\Ab)$ consisting of $(\Ub+\Xb\Ab^\T)\Gb$, $\bGamma\Gb^{-\T}$, and $\Bb-\bGamma\Ab$ gives the same log-likelihood as that of $\bphi$. Among the equivalence class of estimators that maximize the log-likelihood function $L(\Yb | \bphi)$, we choose the solution to our problem to be the one that satisfies the identifiability conditions, which uniquely exists. 
  Therefore, our solution leads to $P(\hat{\bphi}^0)=0$ whereas all other solutions do not satisfy the proposed identifiability conditions as a result of negative penalty terms. Hence we conclude that the estimator from minimizing the objective function $\cL(\Yb | \bphi)$ is equivalent to that obtained by maximizing $L(\Yb | \bphi)$ under identifiability conditions $1^{\prime}$--$2^{\prime}$. We highlight that the choice of any positive $c$ shall yield the same estimation results and that we set $c$ to be smaller than $b_L$ just for convenience in our theoretical analysis. The similar idea of introducing the regularization term into the joint likelihood function and formulating the constrained maximum likelihood estimation has been extensively used in literature~\citep{Aitchison1958mle, silvey1959lagrangian, wang2022maximum,li2023statistical}

The derivations of the theoretical properties of the estimator $\hat\bphi^0$ are inspired by \cite{wang2022maximum}. We would like to emphasize that addressing the new identifiability issues arising from the transformation $\Ab$ and the additional $pq$ parameters for the covariate effects requires substantial efforts beyond existing work. First, the indeterminacy between the covariates $\bX_i$ and the latent factors $\bU_i$ is resolved by imposing additional $Kp$ constraints. These constraints are completely different from those used to address the rotational indeterminacy between factors and loadings in standard factor models without covariates. Specifically, we adopt different approaches for proving the average consistency and the local convexity of the Hessian matrix, compared to \cite{wang2022maximum}. Additionally, our Hessian matrix increases significantly in scale from that of a factor model. The number of parameters related to the covariate effects, $\bbeta_j$, is $pq$, with $p\to\infty$. This substantial increase of parameters requires a different analytical approach to derive tight bounds for the Hessian matrix.

We start with giving the expressions for the derivatives of the objective function in the estimation problem~\eqref{eq:MLE 1prime and 2}. 
For notation convenience, we further define $\Eb_{rl}^{(n)}$ to be the indicator matrix of dimension $n\times n$ with the $(r,l)$th element being $1$ and others being $0$s. 
To better present the block-structure of the Hessian matrix, we define the 
 indices: $K(r)=(r-1)K+r$ and $K(r,l)=(r-1)K+l$, and define index sets: $K_i=\{(i-1)K+1,\cdots,iK\}$ and $P_j=\{(j-1)(K+p)+1,\cdots,j(K+p)\}$.

The score function on $\bphi^0$ is written as a $(qK + qp+ nK )$-dimensional vector:
\begin{equation}
\bS(\bphi^0) = -\frac{\partial}{\partial \bphi} (nq)^{-1}\sum_{i=1}^n \sum_{j=1}^q l_{ij}\big((\bgamma_j^{0})^\T\bU_i^0 +  (\bbeta_j^{0})^\T {\Xb}_i    \big) ,\nonumber  
\end{equation}
which can be written in two parts: $\bS(\bphi)=(\bS_{f}(\bphi^0)^\T,\bS_U(\bphi^0)^\T)^\T$ with each part given as:
\begin{align}
	[\bS_f(\bphi^0)]_{[P_j]} &= -(nq)^{-1}\sum_{i=1}^n  l_{ij}^{\prime}(w_{ij}^0)\bZ_i^0, \nonumber \\
	[\bS_U(\bphi^0)]_{[K_i]} &= -(nq)^{-1}\sum_{j=1}^q  l_{ij}^{\prime}(w_{ij}^0) \bgamma_j ^0\nonumber,
\end{align}
where $w_{ij}^0=(\bgamma_j^0)^\T\bU_i^0+(\bbeta_j^0)^\T\bX_i$. Here it can be verified that the first order derivative of the penalty function on $\bphi^0$ is always zero. 

Next, we compute the second order derivative of the objective function. 
The Hessian matrix is computed as a summation of three parts, $\cH (\bphi)=  \partial_{\bphi\bphi}^2 \cL (\bphi) = \Hb_L(\bphi) + \Hb_{R}(\bphi) + \Hb_P(\bphi)$, with the block form given as 
\begin{equation*}
    \cH(\bphi)=\begin{bmatrix}
        \cH_{ff^{\prime}}(\bphi)&\cH_{fu^{\prime}}(\bphi)\\\cH_{uf^{\prime}}(\bphi)&\cH_{uu^{\prime}}(\bphi)
    \end{bmatrix},
\end{equation*}
where $\cH_{ff^{\prime}}(\bphi)\in \RR^{q(K+p)\times q(K+p)}$ and $\cH_{uu^{\prime}}(\bphi)\in \RR^{nK\times nK}$.
For $\Hb_L$, we have 
\begin{equation*}
	\Hb_L(\bphi) =\begin{bmatrix}
        \Hb_{Lff^\prime}(\bphi)&\Hb_{Lfu^\prime}(\bphi)\\\Hb_{Luf^\prime}(\bphi)&\Hb_{Luu^\prime}(\bphi)
    \end{bmatrix}
\label{eq:HL}
\end{equation*}
where the nonzero parts of the blocks are given as
\begin{align*}
    \big[\Hb_{Lff^\prime}\big]_{[P_j,P_j]}&=-\frac{1}{nq}\sum_{i=1}^nl_{ij}^{\prime\prime}(w_{ij})\bZ_i\bZ_i^\T, \\
    \big[\Hb_{Lfu^\prime}\big]_{[P_j,K_i]}&=-\frac{1}{nq}l_{ij}^{\prime\prime}(w_{ij})\bZ_i\bgamma_j^\T, \\
    \big[\Hb_{Luu^\prime}\big]_{[K_i,K_i]}&=-\frac{1}{nq}\sum_{j=1}^ql_{ij}^{\prime\prime}(w_{ij})\bgamma_j\bgamma_j^\T,
\end{align*}
with $\Hb_{L u f^{\prime}} = \Hb_{L f u^{\prime}}^{\T}$.
The matrix $\Hb_R$ is a complement term due to the chain rule in differentiation, which is written as 
 \begin{equation}
\Hb_R(\bphi)=\left(\begin{array}{lll}
\bm{0} & \Hb_{R f u^{\prime}}(\bphi) \\
\Hb_{R u f^{\prime}}(\bphi) & \bm{0} 
\end{array}\right), \label{eq:JL}
\end{equation}
where $\Hb_{R f u^{\prime}}(\bphi)$ is of dimension $q(K+p) \times nK$ with each block of size $(K+p) \times K$ and the $(i, j)$-th block is given as 
\begin{equation*}
    [\Hb_{R f u^{\prime}}(\bphi)]_{[P_j,K_i]}=-\frac{1}{nq}l_{ij}^{\prime}(w_{ij})\begin{pmatrix}
        \Ib_K\\\zero_{p\times K}
    \end{pmatrix}
\end{equation*}
$\Hb_{R u f^{\prime}}(\bphi)$ is the transpose of $\Hb_{R f u^{\prime}}$. Here $\Hb_L(\bphi)+\Hb_R(\bphi)$ is the Hessian matrix of $-L(\Yb|\bphi)$.
Then the last part $\Hb_P(\bphi)$ is a matrix due to the differentiation of penalty term $P(\bGamma, \Ub)$.
It can be derived that for $r, l \neq h \in [K]$ and $s \in [p]$,
\begin{align}
	\partial_{\bphi} [ ( \frac{\sum_{j=1}^q\gamma_{jr}^2}{q} - \frac{\sum_{i=1}^n U_{ir}^2}{n})^2] & =  4\left( \frac{\sum_{j=1}^q\gamma_{jr}^2}{q} - \frac{\sum_{i=1}^n U_{ir}^2}{n}\right)\Db_{q}^{-1}\bnu_{rr}; \nonumber \\
	\partial_{\bphi\bphi }^2 [ ( \frac{\sum_{j=1}^q\gamma_{jr}^2}{q} - \frac{\sum_{i=1}^n U_{ir}^2}{n})^2] &= 8\Db_{q}^{-1}\bnu_{rr} \bnu_{rr}^{\intercal} \Db_{q}^{-1} \nonumber \\
	  + 4\Biggr( \frac{\sum_{j=1}^q\gamma_{jr}^2}{q} & - \frac{\sum_{i=1}^n U_{ir}^2}{n}\Biggr)\Db_{q}^{-1}\Biggr[\begin{array}{cc}
\Ib_{q} \otimes \Eb_{rr}^{(K+p)}& \bm 0  \\
\bm 0 & \Ib_{n} \otimes \Eb_{rr}^{(K)}
\end{array}\Biggr]; \nonumber \\
	 \label{eq:HP3 1} \\
	 \partial_{\bphi\bphi}^2\big[(\sum_{j=1}^q \gamma_{j h} \gamma_{j l})^2\big]&=2[\bu_{ hl}\bu_{hl}^{\T}+(\sum_{j=1}^q  \gamma_{j r} \gamma_{j l}) \Db_{1,hl}]; \label{eq:HP3 2} \\
	  \partial_{\bphi\bphi}^2\big[(\sum_{i=1}^n U_{i h} U_{i l})^2\big]&=2[\bu_{ lh}\bu_{lh}^{\T}+(\sum_{i=1}^n  U_{i r} U_{il}) \Db_{2,lh}]; \label{eq:HP3 3} \\
	  \partial_{\bphi\bphi}^2\big[ (\sum_{i=1}^nU_{ir}X_{is})^2\big] &= 2 \bb_{rs}\bb_{rs}^{\T},\label{eq:HP3 4}
\end{align}
where for $r,l\neq h\in[K]$, $s\in[p]$, we have the following:
\begin{align*}  \Db_q&=\begin{pmatrix}
        q\Ib_{q(K+p)}&\zero_{q(K+p)\times nK}\\\zero_{nK\times q(K+p)}&n\Ib_{nK}
    \end{pmatrix};\\
     \bnu_{rl}&=\begin{pmatrix}
        \bGamma_{[,l]}\otimes \one_r^{(K+p)}\\-\Ub_{[,r]}\otimes \one_l^{(K)}
    \end{pmatrix};\;\bb_{rs}=\begin{pmatrix}
        \zero_{q(K+p)}\\\Xb_{[,s]}\otimes \one_r^{(K)}
    \end{pmatrix};\;  \\
     \bu_{hl}&=\begin{pmatrix}
        \bGamma_{[,h]}\otimes \one_l^{(K+p)}+\bGamma_{[,l]}\otimes \one_h^{(K+p)}\\\zero_{nK}
    \end{pmatrix}1_{(l<h)}+\begin{pmatrix}
        \zero_{q(K+p)}\\\Ub_{[,h]}\otimes \one_l^{(K)}+\Ub_{[,l]}\otimes \one_h^{(K)}
    \end{pmatrix}1_{(h<l)}\ ;\\
	\Db_{1,rl}&=\left[\begin{array}{cc}
\Ib_q \otimes (\Eb_{rl}^{(K+p)}+\Eb_{lr}^{(K+p)}) & \bm 0 \\
 \bm 0 & \Ib_{nK\times nK}
\end{array}\right];\; \\
\Db_{2,rl}&=\left[\begin{array}{cc}
\zero_{q(K+p)\times q(K+p)} & \bm 0 \\
\bm 0 & \Ib_n \otimes (\Eb_{rl}^{(K)}+\Eb_{lr}^{(K)})
\end{array}\right].\nonumber 
\end{align*}
Here $\otimes $ is the Kronecker product. By the identifiability conditions $1^{\prime}$--$2^{\prime}$, we have
\begin{eqnarray*}
&&	4\left( \frac{\sum_{j=1}^q(\gamma_{jr}^0)^2}{q} - \frac{\sum_{i=1}^n( U_{ir}^0)^2}{n}\right)\Db_{q}^{-1}\begin{bmatrix}
    \Ib_{q} \otimes \Eb_{rr}^{(K+p)}& \bm 0  \\
\bm 0 & \Ib_{n} \otimes \Eb_{rr}^{(K)}
\end{bmatrix} = 0;\\
&&(\sum_{j=1}^q  \gamma_{j h} ^0\gamma_{j l}^0) \Db_{1,hl} = 0; \quad (\sum_{i=1}^n  U_{i h}^0 U_{i l}^0) \Db_{2,lh} = 0.
\end{eqnarray*}
So the third part $\Hb_P(\bphi)$ on parameters $\bphi^0$ can be written as 
\begin{align}
	\Hb_P(\bphi^0) =& c\Big(\sum_{r=1}^K \Db_{q}^{-1}\bnu_{rr}^0 (\bnu_{rr}^0)^{\intercal} \Db_{q}^{-1} +\sum_{r=1}^K \sum_{l=r+1}^K \Db_{q}^{-1}\bu_{ hl}^0(\bu_{hl}^0)^{\T} \Db_{q}^{-1}\nonumber\\
 &+ \sum_{r=1}^K \sum_{l=r+1}^K  \Db_{q}^{-1}\bu_{ lh}^0(\bu_{lh}^0)^{\T}\Db_{q}^{-1}+\sum_{r=1}^K\sum_{s=1}^p\Db_{q}^{-1}\bb_{rs}^0(\bb_{rs}^0)^\T \Db_{q}^{-1}\Big).\label{eq:HP}
\end{align}

Define $\Vb_p\in \RR^{(q(K+p)+nK)\times (K^2+Kp)}$ with column vectors given as, for $r,l\neq h\in[K]$, $s\in[p]$, 
\begin{equation*}
\big[\Vb_p\big]_{[,K(r)]}=\Db_q^{-1/2}\bnu_{rr},\;\big[\Vb_p\big]_{[,(h-1)K+l]}=\Db_q^{-1/2}\bu_{hl},\;\big[\Vb_p\big]_{[,K^2+(r-1)p+s]}=\Db_q^{-1/2}\bb_{rs}.
\end{equation*}
Such  $\Vb_p$ is the low-rank decomposed matrix of the Hessian matrix of the penalty term $P(\bGamma,\Ub)$.
We assemble those decomposed vectors into $\bLambda=(\bLambda_1^\T,\bLambda_2^\T)^\T$
as follows: 
    \begin{equation*}
        \bLambda_1=\sqrt{\frac{c}{q}}[\Vb_p]_{[1:q(K+p),]},\;\bLambda_2=\sqrt{\frac{c}{n}}[\Vb_p]_{[q(K+p)+1:q(K+p)+nK,]},
    \end{equation*}
    Here the entries in the last $q(K+p)\times Kp$ part of $\bLambda_1$ are zeros. We also denote $\bLambda_{1\bullet}\in\RR^{q(K+p)\times K^2} $ as $[\bLambda_1]_{[,1:K^2]}$.
Then the matrix $\Hb_P({\bphi}^0)$ can be expressed as \begin{equation}
\Hb_P({\bphi}^0) = -\partial^2_{\bphi}P(\bphi^0)=\bLambda^0(\bLambda^0)^\T. \label{eq:rank1 decomp}
    \end{equation} 
Here we write $\bLambda^0$ to denote $\bLambda(\bphi^0)$ for short.

Moreover, we further define the following:
\begin{align*}
\bomega_{rs}=\begin{pmatrix}
        -\bGamma_{[,r]}\otimes \one_{K+s}^{(K+p)}\\\Xb_{[,s]}\otimes \one_r^{(K)}
    \end{pmatrix},
\end{align*}
and let  $\Vb_0 \in \RR^{(q(K+p)+nK)\times (K^2+Kp)}$ to be
\begin{align*}
\big[\Vb_0\big]_{[,K(r,l)]}=&\Db_q^{-1/2}\bnu_{rl},\;\big[\Vb_0\big]_{[,K^2+(r-1)p+s]}=\Db_q^{-1/2}\bomega_{rs},
\end{align*}
such that $\max_t\|\Vb_0\|_{[,t]}=O(1)$ and $\max_t\|\Vb_p\|_{[,t]}=O(1)$ inside the space ${\cB}(D)$. It can be verified that the column vectors of $\Vb_0$ form the null space of the Hessian matrix under regularization conditions, and can be approximated by $\Vb_p$, which leads to local convexity of the objective function in some regime. 

\subsection{Proof of Lemma~\ref{prop:average consistency}}
\begin{lemma}[Average Consistency]
    \label{prop:average consistency}
    	Under Assumptions~\ref{assumption: psd covariance}--\ref{assumption:smoothness} and $p=o(n\wedge q)$, $n^{-1}\|\hat{\Ub}^0 - \Ub^0\|_F^2 = O_p(\zeta_{nq,p}^{-2})$, $q^{-1} \|\hat{\bGamma}^0 - \bGamma^0 \|_F^2 = O_p(\zeta_{nq,p}^{-2})$, and  $q^{-1} \|\hat{\Bb}^0 - \Bb^0 \|_F^2 = O_p(\zeta_{nq,p}^{-2})$.
\end{lemma}
These average consistency results may not be sharp enough to guarantee that $\hat{\bphi}^0$ lies in the interior of the parameter space, as the number of parameters increases to infinity. 
However, in the following, we will show the strong convexity of the objective function inside a small region near the true parameter $\bphi^0$. The average consistency restricts our estimation to this region, allowing us to apply the mean value theorem to the score function $\partial_{\bphi}\cL$.
\begin{proof}
    Recall that $w_{ij} = \bgamma_j^{\T} \bU_i + \bbeta_j^{\T} {\Xb}_i$. Here we further define $\theta_{ij} = \bgamma_j^{\T} \bU_i$, in which way we have $w_{ij} = \theta_{ij} + \bbeta_j^{\T} \bX_i$. Define $\bTheta = \Ub \bGamma^{\T}$ and let $\btheta_v = (\theta_{11}, \dots, \theta_{1q}, \theta_{21}, \dots, \theta_{2q}, $ $ \dots, \theta_{n1}, \dots, $ $ \theta_{nq})$. For the new parameter system, we define $\bpsi = (\btheta_v ,\bB_v)^\T$ and a bijective mapping between the two parameter spaces ${\Phi}_0=\{{\bphi}|P({\bphi})=0\}$ and ${\Psi}_0=\{{\bpsi}|\check P({\bpsi})=0,\mbox{rank}(\bTheta)\le K\}$ as \begin{equation}
    \Pi:{\bphi}\mapsto \left(\big(\mbox{vec}(\bGamma\Ub^\T)\big)^\T,\bB_v^\T\right)^\T,\label{bijective}
\end{equation} with the corresponding penalty function defined as 
\begin{equation}
    \check{P} (\bTheta) = -\frac{c}{2q}\sum_{j=1}^q\big\|\frac{1}{n}\sum_{i=1}^n\theta_{ij}{\Xb}_i\big\|^2. \nonumber
\end{equation}
We notice $\check P(\bTheta^0)=0$ as $\sum_{i=1}^n\bU_i^0\bX_i^\T=\zero_{K\times p}$. Since $\text{rank}(\bTheta)$ in ${\bpsi}_0$ is no larger than $K$, we can give a unique rank-K decomposition for $\bTheta=\bGamma\Ub^\T$ where $\Ub\in\RR^{n\times K}$, $\bGamma\in\RR^{q\times K}$ and $q^{-1}\bGamma^\T \bGamma=n^{-1}\Ub^\T\Ub$ are diagonal matrices, the uniqueness (up to a signed permutation) of which can be implied by the singular value decomposition. This induces the inversion of the mapping $\bPi$ as $\bPi^{-1}(\bpsi)= \bPi^{-1}(\bTheta,\Bb)=(\bGamma,\Ub,\Bb)$. Recall that we have argued that $\hat\bphi^0$ can be expressed as $\arg\min_{\bphi\in \cB(D)\cap\Phi_0} -   L(\bphi)$, we define
\begin{equation}
\check{\bpsi}=\mathop{\arg\min}_{{\bpsi\in\check\cB(D)\cap \Psi_0}}-\check{{L}}({\bpsi}),\label{psimle}
\end{equation}
where $\check{{L}}({\bpsi})= -(nq)^{-1}\sum_{i=1}^n\sum_{j=1}^ql_{ij}(\theta_{ij}+\bbeta_j^\T\bX_i)$ and $\cB(D)=\{\bpsi:\|\bpsi\|_{\infty}\le D\}$

For the objective function $\check\cL(\Yb|\bpsi)$, we claim that ${\bphi}\in{\Phi}_0\cap\cB(D)$ is a minimizer of $\cL(\Yb|{\bphi})$ if and only if $\Pi({\bphi})\in{\Psi}_0\cap\check\cB(D)$ is a minimizer of $\check\cL(\Yb|{\bpsi})$. The `if' part is trivial. For the `only if' part, suppose $\bphi_1$ is the minimizer of $-L(\Yb|\bphi)$ and $\check L(\Yb|\Pi({\bphi}_1))<\check L(\Yb|{\bpsi}_2)$ for some ${\bpsi}_2 \in{\bPsi}_0\cap\Omega_{\bpsi}$. Then since ${\bpsi}_2\in{\bPsi}_0$, we can get a rank-K decomposition for $\bTheta_2$ as $\bTheta_2=\mbox{vec} (\bGamma_2 \bU_2^\T)$. Then $L(\Yb|\bU_2,\bGamma_2,\bB_2)>L(\Yb|{\bpsi}_1)$, which is a contradiction. This implies
\begin{equation*}\Pi\big(\mathop{\arg\min}_{{\bphi}\in{\Phi}_0\cap\cB(D)} -L(\Yb|{\bphi})\big) = \mathop{\arg\min}_{{\bpsi}\in{\Psi}_0\cap \check\cB(D)}-\check L(\Yb|{\bpsi}),\end{equation*}
which implies that $\check\bpsi=\Pi(\hat\bphi^0)$. In the following, we use $\hat\bpsi^0$ to denote the solution to \eqref{psimle} for clarity. We will denote ${\bpsi}^0=\Pi({\bphi}^0)$ as the true parameter. Let $\hat\bTheta^0=\hat\bGamma^0(\hat\Ub^0)^\T$ and $\bTheta^0=\bGamma^0(\Ub^0)^\T$.

The derivatives of $ \check{\mathcal{L}}(\bpsi)=-\check L(\bpsi)-\check P(\bpsi)$ with respect to ${\bpsi}$ will be denoted as $ \check {\bS}({\bpsi})={\partial_{\psi}  \check{\mathcal{L}}}(\Yb|{\bpsi})$, $ \check {\bS}_{\theta}({\bpsi})={\partial_{\theta}  \check{\mathcal{L}}}(\Yb|{\bpsi})$, and $ \check S_{\beta}({\bpsi})={\partial_{\beta}  \check{\mathcal{L}}}(\Yb|{\bpsi})$. The Hessian matrix will be denoted as $\check\cH(\bpsi)=\partial^2_{\psi}\check\cL(\bpsi)=\check\Hb_{L}(\bpsi)+\check\Hb_P(\bpsi)$ with $\check\Hb_{L}=\partial_{\psi\psi}^2\big[-\sum_{i=1}^n\sum_{j=1}^ql_{ij}(\theta_{ij}+\bbeta_j^\T\bX_i)\big]$ and $\check\Hb_P=\partial_{\psi\psi}^2\check P(\bTheta)$. 
Expand the objective function $\check\cL(\bpsi)$ at $\hat\bpsi^0$, it follows that 
\begin{align}
\check \cL(\Yb |\hat{\bpsi}^0) 
	&= \check \cL(\Yb|\bpsi^0) + \check\bS(\bpsi^0)^\T(\hat\bpsi^0-\bpsi^0) \nonumber\\&+ \frac{1}{2}\big[\check\Db_{q}^{-1/2}(\hat\bpsi^0-\bpsi^0)\big]^\T\big[\check\Db_{q}^{1/2}\check\cH(\tilde\bpsi)\check\Db_{q}^{1/2}\big]\big[\check\Db_{q}^{-1/2}(\hat\bpsi^0-\bpsi^0)\big],\label{eq:expandto_psi0}
\end{align}
where $\tilde\bphi$ is some points between $\hat\bpsi^0$ and $\bpsi^0$, and $\check\Db_q$ is a scaling matrix defined as
\begin{equation*}
    \check\Db_q=\begin{pmatrix}
        nq\Ib_{nq}\\&q\Ib_{q}
    \end{pmatrix}.
\end{equation*}

Because $\check P(\hat\bTheta^0) \le 0$ and $\check P (\bTheta^0) = 0$, we have $\check L(\Yb| \hat\bpsi^0) + \check P(\hat\Theta^0) \ge \check L(\Yb | \bpsi^0) + \check P (\bTheta^0)$ $  =  \check L(\Yb | \bpsi^0)$. Therefore $\check L(\Yb| \hat\bpsi^0) \ge \check L(\Yb | \bpsi^0)$, which together with \eqref{eq:expandto_psi0}, Lemmas~\ref{lemma:estimation for 1st derivative}, \ref{lemma:Hcheck} gives 
\begin{eqnarray*}
    0\lesssim \delta_{nq}^{-1}\sqrt{\frac{\log q}{nq}}\|\hat\bTheta^0-\bTheta^0\|_F+\sqrt{\frac{\log pq}{n}}\sqrt{\frac{p}{q}}\|\hat\Bb^0-\Bb^0\|_F \\
    -\frac{\check\gamma}{2}\Big[\frac{1}{nq}\|\hat\bTheta^0-\bTheta^0\|_F^2+{\frac{1}{q}}\|\hat\Bb^0-\Bb^0\|_F^2\Big],
\end{eqnarray*}
which implies with $p<q$ that
\begin{equation*}
    \frac{1}{nq}\|\hat\bTheta^0-\bTheta^0\|_F^2\le O_p\Big(\frac{\log n}{q}+\frac{p\log q}{n}\Big),\quad {\frac{1}{q}}\|\hat\Bb^0-\Bb^0\|_F^2\le O_p\Big(\frac{\log n}{q}+\frac{p\log q}{n}\Big).
\end{equation*}
Here $\check\gamma$ is some positive constant specified in Lemma~\ref{lemma:Hcheck} that does not depend on the dimension.

		Let $\rho_1^0, \dots, \rho_K^0$ be the singular values of $n^{-1/2}q^{-1/2}  \Ub^0 (\bGamma^{0})^\T$ and  $\bupsilon_1^0, \dots, \bupsilon_K^0$ be the corresponding left-singular vectors. Let $\hat{\rho}_1, \dots, \hat{\rho}_K$ be the singular values of $n^{-1/2}q^{-1/2} \hat{\Ub}^0  \hat{\bGamma}^{0\T}$ and $\hat{\bupsilon}_1, \dots, \hat{\bupsilon}_K$ be the corresponding left-singular vectors. From a variant version of Davis-Kahan Theorem \citep{yu2015useful}, we have 
		\begin{eqnarray}
			\| \hat{\bupsilon}_k - \bupsilon_k^0\|_2 &\le & \sqrt{2}\|n^{-1/2}q^{-1/2} \hat{\Ub}^0  (\hat{\bGamma}^{0})^{\T}-  n^{-1/2}q^{-1/2}  \Ub^0 (\bGamma^{0})^\T\|_F /\eta, \nonumber	\\
			&\le & \sqrt{2} n^{-1/2}q^{-1/2} \| \hat{\Ub}^0  (\hat{\bGamma}^{0})^{\T} - \Ub^0 (\bGamma^{0})^\T \|_F/\eta	\label{eq:ekhat ek}  \end{eqnarray}
			where $\eta = \min \{|\hat{\rho}_{k-1} - {\rho}_k^0| \wedge |\hat{\rho}_{k+1} - {\rho}_k^0|: k = 1,\dots, K \}$. By Weyl's inequality, for all $k$, $ |\hat{\rho}_k - {\rho}_k^0|\le (nq)^{-1/2}\| \hat{\Ub}^0  (\hat{\bGamma}^{0})^{\T} - \Ub^0 (\bGamma^{0})^\T \|_F = (nq)^{-1/2}\|\hat\bTheta-\bTheta^0\|_F=O_p(\zeta_{nq,p}^{-1})$. Under the Assumption~\ref{assumption: psd covariance} and $p<\delta_{nq}$, $\eta$ is bounded and bounded away from zero in probability and it follows from~\eqref{eq:ekhat ek} that $\| \hat{\bupsilon}_k - \bupsilon_k^0\|_2 = O_p(\zeta_{nq,p}^{-1})$. Under penalty function $P(\bGamma, \Ub)$, we have the $k$-th factor to be $\sqrt{n\hat\rho_k} \hat\bupsilon_k$. Thus we have 
			\begin{eqnarray}
					\|\hat{\Ub}^0 - \Ub^0  \|_F = \|\hat{\Ub}_v^0 - \Ub_v^0   \|_2  \le \sqrt{n}|\sqrt{\hat{\rho}_k} - \sqrt{\rho_k^0} | \|\hat{\bupsilon}_k \|_2 + \sqrt{n\rho_k} \| \hat{\bupsilon}_k - \bupsilon_k^0\|_2 = O_p\big(\sqrt{q}\zeta_{nq,p}^{-1}\big). \nonumber
			\end{eqnarray}
			Similarly we also have $$\|\hat{\bGamma}^0 - \bGamma^0  \|_F =  O_p\big(\sqrt{q}\zeta_{nq,p}^{-1}\big).$$

\end{proof}
\subsection{Proof of Lemma~\ref{prop:ind consistency}}
\begin{lemma}[Individual Consistency]
\label{prop:ind consistency}
    Under Assumptions~\ref{assumption: psd covariance}--\ref{assumption: Scaling}, we have for any $j\in[q]$
    \begin{equation*}
        \|\hat {\bgamma}_j^0-\bgamma_j^0\|_2=O_p\big(\zeta_{nq,p}^{-1}\big),\;\|\hat{\bbeta}_j^0 -\bbeta_j^0\|_2=O_p\big(\zeta_{nq,p}^{-1}\big),
    \end{equation*}
    and for any $i\in[n]$
    \begin{equation*}
        \|\hat\bU_i^0-\bU_i^0\|_2=O_p({\sqrt{p}}{\zeta_{nq,p}^{-1}}).
    \end{equation*}
\end{lemma}
 \begin{proof}
        By Lemma~\ref{lemma: first order condition} we know that $\bS_{f}(\hat\bphi^0)=0$ {\em w.h.p.} We can use the integral form of mean value theorem to expand $\big[\bS_f(\hat\bphi^0)]_{[P_j]}$ 
        as follows:
\begin{align}
0=&[\bS_{f_j}(\hat\bphi^0)]_{[P_j]}\nonumber\\=&-(nq)^{-1}\sum_{i=1}^n\hat\bZ_i^0 l_{ij}^\prime\big[(\bgamma_j^0)^\T\hat \bU_i^0+({\bbeta}_j^0)^\T\bX_i\big]\label{eq:expansf_1}\\&-(nq)^{-1}{\sum_{i=1}^n\int_{[0,1]}l^{\prime\prime}_{ij}\big[(\bbf_j^0)^\T\hat \bZ_i^0+s(\hat{\bbf}_j^0-{\bbf}_j^0)^\T\hat \bZ_i^0\big]\hat \bZ_i^0(\hat \bZ_i^0)^\T(\hat{\bbf}_j^0-\bbf_j^0)ds}\label{eq:expansf_2}
\end{align}
 {For \eqref{eq:expansf_2}, according to Assumption~\ref{assumption:smoothness}(iii), we have
\begin{equation*}
    \|\eqref{eq:expansf_2}\|\ge \lambda_{\min}\big[b_Lq^{-1}(n^{-1}\sum_{i=1}^n\hat{\bZ}_i^0(\hat\bZ_i^0)^\T)\big]\|\hat{\bbf}_j^0-\bbf_j^0\|\gtrsim q^{-1}\|\hat{\bbf}_j^0-\bbf_j^0\|.
\end{equation*}
}

Before we discuss~\eqref{eq:expansf_1}, for notational simplicity, here we denote ${b}^{\prime}(\hat{\bU}_i^0)=b^\prime\big\{(\bgamma_j^0)^\T\hat \bU_i^0$  $+{(\bbeta_j^0)}^\T\bX_i\big\}$ and ${b}^{\prime}({\bU}_i^0)=b^{\prime}\big\{(\bgamma_j^0)^\T\bU_i^0+({\bbeta}_j^0)^\T\bX_i\big\}$. We next decompose~\eqref{eq:expansf_1} into
\begin{align}
      \|\eqref{eq:expansf_1}\|=&\;\Big\|(nq)^{-1} \sum_{i=1}^n\hat\bZ_i^0\big\{Y_{ij}-b^\prime(\hat \bU_i^0)\big\}\Big\|\nonumber \\=&\;
    \Big\|(nq)^{-1}\sum_{i=1}^n\Big[ \bZ_i^0\{Y_{ij}-{b}^{\prime}({\bU}_i^0)\}+(Y_{ij}-{b}^{\prime}({\bU}_i^0))(\hat\bZ_i^0-\bZ_i^0)\nonumber\\&\;+\bZ_i^0\big\{{b}^{\prime}({\bU}_i^0)-{b}^{\prime}(\hat{\bU}_i^0)\big\} +\big[{b}^{\prime}({\bU}_i^0)-{b}^{\prime}(\hat{\bU}_i^0)\big](\hat\bZ_i^0-\bZ_i^0)\Big\}\Big\|\nonumber\\
\le&\;
    q^{-1}\Big\{\big\|n^{-1}\sum_{i=1}^n \bZ_i^0l_{ij}^{\prime}(w_{ij}^0)\big\|+\big\|n^{-1}\sum_{i=1}^nl_{ij}^{\prime}(w_{ij}^0)(\hat\bU_i^0-\bU_i^0)\big\|\nonumber\\&+\big\|n^{-1}\sum_{i=1}^nl_{ij}^{\prime\prime}(\tilde w_{ij})\bZ_i^0\big(\bU_i^0-\hat\bU_i^0\big)^\T\bgamma_j^0\big\|\nonumber\\
    &+\big\|n^{-1}\sum_{i=1}^nl_{ij}^{\prime\prime}(\tilde w_{ij})\big(\bU_i^0-\hat\bU_i^0\big)\big(\bU_i^0-\hat\bU_i^0\big)^\T\bgamma_j^0\big\|\Big\},\label{eq:B2 connect}
\end{align}
where $\tilde{w}_{ij}$ lies in the segment between $w_{ij}^0$ and $(\bgamma_j^0)^\T\hat \bU_i+{(\bbeta_j^0)}^\T\bX_i$. Next, we introduce the following lemma proving the concentration of the gradient to facilicate the subsequent analysis.

\noindent {\bf Lemma~\ref{lemma:first-order derivative}}~(First-order concentration)~Under Assumptions~\ref{assumption: psd covariance}--\ref{assumption:smoothness}, we have estimates for the first order derivatives on $\bphi^0$ as:
 \begin{align*}
       \| n^{1/2} \bS_U  (\bphi^0)\|_2 &= O_p \Big(\sqrt{\frac{\log n}{q}} \Big);\\
        \|q^{1/2}\bS_f(\bphi_v^0)\|_{2} &= O_p\Big( \sqrt{\frac{p\log qp}{n}}  \Big).
    \end{align*}
Scaled by matrix $\Db_p$ we can write
    \begin{equation}
    \|\Db_q^{1/2} \bS(\bphi^0)\|_2 = O_p\big(\zeta_{nq,p}^{-1}\big), \nonumber
\end{equation}
or $\|\Db_q^{1/2}\bS(\bphi^0)\|_{\zeta}\le O_p(p^{1/2}(nq)^{\epsilon-1/2})$ when $\zeta>1/\epsilon$.

For the first term we have $\|{n^{-1}\sum_{i=1}^n\bZ_i^0l_{ij}^{\prime}(w_{ij}^0)}\|=\sqrt{pn^{-1}\log qp}$ from Lemma~\ref{lemma:first-order derivative}. For the second term, by Assumption~\ref{assumption:smoothness}, we have $\sqrt{\sum_{i=1}^n [l_{ij}^{\prime}(w_{ij}^0)]^2}= O_p(n^{1/2})$ and  therefore $$\big\|n^{-1}\sum_{i=1}^nl_{ij}^{\prime}(w_{ij}^0)(\hat\bU_i^0-\bU_i^0)\big\|\le O_p\big(n^{-1/2}\|\hat\Ub^0-\Ub^0\big\|\big)=O_p\Big(\sqrt{\frac{p\log qp}{n} + \frac{\log n}{q}}\Big). $$
For the third term we note that $\sum_{i=1}^n\bU_i^0\bX_i^\T=\sum_{i=1}^n\hat\bU_i^0\bX_i^\T=\zero_{K\times p}$. Then together with Assumption~\ref{assumption: psd covariance} and Lemma~\ref{prop:average consistency},
\begin{align*}
    \big\|n^{-1}\sum_{i=1}^nl_{ij}^{\prime\prime}(\tilde w_{ij})\bZ_i^0\big(\bU_i^0-\hat\bU_i^0\big)^\T\bgamma_j^0\big\|&\lesssim  \big\|n^{-1}\sum_{i=1}^n\bZ_i^0\big(\bU_i^0-\hat\bU_i^0\big)^\T\big\|\\
    &\le n^{-1/2}\big\|\hat\Ub^0-\Ub^0\big\|n^{-1/2}\|\Ub^0\|\\&\le O_p\Big(\sqrt{\frac{p\log qp}{n} + \frac{\log n}{q}}\Big).
\end{align*}
For the fourth term, by Lemma~\ref{prop:average consistency}, it can be bounded at a rate of $O_p(pn^{-1}\log (pq)+q^{-1}\log n)$. Then we conclude that 
\begin{align*}
    \eqref{eq:B2 connect}
    \le O_p\left(\frac{1}{ q}\sqrt{\frac{p\log qp}{n} + \frac{\log n}{q}}\right).
\end{align*}
 Therefore, \eqref{eq:expansf_1}$+$\eqref{eq:expansf_2} $=0$ implies that for any $j \in [q]$
\begin{equation}
    \|\hat {\bbf}_j^0-\bbf_j^0\|=O_p\left(\sqrt{\frac{p\log qp}{n} + \frac{\log n}{q}} \right),
\end{equation}
which gives that for any $j \in [q]$
\begin{align*}
    \|\hat {\bgamma}_j^0-\bgamma_j^0\|=O_p\left(\sqrt{\frac{p\log qp}{n} + \frac{\log n}{q}} \right),\;
    \|\hat {\bbeta}_j^0-\bbeta_j^0\|=O_p\left(\sqrt{\frac{p\log qp}{n} + \frac{\log n}{q}} \right).
\end{align*}

Similarly, we can use the expansion of $[\bS_u(\hat\bphi^0)]_{[K_i]}$ 
and show that 
\begin{align}
0=&[\bS_u(\hat\bphi^0)]_{[K_i]}\nonumber \\=&- (nq)^{-1}\sum_{j=1}^q\hat\bgamma_j^0 l_{ij}^\prime\big[(\hat\bgamma_j^0)^\T \bU_i^0+(\hat{\bbeta}_j^0)^\T\bX_i\big]\label{B87}\\
&-(nq)^{-1} {\sum_{j=1}^q\int_{[0,1]}l^{\prime\prime}_{ij}\big[(\hat\bgamma_j^0)^\T\hat \bU_i^0+s(\hat{\bgamma}_j^0)^\T(\hat \bU_i^0-\bU_i^0)+(\hat {\bbeta}_j^0)^\T\bX_i\big]\hat \bgamma_j^0(\hat \bgamma_j^0)^\T(\hat{\bU}_i^0-\bU_i^0)ds}\label{B88}.
\end{align}
For the first term, we use similar arguments as in bounding~$\|\eqref{eq:expansf_2}\|$:
\begin{equation*}
    \|\eqref{B87}\|\gtrsim n^{-1}\|\hat\bU_i-\bU_i^0\|.
\end{equation*} For \eqref{B88}, denoting  ${b}^{\prime}(\hat{\bbf}_j^0)=b^\prime\big\{(\hat{\bbf}_j^0)^\T \bZ_i^0\big\}$ and ${b}^{\prime}({\bbf}_j^0)=b^{\prime}\big\{(\bbf_j^0)^\T\bZ_i^0\big\}$, we have 
\begin{align*}
    \|\eqref{B88}\|=&\; \Big\|(nq)^{-1}\sum_{j=1}^q\hat{\bgamma}_j \big\{Y_{ij}-{b}^{\prime}(\hat{\bbf}_j^0)\big\}\Big\|\\=&\; \Big\|(nq)^{-1}\sum_{j=1}^q\Big[ \bgamma_j^0\{Y_{ij}-{b}^{\prime}({\bbf}_j^0)\}+\{Y_{ij}-{b}^{\prime}(\bbf_j^0)\}(\hat\bgamma_j^0-\bgamma_j^0)\\&+\bgamma_j^0\big\{{b}^{\prime}(\bbf_j^0)-{b}^{\prime}(\hat\bbf_j^0)\big\} +\big\{{b}^{\prime}(\bbf_j^0)-{b}^{\prime}(\hat\bbf_j^0)\big\}(\hat\bgamma_j^0-\bgamma_j^0)\Big\}\Big\|
    \end{align*}
    Because we have $\|{q^{-1}\sum_{j=1}^q\bgamma_j^0l_{ij}^{\prime}(w_{ij}^0)}\|=\sqrt{q^{-1}\log n}$ from Lemma~\ref{lemma:first-order derivative}, similar to the approach of estimating \eqref{eq:expansf_2}, we have
    \begin{align*}
   &\|\eqref{B88}\|\\
   \leq   &\frac{1}{n}O_p\Big\{\sqrt{\frac{\log n}{q}}+\sqrt{\frac{\log n}{q}}\|\hat\bgamma-\bgamma^0\|+q^{-1}b_U\|\bZ_i^0\|\|\hat\bbf^0-\bbf^0\|\|\bgamma^0\|+q^{-1}b_U\|\hat\bbf^0-\bbf^0\|^2\Big\}\\=&O_p\left\{\frac{1}{n}\sqrt{\frac{p^2\log(pq)}{n}+\frac{p\log n}{q}}\right\}.
\end{align*}
By \eqref{B87} + \eqref{B88} = 0, we show that for any $i \in [n]$
\begin{equation*}
    \|\hat\bU_i^0-\bU_i^0\|=O_p\left(\sqrt{\frac{p^2\log(pq)}{n}+\frac{p\log n}{q}}\right).
\end{equation*}

\end{proof}

\subsection{Proof of Lemma~\ref{prop_consistency_AG}}

\begin{lemma}
    [Consistency of transformation matrices]\label{prop_consistency_AG}
Under Assumption~\ref{assumption: psd covariance}--\ref{assumption:asymptotic normality}, 
   $\hat\Ab^0$ and $\hat\bG^0$ are consistent estimation of $\Ab^0$ and $\Gb^0:=(\Gb^\ddagger)^{-1}$ such that
    \begin{align*}
    \| \hat{\Ab}^0 - \Ab^0 \|_F =&  O_p\Big(\frac{\sqrt{p}}{\zeta_{nq,p}}\Big),\\\| \hat{\Gb}^0 - \Gb^0\|_F =& O_p \left(  \frac{\sqrt{p}}{\zeta_{nq,p}}\vee\frac{{p}^{3/2}(nq)^{3/\xi}}{\zeta_{nq,p}^{2}}\right).
\end{align*}
Further under Assumption~\ref{assumption:A consistency}, we have 
\begin{align*}
    \| \hat{\Ab}^0 - \Ab^0 \|_F =&  O_p\left(\Big(\frac{\sqrt{p}}{\zeta_{nq,p}}\Big)\wedge  \Big(\sqrt{p}\iota_{nq,p}\vee\frac{p^{3/2}(nq)^{3/\xi}}{\zeta_{nq,p}^{2}}\Big)\right),\\\| \hat{\Gb}^0 - \Gb^0\|_F =& O_p \left(  p\iota_{nq,p}\vee\frac{{p}^{3/2}(nq)^{3/\xi}}{\zeta_{nq,p}^{2}}\right).
\end{align*}
In particular, under condition $p^{3/2}\zeta_{nq,p}^{-1} (nq)^{3/\xi}\to0$ and Assumption~\ref{assumption:A consistency} that $p\sqrt{n}\iota_{nq,p}\to 0$, we have $\| \hat{\Ab}^0 - \Ab^0 \|_F,\| \hat{\Gb}^0 - \Gb^0\|_F=o_p\big(  \zeta_{nq,p}^{-1}\big)$.
\end{lemma}

\begin{remark}Establishing the consistency of $\hat{\Ab}^0$ is closely related to the median regression problem with measurement errors.
$\hat\Ab^0$'s consistent rate  $O_p(\sqrt{p}\zeta_{nq,p}^{-1})$ can be derived using the individual rate, which is to be shown in Lemma~\ref{prop:ind consistency}. This result, however, is not enough for establishing asymptotic properties of $\hat\bphi^*$.
To derive a stringent convergence rate for $\hat{\Ab}^0$, we extend the general results of Bahadur representations for $M$-estimators in the median regression framework~\citep{he1996bahadur}. Compared to the results in~\cite{he2000quantile} with measurement errors being independent and following spherically symmetric distributions, our measurement errors, that is, $\hat{\bbeta}_j^0 - \bbeta_j^0$ and $\hat{\bgamma}_j^0 - \bgamma_j^0$, are asymptotically Gaussian with weak dependence. It leaves a non-trivial problem to show the similar consistency results of $\hat{\Ab}^0$ given the weakly dependent errors. The estimator $\hat{\Gb}^0$ is also consistent and the convergence rate is dependent on the asymptotic results of $\hat{\bphi}^0$ and also the estimation consistency of $\hat{\Ab}^0$.
\end{remark}

We start with proving the consistency of estimator $\hat{\Ab}^0$. We first define for $s=2,\cdots,p$ that $L_{s}^0 (\ba) = q^{-1} \| {\Bb}_{[,s]}^0 + {\bGamma}^0 (\ba -\ba_s^0)\|_1 = q^{-1} \sum_{j=1}^q | \beta_{js}^0 + (\bgamma_j^0)^{\T} (\ba -\ba_s^0)| $ and let $\ba_s^0 = \argmin_{\ba} L_s^0(\ba)$. Similarly, we let the loss function $\hat{L}_{s} (\ba)$ to be $ \hat{L}_s(\ba) = q^{-1} \| \hat{\Bb}_{[,s]}^0 + \hat{\bGamma}^0 (\ba -\ba_s^0) \|_1 = q^{-1} \sum_{j=1}^q |\hat{\beta}_{js}^0 + (\hat{\bgamma}_j^0)^{\T} (\ba -\ba_s^0)|$. We denote $\hat{\ba}_s^0 -\ba_s^0= \argmin_{\ba} \hat{L}_s(\ba)$.



\vspace{0.2in}

\noindent {\bf Preliminary Convergence of $\hat{\ba}_s^0$.}
 From the convergence rates for $\hat{\bphi}^0$, we have $| \hat{\bbeta}_{js}^0 - \bbeta_{js}^0|  \le \| \hat{\bbeta}_j^0 - \bbeta_j^0 \| = O_p(  \zeta_{nq,p}^{-1})$ and $\| \hat{\bgamma}_j^0 - \bgamma_j^0\| = O_p(  \zeta_{nq,p}^{-1})$, then for any $\ba$, we have
 \begin{equation}
     |\hat{L}_{s} (\ba) -  L_{s}^0 (\ba)| \le  q^{-1} \sum_{j=1}^q \big[| \hat{\bbeta}_{js}^0 - \bbeta_{js}^0|  + \| \hat{\bgamma}_j^0 - \bgamma_j^0\| \|(\ba -\ba_s^0)\|\big] = O_p\big( (1+\|(\ba -\ba_s^0)\|)\zeta_{nq,p}^{-1}\big). \label{eq:loss function diff}
 \end{equation}
 
 As $\zero = \arg\min_{\ba} L_{s}^0 (\ba)$ is unique, for any $\bv$, $\|\nabla_{\bv}L_s^0(\zero)\| > c$ for some $c > 0$. 
 Also since $L_s^0$ is convex, we have $\|\nabla_{\bv}L_s^0(\zero + t\bv )\| > c^{\prime}$. For any $\ba_s \neq \zero$, we have for certain $t$ that
 \begin{align}
     L_s^0(\ba_s) & = L_s^0(\zero) + \|\nabla_{\bv}L_s^0( t\bv )\|\| \ba_s\| \nonumber \\
     & \geq L_s^0(\zero) + c \| \ba_s\|. \label{eq:Ls0 expand}
 \end{align}
By definition we have $\hat{L}_s (\hat{\ba}_s^0-\ba_s^0) \leq \hat{L}_s (\zero)$. Then by taking $\ba_s=\hat \ba_s^0-\ba_s^0$ in~\eqref{eq:Ls0 expand} we have
\begin{align*}
    \hat{L}_s (\zero) - L_s^0(\zero) &\geq \hat{L}_s ({\hat\ba}_s^0-\ba_s^0) - L_s^0(\zero) \\&\geq c \| \hat{\ba}_s - \ba_s^0\|- |  \hat{L}_s (\hat{\ba}_s^0-\ba_s^0) - L_s^0(\hat{\ba}_s-\ba_s^0)|.
\end{align*}
As a result, we have by \eqref{eq:loss function diff} that
\begin{align*}
    c \| \hat{\ba}_s^0 - \ba_s^0\| &\le |  \hat{L}_s (\hat{\ba}_s^0-\ba_s^0) - L_s^0(\hat{\ba}_s^0-\ba_s^0)| +  \hat{L}_s (\zero) - L_s^0(\zero) \\&\lesssim\big(1+\|\hat\ba_s^0-\ba_s^0\|\big) \zeta_{nq,p}^{-1}.
\end{align*}
Therefore we have $\|\hat\ba_s^0-\ba_s^0\|=O_p(\zeta_{nq,p}^{-1})$ and by definition we know $\hat\Ab^0=(\hat\ba_1^0,\cdots,\hat\ba_{p}^0)$, from which we conclude that 
 \begin{equation}
    \|\hat{\Ab}^0 - \Ab^0\|_F = O_p(p^{1/2}\zeta_{nq,p}^{-1}).  \nonumber
 \end{equation}

\vspace{0.2in}

 \noindent {\bf Refined convergence rates.}
 We next obtain a more stringent rate for $\hat{\Ab}^0$. We write $\hat L_s(\ba)$ as
\begin{align*}
     \hat L_s(\ba)=&q^{-1}\sum_{j=1}^q\big|\hat\beta_{js}^0+{(\hat\bgamma_j^0)}^{\T}(\ba-\ba_s^0)\big|\\=& q^{-1}\sum_{j=1}^q\big|\hat\beta_{js}^0-\beta_{js}^0+ \beta_{js}^0+({\bgamma_j^0}+\hat\bgamma_j^0-\bgamma_j^0)^\T(\ba-\ba_s^0)\big|\\=&\,q^{-1}\sum_{j=1}^q\big|\beta_{js}^*+\hat\beta_{js}^0- \beta_{js}^0 -(\ba_s^0)^\T(\hat\bgamma_j^0 - {\bgamma_j^0})+\ba^\T\Gb^0\bgamma_j^*+\ba^\T(\hat\bgamma_j^0 - {\bgamma_j^0})\big|
\end{align*}

To continue the proof, we introduce the following technical results for further derivation of the derivative of $\hat L_s(\ba)$.

\noindent {\bf Lemma~\ref{lemma:asym_estimate}.}~Under Assumption~\ref{assumption: psd covariance}--\ref{assumption:asymptotic normality}, we have 

   \noindent(\romannumeral1) \begin{align} \big\|[\cH ^{-1}(\bphi^0)\bS(\bphi^0)]_{[P_j]}&-[ \Hb_{Lff^\prime}^{-1}(\bphi^0)\bS_f(\bphi^0)]_{[P_j]}\big\|=O_p\Big(\frac{p(nq)^{3/\xi}}{\sqrt{nq}}\epsilon_{nq}\Big);\nonumber \\\big\|[\cH ^{-1}(\bphi^0)\Rb]_{[P_j]}\big\|&=O_p\Big(\frac{p(nq)^{3/\xi}}{\zeta_{nq,p}^2}\Big).\nonumber \end{align}
\noindent(\romannumeral2) \begin{align}\big\|[\cH ^{-1}(\bphi^0)\bS(\bphi^0)]_{[q(K+p)+K_i]}&-[\Hb_{Luu^\prime}^{-1}(\bphi^0)\bS_u(\bphi^0)]_{[K_i]}\big\|=O_p\Big(\frac{p^{3/2}(nq)^{3/\xi}}{\sqrt{nq}}\epsilon_{nq}\Big);\nonumber \\\big\|[\cH ^{-1}(\bphi^0)\Rb]_{[q(K+p)+K_i]}\big\|&=O_p\Big(\frac{p^{3/2}(nq)^{3/\xi}}{\zeta_{nq,p}^2}\Big) .\nonumber \end{align}

By \eqref{eq:normality_f1} in Lemma~\ref{thm:asymptotic normality} and Lemma~\ref{lemma:asym_estimate},
we write the first order derivative of $\hat L_s(\ba)$ as
\begin{equation*}\partial_{\ba}\hat L_s(\ba)=q^{-1}\sum_{j=1}^q\psi_{js}
\big(\ba, \bdelta_{js}\big)+O_p\big(p(nq)^{3/\xi}\zeta_{nq,p}^{-2}\big),
\end{equation*}
where $\psi_{js} (\ba, \bdelta_{js})$ is defined in~\eqref{eq:phi js def}. We next introduce the following result extended from Corollary 2.2 in 
\cite{he1996bahadur}.
    \begin{theorem}[\citeauthor{he1996bahadur}, \citeyear{he1996bahadur}]
Suppose there exists $\balpha_{s0}$ such that $\chi_{s}(\balpha_{s0})=\sum_{j=1}^q $ $\EE \psi_{js}(\bdelta_{js},\balpha_{s0}) $ $= 0$ with $\|\balpha_{s0}\|\le O_p(v_q)$ for some sequence $v_q=o_p(q^{-1/2})$ as $q\to \infty$. In a neighbourhood of $\balpha_{s0}$, $\chi_{s}(\ba)$ has a nonsingular derivative $D_q$ such that $\{D_q(\balpha_{s0})\}^{-1}=O(q^{-1})$ and $|D_q(\ba)-D_q(\balpha_{s0})|\le kq|\ba-\balpha_{s0}|$. Moreover, assume $q^{-1}\sum_{j=1}^q\psi_{js}(\bdelta_{js},\balpha_{s0})=O_p(v_q)$. 
Let $\hat\balpha_s=\arg\min_{\ba}\sum_{j=1}^q\psi_{js}(\bdelta_{js},\ba)$, then we have $\|\hat{\balpha}_{s}\|=O_p(v_q)$.\label{lemma:A rate}
    \end{theorem}

By Assumption~\ref{assumption:A consistency}, the conditions in Theorem~\ref{lemma:A rate} are satisfied with rate $\iota_{nq,p}$, and therefore the minimizer of $
q^{-1}\sum_{j=1}^q\psi_{js}(\bdelta_{js},\ba)$ denoted as $\hat{\ba}_{\psi}$ satisfies $\|\hat{\ba}_{\psi}\|= O_p(\iota_{nq,p})$. Then since $\hat\ba_s^0-\ba_{s}^0$ is the minimizer of $\hat L_s(\ba)$, it also minimizes $\big|\partial_{\ba}\hat L_s(\ba)\big|$. Then by a similar argument as in proving the preliminary convergence rates, we have $\| \hat{\ba}_s^0-\ba_s^0 - \hat{\ba}_{\psi}\| \lesssim | \hat L_s(\ba) - q^{-1}\sum_{j=1}^q\psi_{js}(\bdelta_{js},\ba)| =  O_p(p(nq)^{3/\xi} \zeta_{nq,p}^{-2})$.
Therefore, $\|\hat\ba_s^0-\ba_s^0\|= O_p(\iota_{nq,p}\vee p(nq)^{3/\xi} \zeta_{nq,p}^{-2})$, which implies
\begin{align}
    \| \hat{\Ab}^0 - \Ab^0 \|_F = O_p\Big(p^{1/2}\iota_{nq,p}\vee p^{3/2}(nq)^{3/\xi} \zeta_{nq,p}^{-2}\Big).\label{eq_consist_A_final}
\end{align}

Next we show the consistency rate of $\hat\Gb^0$. 
We start with bounding $\| (q^{-1} (\hat{\bGamma}^0)^{\T} \hat{\bGamma}^0)^{1/2} -(q^{-1} (\bGamma^0)^{\T} \bGamma^0)^{1/2} \|_F$.

Write 
\begin{align}
    &q^{-1}(\hat{\bGamma}^0)^{\T} \hat{\bGamma}^0 - q^{-1}(\bGamma^0)^{\T} \bGamma^0\nonumber \\
     =& q^{-1}\sum_{j=1}^q \hat{\bgamma}_j^{0} (\hat{\bgamma}_j^0)^{\T} - \bgamma_j^0 (\bgamma_j^0)^{\T} \nonumber\\
     = & q^{-1}\sum_{j=1}^q\Big[ (\hat{\bgamma}_j^0 - \bgamma_j^0)(\hat{\bgamma}_j^0 - \bgamma_j^0)^{\T} + {\bgamma}_j^0 (\hat{\bgamma}_j^0 - \bgamma_j^0)^{\T} + (\hat{\bgamma}_j^0 - \bgamma_j^0) ({\bgamma}_j^0)^{0\T}\Big].\label{eq_bound_trans_gamma}
\end{align}
From Lemma~\ref{prop:average consistency}, the first term on the right side of~\eqref{eq_bound_trans_gamma} can be bounded by $ O_p ( \zeta_{nq,p}^{-2})$. Moreover, combining \eqref{eq:normality_f1}, (i) of Lemma~\ref{lemma:asym_estimate}, we can show that
\begin{align*}
    &\Big\| \sum_{j=1}^q(\hat{\bgamma}_j^0 - \bgamma_j^0) (\bgamma_j^0)^{\T}\Big\|\\
    = &\Big\| \sum_{j=1}^q \Big[\sum_{i=1}^nl^{\prime\prime}_{ij}(w_{ij}^0)\bZ_i^0(\bZ_i^0)^\T \Big]^{-1}\Big(\sum_{i=1}^nl^{\prime}_{ij}(w_{ij}^0)\bZ_i^0\Big)(\bgamma_j^0)^{\T} \Big\|+O_p\Big(\frac{pq(nq)^{3/\xi}}{\zeta_{nq, p}^2}\Big)\\  = &O_p\Big( \sqrt{\frac{pq}{n}}\epsilon_{nq}\vee \frac{pq(nq)^{3/\xi}}{\zeta_{nq, p}^2}\Big).
\end{align*}
Plugging this estimate to the last two terms in \eqref{eq_bound_trans_gamma} we have \begin{equation}\|q^{-1}(\hat{\bGamma}^0)^{\T} \hat{\bGamma}^0 - q^{-1}(\bGamma^0)^{\T} \bGamma^0 \|_F =O_p\big( \sqrt{p/(nq)}\epsilon_{nq}\vee p(nq)^{3/\xi}\zeta_{nq,p}^{-2}\big).\label{eq_for_convergenceG1}\end{equation}

Similarly from \eqref{eq:normality:u1}, 
(ii) of Lemma~\ref{lemma:asym_estimate}, we have 
\begin{equation*}
    \big\|\sum_{i=1}^q(\hat\bU_i^0-\bU_i^0){\bU_i^0}^\T\big\|=O_p\big(\sqrt{p/(nq)}\epsilon_{nq}\vee p^{3/2}(nq)^{3/\xi}n\zeta_{nq,p}^{-2}\big).
\end{equation*}
By Lemma~\ref{prop:average consistency} we know $\|\hat\Ub^0-\Ub^0\|=O_p(n^{1/2}\zeta_{nq,p}^{-1})$. Then by $(\hat\Ub^0)^\T\Xb=(\Ub^0)^\T\Xb=\zero_{K\times p}$, we have 
  \begin{align*}
          (\hat{\Ub}^0 +& \Xb (\hat\Ab^{0})^{\T})^{\T} (\hat{\Ub}^0 + \Xb (\hat\Ab^{0})^{\T}) 
      -  (\Ub^0 + \Xb (\Ab^{0})^{\T})^{\T} (\Ub^0 + \Xb (\Ab^{0})^{\T})\\ =&
      \big[(\hat{\Ub}^0)^\T\hat{\Ub}^0-(\Ub^0)^\T{\Ub}^0\big]+\big[(\hat\Ab^0)^\T(\Xb^\T\Xb)\hat\Ab^0-(\Ab^0)^\T(\Xb^\T\Xb)\Ab^0\big]\\
      =&\sum_{i=1}^n\big[(\hat\bU_i^0-\bU_i^0)(\hat\bU_i^0-\bU_i^0)^\T+\bU_i^0(\hat\bU_i^0-\bU_i^0)^\T+(\hat\bU_i^0-\bU_i^0)(\bU_i^0)^\T\big]\\&+(\hat\Ab^0-\Ab^0)^\T(\Xb^\T\Xb)\hat\Ab^0-(\Ab^0)^\T(\Xb^\T\Xb)(\hat\Ab^0-\Ab^0),
 \end{align*}
together with the convergence rate of $\hat\Ab^0$, we have
\begin{align}
    \Big\|n^{-1}(\hat{\Ub}^0 + \Xb (\hat\Ab^{0})^{\T})^{\T} (\hat{\Ub}^0 + \Xb (\hat\Ab^{0})^{\T}) 
      -  n^{-1}(\Ub^0 + \Xb (\Ab^{0})^{\T})^{\T} (\Ub^0 + \Xb (\Ab^{0})^{\T})\Big\|\nonumber\\=O_p \Big(\frac{\sqrt{p}}{\zeta_{nq,p}}\vee \frac{p^{3/2}(nq)^{3/\xi}}{\zeta_{nq,p}^2}\Big)\label{eq_for_convergenceG2}
\end{align}
Finally combine \eqref{eq_for_convergenceG1}, \eqref{eq_for_convergenceG2} and by Lemma \ref{lemma_for_convergenceG} we have
\begin{equation}\| \hat{\Gb}^0 - \Gb^0\|_F = O_p \left(  \frac{\sqrt{p}}{\zeta_{nq,p}}\vee \frac{p^{3/2}(nq)^{3/\xi}}{\zeta_{nq,p}^2}\right).\label{eq_final_convergence_G2}
\end{equation}
The consistency rate under Assumption~\ref{assumption:A consistency} can be similarly derived for $\hat\Gb^0$ based on \eqref{eq_consist_A_final}.

\subsection{Proof of Lemma~\ref{thm:asymptotic normality}}
\begin{lemma}[Asymptotic Normality]
\label{thm:asymptotic normality}
Under Assumptions~\ref{assumption: psd covariance}--\ref{assumption:asymptotic normality}, 
we have the asymptotic distributions for the constrained maximum likelihood estimators $\hat{\bU}_i^0$ as 
\begin{equation*}
    \sqrt{q}(\bPhi_{i\gamma}^0)^{1/2}(\hat{\bU}_i^0-\bU_i^0) \overset{d}{\to} \cN(\zero_K, \Ib_K )\quad \text{ if }   p^{3/2}\sqrt{q}(nq)^{3/\xi}\zeta_{nq,p}^{-2}\to 0\text{, for all }i\in[n]
\end{equation*}
and $\hat{\bbf}_j^0$ as
\begin{equation*}
     \sqrt{n} \ba^{\T} (\bPhi_{jz}^0)^{1/2}(\hat{\bbf}_j^0-\bbf_j^0)\overset{d}{\to} \cN(0,1)  \quad \text{ if }   p\sqrt{n}(nq)^{3/\xi}\zeta_{nq,p}^{-2}\to 0\text{, for all }j\in[q].
\end{equation*}
for any $\ba\in\RR^{K+p}$ with $\|\ba\|_2 = 1$, 
where asymptotic variance matrices $\ba^\T(\bPhi_{i\gamma}^0)^{-1}\bb$ and $\ba^\T(\bPhi_{jz}^0)^{-1}\bb$
can be consistently estimated by 
    \begin{align*}
\ba^\T\big(\hat\bPhi_{i\gamma}^0\big)^{-1}\bb&=q\ba^\T\big\{\sum_{j=1}^q\hat l_{ij}^{\prime\prime}\hat\bgamma_j^0(\hat\bgamma_j^{0})^{\T}\big\}\big\{\sum_{j=1}^q(\hat l_{ij}^{\prime})^2\hat\bgamma_j^0(\hat\bgamma_j^{0})^{\T}\big\}^{-1}\{\sum_{j=1}^q\hat l_{ij}^{\prime\prime}\hat\bgamma_j^0(\hat\bgamma_j^{0})^{\T}\big\}\bb;\\ 
\ba^\T\big(\hat\bPhi_{jz}^0\big)^{-1}\bb&=n\bb^\T\big\{\sum_{i=1}^n\hat l_{ij}^{\prime\prime}\hat\bZ_i^0(\hat\bZ_i^0)^{\T}\big\}\big\{\sum_{i=1}^n(\hat l_{ij}^{\prime})^2\hat\bZ_i^0(\hat\bZ_i^0)^{\T}\big\}^{-1}\{\sum_{i=1}^n\hat l_{ij}^{\prime\prime}\hat\bZ_i^0(\hat\bZ_i^0)^{\T}\big\}\bb,
\end{align*}
with $\hat\bZ_i^0=((\hat\bU_i^0)^\T,\bX_i^\T)^\T$, $\hat{l}_{ij}^{\prime} = l_{ij}^{\prime} (\hat{\bgamma}_j^{0\intercal}\hat{\bU}_i^0  + \hat{\bbeta}_j^{0\intercal}\bX_i)$ and $\hat l_{ij}^{\prime\prime}= l_{ij}^{\prime\prime} (\hat{\bgamma}_j^{0\intercal}\hat{\bU}_i^0  + \hat{\bbeta}_j^{0\intercal}\bX_i)$ for any $\ba,\bb\in\RR^{K+p}$ with $\|\ba\|=\|\bb\|=1$.
\end{lemma}
\begin{remark}\label{remkr_sandwitch}
    {Here we introduce sandwich estimators for the covariance matrices $\ba^\T\bPhi_{i\gamma}^0\bb$ and $\ba^\T\bPhi_{jz}^0\bb$. For $\bPhi_{i\gamma}^0$, its estimator $\hat\bPhi_{i\gamma}^0$ can be considered as an approximation to
    \begin{align*}
        &q^{-1}\sum_{j=1}^q\EE [l_{ij}^{\prime\prime}(w_{ij}^0)]\bgamma_j^0(\bgamma_j^0)^\T \\&= \big\{q^{-1}\sum_{j=1}^q\EE [l_{ij}^{\prime\prime}(w_{ij}^0)]\bgamma_j^0(\bgamma_j^0)^\T\big\}\big\{q^{-1}\sum_{j=1}^q\EE [l_{ij}^{\prime}(w_{ij}^0)^2]\bgamma_j^0(\bgamma_j^0)^\T\big\}^{-1}\big\{q^{-1}\sum_{j=1}^q\EE [l_{ij}^{\prime\prime}(w_{ij}^0)]\bgamma_j^0(\bgamma_j^0)^\T\big\},
    \end{align*}as $\EE [l_{ij}^{\prime\prime}(w_{ij})] = \EE [l_{ij}^{\prime}(w_{ij}^0)^2]$. Though $\sum_{j=1}^q\hat l_{ij}^{\prime\prime}\hat\bgamma_j^0(\hat\bgamma_j^{0})^{\T}$ and $\sum_{j=1}^q(\hat l_{ij}^{\prime})^2\hat\bgamma_j^0(\hat\bgamma_j^{0})^{\T}$ can also be shown to consistently estimate $\bPhi^0_{i\gamma}$, we adopt the sandwich estimator because it tends to perform better in practice. A similar rationale applies to estimating $\bPhi^0_{jz}$.} Further, we write the asymptotic covariance matrix $(\bPhi_{jz}^0)^{-1}$ in a $2\times 2$ block defined as follows:
\begin{align*}
    \bPhi_{\gamma,j}^0=\big[(\bPhi_{jz}^0)^{-1}\big]_{[1:K,1:K]},\;\bPhi_{\beta,j}^0=\big[(\bPhi_{jz}^0)^{-1}\big]_{[(K+1):(K+p),(K+1):(K+p)]},\\\;\bPhi_{\gamma\beta, j}^0=\big[(\bPhi_{jz}^0)^{-1}\big]_{[1:K,(K+1):(K+p)]},\;\bPhi_{\beta\gamma, j}^0=\big[(\bPhi_{jz}^0)^{-1}\big]_{[(K+1):(K+p),1:K]}.
\end{align*}
Then $\bPhi_{\gamma,j}^0$ can be consistently estimated by the plug-in estimator as
\begin{align*}
      \hat\bPhi_{\gamma,j}^0=n\big(\sum_{i=1}^n\hat l_{ij}^{\prime\prime}\hat\bU_i^0(\hat\bU_i^{0})^{\T}\big)^{-1}\big\{\sum_{i=1}^n(\hat l_{ij}^{\prime})^2\hat\bU_i^0(\hat\bU_i^{0})^{\T}\big\}(\sum_{i=1}^n\hat l_{ij}^{\prime\prime}\hat\bU_i^0(\hat\bU_i^{0})^{\T}\big)^{-1},
      \end{align*}
      and for any $\ba,\bb\in\RR^{p}$ with $\|\ba\|=\|\bb\|=1$, $\bb^\T\bPhi_{\beta,j}^0\ba$ and $\bPhi_{\gamma\beta,j}^0\ba$ can be consistently estimated by the plug-in estimator as
      \begin{align*}
      \bb^\T\hat\bPhi_{\beta,j}^0\ba&=n\bb^\T\big(\sum_{i=1}^n\hat l_{ij}^{\prime\prime}\bX_i\bX_i^\T\big)^{-1}\big\{\sum_{i=1}^n(\hat l_{ij}^{\prime})^2\bX_i\bX_i^\T\big\}(\sum_{i=1}^n\hat l_{ij}^{\prime\prime}\bX_i\bX_i^\T\big)^{-1}\ba;\\
      \hat\bPhi_{\gamma\beta,j}^0\ba&=n\big(\sum_{i=1}^n\hat l_{ij}^{\prime\prime}\hat\bU_i^0(\hat\bU_i^{0})^{\T}\big)^{-1}\big\{\sum_{i=1}^n(\hat l_{ij}^{\prime})^2\hat\bU_i^0\bX_i^{\T}\big\}(\sum_{i=1}^n\hat l_{ij}^{\prime\prime}\bX_i\bX_i^{\T}\big)^{-1}\ba,
\end{align*}
with $\hat\bPhi_{\beta\gamma,j}^0=(\hat\bPhi_{\gamma\beta,j}^0)^\T$.

\end{remark}
\begin{remark}
    Lemma~\ref{thm:asymptotic normality} provides the asymptotic distributions for all individual estimators under Conditions $1^{\prime}$--$2^{\prime}$. For the asymptotic distributions of $\hat{\bbeta}_j^0$ and $\hat{\bgamma}_j^0$, the scaling conditions imply $p = o(n^{1/4} \wedge q/\sqrt{n})$ up to small order term. For the asymptotic normality of $\hat{\bU}_i^0$, the scaling condition implies $p = o\{q^{1/3} \wedge (n^2/q)^{1/5}\}$ up to small order term. These asymptotic results lay the foundation for deriving the asymptotic distributions for $\hat{\bphi}^*$.
\end{remark}
\begin{proof}
To establish the asymptotic distribution for the estimators, we expand the first order condition $S(\hat{\bphi}^0)$ to high orders as follows
\begin{equation}
    0=\bS(\bphi^0)+\mathcal{H}(\bphi^0)(\hat{\bphi}^0-{\bphi}^0)+\frac12\Rb,\label{eq:nomality_1}
\end{equation} where
\begin{align*}
    \Rb&=(\Rb_f^\T ,\Rb_ U^\T  )^\T ;\\
[\Rb_{f}]_{[(j-1)(K+p)+r]}&=(\hat{\bphi}^0-{\bphi}^0)^\T  \partial_{{\bphi}{\bphi} \gamma_{jr}}\cL({\bphi}^\flat)(\hat{\bphi}^0-{\bphi}^0);\\
[\Rb_{f}]_{[(j-1)(K+p)+K+s]}&=(\hat{\bphi}^0-{\bphi}^0)^\T  \partial_{{\bphi}{\bphi} \beta_{js}}\cL({\bphi}^\flat)(\hat{\bphi}^0-{\bphi}^0);\\
[\Rb_{ U}]_{[(i-1)K+r]}&=(\hat{\bphi}^0-{\bphi}^0)^\T  \partial_{{\bphi}{\bphi} U_{ir}}\cL({\bphi}^\flat)(\hat{\bphi}^0-{\bphi}^0),
\end{align*}
for $i\in[n]$, $j\in[q]$, $r\in[K]$, $s\in[p]$ and ${\bphi}^\flat$ in the segment between $\hat{\bphi}^0$ and ${\bphi}^0$.
Invert the matrix $\mathcal{H}(\bphi^0)$ we have
\begin{align}
\hat\bbf_j^0-\bbf_j^0&=-[\mathcal H^{-1}(\bphi^0)\bS(\bphi^0)]_{[P_j]}-\frac12[\mathcal H^{-1}(\bphi^0)\Rb]_{[P_j]}\label{eq:normality_f1};\\
    \hat{\bU}_i^0-\Ub^0_i&=-[\mathcal H^{-1}(\bphi^0)\bS(\bphi^0)]_{[q(K+p)+K_i]}-\frac12[\mathcal H^{-1}(\bphi^0)\Rb]_{[q(K+p)+K_i]}\label{eq:normality:u1}.
\end{align}

Implied by Lemma~\ref{lemma:asym_estimate}, we know that the first term in \eqref{eq:normality_f1} can be approximated by $[ \Hb_{Lff^\prime}^{-1}(\bphi^0)\bS_f(\bphi^0)]_{[P_j]}$, which is asymptotic normal as will be shown in the following. The approximation error and the second term in \eqref{eq:normality_f1} are considered as negligible terms compared to $[ \Hb_{Lff^\prime}^{-1}(\bphi^0)\bS_f(\bphi^0)]_{[P_j]}$ when $p\sqrt{n}(nq)^{3/\xi}\zeta_{nq,p}^{-2}\to 0$. For $\hat{\bU}_i^0-\Ub^0_i$ similarly by Lemma~\ref{lemma:asym_estimate} we know that it can be approximated by $[ \Hb_{Luu^\prime}^{-1}(\bphi^0)\bS_u(\bphi^0)]_{[K_i]}$ with a negligible term when $p^{3/2}\sqrt{q}(nq)^{3/\xi}\zeta_{nq,p}^{-2}\to 0$.

To show the asymptotic normality of the leading term $[ \Hb_{Lff^\prime}^{-1}(\bphi^0)\bS_f(\bphi^0)]_{[P_j]}$ in \eqref{eq:normality_f1}, we verify the Lindeberg-Feller condition \citep{ash2000probability}. For any $\ba\in\RR^{p+K}$ and $\|\ba\|=1$, define the triangular array $\{\bV_{ni,j}(\ba)\}_{i,n}$ for each $j$:
\begin{equation*}
    \bV_{ni,j}(\ba) = \ba^\T\big\{n^{-1}\sum_{t=1}^n
    \EE [l_{tj}^{\prime\prime}(w_{tj}^0)]\bZ_t^0(\bZ_t^0)^\T\big\}^{-1}l_{ij}^{\prime}(w_{ij}^0)\bZ_i^0.
\end{equation*}
Note $\EE\bV_{ni,j}(\ba)=\zero$ as $\EE l_{ij}^{\prime}(w_{ij}^0)=0$. Then by independence we know that \begin{align*}s_n^2 :=&\, n^{-1}\sum_{i=1}^n\text{Var}\big(\bV_{ni,j}(\ba)\big)\\=&\, n^{-1}\sum_{i=1}^n \ba^\T\big\{n^{-1}\sum_{t=1}^n
    \EE [l_{tj}^{\prime\prime}(w_{tj}^0)]\bZ_t^0(\bZ_t^0)^\T\big\}^{-1}\big\{\EE[l_{ij}^{\prime}(w_{ij}^0)^2]\bZ_i^0(\bZ_i^0)^\T\big\}\\&\,\big\{n^{-1}\sum_{t=1}^n
    \EE [l_{tj}^{\prime\prime}(w_{tj}^0)]\bZ_t^0(\bZ_t^0)^\T\big\}^{-1}\ba\\=&\, \ba^\T\big\{n^{-1}\sum_{t=1}^n
    \EE [l_{tj}^{\prime\prime}(w_{tj}^0)]\bZ_t^0(\bZ_t^0)^\T\big\}^{-1}\ba.\end{align*}
    Note that 
    \begin{equation*}
        \bZ_i^0 =\begin{pmatrix}
            (\Gb^{\ddagger})^\T&-(\Gb^{\ddagger})^\T\Ab^\ddagger\\\zero&\Ib_p
        \end{pmatrix}\begin{pmatrix}
            \bU_i^*\\\bX_i
        \end{pmatrix}
    \end{equation*}
    Then by Assumption~\ref{assumption:asymptotic normality} and definition of $\bPhi^0_{jz}$ in \eqref{eq_definephi_jz}, we have
    \begin{align*}s_n^2 :=&\, \ba^\T\begin{pmatrix}
            (\Gb^{\ddagger})^{-\T}&-\Ab^\ddagger\\\zero&\Ib_p
        \end{pmatrix}\big\{n^{-1}\sum_{t=1}^n
    \EE [l_{tj}^{\prime\prime}(w_{tj}^0)]\bZ_t^*(\bZ_t^*)^\T\big\}^{-1}\begin{pmatrix}
            (\Gb^{\ddagger})^{-1}&\zero\\-(\Ab^\ddagger)^\T&\Ib_p
        \end{pmatrix}\ba\\&\overset{p}{\to}\,\ba^\T\big\{\bPhi^0_{jz}\big\}^{-1}\ba.\end{align*}
        Next we compute
        \begin{align*}
            &\,\frac{1}{n}\sum_{i=1}^n\EE\big[\bV_{ni,j}(\ba)^21(|\bV_{ni,j}(\ba)|>\epsilon \sqrt n)\big]\\\le &\,\frac{1}{n}\sum_{i=1}^n\Big\{\EE[\bV_{ni,j}(\ba)^4]\PP\big(|\bV_{ni,j}(\ba)|>\epsilon \sqrt n\big)\Big\}^{1/2}\\\le &\, \max_iM^{2/\xi}\PP\left(|l_{ij}^{\prime}(w_{ij}^0)|>\epsilon \sqrt n\lambda_{\max}\big\{-n^{-1}\sum_{i=1}^n\EE[l_{ij}^{\prime\prime}(w_{ij}^0)]\bZ_i^0(\bZ_i^0)^\T\big\}^{-1}\right)^{1/2}\\\le &\,\max_iM^{2/\xi}\PP\left(|l_{ij}^{\prime}(w_{ij}^0)|>b_U^{-1}\epsilon \lambda_{\max}\big\{n^{-1}\sum_{i=1}^n\bZ_i^0(\bZ_i^0)^\T\big\}^{-1}\sqrt n\right)^{1/2}\\\overset{p}{\to}&\,0,
        \end{align*}
        as $\lambda_{\max}\big\{n^{-1}\sum_{i=1}^n\bZ_i^0(\bZ_i^0)^\T\big\}$ can be bounded from above and $l_{ij}^{\prime}(w_{ij}^0)$ is sub-exponential with $\|l_{ij}^{\prime}(w_{ij})\|_{\varphi_1}\le M$. Then by Lindeberg-Feller central limit theorem we have
        \begin{equation}
            \sqrt{n}\ba^\T\big\{n^{-1}\sum_{t=1}^n
    \EE [l_{tj}^{\prime\prime}(w_{tj}^0)]\bZ_t^0(\bZ_t^0)^\T\big\}^{-1}n^{-1}\sum_{i=1}^nl_{ij}^{\prime}(w_{ij}^0)\bZ_i^0 \overset{d}{\to} \cN\big(0,\ba^\T\big\{\bPhi_{jz}^0\big\}^{-1}\ba\big). \label{eq:lf clt}
        \end{equation}
        Finally by weak law of large numbers (WLLN) we know that $n^{-1}\sum_{t=1}^n
    l_{tj}^{\prime\prime}(w_{tj}^0)\bZ_t^0(\bZ_t^0)^\T\overset{p}{\to}n^{-1}\sum_{t=1}^n
    \EE [l_{tj}^{\prime\prime}(w_{tj}^0)]\bZ_t^0(\bZ_t^0)^\T$. We then conclude with
 $\Hb_{Lff^{\prime}}(\bphi^0)$ is block-diagonal that for any $\ba\in\RR^{K+p}$ with $\|\ba\|=1$,
\begin{equation*}
  \sqrt{n}  \ba^\T(\bPhi_{jz}^0)^{1/2}[ \Hb_{Lff^\prime}^{-1}(\bphi^0)\bS_f(\bphi^0)]_{[P_j]} \overset d\rightarrow \cN (0,1).
\end{equation*}
Expanding~\eqref{eq:normality_f1}, we have 
\begin{align*}
    \sqrt{n}(\hat\bbf_j^0-\bbf_j^0) =& - \sqrt{n}  [ \Hb_{Lff^\prime}^{-1}(\bphi^0)\bS_f(\bphi^0)]_{[P_j]}  \\
    & - \sqrt{n} \{ [\mathcal H^{-1}(\bphi^0)\bS(\bphi^0)]_{[P_j]}-  [ \Hb_{Lff^\prime}^{-1}(\bphi^0)\bS_f(\bphi^0)]_{[P_j]} \} - \frac{\sqrt{n}}{2} [\mathcal H^{-1}(\bphi^0)\Rb]_{[P_j]}.
\end{align*}
Under the scaling condition that $\sqrt{n} {p(nq)^{3/\xi}}/{\zeta_{nq, p}^2} \rightarrow 0 $, the small-order terms in Lemma~\ref{lemma:asym_estimate} (i) can be omitted. Hence we have the following asymptotic distribution
\begin{equation}
    \sqrt{n}\ba^\T(\bPhi_{jz}^0)^{1/2}(\hat\bbf_j^0-\bbf_j^0)  \rightarrow \cN (0,1), \quad \text{ if } p\sqrt{n}(nq)^{3/\xi}\zeta_{nq,p}^{-2}\to 0.\label{eq_aymp_dist_f}
\end{equation}
As a result, for $\ba \in \RR^{K+p}$ with $\|\ba\|=1$
\begin{align*}
    \sqrt{n} \ba^{\T} (\bPhi_{jz}^0)^{1/2} (\hat{\bbeta}_j^0 - \bbeta_j^0) \overset{d}\rightarrow \cN (0, 1). 
\end{align*}

Similarly for $\hat{\bU}_i^0 - \bU_i^0$, we can expand~\eqref{eq:normality:u1} as 
\begin{align*}
   \sqrt{q} (  \hat{\bU}_i^0-\bU_i^0) = & -\sqrt{q} [\Hb_{Luu^\prime}^{-1}(\bphi^0)\bS_u(\bphi^0)]_{[K_i]}  - \sqrt{q}\big \{ [\cH ^{-1}(\bphi^0)\bS(\bphi^0)]_{[q(K+p)+K_i]} \\
   & - [\Hb_{Luu^\prime}^{-1}(\bphi^0)\bS_u(\bphi^0)]_{[K_i]}\big\} - \frac{\sqrt{q}}{2}[\mathcal H^{-1}(\bphi^0)\Rb]_{[q(K+p)+K_i]}.
\end{align*}
A similar procedure to verify the Lindeberg-Feller condition yields
\begin{equation*}
     \sqrt{q} [\Hb_{Luu^\prime}^{-1}(\bphi^0)\bS_u(\bphi^0)]_{[K_i]}  \overset d\rightarrow  \cN(\bm{0}, \bPsi_{i\gamma}^0).
\end{equation*}
Under the scaling condition that $\sqrt{q} {p^{3/2}(nq)^{3/\xi}}/{\zeta_{nq, p}^2} \rightarrow 0 $, the small-order terms in Lemma~\ref{lemma:asym_estimate} (2) are negligible, so we have
\begin{equation*}
     \sqrt{q} (\bPhi_{i\gamma}^0)^{1/2} (  \hat{\bU}_i^0-\bU_i^0)  \overset d\rightarrow  \cN(\bm{0}, \Ib_K),\quad  \text{ if } p^{3/2}\sqrt{q}(nq)^{3/\xi}\zeta_{nq,p}^{-2}\to 0.\label{eq_aymp_dist_for_u}
\end{equation*}
where $\bPhi_{i\gamma}^0 = (\bPsi_{i\gamma}^0)^{-1} \bOmega_{i\gamma}^0(\bPsi_{i\gamma}^0)^{-1}$.
The asymptotic variances $\bPhi_{i\gamma}^0$ and $\bb^\T\bPhi_{jz}^0\ba$ can be consistently estimated by
 \begin{align*}
\hat\bPhi_{i\gamma}^0&=q\big(\sum_{j=1}^q\hat l_{ij}^{\prime\prime}\hat\bgamma_j^0(\hat\bgamma_j^{0})^{\T}\big)\big\{\sum_{j=1}^q(\hat l_{ij}^{\prime})^2\hat\bgamma_j^0(\hat\bgamma_j^{0})^{\T}\big\}^{-1}(\sum_{j=1}^q\hat l_{ij}^{\prime\prime}\hat\bgamma_j^0(\hat\bgamma_j^{0})^{\T}\big);\\ 
\bb^\T\hat\bPhi_{jz}^0\ba&=n\bb^\T\big(\sum_{i=1}^n\hat l_{ij}^{\prime\prime}\hat\bZ_i^0(\hat\bZ_i^0)^{\T}\big)\big\{\sum_{i=1}^n(\hat l_{ij}^{\prime})^2\hat\bZ_i^0(\hat\bZ_i^0)^{\T}\big\}^{-1}(\sum_{i=1}^n\hat l_{ij}^{\prime\prime}\hat\bZ_i^0(\hat\bZ_i^0)^{\T}\big)\ba,
\end{align*}
for any $\ba,\bb\in\RR^{K+p}$ with $\|\ba\|=\|\bb\|=1$,
and the consistency of $\hat\bPhi_{i\gamma}$ and $\hat\bPhi_{jz}^0$ can be shown by Assumption~\ref{assumption:smoothness}, Lemma~\ref{prop:ind consistency} and WLLN.

\end{proof}

\section{Proofs of Other Technical Lemmas}\label{sec:prove_other}
\subsection{Proof of Lemma~\ref{lemma_for_convergenceG}}
\begin{lemma}
    Suppose $\Ab,\Bb\in\mathbb{R}^{K\times K}$ are positive-definite matrices and the eigen-gap of $\Ab\Bb$ is positive. If symmetric matices $\hat\Ab$ and $\hat\Bb $ are consistent estimates of $\Ab$ and $\Bb $ with rate $\nu\to0$.
    \begin{equation*}
        \|\hat\Ab-\Ab\|=O_p(\nu),\quad\|\hat\Bb -\Bb \|=O_p(\nu).
    \end{equation*}
     Suppose $\Gb$ satisfies $
        \Gb^\T\Ab \Gb= \Gb^{-1}\Bb \Gb^{-\T} =$ {diagonal} and correspondingly $\hat \Gb$ satisfies $
         \hat \Gb^\T\hat\Ab\hat \Gb=\hat \Gb^{-1}\hat\Bb \hat \Gb^{-\T}=$  diagonal.
    Then
    \begin{equation*}
        \|\hat \Gb- \Gb\|=O_p(\nu).
    \end{equation*}\label{lemma_for_convergenceG}
\end{lemma}
\begin{proof}
Since $\hat\Ab$ and $\hat \Bb $ are of finite dimension, we have $\|\hat\Ab-\Ab\|_{\max}=O_p(\nu)$, $\|\hat\Bb -\Bb \|_{\max}=O_p(\nu)$, with \begin{align*}\|\hat\Ab\hat\Bb -\Ab\Bb \|&\le \|\hat\Ab-\Ab\|\|\hat\Bb-\Bb\|+\|\Ab\|\|\hat\Bb-\Bb\|+\|\hat\Ab-\Ab\|\|\Bb\|\\&=O_p(\nu).\end{align*} Let $\gamma$ be the eigen-gap of $\Ab\Bb $. By Weyl's theorem, we know that when $\nu$ is small enough, the eigen-gap of $\hat\Ab\hat\Bb $ is larger than $\gamma/2$. Similarly if we denote $\gamma_L=\min\big\{\lambda_{\min}(\Ab),\lambda_{\min}(\Bb )\big\}>0$ and $\gamma_U=\max\big\{\lambda_{\max}(\Ab),\lambda_{\max}(\Bb )\big\}>0$, then we have $$\min\big\{\lambda_{\min}(\hat\Ab),\lambda_{\min}(\hat\Bb )\big\}>\gamma_L/2\text{ and }\max\big\{\lambda_{\max}(\hat\Ab),\lambda_{\max}(\hat\Bb )\big\}<2\gamma_U$$ when $\nu$ exceeds some threshold. Next we discuss only when $\nu$ exceeds these two thresholds.

    Next let the singular value decomposition (SVD) of $\Ab$ be $\cU_1\bUpsilon_1\cU_1^\T$, and the singular value decomposition of $\Ab^{1/2}\Bb \Ab^{1/2}$ be $\cU_2\bUpsilon_2\cU_2^\T$ with $\Ab^{1/2}$ defined as $\cU_1\bUpsilon_1^{1/2}\cU_1^\T$. Then $\Gb=\Ab^{-1/2}\cU_2\bUpsilon_2^{-1/4}$ is the unique solution up to a permutation of $\bUpsilon_2$. The uniqueness is a result of the eigenvalues of $\Ab\Bb $ being different, which implies that there exists a unique order for the diagonal entries of $\bUpsilon_2$, and thus $\cU_2$ is uniquely determined.
    Similarly we write the SVD of $\hat \Ab$ and $\hat \Ab^{1/2}\hat \Bb \hat \Ab^{1/2}$ as $\hat \cU_1\hat \bUpsilon_1\hat \cU_1^\T$ and $\hat \cU_2\hat \bUpsilon_2\hat \cU_2^\T$. Therefore $\hat\Gb$ can be uniquely expressed as $\hat \Ab^{-1/2}\hat \cU_2\bUpsilon^{-1/4}$. The existence of $\hat\Ab^{-1/2}$ is implied by $\lambda_{\min}(\hat\Ab)>\gamma_L/2$, $\lambda_{\max}(\hat\Ab)<\gamma_U/2$ and uniqueness is implied by the fact that the eigen-gap of $\hat\Ab\hat\Bb $ is larger than $\gamma/2$.

    Implied by $\|\hat\Ab-\Ab\|=O_p(\nu)$ and that $\lambda_{\max}(\hat\Ab),\lambda_{\max}(\Ab), \lambda_{\min}(\hat\Ab),\lambda_{\max}(\Ab)$ are bounded in $(2\gamma_U,\gamma_L/2)$, we have
    \begin{align}
        \big\|\hat\Ab^{-1/2}-\Ab^{-1/2}\big\|&\le \big\|\big(\hat\Ab^{-1/2}+\Ab^{-1/2}\big)\big\|\big\|\hat\Ab^{-1}-\Ab^{-1}\big\|\nonumber\\&\le \big\|\big(\hat\Ab^{-1/2}+\Ab^{-1/2}\big)\big\|\big\|\hat\Ab^{-1}\big\|\big\|\Ab^{-1}\big\|\big\|\hat\Ab-\Ab\big\|\nonumber\\&=O_p(\nu).\label{lemma_eq_dkthm}
    \end{align}
     Next by variant of Davis-Kahan theorem \citep{yu2015useful},
    \begin{equation*}
        \big\|\hat \cU_2-\cU_2\big\|=O_p(\nu)\mbox{, and }\big|[\hat\bUpsilon_2]_{[r,r]}-[\bUpsilon_2]_{[r,r]}\big|=O_p(\nu), \forall r\in[K].
    \end{equation*}
    as we have fixed the order of the distinctive diagonal entries in $\hat\Lambda_2$ and $\Lambda_2$. Combine these all together we have 
    \begin{equation*}
        \|\hat \Gb- \Gb\| = O_p(\nu).
    \end{equation*}
\end{proof}

\subsection{Proof of Lemma~\ref{lemma:first-order derivative}}
\begin{lemma}[First-order concentration]
\label{lemma:first-order derivative}

Under Assumptions~\ref{assumption: psd covariance}--\ref{assumption:smoothness}, we have estimates for the first order derivatives on $\bphi^0$ as:
 \begin{align*}
       \| n^{1/2} \bS_U  (\bphi^0)\|_2 &= O_p \Big(\sqrt{\frac{\log n}{q}} \Big);\\
        \|q^{1/2}\bS_f(\bphi_v^0)\|_{2} &= O_p\Big( \sqrt{\frac{p\log qp}{n}}  \Big).
    \end{align*}
Scaled by matrix $\Db_p$ we can write
    \begin{equation}
    \|\Db_q^{1/2} \bS(\bphi^0)\|_2 = O_p\big(\zeta_{nq,p}^{-1}\big), \nonumber
\end{equation}
or $\|\Db_q^{1/2}\bS(\bphi^0)\|_{\zeta}\le O_p(p^{1/2}(nq)^{\epsilon-1/2})$ when $\zeta>1/\epsilon$.
\end{lemma}

\begin{proof}
Assumption~\ref{assumption:smoothness} in the main text imposes regularity conditions for parameters $\bphi^0$ that $\bU_i^0$ and $\bgamma_j^0$ are bounded for $i \in [n]$ and $j \in [q]$. Based on the assumptions, we have 
$\|l_{ij}^{\prime} U_{ir}^0\|_{\varphi_1} \le 2M^2$, for $j=1,\dots, q$ and $r = 1, \dots ,K$.


 The variable $l_{ij}^{\prime} U_{ir}^0$ is sub-exponential with $\|l_{ij}^{\prime} U_{ir}^0\|_{\varphi_1} \le 2M^2$ and $\EE(l_{ij}^{\prime} U_{ir}^0) = 0$. By the definition of sub-exponential variable, there exists a universal constant $c >0$ such that, for any $0 < \nu < (2cM^2)^{-1}$, the following inequality holds:
\begin{equation}
	\EE[\exp\{\nu l_{ij}^{\prime}(w_{ij}) U_{ir}^0\}] \le\exp\{2\nu^2c^2M^2\}, \nonumber
\end{equation}
and further we have 
\begin{align}
	 P\Big\{(nq)^{-1}\sum_{i=1}^n |l_{ij}^{\prime}(w_{ij}) U_{ir}^0| \ge t \Big\} 
 & = P\left[ \exp\left\{\nu\sum_{i=1}^n |l_{ij}^{\prime}(w_{ij}) U_{ir}^0|\right\} \ge \exp(nq\nu t)\right] \nonumber \\
	  \le & \EE \left[ \exp\left\{\nu\sum_{i=1}^n |l_{ij}^{\prime}(w_{ij}) U_{ir}^0|\right\}\right] \exp(-nq\nu t) \nonumber \\
	  = & \exp(-nq\nu t)  \prod_{i=1}^n\EE[\exp\{\nu l_{ij}^{\prime}(w_{ij}) U_{ir}^0\}] \nonumber \\
	  \le & \exp \{2nc^2M^2 \nu^2 - nqt\nu\}. \label{lemma5: bernstein exp}
\end{align}
Next we will minimize~\eqref{lemma5: bernstein exp} by minimizing the quadratic expression of $\nu$ in the exponent. Differentiating the quadratic term with respect to $\nu$ and set it to 0 gives the optimizer
\begin{equation}
	\nu^{*} = \frac{qt}{4c^2M^2}. \nonumber
\end{equation}
Under $\nu = \nu^*$, \eqref{lemma5: bernstein exp} $\le \exp \{-nq^2t^2/(8c^2M^2)\}$. If $\nu^*$ falls outside of $(0,(2cM^2)^{-1})$, then \eqref{lemma5: bernstein exp} is minimized at $\nu = (2cM^2)^{-1}$. So \eqref{lemma5: bernstein exp} $\le \exp \{(2nc -nqt)/(2cM^2) \}$. In conclusion,
\begin{align}
	P\Big\{|(nq)^{-1}\sum_{i=1}^n l_{ij}^{\prime}(w_{ij}) U_{ir}^0| \ge t \Big\} \le  2\exp \Big\{-\min \Big(\frac{nq^2t^2}{8c^2M^2}, \frac{nqt}{2cM^2} \Big) \Big\}. \nonumber
\end{align}
 By the union bound inequality, we have
\begin{equation*}
	P\Big\{\max_j \| (nq)^{-1} \sum_{i=1}^n l_{ij}^{\prime}(w_{ij}) \bU_i^0\|_{\infty}\ge t \Big\} \le  2qK\exp \Big\{-\min \Big(\frac{nq^2t^2}{8c^2M^2}, \frac{nqt}{2cM^2} \Big) \Big\}. \nonumber
\end{equation*}
Take $t= cM (\log q )^{1/2}n^{-1/2}q^{-1}$, and we have
\begin{equation}
    \| \bS_{\gamma} (\bphi^0)\|_{\infty} =  \max_j \| (nq)^{-1} \sum_{i=1}^n l_{ij}^{\prime}(w_{ij}) \bU_i^0\|_{\infty} = O_p \Big(\sqrt{\frac{\log q}{nq^2}} \Big), \nonumber
\end{equation}
which further leads to
\begin{equation}
   \|q^{1/2} \bS_{\gamma} (\bphi^0)\|_{2}  = O_p \Big(\sqrt{\frac{\log q}{n}} \Big). \label{eq:S_gamma}
\end{equation}

Using similar techniques, we have
\begin{equation}
      \| \bS_{U} (\bphi^0)\|_{\infty} = O_p \Big(\sqrt{\frac{\log n}{n^2q}} \Big)\text{, and }
    \| n^{1/2} \bS_U  (\bphi^0)\|_2 = O_p \Big(\sqrt{\frac{\log n}{q}} \Big). \label{eq:S_U}
\end{equation}
For $\|q^{1/2} \bS_{\beta}(\bphi^0)\|_2$, because  \begin{equation*}
	P\Big\{\max_j \| (nq)^{-1} \sum_{i=1}^n l_{ij}^{\prime}(w_{ij}) \bX_i\|_{\infty}\ge t \Big\} \le  2qp\exp \Big\{-\min \Big(\frac{nq^2t^2}{c^2M^2}, \frac{nqt}{2cM^2} \Big) \Big\}. 
\end{equation*}
Take $t= cM(\log qp)^{1/2}n^{-1/2}q^{-1}$, and we have
\begin{equation}
    \|\bS_{\beta}(\bphi^0)\|_{\infty} = \max_j \| (nq)^{-1} \sum_{i=1}^n l^{\prime}_{ij}(w_{ij}) \bX_i\|_{\infty} = O_p \left( \sqrt{\frac{\log qp}{nq^2}}  \right), \nonumber
\end{equation}
and similarly we obtain
\begin{equation}
     \|q^{1/2} \bS_{\beta}(\bphi^0)\|_{2} = O_p\left( \sqrt{\frac{p\log qp}{n}}  \right).\label{eq:S_beta}
\end{equation}

\end{proof}

\subsection{Proof of Lemma~\ref{lemma:estimation for 1st derivative}}
\begin{lemma}
\label{lemma:estimation for 1st derivative}
Under Assumption~\ref{assumption:smoothness}, the first order derivatives of $\check\cL$ on $\bpsi^0$ has following estimates:
\begin{align*}
    \Big| \check \bS_{\theta}({\bpsi}^0)^\T(\hat\btheta-\btheta^0)\Big|\lesssim { \delta_{nq}^{-1/2}} \sqrt{\frac{\log n}{nq}}\|\check\bTheta-\bTheta^0\|_F;\\
\big|\check \bS_b(\bpsi^0)^\T(\check\bbeta-\bbeta^0)\Big|\lesssim\sqrt{\frac{p\log pq}{n}}\frac{1}{\sqrt{q}}\|\check\bbeta-\bbeta^0\|.
\end{align*}
\end{lemma}
Define $\Lb_\theta $ as $q \times n$ matrix with $\partial_{\theta_{ij}}  l_{ij}(\theta_{ij}+(\bbeta_j)^\T\bX_i)$ in the $j$-th row and $i$-th column. By Assumption~\ref{assumption:smoothness} we know that entries of $\Lb(\bpsi^0)$ are independent and have a finite fourth moment, therefore by the concentration inequality of \cite{latala2005some}, we have 
\begin{equation}\|\Lb_{\theta}(\bpsi^0)\|  = O_p\big(\sqrt {(n + q)\log n}\big)\label{eq_ltheta_est}\end{equation}

Next for the first order derivatives of $\check\cL$ on $\bpsi^0$ with respect to $\btheta$:
    \begin{equation}
 [\check\bS_{\theta}(\bpsi^0)]_{[(j-1)q+i]}=\frac{\partial L(\Yb|{\bpsi})}{\partial \theta_{ij} }  = -(nq)^{-1} l_{ij}^{\prime}(\theta_{ij}+\bbeta_j^\T\bX_i), \label{eq:order1 pi_ij} 
\end{equation}
by $Tr(AB)\le\|A\|\|B\|_*\le rank(B)\|A\|\|B\|_F$ we have
\begin{equation}\Big| \check \bS_\theta({\bpsi}^0)^\T(\hat\btheta_v-\btheta_v^0)\Big|= Tr\big[ \Lb_{\theta}({\bpsi}^0)(\hat\bTheta-\bTheta^0)\big]\le 2K(nq)^{-1}\| \Lb_{\theta}\| \|\hat\bTheta-\bTheta^0\|_F.\label{B7}\end{equation}
And for the first order derivatives of $\check\cL$ on $\bpsi^0$ with respect to $\bbeta$
\begin{equation*}
    [\check\bS_{\beta}(\bpsi^0)]_{[(j-1)p+1:jp]}= -(nq)^{-1} \sum_{i=1}^n  l_{ij}^{\prime}(\theta_{ij}^0+(\bbeta_j^0)^\T\bX_i) \bX_i \label{eq:order1 beta_j}
\end{equation*}
Then the estimation results for the first order derivative can be derived directly from \eqref{eq_ltheta_est}, \eqref{B7} and Lemma \ref{lemma:first-order derivative}.

\subsection{Proof of Lemma~\ref{lemma:Hcheck}}
\begin{lemma}
\label{lemma:Hcheck}
Under Assumptions~\ref{assumption: psd covariance}--\ref{assumption:smoothness}, there exist some $ \check\gamma>0$ such that $\lambda_{\min}( \check \Db_{q}^{1/2} \check{\mathcal{H}}({\bpsi}) \check \Db_{q}^{1/2})$ $> \check\gamma$ for any $\psi\in\check\cB(D)$ where
\begin{equation}
     \check{\mathcal{H}}({\bpsi})=\frac{\partial^2[ \check{\mathcal{L}}({\bpsi})]}{\partial^2{\bpsi}},
\end{equation}
and $\check\Db_q$ is defined as 
\begin{equation*}
    \check\Db_q=\begin{pmatrix}
        nq\Ib_{nq}\\&q\Ib_{qp}
    \end{pmatrix}.
\end{equation*}
\end{lemma}
    The Hessian matrix will be given by $\check\cH(\bpsi)=\check\Hb_L(\bpsi)+\check\Hb_P(\bpsi)$ with
\begin{equation*}
    \check{\Hb}_L=\begin{bmatrix}
        \check\Hb_{L\theta\theta^{\prime}}&\check\Hb_{L\theta \beta^{\prime}}\\
        \check\Hb_{L\beta\theta^{\prime}}&\check\Hb_{L\beta \beta^{\prime}}
    \end{bmatrix},
\end{equation*}
and 
\begin{equation*}
    \check\Hb_{P}=c\begin{bmatrix}
        \frac{1}{n^2q}\sum_{s=1}^p\Ib_q\otimes(\Xb_{[,s]}\Xb_{[,s]}^\T)\\&\zero_{pq\times qp},
    \end{bmatrix}
\end{equation*}
where the non-zero parts of blocks of $\check\Hb_L$ are given as 
\begin{align*}
    \big[\check\Hb_{L\theta\theta^{\prime}}(\bpsi)\big]&_{[(j-1)q+i,(j-1)q+i]}=-(nq)^{-1}l_{ij}^{\prime\prime}(\theta_{ij}+\bbeta_j^\T\bX_i);\\
    \big[\check\Hb_{L\theta\beta^{\prime}}(\bpsi)\big]&_{[(j-1)q+i,(j-1)p+1:jp]}=-(nq)^{-1}l_{ij}^{\prime\prime}(\theta_{ij}+\bbeta_j^\T\bX_i)\bX_i;\\
    \big[\check\Hb_{L\beta\beta^{\prime}}(\bpsi)\big]&_{[(j-1)p+1:jp,(j-1)p+1:jp]}=-(nq)^{-1}\sum_{i=1}^n l_{ij}^{\prime\prime}(\theta_{ij}+\bbeta_j^\T\bX_i)\bX_i\bX_i^\T.
\end{align*}
First note that
\begin{align*}
    \check\cH_j:=&\begin{bmatrix}
        nq\big[\check\Hb_{L\theta\theta^{\prime}}(\bpsi)\big]_{[((j-1)q+1):jq,((j-1)q+1):jq]}&n^{1/2}q\big[\check\Hb_{L\theta\beta^{\prime}}(\bpsi)\big]_{[((j-1)q+1):jq,(j-1)p+1:jp]}\\n^{1/2}q\big[\check\Hb_{L\theta\beta^{\prime}}(\bpsi)\big]_{[(j-1)p+1:jp,((j-1)q+1):jq]}&q\big[\check\Hb_{L\beta\beta^{\prime}}(\bpsi)\big]_{[(j-1)p+1:jp,(j-1)p+1:jp]}
    \end{bmatrix}\\&+c\begin{bmatrix}
        n^{-1}\Xb\Xb^\T&-n^{-1/2}\Xb\\-n^{-1/2}\Xb^\T&\Ib_p
    \end{bmatrix}\\=&\sum_{i=1}^n(-l_{ij}^{\prime\prime}(w_{ij})-c)\begin{pmatrix}
        \one_i^{(n)}\\n^{-1/2}\bX_i
    \end{pmatrix}\begin{pmatrix}
        \one_i^{(n)}\\n^{-1/2}\bX_i
    \end{pmatrix}^\T+c\begin{bmatrix}
        \Ib_n+n^{-1}\Xb\Xb^\T&\\&n^{-1}\Xb^\T\Xb+\Ib_p
    \end{bmatrix}.
\end{align*}
We let $c<b_l$ and the first term is positive definite. For the second term, since $\|n^{-1}\Xb\Xb^\T\|_{\max}$ $=O(p/n)$, there exist some $\gamma^{\prime}<1$ such that $\lambda_{\min}\big(\Ib_n+n^{-1}\Xb\Xb^\T\big)>\gamma^{\prime}$. For the lower-right part, by Assumption~\ref{assumption: psd covariance}, we also have $\lambda_{\min}\big(\Ib_p+n^{-1}\Xb^\T\Xb\big)>\gamma^{\prime}$.
 Furthermore, the zero eigenvalues for $ \check \Db_{q}^{1/2}\check{\Hb}_L\check \Db_{q}^{1/2}$ has multiplicity $pq$ with corresponding eigenvectors:
\begin{equation*}
     \check{\bnu}_{js}= \sqrt{q}\check{\bD}_{q}^{-1/2}\left\{\bm{1}_j^{(q)}\otimes\begin{pmatrix} \Xb_{[,s]}\\- \bm{1}_s^{(p)}\end{pmatrix}\right\},\;j\in[q],s\in[p].
\end{equation*}
Note that the eigenvectors do not depend on the parameters. Here $\bm{1}_j^{(q)}$ is a $q$-dimensional indicator vector, $\bm{1}_s^{(p)}$ is a $p$-dimensional indicator vector. Note
\begin{equation}
     \check{\bD}_{q}^{1/2}\check\Hb_P \check{\bD}_{q}^{1/2}= c\begin{bmatrix}\frac{1}{n}\sum_{s=1}^pI_q\otimes(\Xb_{[,s]}\Xb_{[,s]}^\T)\\&\bm0_{pq\times pq}\end{bmatrix}= c\sum_{j=1}^q\sum_{s=1}^p \check{\bnu}_{js}^{(1)} (\check{\bnu}_{js}^{(1)})^\T,\label{BB9}
\end{equation}
where $ \check{\bnu}_{js}^{(1)}=n^{-1/2}((\one_j^{(q)})^\T\otimes \bm \Xb_{[s,]}^\T,\bm0_{pq}^\T)^\T$. Further define $\check{\bnu}_{js}^{(2)}=\big(\zero_{nq}^\T,(\one_{(j-1)p+s}^{(qp)})^\T\big)$.
Then $ \check{\bnu}_{js}= \check{\bnu}_{js}^{(1)}- \check{\bnu}_{js}^{(2)}$.

Suppose $ \check{\mathcal{V}}=\mbox{span}\big\{ \check\Db_{q}^{1/2}\check{\bnu}_{js}:j=1,\cdots,q,s=1,\cdots,p\big\}$ is the null space of $ \check\Db_{q}^{1/2}\check{H}_L\check\Db_{q}^{1/2}$. Here $ \check{\mathcal{V}}$ is a constant space that does not depend on the parameters (but depends on the covariates $\bX_i$). For any $\bw\in\mathbb{R}^{nq+qp}$ $\|\bw\|=1$, let $\bv\in \check{\mathcal{V}}$ and $\bu\in \check{\mathcal{V}}^\perp$ such that $\bw=\alpha \bv+\beta\bu$ and $\|\bv\|=\|\bu\|=1$. So we have $\alpha^2+\beta^2=1$. Let $\bv=\sum_{j=1}^q\sum_{s=1}^p\lambda_{js} \check{\bnu}_{js}$. The coefficients $\{\lambda_{js}\}_{j,s}$ are unique as $\{ \check{\bnu}_{js}\}$ are linear independent. We also have $\bv^\T\check\bnu_{js}=\lambda_{js}+n^{-1}\Xb_{[,s]}^\T\big(\sum_{r=1}^p\lambda_{jr}\Xb_{[,r]}\big):=\lambda_{js}+\lambda_{js}(\Xb)$ and $\bv^\T\check\bnu_{js}^{(2)}=-\lambda_{js}$. Here $\lambda_{js}(\Xb)$ is the $s$th column of $(\lambda_{j1},\cdots,\lambda_{jp})^\T n^{-1}\Xb^\T\Xb$. Then \begin{equation}
\sum_{j=1}^q\sum_{s=1}^p\big(\lambda_{js}(\Xb)\big)^2\ge \lambda_{\min}(n^{-1}\Xb^\T\Xb)^2\sum_{j=1}^q\sum_{s=1}^p\lambda_{js}^2=O(1).\label{eq_contr_l2_0}\end{equation}
Then
\begin{align*}
    \bw^\T  \check \Db_{q}^{1/2} \check{\mathcal{H}}({\bpsi}) \check \Db_{q}^{1/2}\bw 
   =&\,(\alpha \bv+\beta\bu)^\T \Big[ \check{\bD}_{q}^{1/2} \check{H}_L \check{\bD}_{q}^{1/2}+ \check{\bD}_{q}^{1/2}\frac{\partial^2 \check P({\bpsi})}{\partial^2{\bpsi}} \check{\bD}_{q}^{1/2}\Big](\alpha \bv+\beta\bu)\\
    =&\,\beta^2\bu^\T  \check{\Hb}_L\bu+ c(\alpha \bv+\beta\bu)^\T\Big(\sum_{j=1}^q\sum_{s=1}^s \check{\bnu}_{js}^{(1)} \big(\check{\bnu}_{js}^{(1)}\big)^\T\Big)(\alpha \bv+\beta\bu)\\
    \ge&\,\beta^2 {\check{\gamma}}^{\prime}+c\sum_{j=1}^q\sum_{s=1}^p\Big(\beta\bu^\T
    \check\bnu_{js}^{(2)}+\alpha\lambda_{js}(\Xb)\Big)^2.
\end{align*}
It suffices to prove that there exist some $ \check\gamma>0$ that does not depend on $n,q$, such that 
\begin{equation}
    \beta^2 \check{\gamma}^{\prime}+c\sum_{j=1}^q\sum_{s=1}^p\Big(\beta\bu^\T
    \check\bnu_{js}^{(2)}+\alpha\lambda_{js}(\Xb)\Big)^2\ge  \check\gamma.\label{BB11}
\end{equation}
Suppose on the contrary that for any $\epsilon_0>0$ we can select a set of parameters such that the above lower bound for $\bw^\T  \check \Db_{q}^{1/2} \check{\mathcal{H}}({\bpsi}) \check \Db_{q}^{1/2}\bw $ is smaller than $\epsilon_0$.
Since all there terms are positive, we have 
\begin{align}
    \beta<\sqrt{\frac{\epsilon_0}{ \check \gamma}},&\quad\alpha>\sqrt{1-\frac{\epsilon_0}{ \check \gamma}};\nonumber\\
    \sum_{j=1}^q\sum_{s=1}^p\Big(\frac{\beta}{\alpha}\bu^\T  \check{\bnu}_{js}^{(2)}&+\lambda_{js}(\Xb)\Big)^2<\frac{\epsilon_0}{c\alpha^2}<\frac{\check\gamma \epsilon_0}{c(\check\gamma-\epsilon_0)}\label{eq_contrad1}.
\end{align}
In this case, by noting that $\{ \check{\bnu}_{js}^{(2)}\}_j$ is a set of orthogonal indicators, we know
\begin{equation}
    \frac{\beta^2}{\alpha^2}\sum_{j=1}^q\sum_{s=1}^p\big(\bu \check {\bnu}_{js}^{(2)}\big)^2< \frac{\epsilon_0 }
{ \check\gamma^{\prime}-\epsilon_0}.\label{eq_contrad2}\end{equation}
If we select $\epsilon_0$ small enough, \eqref{eq_contrad1} and \eqref{eq_contrad2} will contradict \eqref{eq_contr_l2_0}. This contradiction implies \eqref{BB11}.

\subsection{Proof of Lemma \ref{prop:min eigenvalue}}
\begin{lemma}[Local Convexity of $\cL$]
\label{prop:min eigenvalue}
    Under Assumptions~\ref{assumption: psd covariance}--\ref{assumption: Scaling}, 
    there exist some $\epsilon>0$ and $m$  such that\begin{equation}
    \Pr\Big\{\min_{{\bphi}\in\cB(D),\sqrt{p}\|\Db_{q}^{-1/2}({\bphi}-{\bphi}^0)\|\le m}\lambda_{\min}\big(\Db_{q}^{1/2}\mathcal{H}({\bphi})\Db_{q}^{1/2}\big)\geq\epsilon\Big\}\to1\label{eq_local_convex},
\end{equation}
where $\mathcal{H}({\bphi})= \partial_{\bphi\bphi}^2\cL(\Yb|{\bphi}) = - \partial_{\bphi\bphi}^2 L(\Yb|{\bphi})  - \partial_{\bphi\bphi}^2 P({\bphi})$.
\end{lemma}
 \begin{proof}
     The strategy for this proof is that we will first show the positive definiteness of the scaled Hessian matrix at $\bphi^0$, i.e. we first show $\lambda_{\min}\big(\Db_{q}^{1/2}\mathcal{H}({\bphi}^0)\Db_{q}^{1/2}\big)\ge \gamma^0$ for some positive $\gamma^0$. Then we prove the local convexity in the region $\cB(D)\cap \{\bphi:\sqrt{p}\|\Db_p^{-1/2}(\bphi-\bphi^0)\|\le m\}$.

\noindent {\bf Positive definiteness at $\bphi^0$.} First, it can be verified that all eigenvalues of $\Db_{q}^{1/2}{\Hb}_L\Db_{q}^{1/2}$ are all nonnegative. We also observe that
\begin{align}
   & \Db_{q}^{1/2}{\Hb}_L\Db_{q}^{1/2}+c\Vb_0\Vb_0^\T \nonumber\\
    =&-\Db_q^{-1/2}\sum_{i=1}^n\sum_{j=1}^q\{l_{ij}^{\prime\prime}(w_{ij}^0)+c\}\begin{pmatrix}
        \one_j^{(q)}\otimes \bZ_i^0\\\one_i^{(n)}\otimes \bgamma_j^0
    \end{pmatrix}\begin{pmatrix}
        \one_j^{(q)}\otimes \bZ_i^0\\\one_i^{(n)}\otimes \bgamma_j^0
    \end{pmatrix}^\T\Db_q^{-1/2}\nonumber\\
    &+c\begin{bmatrix}
        \Ib_q\otimes {c}{n}^{-1}\sum_{i=1}^n\bZ_i^0(\bZ_i^0)^\T\\&\Ib_n\otimes {c}{q}^{-1}\sum_{j=1}^q\bgamma_j^0(\bgamma_j^0)^\T
    \end{bmatrix}\nonumber
\\&+\begin{bmatrix}
    [V_0V_0^\T]_{[1:q(K+p),1:q(K+p)]} \\&[V_0V_0^\T]_{[q(K+p)+1:q(K+p)+nK,q(K+p)+1:q(K+p)+nK]}
\end{bmatrix}.\label{eq:semi definite decomposition}
\end{align} By selecting $0<c<b_L$, the first and third terms are positive semi-definite. In addition, as $n,q$ goes to infinity, and by the transition of true parameters $\bphi^*$ to the working parameters $\bphi^0$ satisfying identifiability conditions $1^{\prime}$--$2^{\prime}$, we have ${q}^{-1}\sum_{j=1}^q\bgamma_j^0(\bgamma_j^0)^\T$ and ${n}^{-1}\sum_{i=1}^n\bU_i^0(\bU_i^0)^\T$ converge to the diagonal matrix of eigenvalues of 
$\{ (\bSigma_u^0)^{1/2} \bSigma_{\gamma}^0 (\bSigma_u^0)^{1/2}\}^{1/2}$. Besides, we also have
$n^{-1}\sum_{i=1}^n\bX_i\bX_i^\T$ converges to  $\bSigma_x$ with $\lambda_{\min} (\bSigma_x) \ge \kappa^2$ and $n^{-1}\sum_{i=1}^n\bU_i^0\bX_i^{\T}=\bm{0}$. Hence, we know that 
\begin{equation}
    \lambda_{\min}\Big(\Db_{q}^{1/2}{\Hb}_L\Db_{q}^{1/2}+c\Vb_0\Vb_0^\T\Big)\ge \gamma,\label{eq_lowboundh_l}
\end{equation}
for some $0<\gamma< \min\{c,b_L\}$.
\eqref{eq:semi definite decomposition} also implies that there are a total of $K^2+Kp$ zero eigenvalues for $\Db_{q}^{1/2}{\Hb}_L\Db_{q}^{1/2}$ with corresponding eigenvectors:
\begin{align*}\Db_{q}^{-1/2}&\bnu_{rl},1\le r, l\le K;\\
\Db_{q}^{-1/2}&\bomega_{rs},1\le r\le K,1\le s\le p.
\end{align*}
Therefore, the column vectors of $\Vb_0$ form a non-degenerate basis of the null space of $\Db_{q}^{1/2}{\Hb}_L\Db_{q}^{1/2}$ when $P(\bphi)=0$. 
Note that when ${\bphi}={\bphi}^0$ the $l_2$-norm of columns of $\Vb_0$ is uniformly bounded. Denote by $L=\max_j\|[\Vb_0]_{[,j]}\|<\infty$. The inner products of these vectors are given as follows:
\begin{align*}
\langle\Db_q^{-1/2}\bnu_{r_1l_1},\Db_q^{-1/2}\bnu_{r_2l_2}\rangle &=\frac{1}{q}\sum_{j=1}^q\gamma_{jl_1}\gamma_{jl_2}1_{(r_1=r_2)}+\frac{1}{n}\sum_{i=1}^nU_{ir_1}U_{ir_2}1_{(l_1=l_2)}; \\
\langle\Db_q^{-1/2}\bnu_{rl},\Db_q^{-1/2}\bomega_{ks}\rangle &=\frac{1}{n}\sum_{i=1}^nU_{ir}X_{is}=0; \\
\langle\Db_q^{-1/2}\bomega_{r_1s_1},\Db_q^{-1/2}\bomega_{r_2s_2}\rangle &=\frac{1}{q}\sum_{j=1}^q\gamma_{jr_1}\gamma_{jr_2}1_{(s_1=s_2)}+\frac{1}{n}\sum_{i=1}^nX_{is_1}X_{is_2}1_{(r_1=r_2)}.
\end{align*}
Here the first $K^2$ column vectors are orthogonal to each other and to the last $Kp$ column vectors. But the last $Kp$ column vectors are not orthogonal to each other. However, since rank$(\Xb)=p$, these vectors are still linear independent and form a basis for the null space.

To further bound the eigenvalues of $\Db_{q}^{1/2} \mathcal{H}({\bphi}^0) \Db_{q}^{1/2}$, we define function $\mathscr{N}:\mathbb{R}^{q(K+p)+nK}$ $\to\mathbb{R}^{K^2+Kp}$ that maps any $\bm x\in\mathbb{R}^{nK+qK+qp}$ to its projection on the null space of $\Db_{q}^{1/2}\Hb_L(\bphi^0)$ $\Db_{q}^{1/2}$ under basis $\Vb_0(\bphi^0)$. Since this matrix is full-column rank, i.e., the column vectors are linearly independent as shown above, the projection is unique. Specifically, $\bx-\Vb_0\mathscr{N}(\bx)\perp null\big(\Db_{q}^{1/2}\Hb_L \Db_{q}^{1/2}\big)$. We write $\mathscr{N}(\bx)=(\mathscr{N}_1(\bx),\cdots,\mathscr{N}_{K^2+Kp}(\bx))^\T $ where $\mathscr{N}_i$ is the $i$th coordinates corresponding to vector $\Vb_{0[,i]}$. Suppose $\tau_r= ({n^{-1}
\sum_{i=1}^n(U_{ir}^0)^2})^{1/2}=({q^{-1}\sum_{j=1}^q(\gamma_{jr}^0)^2})^{1/2}$. 
We first observe that at $\bphi^0$, we have for $r_2\neq l_2$
\begin{align*}
    \langle\Db_q^{-1/2}\bu_{r_2l_2},\Db_q^{-1/2}\bnu_{r_1l_1}\rangle=&~\frac{1}{q}(\sum_{j=1}^q\gamma_{jl_1}^0\gamma_{jl_2}^01_{(r_1=r_2)}+\sum_{j=1}^q\gamma_{jl_1}^0\gamma_{jr_2}^01_{(r_1=l_2)})1_{(r_2> l_2)}\\&-\frac{1}{n}(\sum_{i=1}^nU_{ir_1}^0U_{ir_2}^01_{(l_1=l_2)}+\sum_{i=1}^nU_{ir_1}^0U_{il_2}^01_{(l_1=r_2)})1_{(r_2< l_2)}\\=&~\tau_{r_1}^21_{(r_1=r_2>l_2=l_1)}+\tau_{l_1}^21_{(l_1=r_2> l_2= r_1)}\\&-\tau_{r_1}^21_{(r_1=r_2<l_2=l_1)}-\tau_{l_1}^21_{(l_1=r_2< l_2= r_1)}; \\
\langle\Db_q^{-1/2}\bu_{rl},\Db_q^{-1/2}\bomega_{ks}\rangle=&~\frac{1}{n}(\sum_{i=1}^nU_{il}^0X_{is}1_{(r=k)}+\sum_{i=1}^nU_{ir}^0X_{is}1_{(l=k)})=0;\\
\langle\Db_q^{-1/2}\bb_{ks},\Db_q^{-1/2}\bnu_{rl}\rangle=&~\frac{1}{n}(\sum_{i=1}^nX_{is}U_{il}^01_{(r=k)}+\sum_{i=1}^nX_{is}U_{ir}^01_{(l=k)})=0; \\
\langle\Db_q^{-1/2}\bb_{k_1s_1},\Db_q^{-1/2}\bomega_{k_2s_2}\rangle=&~\frac{1}{n}\sum_{i=1}^nX_{is_1}X_{is_2}1_{(r_1=r_2)}.
\end{align*}
Therefore, we get the projection of first $K^2$ column vectors of $\Vb_p$ given by
\begin{align*}
    \mathscr{N}&\big([\Vb_p]_{[K(r)]}\big)=\one^{(K^2+Kp)}_{[K(r)]}\\
    \mathscr{N}&\big([\Vb_p]_{[K(r,l)]}\big)=\frac{\tau_l^2}{\tau_l^2+\tau_h^2}\one^{(K^2+Kp)}_{[K(r,l)]}+\frac{\tau_h^2}{\tau_h^2+\tau_l^2}\one^{(K^2+Kp)}_{[(l-1)K+r]}.
\end{align*}
For the projection of $\Db_q^{-1/2}\bb_{ks}$, the first $K^2$ components of $\mathscr{N}(\Db_q^{-1/2}\bb_{ks})$ are zero. Also note that $[\Vb_p]_{[K^2+(k-1)p+s]}-\Db_q^{-1/2}\omega_{ks}$ is orthogonal to all columns vectors of $\Vb_0$. Then the projection on span($[\Vb_0]_{[K^2+1:K^2+Kp]})$ is $\Db_q^{-1/2}\bomega_{ks}$ and therefore 
\begin{equation*}
    \mathscr{N}\big([\Vb_p]_{[K^2+(k-1)p+s]}\big)=\one^{(K^2+Kp)}_{[K^2+(k-1)p+s]}
\end{equation*}
 Define $\mathcal{N}$ the assembled coordinates, i.e., assemble the coordinates under $\Vb_0$ as the column vectors of $\mathcal{N}$. The matrix $\mathcal{N}({\bphi}^0)$ is an invertible matrix when set on $\bphi^0$ as $\tau_h\neq \tau_l$ for any $h\neq l$ when $n,q$ are large enough where $\tau_r$ converges to the $r$th diagonal element in $\{ (\bSigma_u^0)^{1/2} \bSigma_{\gamma}^0 (\bSigma_u^0)^{1/2}\}^{1/2}$. By Assumption~\ref{assumption: psd covariance}, the diagonal values in this matrix are distinct. 
 We also define the following residual vectors:
\begin{align*}
    \bnu_{hl}^{(R)}=&\Db_{q}^{-1/2}\bnu_{hl}-\Vb_0\mathscr{N}\big(\Db_{q}^{-1/2}\bnu_{hl}\big);\\\bm \bb_{ks}^{(R)}=&\Db_{q}^{-1/2}\bb_{ks}-\Vb_0\mathscr{N}\big(\Db_{q}^{-1/2}\bb_{ks}\big)=\Db_{q}^{-1/2}\begin{pmatrix}
        \bGamma_{[,k]}\otimes\one_{K+s}^{(K+p)}\\\zero_{nK}
    \end{pmatrix},
\end{align*} and $\mathcal{R}=\big[\mathcal{R}_1,\mathcal{R}_2\big]$,
where 
\begin{align*}
    \mathcal{R}_1=&\big[ \zero_{K^2+Kp},\bnu_{12}^{(R)},\cdots,\bnu_{1D}^{(R)},\bnu_{21}^{(R)},\zero_{K^2+Kp},\bnu_{23}^{(R)},\cdots,\bnu_{KK-1}^{(R)},\zero_{K^2+Kp}\big];\\\mathcal{R}_2=&\big[\bb_{1}^{(R)},\cdots,\bb_{p}^{(R)},\bb_{21}^{(R)},\cdots,\bb_{qp}^{(R)}\big].
\end{align*}
Here all column vectors of $\cR_1$ and $\cR_2$ are orthogonal to span$(\Vb_0)$. Next on ${\bphi}^0$ we have 
\begin{align*}\Db_{q}^{1/2}\Hb_P({\bphi}^0)\Db_{q}^{1/2}=&c\Db_{q}^{-1/2}\big(\sum_{r=1}^K\bnu_{rr}^0(\bnu_{rr}^0)^\T+\sum_{h\neq l}\bu_{hl}^0(\bu_{hl}^0)^\T+\sum_{r=1}^K\sum_{s=1}^p\bb_{rs}^0(\bb_{rs}^0)^\T\big)\Db_{q}^{-1/2}\\=&c(\Vb_0({\bphi}^0)\mathcal{N}({\bphi}^0)+\mathcal{R}({\bphi}^0)) (\Vb_0({\bphi}^0)\mathcal{N}({\bphi}^0)+\mathcal{R}({\bphi}^0))^\T\\=&c\Vb_0({\bphi}^0)\mathcal{N} ({\bphi}^0) \mathcal{N}({\bphi}^0)^\T \Vb_0({\bphi}^0)^\T + c\begin{bmatrix}\mathcal{R}_1({\bphi}^0)&\mathcal{R}_2({\bphi}^0)\end{bmatrix}\mathcal{N}({\bphi}^0)^\T \Vb_0({\bphi}^0)^\T\\&+c\Vb_0({\bphi}^0)\mathcal{N}({\bphi}^0)\begin{bmatrix}\mathcal{R}_1({\bphi}^0)^\T\\\mathcal{R}_2({\bphi}^0)^\T\end{bmatrix}+c\mathcal{R}_2({\bphi}^0)\mathcal{R}_1({\bphi}^0)^\T+c\mathcal{R}_2({\bphi}^0)\mathcal{R}_2({\bphi}^0)^\T.
\end{align*}
Define by $\mathcal{V}$ the space spanned by the column vectors of $\Vb_{0}$, namely, the null space of $\Db_{q}^{1/2}\Hb_L\Db_{q}^{1/2}$. Then for any $\bz$, $\|\bz\|=1$, let $\bv\in\mathcal{V}$ and $\bu\in \mathcal{V}^\perp$ such that $\bz=\alpha\bv+\beta\bu$ and $\|\bv\|=\|\bu\|=1$. Still, we have $\alpha^2+\beta^2=1$. Denote $\bn_0=\mathscr{N}(\bz)=\mathscr{N}(\bv)$ and ${\bn}_v=\mathcal{N}({\bphi}^0)^\T \Vb_{0}({\bphi}^0)^\T \Vb_{0}({\bphi}^0)\bn_0$. By definition, $\bv=\Vb_{0}({\bphi}^0)\bn_0$. By \eqref{eq_lowboundh_l}
\begin{align}
    \bz^\T \Db_{q}^{1/2}&\{\Hb_L(\bphi^0)+\Hb_P(\bphi^0)\}\Db_{q}^{1/2}\bz\nonumber\\&=\beta^2\bu^\T  \Db_{q}^{1/2}\Hb_L({\bphi}^0)\Db_{q}^{1/2}\bu +c\alpha^2\bv^\T \Vb_{0}(\bphi^0)\mathcal{N}(\bphi^0)  \mathcal{N}(\bphi^0)^\T \Vb_{0}(\bphi^0)^\T\bv\nonumber\\&+2c\alpha\beta\bu^\T\begin{bmatrix}\mathcal{R}_1(\bphi^0),\mathcal{R}_2(\bphi^0)\end{bmatrix}\mathcal{N}({\bphi}^0)^\T \Vb_{0}({\bphi}^0)^\T \bv\nonumber\\
    &+b\beta^2\bu^\T [\mathcal{R}_1(\bphi^0)\mathcal{R}_1(\bphi^0)^\T+\mathcal{R}_2(\bphi^0)\mathcal{R}_2(\bphi^0)^\T]\bu\nonumber\\\le& \beta^2\gamma+c\alpha^2\Big[\sum_{r=1}^{K}[{\bn}_v(\bphi^0)]_{K(r)}^2\Big]+c\sum_{1\le h\neq l\le K}\Big[(\bnu_{hl}^{(R)}(\bphi^0))^\T\bz+\alpha[{\bn}_v(\bphi^0)]_{[K(h,l)]}\Big]^2\nonumber\\&+c\sum_{r=1}^K\sum_{s=1}^p\Big[(\bb_{rs}^{(R)}(\bphi^0))^\T \bz+\alpha[{\bn}_v(\bphi^0)]_{[K^2+(r-1)p+s]}\Big]^2.\label{BB15}
\end{align}
Next we prove \eqref{BB15} has positive upper-bound independent of $n$ and $q$ by contradiction. Suppose for any $\epsilon>0$ we have $\eqref{BB15}<\epsilon$. Let $L_m=\|\mathcal{N}(\bphi^0)^\T \Vb_{0}(\bphi^0)^\T\|$. It is not hard to see $L_m^2=\sum_{r=1}^{K^2+Kp}[{\bn}_v(\bphi^0)]_{r}^2<\infty$. Then by $\eqref{BB15}<\epsilon$ we know that each term in \eqref{BB15} will be smaller than $\epsilon$, which implies
\begin{align}
&\beta<\sqrt{\frac{\epsilon}{\gamma}}\label{eq_contr_1}\\
    &\sum_{r=1}^K[\bn_v(\bphi^0)]_{K(r)}^2<\frac{\epsilon}{c\alpha^2};\label{eq_contr_2}\\&
    -\big|(\bnu_{hl}^{(R)}(\bphi^0))^\T\bz\big|-\sqrt{\frac{\epsilon}{c}}<\alpha[{\bn}_v(\bphi^0)]_{[K(h,l)]}< -\big|(\bnu_{hl}^{(R)}(\bphi^0))^\T\bz\big|+\sqrt{\frac{\epsilon}{c}};\label{eq_contr_3}\\&-\big|(\bb_{rs}^{(R)}(\bphi^0))^\T \bz\big|-\sqrt{ \frac{\epsilon}{c}}<\alpha[{\bn}_v(\bphi^0)]_{[K^2+(r-1)p+s]}<-\big|(\bb_{rs}^{(R)}(\bphi^0))^\T \bz\big|+\sqrt{ \frac{\epsilon}{c}}.\label{eq_contr_4}
\end{align}
First we observe \eqref{eq_contr_1} implies $\alpha>\sqrt{1-{\epsilon}/{\gamma}}$.
Takes \eqref{eq_contr_1} to \eqref{eq_contr_2} and $L_m^2=\sum_{r=1}^{K^2+Kp}$ $[{\bn}_v(\bphi^0)]_{r}^2$ we have
\begin{equation*}
    \sum_{r=1}^{K^2+Kp}[{\bn}_v(\bphi^0)]_{r}^2-\sum_{r=1}^K[\bn_v(\bphi^0)]_{K(r)}^2> L_m^2-\frac{\epsilon\gamma}{c(\gamma-\epsilon)}.
\end{equation*}
Then plug in \eqref{eq_contr_3} and \eqref{eq_contr_4} we have 
\begin{align*}
    &L_m^2-\frac{\epsilon\gamma}{c(\gamma-\epsilon)}\\
    <&(1-\frac{\epsilon}{c})\Big\{\sum_{h\neq l}\big[(\bnu_{hl}^{(R)}(\bphi^0))^\T \bz\big]^2+\sum_{r=1}^K\sum_{s=1}^p\big[(\bb_{rs}^{(R)}(\bphi^0))^\T \bz\big]^2\Big\}\\<&(1-\frac{\epsilon}{c})\beta^2\Big\{\sum_{h\neq l}\big[(\bnu_{hl}^{(R)}(\bphi^0))^\T \bu\big]^2+\sum_{r=1}^K\sum_{s=1}^p\big[(\bb_{rs}^{(R)}(\bphi^0))^\T \bu\big]^2\Big\}\to 0\mbox{ as }\beta<\sqrt{\epsilon/\gamma}
\end{align*}
which is a contradiction.
Then by boundedness of $\bgamma_j^0$, $\bU_i^0$ and Lemma \ref{lemma:estimation for 1st derivative}, we have $\|\Hb_{Ruf^{\prime}}(\bphi^0)\|$ $=\|\Hb_{Rfu ^{\prime}}(\bphi^0)\|=O_p\big(\sqrt{n}+\sqrt{q}\big)$, which implies
\begin{equation*}
    \bz^\T \Db_{q}^{1/2}\Hb_R({\bphi}^0)\Db_{q}^{1/2}\bz=O_p\big(\delta_{nq}^{-1}\big).
\end{equation*}
To sum up, we get $\lambda_{\min}\big[\Db_{q}^{1/2}\mathcal{H}({\bphi}^0)\Db_{q}^{1/2}\big]\le c$ for some $c>0$.

\noindent {\bf Local convexity in a neighborhood of $\bphi^0$.}
For arbitrary ${\bphi}\in\cB(D)\cap\{{\bphi}:\sqrt{p}\|\Db_{q}^{-1/2}({\bphi}-{\bphi}^0) \| \le m\}$, we prove for any $\bz$, $\|\bz\|=1$, \begin{equation*}\bz^\T \Db_{q}^{1/2}\mathcal{H}({\bphi})\Db_{q}^{1/2}\bz=\bz^\T \Db_{q}^{1/2}\Big[\Hb_L({\bphi})+\Hb_R(\bphi)+\Hb_P({\bphi})\Big]\Db_{q}^{1/2}\bz\end{equation*} has positive upper bound. 
Note here we consider a neighborhood region of $\{{\bphi}:\sqrt{p}\|\Db_{q}^{-1/2}({\bphi}-{\bphi}^0)\|\le m\}$ where it has been proved in Lemma~\ref{prop:average consistency} that the constrained MLEs will fall into this region with any small $m$ when $p\lesssim \delta_{nq}$.

For ${\bphi}\in\cB(D)$, we omit $\bphi$ in the notations including $\Vb_0(\bphi)$, $\cN(\bphi)$, $\cR(\bphi)$ and other vectors dependent on $\bphi$ for simplicity. We first focus on the expression related to $\Hb_P$, that is, 
\begin{align}
    \bz^\T &\Db_{q}^{1/2}\Hb_P\Db_{q}^{1/2}\bz=(\Vb_{0}\mathcal{N}+\mathcal{R}) (\Vb_{0}\mathcal{N}+\mathcal{R})^\T\label{BB16_1}\\&+{2}c\bz^\T \sum_{r=1}^K(n^{-1}\sum_{i=1}^nU_{ir}^2-q^{-1}\sum_{j=1}^q\gamma_{jr}^2)\begin{bmatrix}
        \Ib_{n+q}\otimes \Eb_{rr}^{(K+p)}\\&\Ib_n\otimes\Eb_{rr}^{(K)}
    \end{bmatrix}\bz\nonumber\\&+c\bz^\T \Db_{q}^{-1/2}\sum_{r<l}\Big[
        (\sum_{i=1}^nU_{ih}U_{il})\Db_{1,rl}+(\sum_{j=1}^q \gamma_{jh} \gamma_{jl})\Db_{2,rl}\Big]\Db_{q}^{-1/2}\bz.\label{BB16}
\end{align}
Notice that $P(\bphi) = 0$ does not necessarily hold for the general ${\bphi}\in\cB(D)$, which results in extra terms in~\eqref{BB16}. Because $\sqrt{p}\|\Db_{q}^{-1/2}({\bphi}-{\bphi}^0)\|\le m$ and by Cauchy-Schwarz Inequality, we have
\begin{align}
    \Big|\frac{1}{n}\sum_{i=1}^n\big[U_{ir}^2-(U_{ir}^0)^2\big]\Big|&\le 2n^{-1/2}\|\hat\Ub^0-\Ub^0\|n^{-1/2}\|\Ub^0\|+n^{-1}\|\hat\Ub^0-\Ub^0\|\nonumber\\&\le2mn^{-1/2}\|\Ub^0\|/p+m^2/p^2;\label{BB17}
\end{align}
Similarly, we have 
\begin{align}
    \big|\frac{1}{q}\sum_{j=1}^q \gamma_{jr}^2-\frac{1}{q}\sum_{j=1}^q (\gamma_{jr})^2\big|\le 2mq^{-1/2}\|\bGamma^0\|/p+m^2/p^2;\\
    \big|\frac{1}{n}\sum_{h<l}^nU_{ih}U_{il}\big|\le 2mn^{-1/2}\|\Ub^0\|/p+m^2/p^2;\\
    \big|\frac{1}{q}\sum_{h<l}^q \gamma_{jh}\gamma_{jl}\big|\le 2mq^{-1/2}\|\bGamma^0\|/p+m^2/p^2.\label{BB20}
\end{align}
From \eqref{BB17}-\eqref{BB20}, the two terms in \eqref{BB16} can be bounded as follows
\begin{align}
    | \eqref{BB16}| \le cK(4K+3) (mn^{-1/2}\|\Ub^0\|/p+ mq^{-1/2}\|\bGamma^0\|/p +m^2/p^2)\lesssim \frac{m}{p}. \label{eq:prop 2 Hp term2}
\end{align}
We next compute $\bz^\T \Db_{q}^{-1/2}\Hb_R(\bphi)\Db_{q}^{-1/2}\bz$. Since
\begin{align*}
   & \bz^\T \Db_{q}^{-1/2}\Hb_R({\bphi})\Db_{q}^{-1/2}\bz  \\
    &=\frac{2}{\sqrt{nq}}\sum_{i=1}^n\sum_{j=1}^q [y_{ij}-b^\prime(w_{ij}^0)+b^\prime(w_{ij}^0)-b^\prime(w_{ij})]\sum_{r=1}^K[\bz]_{[(i-1)K+r]}[\bz]_{[(j-1)K+r]} \\
    &= \bz^\T \Db_{q}^{1/2}\Hb_R({\bphi^0})\Db_{q}^{1/2}\bz + \frac{2}{\sqrt{nq}}\sum_{i=1}^n\sum_{j=1}^q (b^\prime(w_{ij}^0)-b^\prime(w_{ij})) \sum_{r=1}^K[\bz]_{[(i-1)K+r]}[\bz]_{[(j-1)K+r]} \\
    & \le  \bz^\T \Db_{q}^{1/2}\Hb_R({\bphi^0})\Db_{q}^{1/2}\bz \\
    &+\frac{2}{\sqrt{nq}} \big[ \sum_{i=1}^n\sum_{j=1}^q \{b^\prime(w_{ij}^0)-b^\prime(w_{ij})\}^2\big]^{1/2} \sum_{r=1}^K[\bz]_{[(i-1)K+r]}[\bz]_{[(j-1)K+r]} ,
\end{align*}
and by Assumption~\ref{assumption:smoothness} and the mean value theorem, we have
\begin{align*}
  & \sum_{i=1}^n\sum_{j=1}^q\big\{b^{\prime}(w_{ij}^0)-b^{\prime}(w_{ij})\big\}^2 \\ & \le  \sum_{i=1}^n\sum_{j=1}^q \{ b^{\prime\prime} (\tilde{w}_{ij})(w_{ij}^0 - w_{ij})\}^2  \\
   &\le 2b_U\sum_{i=1}^n\sum_{j=1}^q(\bgamma_j^\T\bU_i-(\bgamma_j^0)^\T\bU_i^0)^2+2b_U\sum_{j=1}^q(\bbeta_j-\bbeta_j^0)\big(\sum_{i=1}^n\bX_i\bX_i^\T\big)(\bbeta_j-\bbeta_j^0)\\&\lesssim q\|\bU-\bU^0\|_F^2+n\|\bGamma-\bGamma^0\|_F^2+n\lambda_{\max}\big(\frac{1}{n}\sum_{i=1}^n\bX_i\bX_i^\T\big)\|\bB-\Bb^0\|_F^2\\&
   \lesssim nq/p ,
\end{align*}
where the last equality is from ${\bphi}\in\cB(D)\cap\{{\bphi}:\sqrt{p}\|\Db_{q}^{-1/2}({\bphi}-{\bphi}^0)\|\le m\}$. Therefore combining above results with the estimates for $\Hb_R$ on $\bphi^0$, we have
\begin{equation}
      \bz^\T \Db_{q}^{-1/2}\Hb_R({\bphi})\Db_{q}^{-1/2}\bz \lesssim  \big\|\Db_{q}^{1/2}\Hb_R({\bphi^0})\Db_{q}^{1/2}\big\|+  \frac{1}{\sqrt p} \lesssim \frac{1}{{\delta_{nq}}}. \label{eq:prop 2 HR}
\end{equation}
Lastly, we write the remaining terms in $\bz^\T \Db_{q}^{1/2}\mathcal{H}({\bphi})\Db_{q}^{1/2}\bz$ as
\begin{align*}
    &\bz^\T \Db_{q}^{1/2}\Hb_L(\bphi)\Db_{q}^{1/2}\bz+\bz^\T(\Vb_{0}\mathcal{N}+\mathcal{R}) (\Vb_{0}\mathcal{N}+\mathcal{R})^\T\bz\\&\le \beta^2\gamma+c\alpha^2\Big[\sum_{i=1}^{K}[ {\bn}(\bphi^0)_v]_{K(i)}^2\Big]+c\alpha^2\Big[\sum_{i=1}^K([ {\bn}_v]_{[K(i)]}^2-[ {\bn}(\bphi^0)_v]_{[K(i)]}^2)\Big]\\&+c\sum_{h\neq l}\big[(\bnu_{hl}^{(R)}(\bphi^0))^\T\bz-\alpha[ {\bn}(\bphi^0)_v]_{[K(h,l)]}\big]^2 \\
    &+c\sum_{r=1}^K\sum_{s=1}^p\big[(\bb_{rs}^{(R)}(\bphi^0))^\T \bz-\alpha[ {\bn}(\bphi^0)_v]_{[K^2+(r-1)p+s]}\big]^2\\&+c\sum_{h\neq l}\Big\{\big[(\bnu_{hl}^{(R)})^\T\bz-(\bnu_{lh}^{(R)}(\bphi^0))^\T\bz\big]^2+\big[\alpha[ {\bn}_v]_{[K(h,l)]}-\alpha[ {\bn}(\bphi^0)_v]_{[K(h,l)]}\big]^2\Big\}\\&+c\sum_{r=1}^K\sum_{s=1}^p\Big\{\big[(\bb_{rs}^{(R)})^\T \bz-(\bb_{rs}^{(R)}(\bphi^0))^\T \bz\big]^2+\big[\alpha[ {\bn}_v]_{[K^2+(r-1)p+s]}-\alpha[ {\bn}(\bphi^0)_v]_{[K^2+(r-1)p+s]}\big]^2\Big\}.
\end{align*}
Here $\bnu_{pq}^{(R)}(\bphi^0)$, $\bb_r^{(R)}(\bphi^0)$ are previously defined vectors on parameter ${\bphi}^0$ and $ {\bn}_v(\bphi^0)=\mathcal{N}(\bphi^0)^\T\Vb_{0}(\bphi^0)^\T\Vb_{0}\bn_0$. Then by our definition, $ {\bn}_v(\bphi^0)- {\bn}_v=\big[\Vb_{0}\mathcal{N}-\Vb_{0}(\bphi^0)\mathcal{N}(\bphi^0)\big]^\T\Vb_{0}\bn_0$ and 
\begin{align*}
    &[\Vb_{0}\mathcal{N}]_{[,K(r)]}=[\Vb_{0}]_{[,K(r)]},1\le r\le K;\\
    &[\Vb_{0}\mathcal{N}]_{[,D+K(h,l)]}=\frac{\tau_q^2}{\tau_q^2+\tau_p^2}[\Vb_0]_{[,K(h,l)]}+\frac{\tau_p^2}{\tau_q^2+\tau_p^2}[\Vb_0]_{[,K(l,h)]},1\le l\neq h\le K;\\
    &[\Vb_{0}\mathcal{N}]_{[,K^2+s]}=[\Vb_{0}]_{[,K^2+s]},1\le s\le Kp,
\end{align*}
we conclude that inside this restricted region, $\|\Vb_0-\Vb_0(\bphi^0)\|
\le\sqrt{K^2+Kp}\|\Db_q^{-1/2}(\bphi-\bphi^0)\|$. Therefore we have 
\begin{equation*}
    \|{\bn}_v(\bphi^0)-{\bn}_v\|\le \|\Vb_{0}\mathcal{N}-\Vb_{0}(\bphi^0)\mathcal{N}(\bphi^0)\|\|\bv\|\lesssim \sqrt{p}\|\Db_{q}^{-1/2}({\bphi}-{\bphi}^0)\|\lesssim m.
\end{equation*}
Also, by 
\begin{equation*}
    \bnu_{hl}^{(R)}(\bphi^0)-\bnu_{hl}(\bphi^0)=\Db_{q}^{-1/2}(\bnu_{hl}(\bphi^0)-\bnu_{hl})-\big\{[\Vb_{0}(\bphi^0)\mathcal{N}(\bphi^0)]_{[,K(h,l)]}-[\Vb_{0}\mathcal{N}]_{[,K(h,l)]}\},
\end{equation*}
and $\bb_{rs}^{(G,R)}-\bb_r^{(R)}=\zero_{qp}$, we have 
\begin{align*}
    \sum_{h\neq l}\Big\{\big[(\bnu_{hl}^{(R)})^\T\bz-(\bu_{hl}^{(R)}(\bphi^0))^\T\bz\big]^2-&\big[\alpha[\check{\bn}_v]_{[K(h,l)]}-\alpha[\check{\bn}(\bphi^0)_v]_{[K(h,l)]}\big]^2\Big\}\\&\lesssim \|\Db_{q}^{-1/2}({\bphi}-{\bphi}(\bphi^0))\|\\
    &\lesssim  m;
    \end{align*}
    \begin{align*}
    \sum_{r=1}^K\sum_{s=1}^p\Big\{\big[(\bb_{rs}^{(R)})^\T \bz-(\bb_r^{(R)}(\bphi^0))^\T \bz\big]^2-&\big[\alpha[\check{\bn}_v]_{[K^2+(r-1)p+s]}-\alpha[\check{\bn}(\bphi^0)_v]_{[K^2+(r-1)+s]}\big]^2\Big\}\\&\lesssim \sqrt{p}\|\Db_{q}^{-1/2}({\bphi}-{\bphi}(\bphi^0))\|\\&\le m.
\end{align*}
Therefore, $\bz^\T \Db_{q}^{1/2}\Hb_L\Db_{q}^{1/2}\bz+\bz^\T(\Vb_{0}\mathcal{N}+\mathcal{R}) (\Vb_{0}\mathcal{N}+\mathcal{R})^\T\bz$ has strict negative upper-bound for small enough $m$.

Combining the above results with~\eqref{eq:prop 2 Hp term2} and~\eqref{eq:prop 2 HR}, we complete the proof for Lemma~\ref{prop:min eigenvalue}.
 \end{proof}

\subsection{Proof of Lemma~\ref{lemma: l1 estimates}}
\begin{lemma}\label{lemma: l1 estimates}
    Under Assumption~\ref{assumption: psd covariance}--\ref{assumption:asymptotic normality}, we have estimates for the Hessian matrix on $\bphi^0$ as follows:
    \begin{enumerate}[(i)]
        \item $\|\Hb_{Lff^{\prime}}^{-1}\|=O_p(q)$, $\|\cH_{ff^{\prime}}^{-1}\|=O_p(q)$, and $\|\cH_{ff^{\prime}}^{-1}\|_1=O_p(\sqrt{p}q)$.
        \item $\|\Hb_{Luu^{\prime}}^{-1}\|=O_p(n)$, $\|\cH_{uu^{\prime}}^{-1}\|=O_p(n)$, and $\|\cH_{uu^{\prime}}^{-1}\|_1=O_p(\sqrt{p}n)$.
        \item $\max_i\big\|[\cH_{uf^{\prime}}]_{[K_i,]}\big\|=O_p\big((p^{1/2}+(nq)^{1/\xi})n^{-1}q^{-1/2}\big)$, $\|\cH_{uf^{\prime}}\|_1=O_p\big((nq)^{1/\xi}q^{-1}\big)$, and $\|\cH_{uf^{\prime}}\|=O_p\big((nq)^{-1/2+1/\xi}\big)$. 
        \item $\max_j\big\|[\cH_{fu^{\prime}}]_{[P_j,]}\big\|=O_p\big((nq)^{1/\xi}n^{-1}q^{-1/2}\big)$, $\|\cH_{fu^{\prime}}\|_1=O_p\big((p+(nq))^{1/\xi}n^{-1}\big)$, and $\|\cH_{fu^{\prime}}\|=O_p\big((nq)^{-1/2+1/\xi}\big)$. 
        \item $\max_j\big\|[\cH_{ff^{\prime}}^{-1}\cH_{fu^{\prime}}]_{[P_j,]}\big\|=O_p\big((nq)^{1/\xi}n^{-1/2}\big)$.
        \item $\max_i\big\|[\cH_{uu^{\prime}}^{-1}\cH_{uf^{\prime}}]_{[K_i,]}\big\|=O_p\big((p^{1/2}+(nq)^{1/\xi})q^{-1/2}\big)$.
        \item $\|[\bLambda_{1}]_{[P_j,\cdot]}\|=O(q^{-1})$ and $\|\bLambda_{1}\|=O(q^{-1/2})$.
        \item $ \|[\bLambda_{2}]_{[K_i,\cdot]}\|=O(p^{1/2}n^{-1})$ and $\|\bLambda_{2}\|=O(n^{-1/2})$.
    \end{enumerate}
    Here we omitted the dependence on $\bphi^0$, also in the following proof.
\end{lemma}
\begin{proof}
By the definition of $\bLambda_1$ in section \ref{sec:prelim}, it is obvious that
\begin{equation}\|[\bLambda_{1}]_{[P_j,\cdot]}\|=O(q^{-1}),\quad\|\bLambda_{1}^\T \|_1=O(q^{-1}),\quad\|\bLambda_{1}\|_1=O(1),\quad\|\bLambda_{1}\|=O(q^{-1/2});\label{eq_estimate_lambda1}
\end{equation}
which implies (vii).
For $\bLambda_2$, we verify that $\|[\bLambda_2]_{[K_i,1:K^2]}\|=O(n^{-1})$ and $\|[\bLambda_2]_{[K_i,1:K^2+1:K^2+Kp]}\|$ $=O(p^{1/2}n^{-1})$. Also note that \begin{equation*}\|[\bLambda_2]_{[,K^2+1:K^2+Kp]}\|=\lambda_{\max}\big(n^{-1}\bX^\T\bX\otimes I_K\big)=O(1).\end{equation*}Then the bounds for $\bLambda_2 $ are given as
\begin{equation}
    \|[\bLambda_{2}]_{[K_i,\cdot]}\|=O_p(p^{1/2}n^{-1}),\quad\|\bLambda_{2}^\T \|_1=O_p(pn^{-1}),\quad\|\bLambda_{2}\|_1=O_p(1),\quad\|\bLambda_{2}\|=O_p(n^{-1/2}),\label{eq_estimate_lambda2}
    \end{equation}
    which implies  and (viii).

    Next for (i),
    because $\Hb_{L f f^{\prime}}$ is block diagonal, by Assumptions~\ref{assumption: psd covariance}--\ref{assumption:smoothness}, we have 
\begin{equation}
        \|\Hb_{Lff^\prime}^{-1}\|= \max_j\Big\|\big[-(nq)^{-1}\sum_{i=1}^nl_{ij}^{\prime\prime}\bZ_i\bZ_i^\T\big]^{-1}\Big\|
        \le \frac{q}{b_L}\max_j\big\|n^{-1}\sum_{i=1}^n\bZ_i\bZ_i^\T\big\|^{-1}=
        O_p(q),\label{eq:lemma11eq1}
    \end{equation}
and also for the $l_1$-norm, we have
    \begin{equation}
\big\|\Hb_{Lff^{\prime}}^{-1}\big\|_1=\big\|\big((nq)^{-1}\sum_{i=1}^nl_{ij}^{\prime\prime}\bZ_i\bZ_i^\T\big)^{-1}\big\|_1\le \sqrt{q}\big\|\big((nq)^{-1}\sum_{i=1}^nl_{ij}^{\prime\prime}\bZ_i\bZ_i^\T\big)^{-1}\big\|=O_p(\sqrt{p}q).\label{eq:lemma11eq2}
    \end{equation}
 Since $\bLambda_{1}^\T \Hb_{Lff^\prime}^{-1}\bLambda_{1}$ is positive semi-definite and $\|\bLambda_{1}^\T \Hb_{Lff^\prime}^{-1}\bLambda_{1}\|\le\|\bLambda_{1}\|^2\|\Hb_{Lff^\prime}^{-1}\|\le O(1)$,  the $l_1$-norm for $(\Ib_{K^2}+\bLambda_{1}^\T \Hb_{Lff^\prime}^{-1}\bLambda_{1})^{-1}$ can be bounded by
    \begin{equation*}
        \big\|(\Ib_{K^2}+\bLambda_{1}^\T \Hb_{Lff^\prime}^{-1}\bLambda_{1})^{-1}\big\|_1\le K\big\|(I_{K(K+1)/2}+\bLambda_{1}^\T \Hb_{Lff^\prime}^{-1}\bLambda_{1})^{-1}\big\|\le O_p(1).
    \end{equation*} Next for $\cH_{ff^\prime}= \Hb_{Lff^\prime}+\bLambda_{1}\bLambda_{1}^\T$, we use Woodbury identity to give
    \begin{align}
        \Big\|\cH_{ff^\prime}^{-1}\Big\|=&\Big\| \Hb_{Lff^\prime}^{-1}- \Hb_{Lff^\prime}^{-1}\bLambda_{1}\big(\Ib_{K^2}+\bLambda_{1}^\T \Hb_{Lff^\prime}^{-1}\bLambda_{1}\big)^{-1}\bLambda_{1}^\T  \Hb_{Lff^\prime}^{-1}\Big\|\nonumber\\=&\big\| \Hb_{Lff^\prime}^{-1}\big\|+\big\| \Hb_{Lff^\prime}^{-1}\big\|\big\|\bLambda_{1}\big\|\big\|\big(\Ib_{K^2}+\bLambda_{1}^\T \Hb_{Lff^\prime}^{-1}\bLambda_{1}\big)^{-1}\big\|\big\|\bLambda_{1}^\T \big\|\big\| \Hb_{Lff^\prime}^{-1}\big\|\nonumber\\\le& O_p(q),\label{eq:lemma11eq3}
    \end{align}
    For the $l_1$ estimate of $\cH_{ff^{\prime}}^{-1}$, by \eqref{eq_estimate_lambda1},
    \begin{align*}
        \Big\|\Hb_{Lff^\prime}^{-1}&\bLambda_{1}\big(\Ib_{K^2}+\bLambda_{1}^\T \Hb_{Lff^\prime}^{-1}\bLambda_{1}\big)^{-1}\bLambda_{1}^\T  \Hb_{Lff^\prime}^{-1}\Big\|_{1}\\&=\max_j\Big\|\Hb_{Lff^\prime}^{-1}\bLambda_1\big(\Ib_{K^2}+\bLambda_{1}^\T \Hb_{Lff^\prime}^{-1}\bLambda_{1}\big)^{-1}[\bLambda_{1}]_{[P_j,]}^\T  [\Hb_{Lff^\prime}^{-1}]_{[P_j,P_j]}\Big\|_{1}\\
        &\le \sqrt{pq}\max_j\big\|\Hb_{Lff^\prime}^{-1}\big\|\big\|\bLambda_{1}\big\|\big\|\big(\Ib_{K^2}+\bLambda_{1}^\T \Hb_{Lff^\prime}^{-1}\bLambda_{1}\big)^{-1}\big\|\big\|[\bLambda_{1}]_{[P_j,]}^\T \big\|\big\| [\Hb_{Lff^\prime}^{-1}]_{[P_j,P_j]}\big\|\\&=O_p(\sqrt{p}q).
    \end{align*}
    Together with $\|\Hb_{Lff^{\prime}}^{-1}\|_1=O_p(\sqrt{p}q)$ we have $
        \big\|\cH_{Lff^{\prime}}^{-1}\big\|_1=O_p(\sqrt{p}q)$.
    
    For (ii), similarly by (ii) of Assumption~\ref{assumption: psd covariance} 
and Assumption~\ref{assumption:smoothness}, we have the bound for the $l_2$- and $l_1$-norm of $\Hb_{Luu^{\prime}}$ given as
    \begin{equation*}
        \big\|\Hb_{Luu^\prime}^{-1}\big\|\le\max_i\big\|(nq)^{-1}\big(\sum_{j=1}^ql_{ij}^{\prime\prime}(w_{ij})\bgamma_j\bgamma_j^{\T}\big)^{-1}\big\|\le\frac{\sqrt{K}nq}{b_L}\Big\|\big(\sum_{j=1}^q\bgamma_j\bgamma_j^{\T}\big)^{-1}\Big\|=O_p(n),
    \end{equation*}
    and 
    \begin{equation*}
        \big\|\Hb_{Luu^\prime}^{-1}\big\|_1\le\sqrt{K}\max_i\big\|(nq)^{-1}\big(\sum_{j=1}^ql_{ij}^{\prime\prime}\bgamma_j\bgamma_j^{\T}\big)^{-1}\big\|\le\frac{\sqrt{K}nq}{b_L}\Big\|\big(\sum_{j=1}^q\bgamma_j\bgamma_j^{\T}\big)^{-1}\Big\|=O_p(n).
    \end{equation*}
    We get $\cH_{uu^\prime}^{-1}=\Hb_{Luu^\prime}^{-1}-\Hb_{Luu^\prime}^{-1}\bLambda_2(\Ib_{K^2+Kp}+\bLambda_2^\T \Hb_{Luu^\prime}^{-1}\bLambda_2)^{-1}\bLambda_2^\T \Hb_{Luu^\prime}^{-1}$ by applying Woodbury identity again.
        By \eqref{eq_estimate_lambda2}, we know $\|\bLambda_{2}^\T \Hb_{Luu^\prime}^{-1}\bLambda_{2}\|\le \|\Hb_{Luu^\prime}^{-1}\|\|\blambda_2\|^2\le O(1)$. Therefore we have estimate $
        \big\|(\Ib_{K^2+Kp}+\bLambda_{2}^\T \Hb_{Luu^\prime}^{-1}\bLambda_{2})^{-1}\big\|=O_p(1)$, 
    which implies 
    \begin{align*}
        &\Big\|\Hb_{Luu^\prime}^{-1}\bLambda_{2}\big(\Ib_{K^2+Kp}+\bLambda_{2}^\T \Hb_{Luu^\prime}^{-1}\bLambda_{2}\big)^{-1}\bLambda_{2}^\T  \Hb_{Luu^\prime}^{-1}\Big\|\\&\le 
        \big\|\Hb_{Luu^\prime}^{-1}\big\|\big\|\bLambda_{2}\big\|\big\|\big(-\Ib_{K^2+Kp}+\bLambda_{2}^\T \Hb_{Luu^\prime}^{-1}\bLambda_{2}\big)^{-1}\big\|\big\|\bLambda_{2}^\T\big\|\big\|  \Hb_{Luu^\prime}^{-1}\big\|\\&=O_p(n),
    \end{align*}
    and
    \begin{align*}
        &\Big\|\Hb_{Luu^\prime}^{-1}\bLambda_{2}\big(\Ib_{K^2+Kp}+\bLambda_{2}^\T \Hb_{Luu^\prime}^{-1}\bLambda_{2}\big)^{-1}\bLambda_{2}^\T  \Hb_{Luu^\prime}^{-1}\Big\|_1\\&\le 
        \max_i\Big\|\Hb_{Luu^\prime}^{-1}\bLambda_{2}\big(\Ib_{K^2+Kp}+\bLambda_{2}^\T \Hb_{Luu^\prime}^{-1}\bLambda_{2}\big)^{-1}[\bLambda_{2}]_{[K_i,]}^\T  [\Hb_{Luu^\prime}^{-1}]_{[K_i,K_i]}\Big\|_1
        \\&\le 
        \sqrt{n}\max_i\big\|\Hb_{Luu^\prime}^{-1}\big\|\big\|\bLambda_{2}\big\|\big\|\big(\Ib_{K^2+Kp}+\bLambda_{2}^\T \Hb_{Luu^\prime}^{-1}\bLambda_{2}\big)^{-1}\big\|\big\|[\bLambda_{2}]_{[K_i,]}^\T \big\|\big\| [\Hb_{Luu^\prime}^{-1}]_{[K_i,K_i]}\big\|\\&=O_p(\sqrt{p}n).
    \end{align*}
For $\Hb_{Luu^{\prime}}$, similarly with (i), by Assumption~\ref{assumption: psd covariance} we have
    \begin{equation*}
    \big\|\Hb_{Luu^\prime}^{-1}\big\|_1\le\sqrt{K}\max_i\big\|(nq)^{-1}(\sum_{j=1}^ql_{ij}^{\prime\prime}\bgamma_j\bgamma_j^{\T})^{-1}\big\|\le\frac{\sqrt{K}nq}{b_L}\Big\|(\sum_{j=1}^q\bgamma_j\bgamma_j^{\T})^{-1}\Big\|=O_p(n).\label{B30}
    \end{equation*}
    Therefore the $l_2$-norm of $\cH_{uu^{\prime}}^{-1}$ can be bounded by
    \begin{equation*}
        \Big\|\cH_{uu^\prime}^{-1}\Big\|\le\Big\| \Hb_{Luu^\prime}^{-1}\Big\|+\Big\| \Hb_{Luu^\prime}^{-1}\bLambda_{2}\big(\Ib_{K^2}+\bLambda_{2}^\T H_{Luu^\prime}^{-1}\bLambda_{2}\big)^{-1}\bLambda_{2}^\T  \Hb_{Luu^\prime}^{-1}\Big\|=O_p({n}),
    \end{equation*}
    and its $l_1$-norm can be bounded by
    \begin{equation*}
        \Big\|\cH_{uu^\prime}^{-1}\Big\|_1\le\Big\| \Hb_{Luu^\prime}^{-1}\Big\|_1+\Big\| \Hb_{Luu^\prime}^{-1}\bLambda_{2}\big(\Ib_{K^2}+\bLambda_{2}^\T H_{Luu^\prime}^{-1}\bLambda_{2}\big)^{-1}\bLambda_{2}^\T  \Hb_{Luu^\prime}^{-1}\Big\|_1=O_p(\sqrt{p}{n}).
    \end{equation*}
Next we prove (iii) and (iv). For the off-diagonal blocks, note 
\begin{equation*}
    [\cH_{fu^{\prime}}]_{[P_j,K_i]}=-(nq)^{-1}l_{ij}^{\prime\prime}\bgamma_j\bZ_i^\T-(nq)^{-1}l_{ij}^{\prime}\begin{pmatrix}\Ib_k\\\zero_{(K+p)\times K}\end{pmatrix}+[\bLambda_1]_{[P_j,]}[\bLambda_2]_{[,K_i]},
\end{equation*}
For the first term, we have from Assumption~\ref{assumption:smoothness} that
\begin{equation*}
    \max_j\big\|[\Hb_{Lfu^{\prime}}]_{[P_j,]}\big\|\le \frac{b_U}{nq}\Big[\lambda_{\max}\big(\sum_{i=1}^n\bZ_i\bZ_i^\T\big)\Big]^{1/2}\max_j\|\bgamma_j\|^2=O_p(n^{-1/2}q^{-1}),
\end{equation*}
which also implies $\|\Hb_{Lfu{\prime}}\|=O_p((nq)^{-1/2})$. For the second term, by Assumption~\ref{assumption:smoothness} we have $\max_i|l_{ij}^{\prime}(w_{ij})|=O_p(n^{1/\xi})$, $\max_j|l_{ij}^{\prime}(w_{ij})|=O_p(q^{1/\xi})$, $\max_{i,j}|l_{ij}^{\prime}(w_{ij})|=O_p((nq)^{1/\xi})$. Together with \eqref{eq_estimate_lambda1} and \eqref{eq_estimate_lambda2},
 { we have $\|[\cH_{fu^{\prime}}]_{[P_j,]}\|_1=O_p((p+(nq)^{1/\xi})n^{-1}q^{-1})$, \\ $\max_j\|[\cH_{fu^{\prime}}]_{[P_j,]}\|=O_p((nq)^{1/\xi}n^{-1/2}q^{-1})$, $\|[\cH_{uf^{\prime}}]_{[K_i,]}\|_1=O_p((nq)^{1/\xi}n^{-1}q^{-1})$, \\$\max_i\|[\cH_{uf^{\prime}}]_{[K_i,]}\|=O_p(\sqrt{p}+(nq)^{1/\xi})n^{-1}q^{-1/2})$} and $\|\cH_{uf^{\prime}}\|=\|\cH_{fu^{\prime}}\|=O_p\big((nq)^{-1/2+1/\xi}\big)$.
Finally (v) and (vi) are obtained by
\begin{align*}
    \max_{j}\big\|[\cH_{ff^{\prime}}^{-1}\cH_{fu^{\prime}}]_{[P_{j},]}\big\|\le&\max_j\big\|[\Hb_{Lff^{\prime}}^{-1}\cH_{fu^{\prime}}]_{[P_{j},]}\big\|\\&+\max_j\big\|[\Hb_{Lff^{\prime}}^{-1}\bLambda_1\big(\Ib_{K^2}+\bLambda_1^\T\Hb_{Lff^{\prime}}\bLambda_1\big)^{-1}\bLambda_1^\T\Hb_{Lff^{\prime}}^{-1}\cH_{fu^{\prime}}]_{[P_{j},]}\big\|\\
    \le &\max_j\big\|[\Hb_{Lff^{\prime}}^{-1}]_{[P_j,P_j]}[\cH_{fu^{\prime}}]_{[P_{j},]}\big\|\\&+\max_j  \big\|[\Hb_{Lff^{\prime}}^{-1}]_{[P_j,P_j]}[\bLambda_1]_{[P_j,]}\big(\Ib_{K^2}+\bLambda_1^\T\Hb_{Lff^{\prime}}\bLambda_1\big)^{-1}\bLambda_1^\T\Hb_{Lff^{\prime}}^{-1}\cH_{fu^{\prime}}\big\|\\
    =&O_p\big((nq)^{1/ \xi}n^{-1/2}\big),
\end{align*}
from (i), (iv) and (vi) and 
\begin{align*}
    &\,\max_{i}\big\|[\cH_{uu^{\prime}}^{-1}\cH_{uf^{\prime}}]_{[K_{i},]}\big\|\\\le&\max_j\big\|[\Hb_{Luu^{\prime}}^{-1}\cH_{uf^{\prime}}]_{[K_{i},]}\big\|\\&+\max_j\big\|[\Hb_{Luu^{\prime}}^{-1}\bLambda_2\big(\Ib_{K^2+Kp}+\bLambda_2^\T\Hb_{Luu^{\prime}}\bLambda_2\big)^{-1}\bLambda_2^\T\Hb_{Luu^{\prime}}^{-1}\cH_{uf^{\prime}}]_{[K_{i},]}\big\|\\
    \le &\max_j\big\|[\Hb_{Luu^{\prime}}^{-1}]_{[K_i,K_i]}[\cH_{uf^{\prime}}]_{[K_{i},]}\big\|\\&+\max_j  \big\|[\Hb_{Luu^{\prime}}^{-1}]_{[K_i,K_i]}[\bLambda_2]_{[K_i,]}\big(\Ib_{K^2+Kp}+\bLambda_2^\T\Hb_{Luu^{\prime}}\bLambda_2\big)^{-1}\bLambda_2^\T\Hb_{Luu^{\prime}}^{-1}\cH_{uf^{\prime}}\big\|\\
    =&O_p\big((\sqrt{p}+(nq)^{1/ \xi})q^{-1/2}\big),
\end{align*}
from (ii), (iii), and (vii), respectively.
\end{proof}

\subsection{Proof of Lemma~\ref{lemma: colsum norm cH tilde}}
\begin{lemma}
\label{lemma: colsum norm cH tilde}
    Recall that we have defined that $\tilde\cM(\bphi)=[\int_0^1\cH\big(\bphi^0+s({\bphi}-\bphi^0)\big)ds]^{-1}$. Use the notation in the proof of Lemma~\ref{lemma: first order condition}, denote $\tilde\cM=\tilde\cM(\tilde\bphi)$ and express it in the following block form:
\begin{equation*}\tilde{\cM}=\begin{bmatrix}\tilde{\cM}_{ff^\prime}&\tilde{\cM}_{f u^\prime}\\\tilde{\cM}_{uf^\prime}&\tilde{\cM}_{uu^\prime}\end{bmatrix}.\end{equation*}
Under Assumptions~\ref{assumption: psd covariance}--\ref{assumption:asymptotic normality},
    there exist some $\epsilon>0$ and $m$ such that\begin{equation}
    \Pr\Big\{\min_{{\bphi}\in\cB(D),\|\Db_{q}^{-1/2}({\bphi}-{\bphi}^0)\|\le m}\lambda_{\min}\big(\Db_{q}^{1/2}\tilde{\cM}({\bphi})\Db_{q}^{1/2}\big)\geq\epsilon\Big\}\to1\label{eq_sub_local_convex}.
\end{equation}
    and each blocks of $l_{\infty}$-norm of $\tilde{\cM}$ can be estimated as follows:
  \begin{equation}
      O_p\begin{pmatrix}
      \sqrt pq(nq)^{2/\xi}&\sqrt pn(nq)^{3/\xi}\\ pq(nq)^{3/\xi}&\sqrt pn(nq)^{2/\xi}
    \end{pmatrix}. \label{eq_hessian_infinitybound}
  \end{equation}
\end{lemma}

  \begin{proof}
      The first part \eqref{eq_local_convex} is proved by Lemma \ref{prop:min eigenvalue} and
      \begin{equation*}
          \lambda_{\min}\big(\Db_{q}^{1/2}\tilde{\cM}({\bphi})\Db_{q}^{1/2}\big)\ge \min_s\lambda_{\min}\Big[\big(\Db_{q}^{1/2}{\cH}({\bphi^0}+s(\bphi-\bphi^0))\Db_{q}^{1/2}\big)^{-1}\Big].
      \end{equation*}
      Next we will focus on the second part. We first show the estimates are valid using the bound established in Lemma~\ref{lemma: l1 estimates}. For now, we focus our discussion on $\bphi^0$ and its dependence will be omitted. Again use the notation in the proof of Lemma~\ref{lemma:asym_estimate} 
      \begin{equation*}
          \cH(\bphi^0)^{-1}=\begin{pmatrix}
              \tilde\cH_{ff^{\prime}}&\tilde\cH_{fu^{\prime}}\\
              \tilde\cH_{uf^{\prime}}&\tilde\cH_{uu^{\prime}}
          \end{pmatrix}
      \end{equation*}
      Also recall we have defined the Schur complement $\cH_{-f}=(\cH_{ff^{\prime}}-\cH_{fu^{\prime}}\cH_{uu^{\prime}}^{-1}\cH_{uf^{\prime}})^{-1}$ and $\cH_{-u}=(\cH_{uu^{\prime}}-\cH_{uf^{\prime}}\cH_{ff^{\prime}}^{-1}\cH_{fu^{\prime}})^{-1}$ with $\|\cH_{-f}\|=O_p(q)$ and $\|\cH_{-u}\|=O_p(n)$.

      For the upper-left block $\tilde\cH_{ff^{\prime}}$,
      we have $\tilde\cH_{ff^{\prime}}=\cH_{ff^{\prime}}^{-1}+\cH_{ff^{\prime}}^{-1}\cH_{fu^{\prime}}\cH_{-u}\cH_{uf^{\prime}}\cH_{ff^{\prime}}^{-1}$ by Woodbury formula. By (i) of Lemma~\ref{lemma: l1 estimates} we know $\big\|\cH_{ff^{\prime}}^{-1}\big\|_{\infty}=\big\|\cH_{ff^{\prime}}^{-1}\big\|_1=O_p(p^{1/2}q)$ and by (i), (v) of Lemma~\ref{lemma: l1 estimates} we have
\begin{align*}
        \big\|\cH_{ff^{\prime}}^{-1}\cH_{fu^{\prime}}\cH_{-u}\cH_{uf^{\prime}}\cH_{ff^{\prime}}^{-1}\big\|_{\infty}\le& \max_{j}\big\|[\cH_{ff^{\prime}}^{-1}\cH_{fu}]_{[P_{j},]}\cH_{-u}\cH_{uf^{\prime}}\cH_{ff^{\prime}}^{-1}\big\|_{\infty}\\
        \le &(pq)^{1/2}\max_j\big\|[\cH_{ff^{\prime}}^{-1}\cH_{fu^{\prime}}]_{[P_{j},]}\big\|_{2}\|\cH_{-u}\|\|\cH_{uf^{\prime}}\|\|\cH_{ff^{\prime}}^{-1}\|\\=&O_p\big(\sqrt{p}(nq)^{2/\xi}q\big).
    \end{align*}
    So we conclude that $\big\|\tilde\cH_{ff^{\prime}}\big\|_{\infty}\le O_p\big(\sqrt{p}q(nq)^{2/\xi}\big)$.
    
    For the lower-right block $\tilde{\cH}_{uu^{\prime}}$, we have $\tilde\cH_{uu^{\prime}}=\cH_{uu^{\prime}}^{-1}+\cH_{uu^{\prime}}^{-1}\cH_{uf^{\prime}}\cH_{-f}\cH_{fu^{\prime}}\cH_{uu^{\prime}}^{-1}$. By (ii) of Lemma~\ref{lemma: l1 estimates} we know $\big\|\cH_{uu^{\prime}}^{-1}\big\|_{\infty}=\big\|\cH_{uu^{\prime}}^{-1}\big\|_1=O_p(\sqrt{p}n)$ and by (ii), (vi) of Lemma~\ref{lemma: l1 estimates} we have
    \begin{align*}
        \big\|\cH_{uu^{\prime}}^{-1}\cH_{uf^{\prime}}\cH_{-f}\cH_{fu^{\prime}}\cH_{uu^{\prime}}^{-1}\big\|_{\infty}&= \max_i\big\|[\cH_{uu^{\prime}}^{-1}\cH_{uf^{\prime}}]_{[K_i,]}\cH_{-f}\cH_{fu^{\prime}}\cH_{uu^{\prime}}^{-1}\big\|_{\infty}
        \\
        &\le \sqrt{n}\max_i\big\|[\cH_{uu^{\prime}}^{-1}\cH_{uf^{\prime}}]_{[K_i,]}\big\|\big\|\cH_{-f}\big\|\big\|\cH_{fu^{\prime}}\big\|\big\|\cH_{uu^{\prime}}^{-1}\big\|\\
         &{\le} O_p\big((\sqrt{p}+(nq)^{1/\xi})(nq)^{1/\xi}n\big)
    \end{align*}
     So we conclude that $\big\|\tilde\cH_{uu^{\prime}}\big\|_{\infty}\le O_p\big(\sqrt{p}n(nq)^{2/\xi}\big)$.

     For the off-diagonal term $\tilde\cH_{fu^{\prime}}$, again by the formula of Sherman–Morrison–Woodbury formula, 
     \begin{equation*}
     \tilde\cH_{fu^{\prime}}=\tilde\cH_{ff^{\prime}}\cH_{fu^{\prime}}\cH_{uu^{\prime}}^{-1}=\cH_{ff^{\prime}}^{-1}\cH_{fu^{\prime}}\cH_{uu^{\prime}}^{-1}+\cH_{ff^{\prime}}^{-1}\cH_{fu^{\prime}}\cH_{-u}\cH_{uf^{\prime}}\cH_{ff^{\prime}}^{-1}\cH_{fu^{\prime}}\cH_{uu^{\prime}}^{-1}.
     \end{equation*}
     By (i) (iii) and (v) of Lemma~\ref{lemma: l1 estimates} we have for the first term that
     \begin{align*}\big\|\cH_{ff^{\prime}}^{-1}\cH_{fu^{\prime}}\cH_{uu^{\prime}}^{-1}\big\|_{\infty}=&\max_j\big\|[\cH_{ff^{\prime}}^{-1}\cH_{fu^{\prime}}\cH_{uu^{\prime}}^{-1}]_{[P_j,]}\\\le& \sqrt {n}\max_j\big\|[\cH_{ff^{\prime}}^{-1}\cH_{fu^{\prime}}]_{[P_j,]}\big\|\big\|\cH_{uu^{\prime}}^{-1}\big\|\\=&O_p\big(\sqrt{p}(nq)^{1/\xi}n\big),\end{align*} and for the second term 
\begin{align*}
    \big\|\cH_{ff^{\prime}}^{-1}&\cH_{fu^{\prime}}\cH_{-u}\cH_{uf^{\prime}}\cH_{ff^{\prime}}^{-1}\cH_{fu^{\prime}}\cH_{uu^{\prime}}^{-1}\big\|_{\infty}\\&\le \max_{j}\big\|[\cH_{ff^{\prime}}^{-1}\cH_{fu^{\prime}}]_{[P_j,]}\cH_{-u}\cH_{uf^{\prime}}\cH_{ff^{\prime}}^{-1}\cH_{fu^{\prime}}\cH_{uu^{\prime}}^{-1}\big\|_{\infty}\\&\le \sqrt{n}\max_{j}\big\|[\cH_{ff^{\prime}}^{-1}\cH_{fu^{\prime}}]_{[P_j,]}\cH_{-u}\cH_{uf^{\prime}}\cH_{ff^{\prime}}^{-1}\cH_{fu^{\prime}}\cH_{uu^{\prime}}^{-1}\big\|\\&\le \sqrt{n}\max_j\big\|[\cH_{ff^{\prime}}^{-1}\cH_{fu^{\prime}}]_{[P_j,]}\big\|\big\|\cH_{-u}\big\|\big\|\cH_{uf^{\prime}}\big\|\big\|\cH_{ff^{\prime}}^{-1}\big\|\big\|\cH_{fu^{\prime}}\big\|\big\|\cH_{uu^{\prime}}^{-1}\big\|\\
    &\le O_p\Big({ \sqrt p(nq)^{3/\xi}}n\Big).
\end{align*}
Then  we obtain $\big\|\tilde\cH_{fu^{\prime}}\big\|=O_p( \sqrt p(nq)^{3/\xi}\sqrt{n})$
For the off-diagonal term $\tilde\cH_{uf^{\prime}}$, we expand it similarly by the formula of Sherman–Morrison–Woodbury formula as
     \begin{equation*}
     \tilde\cH_{uf^{\prime}}=\tilde\cH_{uu^{\prime}}\cH_{uf^{\prime}}\cH_{ff^{\prime}}^{-1}=\cH_{uu^{\prime}}^{-1}\cH_{uf^{\prime}}\cH_{ff^{\prime}}^{-1}+\cH_{uu^{\prime}}^{-1}\cH_{uf^{\prime}}\cH_{-f}\cH_{fu^{\prime}}\cH_{uu^{\prime}}^{-1}\cH_{uf^{\prime}}\cH_{ff^{\prime}}^{-1}.
     \end{equation*}
     By (ii) (vi) and (vi) of Lemma~\ref{lemma: l1 estimates} the first term can be estimated by
     \begin{align*}\big\|\cH_{uu^{\prime}}^{-1}\cH_{uf^{\prime}} \cH_{ff^{\prime}}^{-1}\big\|_{\infty}&\le \big\|\cH_{uu^{\prime}}^{-1}\cH_{uf^{\prime}}\cH_{ff^{\prime}}^{-1}\big\|_{\infty}\\&\le \max_i\big\|[\cH_{uu^{\prime}}^{-1}\cH_{uf^{\prime}}]_{[K_i,]}\cH_{ff^{\prime}}^{-1}\big\|_{\infty}\\&\le \sqrt{pq}\max_i\big\|[\cH_{uu^{\prime}}^{-1}\cH_{uf^{\prime}}]_{[K_i,]}\big\|\big\|\cH_{ff^{\prime}}^{-1}\big\|\\&=O_p\big((\sqrt{p}+(nq)^{1/\xi})q\big),\end{align*} and the second term can be similarly calculated as
\begin{align*}
        \big\|\cH_{uu^{\prime}}^{-1}&\cH_{uf^{\prime}}\cH_{-f}\cH_{fu^{\prime}}\cH_{uu^{\prime}}^{-1}\cH_{uf^{\prime}} \cH_{ff^{\prime}}^{-1}\big\|_{\infty}\\&\le \max_{i}\big\|[\cH_{uu^{\prime}}^{-1}\cH_{uf^{\prime}}]_{[K_i,]}\cH_{-f}\cH_{fu^{\prime}}\cH_{uu^{\prime}}^{-1}\cH_{uf^{\prime}} \cH_{ff^{\prime}}^{-1}\big\|_{\infty}\\&\le \sqrt{pq}\max_{i}\big\|[\cH_{uu^{\prime}}^{-1}\cH_{uf^{\prime}}]_{[K_i,]}\cH_{-f}\cH_{fu^{\prime}}\cH_{uu^{\prime}}^{-1}\cH_{uf^{\prime}}\cH_{ff^{\prime}}^{-1} \big\|\\&\le \sqrt{pq}\max_i\big\|[\cH_{uu^{\prime}}^{-1}\cH_{uf^{\prime}}]_{[K_i,]}\big\|\big\|\cH_{-f}\big\|\big\| \cH_{fu^{\prime}}\big\|\big\|\cH_{uu^{\prime}}^{-1}\big\|\big\|\cH_{uf^{\prime}}\big\|\big\|\cH_{ff^{\prime}}^{-1}\big\|\\
    &\le O_p\Big({\sqrt{p}(\sqrt{p}+(nq)^{1/\xi})(nq)^{2/\xi}q}\Big).
\end{align*}
     Therefore we conclude that the infinity bound for $\tilde\cH_{uf^{\prime}}$ is $O_p(p(nq)^{3/\xi}q)$.
     
     Till here we have proved the estimates in infinity norm for the blocks of Hessian matrix on $\bphi^0$. It suffices to prove that the estimates in Lemma~\ref{lemma: l1 estimates} are still valid when $\cH(\bphi^0)$ is replaced with $\cM(\tilde\bphi)$. Still, from now we omit the dependence on $\tilde\bphi$ unless otherwise specified. Specifically, we merely need to verify the following hold:
     \begin{enumerate}
        \item $\|\Mb_{Lff^{\prime}}^{-1}\|=O_p(q)$, $\|\cM_{ff^{\prime}}^{-1}\|=O_p(q)$, and $\|\cM_{ff^{\prime}}^{-1}\|_1=O_p(\sqrt{p}q)$.
        \item $\|\Mb_{Luu^{\prime}}^{-1}\|=O_p(n)$, $\|\cM_{uu^{\prime}}^{-1}\|=O_p(n)$, and $\|\cH_{uu^{\prime}}^{-1}\|_1=O_p(\sqrt{p}n)$.
        \item $\max_i\big\|[\cM_{uf^{\prime}}]_{[K_i,]}\big\|=O_p\big((p^{1/2}+(nq)^{1/\xi})n^{-1}q^{-1/2}\big)$, $\|\cM_{uf^{\prime}}\|_1=O_p\big((nq)^{1/\xi}q^{-1}\big)$, and $\|\cM_{uf^{\prime}}\|=O_p\big((nq)^{-1/2+1/\xi}\big)$. 
        \item $\max_j\big\|[\cM_{fu^{\prime}}]_{[P_j,]}\big\|=O_p\big((nq)^{1/\xi}n^{-1}q^{-1/2}\big)$, $\|\cM_{fu^{\prime}}\|_1=O_p\big((p+(nq))^{1/\xi}n^{-1}\big)$, and $\|\cM_{fu^{\prime}}\|=O_p\big((nq)^{-1/2+1/\xi}\big)$. 
        \item $\|[\tilde\bLambda_{1}]_{[P_j,\cdot]}\|=O_p(q^{-1})$ and $\|\tilde\bLambda_{1}\|=O_p(q^{-1/2})$.
        \item $ \|[\tilde\bLambda_{2}]_{[K_i,\cdot]}\|=O_p(p^{1/2}n^{-1})$ and $\|\tilde\bLambda_{2}\|=O_p(n^{-1/2})$.
    \end{enumerate}
    First, it can be verified that inside this region, the residuals terms in the Hessian of penalty functions, i.e., the last two terms in \eqref{eq:HP}, can be incorporated into these estimates. Here $\tilde\bLambda_{1}$ and $\tilde\bLambda_{2}$ are defined in the following.
      Define $\Hb_{Lff^{\prime}}(s)$, $\Hb_{Lfu^{\prime}}(s)$, $\Hb_{Luf^{\prime}}(s)$, $\Hb_{Luu^{\prime}}(s)$, $\Hb_{Rfu^{\prime}}(s)$, $\Hb_{Ruf^{\prime}}(s)$ and $\bLambda(s)$ as $\Hb_{Lff^{\prime}}(\bphi)$, $\Hb_{Lfu^{\prime}}(\bphi)$, $\Hb_{Luf^{\prime}}(\bphi)$, $\Hb_{Luu^{\prime}}(\bphi)$, $\Hb_{Rfu^{\prime}}(\bphi)$, $\Hb_{Ruf^{\prime}}(\bphi)$ and $\bLambda(\bphi)$ on $\bphi^0+s(\tilde\bphi-\bphi^0)$. Here $\tilde\bphi$ is defined in Lemma~\ref{lemma: first order condition}. Then
      \begin{align*}
          \cM&=\begin{pmatrix}
              \cM_{ff^{\prime}}&\cM_{fu^{\prime}}\\\cM_{uf^{\prime}}&\cM_{uu^{\prime}}
          \end{pmatrix}\\&=\int_0^1\begin{pmatrix}
              \Hb_{Lff^{\prime}}(s)-\bLambda_1(s)\bLambda_1^\T(s)&\Hb_{Lfu^{\prime}}(s)+\Hb_{Rfu^{\prime}}(s)-\bLambda_1(s)\bLambda_2^\T(s)\\\Hb_{Luf^{\prime}}(s)+\Hb_{Ruf^{\prime}}(s)-\bLambda_2(s)\bLambda_1^\T(s)&\Hb_{Luu^{\prime}}(s)-\bLambda_2(s)\bLambda_2^\T(s)
          \end{pmatrix}ds.
      \end{align*}
      For the upper-left block $\cM_{ff^{\prime}}$, since $\bLambda_1(s)$ is linear with respect to $s$, we get
           $\cM_{ff^{\prime}}=\Mb_{Lff^{\prime}}-\tilde\bLambda_1\tilde\bLambda^\T$,
       where $\Mb_{Lff^{\prime}}=\int_0^1\Hb_{Lff^{\prime}}(s)ds$ and $\tilde\bLambda_1=\big(\bLambda_1(1/2),(\bLambda_1(1)-\bLambda_1(0))/(2\sqrt{3})\big)$. Then agian by Woodbury identity
       \begin{equation*}
           \cM_{ff^{\prime}}^{-1}=\Mb_{Lff^{\prime}}^{-1}-\Mb_{Lff^{\prime}}^{-1}\tilde\bLambda_1\big[\Ib_{2K^2}+\tilde\bLambda_1^\T\Mb_{Lff^{\prime}}^{-1}\tilde\bLambda_1\big]^{-1}\tilde\bLambda_1^\T\Mb_{Lff^{\prime}}^{-1}
       \end{equation*}
       Similar to (vii) in Lemma \ref{lemma: l1 estimates}, we have
       \begin{equation*}\|[\tilde\bLambda_{1}]_{[P_j,\cdot]}\|=O_p(q^{-1}),\quad\|\tilde\bLambda_{1}^\T \|_1=O_p(q^{-1}),\quad\|\tilde\bLambda_{1}\|_1=O_p(1),\quad\|\tilde\bLambda_{1}\|=O_p(q^{-1/2}).
\end{equation*}
This implies (v'). Next for $\Mb_{Lff^{\prime}}=\int_0^1\Hb_{Lff^{\prime}}(s)ds$, it is block-diagonal and when $\tilde\bphi\in\cB(D)$, $\sqrt{p}\big\|\Db_q^{-1/2}(\tilde\bphi-\bphi^0)\big\|\le c$ for some small enough $c$, as $\tilde\bphi$ lies in the line segment between $\hat\bphi$ and $\bphi^0$. Under Assumptions~\ref{assumption: psd covariance}--\ref{assumption:smoothness}, with $p=o(\sqrt{\delta_{nq}})$, 
\begin{align}
    \rho_{\min}(\Mb_{Lff^{\prime}})=&\min_{j}\rho_{\min}[\Mb_{Lff^{\prime}}]_{[P_j,P_j]}\nonumber\\\ge& \min_{j,0\le s\le 1}\rho_{\min}[\Hb_{Lff^{\prime}}]_{[P_j,P_j]}\nonumber\\
    \ge &b_L\rho_{\min}(\big(nq)^{-1}\sum_{i=1}^n\bZ_i^0(\bZ_i^0)^\T\big)\nonumber\\&-(nq)^{-1}\max_{j,s}\Big\|\sum_{i=1}^n\big[l_{ij}^{\prime\prime}(w_{ij}^0)\bZ_i^0(\bZ_i^0)^\T-l_{ij}^{\prime\prime}(w_{ij}(s))\bZ_i(s)\bZ_i(s)^\T\Big\|\nonumber\\\gtrsim &\frac{1}{q},\label{eq_minrhominff}
\end{align}
where $c$ is selected to be small enough.
Then we conclude that as $\rho_{\min}(\Mb_{Lff^{\prime}})\ge O_p(q^{-1}) $.
Similar to Lemma \ref{lemma: l1 estimates}, we obtain $\big\|\Mb_{Lff^{\prime}}^{-1}\big\|=O_p(q)$ and $\big\|\cM_{ff^{\prime}}^{-1}\big\|_1=O_p(\sqrt pq)$, which is (i'). For $\cM_{uu^{\prime}}$, we first compute by $\tilde\bLambda_2=\big(\bLambda_2(1/2),(\bLambda_2(1)-\bLambda_2(0))/(2\sqrt{3})\big)$ followed by a similar argument in proving (vii) of Lemma \ref{lemma: l1 estimates} that
       \begin{equation*}\|[\tilde\bLambda_{2}]_{[K_i,\cdot]}\|=O_p(\sqrt{p}n^{-1}),\quad\|\tilde\bLambda_{2}^\T \|_1=O_p(pn^{-1}),\quad\|\tilde\bLambda_{2}\|_1=O_p(1),\quad\|\tilde\bLambda_{2}\|=O_p(n^{-1/2}).
\end{equation*}
This implies (vi'). For $\Mb_{Luu^{\prime}}=\int_0^1\Hb_{Luu^{\prime}}(s)ds$, it is block-diagonal and when $\tilde\bphi\in\cB(D)$, $\sqrt{p}\big\|\Db_q^{-1/2}(\tilde\bphi-\bphi^0)\big\|\le c$, we have by a similar approach of computing \eqref{eq_minrhominff} that $\rho_{\min}(\Mb_{Luu^{\prime}})$ $\ge O_p(n^{-1})$.
Similar to Lemma \ref{lemma: l1 estimates}, we obtain $\big\|\Mb_{Luu^{\prime}}^{-1}\big\|_1=O_p(n)$ and $\big\|\cM_{uu^{\prime}}^{-1}\big\|_1=O_p(n)$, which implies (ii') For the off-diagonal blocks,

Define $\cM_{-f}=(\cM_{ff^{\prime}}-\cM_{fu^{\prime}}\cM_{uu^{\prime}}^{-1}\cM_{uf^{\prime}})^{-1}$ and $\cM_{-u}=(\cM_{uu^{\prime}}-\cM_{uf^{\prime}}\cM_{ff^{\prime}}^{-1}\cM_{fu^{\prime}})^{-1}$. By Lemma \ref{prop:min eigenvalue} and $\rho_{\min}(\Db_q^{1/2}\cM\Db_q^{1/2})\ge \min_s\rho_{\min}(\Db_q^{1/2}\cH(s)\Db_q^{1/2})$ we have $\|\cM_{-f}\|=O_p(q)$ and $\|\cM_{-u}\|=O_p(n)$.
By similar argument in Lemma~\ref{lemma: l1 estimates}, we know $\|\cM_{fu^{\prime}}\|=\|\cM_{uf^{\prime}}\|=O_p((nq)^{-1/2})$. For the off-diagonal blocks, since inside $\cB(D)\cap\sqrt{p}\big\|\Db_q^{-1/2}(\tilde\bphi-\bphi^0)\big\|\le m$, $\max_{i,j}\big|l_{ij}^{\prime}(w_{ij}(s))-l_{ij}^{\prime}(w_{ij}^0)\big|$ is bounded by Assumption~\ref{assumption:smoothness}. The bounds for $\Hb_{Lfu^{\prime}}$ and $[\blambda_1]_{[P_j,]}[\bLambda_{[,K_i]}$ can be derived similarly as in Lemma~\ref{lemma: l1 estimates}. Therefore we get (iii') and similarly, we can derive (iv').
  \end{proof}

\subsection{Proof of Lemmas \ref{prop:infinty bound} and \ref{lemma: first order condition}}
\begin{lemma}[Infinity Bound]
	\label{prop:infinty bound}
	Under Assumptions~\ref{assumption: psd covariance}--\ref{assumption: Scaling}, we have for any small $\epsilon$, that
	\begin{align*}
			\big\|\hat\bU_v-\bU^0_v\big\|_{\infty}&\,=O_p\Big(\frac{p}{\sqrt{n\wedge(pq)}}(nq)^{3/\xi+\epsilon}\Big);\\\big\|\hat\bbf_v-\bbf^0_v\big\|_{\infty}&\,=O_p\Big(\frac{\sqrt p}{\sqrt{n\wedge q}}(nq)^{3/\xi+\epsilon}\Big).	\end{align*}
	
\end{lemma}

 Under the scaling condition in Assumption~\ref{assumption: Scaling} we know that $\|\hat\bphi^0-\bphi^0\|_{\infty}$ converges to 0 with high probability, and therefore each parameter of the estimation $\hat\bphi^0$ converges to the that of $\bphi^0$, which implies that $\hat\bphi^0$ is an interior point of $\cB(D)$ with high probability. This immediately implies the following:

\begin{lemma}[First Order Condition]
\label{lemma: first order condition}
     Under Assumptions~\ref{assumption: psd covariance}--\ref{assumption: Scaling}, $\partial_{\bphi}\cL(\Yb|{ \hat{\bphi}^0}) = \zero$ {\em w.h.p.}
\end{lemma}

The above first order condition plays a foundational role in establishing the asymptotic normality of our estimation. This result is nontrivial in the sense that the dimension of the model parameters goes to infinity along with $n,q,p$. Also, the dimensions of factors, loadings, and the regression coefficients  differ, not only in magnitude but also in their order.
Based on the first order condition, we are able to derive the individual consistency rate for $\hat\bphi^0$:

Suppose $\tilde{{\bphi}}$ is the solution to 
\begin{equation}
    \mathop{\arg\min}_{{\bphi}\in\cB(D),\sqrt{p}\|\Db_{q}^{-1/2}({\bphi}-{\bphi}^0)\|\le m}\|\Db_{q}^{1/2}\bS({\bphi})\|_\zeta.\label{eq:firstorderopti}
\end{equation}
The idea is that we first show that $\tilde\bphi$ is the interior point of the space $\cB(D)\cap\big\{\bphi:\sqrt{p}\big\|\Db_q^{-1/2}(\bphi-\bphi^0)\big\|\le m\big\}$ {\em w.h.p.} To show this we need to bound $\|\hat\bbf_v^0-\bbf_v^G\|_{\infty}$ and $\|\hat\bU_v^0-\bU_v^G\|_{\infty}$. Then we conclude that $\tilde\bphi$ is the maximizer of the objective function $\cL(\bphi)$ inside this space with $\partial_{\bphi}\|\Db_{q}^{1/2}\bS({\bphi})\|_\zeta=0$ {\em w.h.p.}, which implies that $\bS(\tilde\bphi)=0$, {\em w.h.p.} Recall that $\hat\bphi^0$ is a global maximizer of the objective function in $\cB(D)$, we conclude that $\tilde\bphi=\hat\bphi^0$ and thus we proved 
Lemma \ref{prop:infinty bound} and \ref{lemma: first order condition}
\begin{proof}

Expanding $\bS(\tilde\bphi)$ at $\bphi^0$ using mean value theorem, we have 
\begin{equation}
    \tilde\bphi-\bphi^0=\tilde\cM(\tilde\bphi)\big(\bS(\tilde\bphi)-\bS(\bphi^0)\big), \label{eq:lemma2 mvt expansion}
\end{equation}
where $\tilde \cM$ is defined as 
\begin{equation*}\tilde\cM(\bphi)=\Big[\int_0^1\cH\big(\bphi^0+s(\bphi-\bphi^0)\big)ds\Big]^{-1},
\end{equation*}
and $\cM(\bphi)$ defined as $\int_0^1\cH\big(\bphi^0+s(\bphi-\bphi^0)\big)ds$.

We proceed to show $\tilde{\bphi}$ lies in the interior of $\{\bphi:\sqrt{p}\big\|\Db_q^{-1/2}(\bphi-\bphi^0)\big\|\le m \}$. By reorganizing the expression~\eqref{eq:lemma2 mvt expansion}, we have $\Db_{q }^{-1/2}(\tilde{\bphi} -{\bphi}^0)=\big(\Db_{q}^{1/2}\tilde{\cM}(\tilde{\bphi})\Db_{q}^{1/2}\big)^{-1}\big[\Db_{q}^{1/2}\big(\bS(\tilde{\bphi} )-\bS({\bphi}^0)\big)\big]$. Further, according to Lemma~\ref{lemma:first-order derivative}, Lemma~\ref{lemma: colsum norm cH tilde} and Assumption~\ref{assumption: Scaling}, we have\begin{align}
    \big\|\Db_{q}^{-1/2}(\tilde{\bphi} -{\bphi}^0)\big\|&\le \big\|\Db_{q}^{1/2}\big(\bS(\tilde{\bphi} )-\bS({\bphi}^0)\big)\big\|\nonumber\\
    &\le (n+pq)^{1/2-1/\zeta}\big\|\Db_{q}^{1/2}\big(\bS(\tilde{\bphi} )-\bS({\bphi}^0)\big)\big\|_\zeta\nonumber\\
    &\le 2(n+pq)^{1/2-1/\zeta}\big\|\Db_{q}^{1/2}\bS({\bphi}^0)\big\|_\zeta\nonumber\\
    &\le O_p\Big(\frac{p}{\sqrt{pq\wedge n}}\epsilon_{nq}\Big)=o_p(1).\label{BB55}\end{align}

  Next, we show  $\tilde{\bphi}$ lies in the interior of $\cB (D)$. Again by Lemma~\ref{lemma:first-order derivative} and Lemma~\ref{lemma: colsum norm cH tilde}, we have
\begin{align}
    &\,\| \tilde{\bbf}_v - \bbf_v^0\|_{\infty} \nonumber\\\le& \|q^{-1/2}\tilde \cM_{ff^{\prime}} \|_{\infty} \|\sqrt{q}(\bS_{f}(\tilde\bphi)-\bS_f(\bphi^0) )\|_{\infty}+\|n^{-1/2}\tilde \cM_{uf^{\prime}} \|_{\infty} \|\sqrt{n}(\bS_{u}(\tilde\bphi)-\bS_u(\bphi^0)) \|_{\infty} \nonumber \\\le& \|q^{-1/2}\tilde \cM_{ff^{\prime}} \|_{\infty} \|\Db_q^{1/2}(\bS(\tilde\bphi)-\bS(\bphi^0)) \|_{\zeta}+\|n^{-1/2}\tilde \cM_{uf^{\prime}} \|_{\infty} \|\Db_q^{1/2}(\bS(\tilde\bphi)-\bS(\bphi^0)) \|_{\zeta}\nonumber\\
    \le& 2\big(q^{-1/2} \|\tilde \cM_{ff^{\prime}} \|_{\infty}+n^{-1/2}\|\tilde \cM_{uf^{\prime}} \|_{\infty}\big)\|\Db_q^{1/2}\bS(\bphi^0) \|_{\zeta}\nonumber
    \\\le &O_p\left(\frac{\sqrt p(\sqrt n+\sqrt {q})}{\sqrt{nq}}(nq)^{3/\xi+\epsilon}\right)=o_p(1),\nonumber
    \end{align}
under Assumption~\ref{assumption: Scaling}. Similarly, we have
\begin{align}
    \|\tilde{\bU}_v - \bU_v^0 \|_{\infty} &\le O_p\left(\frac{\sqrt p(\sqrt {pq}+\sqrt n)}{\sqrt{nq}}(nq)^{3/\xi+\epsilon}\right)  = o_p(1). \nonumber
\end{align}
    Hence, we show $\tilde{\bphi}$, the solution to \eqref{eq:firstorderopti}, to be an interior point of the specified space $\cB(D)\cap\big\{\bphi:\sqrt{p}
    \big\|\Db_q^{-1/2}(\bphi-\bphi^0)\big\|\le m\big\}$ {\em w.h.p.}.

    By the definition of $\tilde{\bphi}$, we have 
$\partial_{\bphi} \|\Db_{q}^{1/2}\bS({\bphi})\|_\zeta^\zeta |_{\bphi = \tilde{\bphi}}\;=0$, which further gives $\zeta  H (\tilde{\bphi}) $ $ \Db_q^{1/2} [\Db_{q}^{1/2}\bS(\tilde{\bphi})]^{\zeta-1}=0$. By Lemma~\ref{prop:min eigenvalue}, we know that $\Db_q^{1/2} \cH (\tilde{\bphi}) \Db_q^{1/2}$ is positive definite within the space $\cB(D)\cap\big\{\bphi:\sqrt {p}\big\|\Db_q^{-1/2}(\bphi-\bphi^0)\big\|\le m\big\}$. Hence, we have $\bS(\tilde\bphi)=0$  {\em w.h.p.}.
    
    Lastly, according to average consistency results in Lemma~\ref{prop:average consistency}, we know that the global maximizer 
 $\hat{\bphi}^0$ to objective function $\cL(\bphi)$ in $\cB(D)$ lies in the space $\cB(D)\cap\big\{\bphi:\sqrt{p}\big\|\Db_q^{-1/2}(\bphi-\bphi^0)\big\|\le m\big\}$ {\em w.h.p.}. Hence, we conclude that $\tilde\bphi=\hat\bphi^0$ and therefore $S(\hat\bphi^0)=0$ {\em w.h.p.}.


\end{proof}


  \subsection{Proof of Lemma~\ref{lemma: tech_for_normality}}
\begin{lemma}\label{lemma: tech_for_normality}
 Under Assumptions~\ref{assumption: psd covariance}--\ref{assumption:asymptotic normality}, we have estimates on $\bphi^0$ as follows:

 (i) $\big\|\bLambda_1^\T\Hb_{Lff^{\prime}}^{-1}\bS_f\big\|=O_p\big(\sqrt{p/(nq) }\epsilon_{nq})\big)$; (ii) $\big\|\bLambda_2^\T\Hb_{Lff^{\prime}}^{-1}\bS_f\big\|=O_p\big(\sqrt{p/(nq) }\epsilon_{nq}\big)$
 
 (iii) $\big\|[\cH_{uf^{\prime}}  \cH_{ff^\prime}^{-1}\bS_f]_{[K_i]}\big\|=O_p\big(pn^{-3/2}q^{-1/2} (nq)^{1/\xi}\epsilon_{nq}\big)$;
 
 (iv) $\big\|[\cH_{fu^{\prime}}  \cH_{uu^\prime}^{-1}\bS_u]_{[P_j]}\big\|=O_p\big(pn^{-1/2}q^{-3/2}(nq)^{1/\xi}\epsilon_{nq} \big)$;

 (v)$\big\|\cH_{uf^{\prime}}  \cH_{ff^\prime}^{-1}\bS_f\big\|=O_p\big(pn^{-1}q^{-1/2}(nq)^{1/\xi}\epsilon_{nq}\big)$;
 
 (vi) $\big\|\cH_{fu^{\prime}}  \cH_{uu^\prime}^{-1}\bS_u\big\|=O_p\big(pn^{-1/2}q^{-1}(nq)^{1/\xi}\epsilon_{nq}\big)$.
\end{lemma}
\begin{proof}
For (i), since $\bLambda_{1}=(\bLambda_{1\bullet},\zero_{q(K+p)\times Kp})$, we only need to prove first $K^2$ entries of $\bLambda_1^\T\Hb_{Lff^{\prime}}^{-1}\bS_f$ can be bounded by $O_p\big(\sqrt{p }n^{-1/2}\delta_{nq}^{-1}\epsilon_{nq}\big)$. By (ii) of Assumption~\ref{assumption:asymptotic normality}, for all $r\in[K]$, the $K(r)$th element can be bounded as
 \begin{align*}
    \Big|\big[\bLambda_1^\T   \Hb_{Lff^\prime}^{-1}\bS_f\big]_{[K(r)]}\Big|&=\Big|\frac{1}{q}\sum_{j=1}^q\gamma_{jr}^0 (\one_r^{(K+p)})^\T \Big[\sum_{t=1}^nl^{\prime\prime}_{tj}(w_{ij}^0)\bZ_t^0(\bZ_t^0)^\T \Big]^{-1}\Big(\sum_{i=1}^nl^{\prime}_{ij}(w_{ij}^0)\bZ_i^0\Big)\Big|\\&=\Big|\frac{1}{q}\sum_{j=1}^q\sum_{i=1}^nl^{\prime}_{ij}(w_{ij}^0)\gamma_{jr}^0 (\one_r^{(K+p)})^\T \Big[\sum_{t=1}^nl^{\prime\prime}_{tj}(w_{tj}^0)\bZ_t^0(\bZ_t^0)^\T \Big]^{-1}\bZ_i^0\Big|\\&\le \Big|\frac{1}{q}\sum_{j=1}^q\sum_{i=1}^nl^{\prime}_{ij}(w_{ij}^0)b_U\Big\{\lambda_{\min}\big[\sum_{t=1}^n\bZ_t^0(\bZ_t^0)^\T\big] \Big\}^{-1}\gamma_{jr}^0\big\|\bZ_i^0\big\|\Big|\\&= O_p\Big(\sqrt{\frac{p}{nq}}\epsilon_{nq}\Big)\end{align*}
    Here the term $\epsilon_{nq}$ arises from using Bernstein bound as $\{l_{ij}^{\prime}\}_{i,j}$ are independent with zero mean. We suppress the commonly used $\log$ terms and replace them with $\epsilon_{nq}$ for brevity.
    and for $h\neq l\in[K]$, the $K(h,l)$th element can be bounded as
    \begin{align*}
    &\,\Big|\big[\bLambda_1^\T   \Hb_{Lff^\prime}^{-1}\bS_f\big]_{[K(h,l)]}\Big|\\&=\Big|\frac{1}{q}\sum_{j=1}^q \big(\gamma_{jh}^0(\one_l^{(K+p)})+\gamma_{jl}^0(\one_h^{(K+p)})\big)^\T \Big[\sum_{t=1}^nl^{\prime\prime}_{tj}(w_{ij}^0)\bZ_t^0(\bZ_t^0)^\T \Big]^{-1}\Big[\sum_{i=1}^nl^\prime_{ij}(w_{ij}^0)\bZ_i^0\Big]\Big|\\&\le \Big|\frac{1}{q}\sum_{j=1}^q\sum_{i=1}^nb_L\Big\{\lambda_{\min}\big[\sum_{t=1}^n(w_{tj}^0)\bZ_t^0(\bZ_t^0)^\T\big] \Big\}^{-1}2|\gamma_{jh}^0+\gamma_{jl}|\big\|\bZ_i^0\big\|(\one_r^{(K+p)})^\T l^{\prime}_{ij}(w_{ij}^0)\Big|\\&= O_p\Big(\sqrt{\frac{p}{nq}}\epsilon_{nq}\Big)
\end{align*}
So we proved (i).

For (ii), note that $\cH_{uf^{\prime}}$ can be written as a sum of three parts: $
    \Hb_{Luf^{\prime}}+\Hb_{Ruf^{\prime}}+\bLambda_2\bLambda_1^\T$, $[\cH_{uf^{\prime}}  \cH_{ff^\prime}^{-1}\bS_f]_{[K_i]}$ can be bounded by
    \begin{equation*}
    \|[\cH_{uf^{\prime}}  \cH_{ff^\prime}^{-1}\bS_f]_{[K_i]}\|\le \big\|[\Hb_{Luf^{\prime}}\cH_{ff^{\prime}}^{-1}\bS_f]_{[K_i,]}\big\|+ \big\|[\Hb_{Ruf^{\prime}}\cH_{ff^{\prime}}^{-1}\bS_f]_{[K_i,]}\big\|+ \big\|[\bLambda_2\bLambda_1^\T\cH_{ff^{\prime}}^{-1}\bS_f]_{[K_i,]}\big\|,
\end{equation*}
with $\cH_{ff^{\prime}}^{-1}=\Hb_{Lff}^{-1}-\Hb_{Lff}^{-1}\bLambda_1(\Ib_{K^2}+\bLambda_1^\T\Hb_{Lff^{\prime}}^{-1}\bLambda_1)^{-1}\bLambda_1^\T\Hb_{Lff^{\prime}}^{-1}$.
Then for the first term, by Assumptions~\ref{assumption: psd covariance}--\ref{assumption:smoothness}, we have
\begin{align}
       &\,\Big|[\Hb_{Luf^{\prime}}  \Hb_{Lff^\prime}^{-1}\bS_f]_{[(i-1)K+r]}\Big|\nonumber\\&=\Big|(nq)^{-1}\sum_{j=1}^ql_{ij}^{\prime\prime}(w_{ij}^0)\gamma_{jr}^0{\bZ_i^0}^\T\big(\sum_{t=1}^nl_{tj}^{\prime\prime}(w_{tj}^0)\bZ_t^0{\bZ_t^0}^\T\big)^{-1}\big(\sum_{t=1}^nl_{tj}^{\prime}(w_{tj}^0)\bZ_t^0\big)\Big|\nonumber\\&\le \frac{1}{nq}\Big|\sum_{j=1}^q\sum_{t=1}^nl_{tj}^{\prime}(w_{tj}^0)b_U^2\Big\{\lambda_{\min}\big[\sum_{s=1}^n\bZ_s^0(\bZ_s^0)^\T\big]\Big\}^{-1}\|\bZ_t\|\|\bZ_i^0\||\gamma_{jr}|\Big|\nonumber\\&= O_p\big(pn^{-3/2}q^{-1/2}\epsilon_{nq}\big),\label{eq_lemma8_eq1}
\end{align}
and by (i), (iii), (vii) of Lemma~\ref{lemma: l1 estimates} and the previous bound for $\|\bLambda_1^\T\Hb_{Lff^{\prime}}^{-1}\bS_f\|$,
\begin{align}
    \big\|[\Hb_{Luf^{\prime}}&\Hb_{Lff}^{-1}\bLambda_1(\Ib_{K^2}+\bLambda_1^\T\Hb_{Lff^{\prime}}^{-1}\bLambda_1)^{-1}\bLambda_1^\T\Hb_{Lff^{\prime}}^{-1}\bS_f]_{[K_i,]}\big\|\nonumber\\&\le\big\|[\Hb_{Luf^{\prime}}]_{[K_i,]}\big\|\big\|\Hb_{Lff}^{-1}\bLambda_1(\Ib_{K^2}+\bLambda_1^\T\Hb_{Lff^{\prime}}^{-1}\bLambda_1)^{-1}\big\|\big\|\bLambda_1^\T\Hb_{Lff^{\prime}}^{-1}\bS_f\big\|\nonumber\\&= O_p\big(p(nq)^{1/\xi}n^{-3/2}q^{-1/2}\epsilon_{nq}\big).\label{eq_lemma8_eq2}
\end{align}
\eqref{eq_lemma8_eq1} and \eqref{eq_lemma8_eq2} together imply $\|[\Hb_{Luf^{\prime}}  \cH_{ff^\prime}^{-1}\bS_f]_{[K_i]}\|= O_p\big(p(nq)^{1/\xi}n^{-3/2}q^{-1/2}\epsilon_{nq}\big)$. For the second term $[\Hb_{Luf^{\prime}}\cH_{ff^{\prime}}^{-1}\bS_f]_{[K_i,]}$,
\begin{align*}
       \Big|[\Hb_{Ruf^{\prime}}  \Hb_{Lff^\prime}^{-1}\bS_f]_{[(i-1)K+r]}\Big|\le &\left|\frac{1}{nq}\sum_{j=1}^q\Big|\sum_{t=1}^nl_{tj}^{\prime}(w_{tj}^0)\|\bZ_t\|b_U\big\{\lambda_{\min}\big[\sum_{s=1}^n\bZ_s^0(\bZ_s^0)^\T\big]\big\}^{-1}\Big|l_{ij}^{\prime}(w_{ij}^0)\right| \\= &O_p\big(pn^{-3/2}q^{-1/2}\epsilon_{nq} \big).
\end{align*}
Similar with the estimates \eqref{eq_lemma8_eq2} we have 
\begin{equation*}
    \big\|[\Hb_{Ruf^{\prime}}\Hb_{Lff}^{-1}\bLambda_1(\Ib_{K^2}+\bLambda_1^\T\Hb_{Lff^{\prime}}^{-1}\bLambda_1)^{-1}\bLambda_1^\T\Hb_{Lff^{\prime}}^{-1}\bS_f]_{[K_i,]}\big\|= O_p\big(pn^{-3/2}q^{-1/2}(nq)^{1/\xi}\epsilon_{nq}\big),
\end{equation*}
and consequently, we have 
 $\|[\Hb_{Ruf^{\prime}}  \cH_{ff^\prime}^{-1}\bS_f]_{[K_i]}\|= O_p\big(pn^{-3/2}q^{-1/2}(nq)^{1/\xi}\epsilon_{nq}\big)$. Next for \\$[\bLambda_2\bLambda_1^\T  \cH_{ff^\prime}^{-1}\bS_f]_{[K_i]}$, by the previous bound for $\|\bLambda_1^\T\Hb_{Lff^{\prime}}^{-1}\bS_f\|$, estimates of \eqref{eq_lemma8_eq2}, (viii) of Lemma~\ref{lemma: l1 estimates}, we have
\begin{align*}
       \big\|[\bLambda_2\bLambda_1^\T  \Hb_{Lff^\prime}^{-1}\bS_f]_{[K_i]}&\le \big\|[\bLambda_2]_{[K_i,]}\big\|\big\|\big\|\bLambda_1^\T  \Hb_{Lff^\prime}^{-1}\bS_f\big\|=O_p\big(pn^{-3/2}q^{-1/2}\epsilon_{nq}\big),
\end{align*}
and 
\begin{equation*}
   \big\| [\bLambda_2\bLambda_1^\T  \Hb_{Lff}^{-1}\bLambda_1(\Ib_{K^2}+\bLambda_1^\T\Hb_{Lff^{\prime}}^{-1}\bLambda_1)^{-1}\bLambda_1^\T\Hb_{Lff^{\prime}}^{-1}\bS_f]_{[K_i]}\big\|= O_p\big(pn^{-3/2}q^{-1}/2\epsilon_{nq}\big).
\end{equation*}
All these together imply (ii) and (iii).

Next for (iv), since $\bLambda_{2}\in\RR^{nK\times K^2+Kp}$, we show that each element can be bounded by $O_p\big(\sqrt{p}n^{-1/2}q^{-1/2}\epsilon_{nq}\big)$. By Lemma~\ref{lemma:first-order derivative} and (ii) of Lemma~\ref{lemma: l1 estimates}, for all $r\in[K]$, the $K(r)$th element can be bounded as
    \begin{align*}
    \Big|\big[\bLambda_2^\T   \Hb_{Luu^\prime}^{-1}\bS_u\big]_{[K(r)]}\Big|&=\Big|\frac{1}{n}\sum_{i=1}^nU_{ir}^0 (\one_r^{(K)})^\T \Big[\sum_{t=1}^ql^{\prime\prime}_{it}(w_{it}^0)\bgamma_t^0(\bgamma_t^0)^\T \Big]^{-1}\Big(\sum_{j=1}^ql^{\prime}_{ij}(w_{ij}^0)\bgamma_j^0\Big)\Big|\\&\le \Big|\frac{1}{n}\sum_{i=1}^n\sum_{j=1}^ql^{\prime}_{ij}(w_{ij}^0)\|\bgamma_j^0\||U_{ir}^0|\Big\{b_U\lambda_{\min}\big[\sum_{t=1}^q\bgamma_t^0(\bgamma_t^0)^\T \big]\Big\}^{-1}\\&= O_p\Big(\sqrt{\frac{1}{nq}}\epsilon_{nq}\Big),\end{align*}
    and similarly for all $h\neq l\in[K]$, the $K(h,l)$th element can be bounded as
    \begin{align*}
    &\,\Big|\big[\bLambda_2^\T   \Hb_{Luu^\prime}^{-1}\bS_u\big]_{[K(l,h)]}\Big|\\&=\Big|\frac{1}{n}\sum_{i=1}^n\big(U_{ih}^0 \one_l^{(K)}+U_{il}^0\one_{h}^{(K)}\big)^\T \Big[\sum_{t=1}^ql^{\prime\prime}_{it}(w_{it}^0)\bgamma_t^0(\bgamma_t^0)^\T \Big]^{-1}\Big(\sum_{j=1}^ql^{\prime}_{ij}(w_{ij}^0)\bgamma_j^0\Big)\Big|\\&\le \Big|\frac{1}{n}\sum_{i=1}^n\sum_{j=1}^ql^{\prime}_{ij}(w_{ij}^0)\|\bgamma_j^0\|2|U_{il}^0+U_{ih}^0|\Big\{b_U\lambda_{\min}\big[\sum_{t=1}^q\bgamma_t^0(\bgamma_t^0)^\T \big]\Big\}^{-1}\\&= O_p\Big(\sqrt{\frac{1}{nq}}\epsilon_{nq}\Big),\end{align*}
    for all $r\in[K],s\in[p]$, the $\{K^2+(r-1)p+s\}$th element can be bounded as
    \begin{align*}
    &\,\big[\bLambda_2^\T   \Hb_{Luu^\prime}^{-1}\bS_u\big]_{[K^2+(r-1)p+s]}\\&=\frac{1}{n}\sum_{i=1}^nX_{is} (\one_r^{(K)})^\T \Big[\sum_{t=1}^ql^{\prime\prime}_{it}(w_{it}^0)\bgamma_t^0(\bgamma_t^0)^\T \Big]^{-1}\Big(\sum_{j=1}^ql^{\prime}_{ij}(w_{ij}^0)\bgamma_j^0\Big)\\&\le \Big|\frac{1}{n}\sum_{i=1}^n\sum_{j=1}^ql^{\prime}_{ij}(w_{ij}^0)\|\bgamma_j^0\||X_{is}|\Big\{b_U\lambda_{\min}\big[\sum_{t=1}^q\bgamma_t^0(\bgamma_t^0)^\T \big]\Big\}^{-1}\\&= O_p\Big(\sqrt{\frac{1}{nq}}\epsilon_{nq}\Big),\;1\le r\le K\end{align*}
Together we obtain (iv). For (v) and (vi), it suffices to show that each term on the right side of the following
\begin{equation*}
    \|[\cH_{fu^{\prime}}  \cH_{uu^\prime}^{-1}\bS_u]_{[P_j]}\|\le \big\|[\Hb_{Lfu^{\prime}}\cH_{uu}^{-1}\bS_u]_{[P_j,]}\big\|+\big\|[\Hb_{Rfu^{\prime}}\cH_{uu}^{-1}\bS_u]_{[P_j,]}\big\|+\big\|[\bLambda_1\bLambda_2^\T\cH_{uu}^{-1}\bS_u]_{[P_j,]}\big\|
\end{equation*}
 can be bounded by $O_p\big(pn^{-1/2}q^{-3/2}(nq)^{1\xi}\epsilon_{nq}\big)$, which are followed by the technique of proving (ii) and (iii), (ii), (iv), (viii) of Lemma~\ref{lemma: l1 estimates}.
\end{proof}
\subsection{Proof of Lemma~\ref{lemma:asym_prop_res}}
\begin{lemma}\label{lemma:asym_prop_res}
Under Assumptions~\ref{assumption: psd covariance}-\ref{assumption:asymptotic normality}, the residuals defined in the proof of Lemma~\ref{thm:asymptotic normality} can be bounded as follows:
    \begin{equation}
     \big\|\big[\Rb_f\big]_{[P_j]}\big\|=O_p\Big(\frac{p}{q\zeta_{nq,p}^2}\Big),\quad
\big\|\Rb_f\big\|=O_p\Big(\frac{p}{\sqrt q\zeta_{nq,p}^2}\Big).\label{estimateRt}
\end{equation}
\begin{equation}
    \big\|\big[\Rb_u\big]_{[K_i]}\big\|=O_p\Big(\frac{p}{n\zeta_{nq,p}^2}\Big),\quad
\big\|\Rb_u\big\|=O_p\Big(\frac{p}{\sqrt n\zeta_{nq,p}^2}\Big).\label{estimateRt2}
\end{equation}
\end{lemma}
\begin{proof}
Recall we have assumed that the individual log-likelihood function is three times differentiable, i.e., $l^{\prime\prime\prime}_{ij}$ is continuous. In the following we denote $\delta \upsilon=\hat\upsilon-\upsilon^0$; here $\upsilon$ can be ${\bphi}, \Ub,\bGamma,\bbeta$. Then derivative with respect to $\bbf_{jr}$ of $\delta\bphi^\T\Hb(\bphi^\flat)\delta\bphi$ is given, with $\bphi^\flat$ on the segment between $\hat\bphi$ and $\bphi^0$, as
\begin{subequations}
\begin{align}
    \delta{\bphi_v}^\T \partial_{ \gamma_{jr}}\Hb(\bphi^\flat)\delta{\bphi_v}=&-(nq)^{-1}\Big\{\delta\bbf_j^\T\big(\sum_{i=1}^n  U_{ir}^\flat l^{\prime\prime\prime}_{ij}(w_{ij}^\flat)\bZ_i^\flat(\bZ_i^\flat)^\T\big)\delta\bbf_j\label{B71a}\\& +\sum_{i=1}^n\Big[U_{ir}^\flat l^{\prime\prime\prime}_{ij}(w_{ij}^\flat)\big(\delta\bU_i^\T  \bgamma_j^\flat\big)\big(\delta\bbf_j^\T  \bZ_i^\flat\big)+ l^{\prime\prime}_{ij}(w_{ij}^\flat) \big(\delta\bbf_j^\T\bZ_i^\flat\big)\delta U_{ir} \Big]\label{B71b}\\&+
    \sum_{i=1}^n\Big[U_{ir}^\flat l_{ij}^{\prime\prime\prime}(w_{ij}^\flat)\big(\delta\bU_i^\T\bgamma_j^\flat\big)^2+2l_{ij}^{\prime\prime}(w_{ij}^\flat)\big(\delta\bU_i^\T\bgamma_j^\flat\big)\delta U_{ir}\Big]\Big\},\label{B71c}
\end{align}\end{subequations}
for $r\in[K]$, and 
\begin{subequations}
\begin{align}
    \delta{\bphi_v}^\T \partial_{ \beta_{js}}\Hb(\bphi^\flat)\delta{\bphi_v}=&-(nq)^{-1}\Big\{\delta\bbf_j^\T \big(\sum_{i=1}^n  X_{is} l^{\prime\prime\prime}_{ij}(w_{ij}^\flat)\big(\bZ_i^\flat(\bZ_i^\flat)^\T\big)\delta\bbf_j\label{B71aa}\\& +\sum_{i=1}^n\Big[X_{is} l^{\prime\prime\prime}_{ij}(w_{ij}^\flat)\big(\delta\bU_i^\T  \bgamma_j^\flat\big)\big(\delta\bbf_j^\T  \bZ_i^\flat\big)\Big]\label{B71bb}\\&+
    \sum_{i=1}^n\Big[X_{is}l_{ij}^{\prime\prime\prime}(w_{ij}^\flat)\big(\delta\bU_i^\T\bgamma_j^\flat\big)^2\Big]\Big\},\label{B71cc}
\end{align}\end{subequations}
for $r=s+K$ with $s\in[p]$.
We first give an estimate for \eqref{B71a} and \eqref{B71aa}. Note that for any $r\in[K+p]$, $j\in[q]$, similar to estimating \eqref{eq_minrhominff},
\begin{align*}
    &\,\delta\bbf_j^\T \big(\sum_{i=1}^n  \bZ_{ir}^\flat l^{\prime\prime\prime}_{ij}(w_{ij}^\flat)\big(\bZ_i^\flat(\bZ_i^\flat)^\T\big)\delta\bbf_j\\\le& C\Big[\delta\bbf_j^\T\sum_{i=1}^n\bZ_i^0(\bZ_i^0)^\T\delta\bbf_j+\delta\bbf_j^\T\max_{s}\big(\bZ_i(s)\bZ_i(s)^\T-\bZ_i^0(\bZ_i^0)^\T\big)\delta\bbf_j\Big]
    \\\le &C\Big[n\big\|\delta\bbf_j\big\|^2+\frac{np}{\sqrt{\delta_{nq}}}\big\|\delta\bbf_j\big\|^2\Big],
\end{align*}
with ${\bphi}^\flat\in\cB(D)\cap\{\bphi:\sqrt{p}\|\Db_q^{-1/2}(\bphi-\bphi^0)\le m\}$. Therefore by the results from Lemma~\ref{prop:average consistency} and Lemma~\ref{prop:ind consistency}, we have
\begin{align*}
    \big|\eqref{B71a}\big|=& O_p\big(q^{-1}\|\delta\bbf_j\|^2\big)= O_p\Big(\frac{1}{q\zeta_{nq,p}^2}\Big) \\\big|\eqref{B71b}\big|=&O_p\big((nq)^{-1}\sqrt{p}\|\delta\bbf_j\|\sum\nolimits_{i=1}^n\|\delta\bU_i\|\big)= O_p\Big(\frac{\sqrt p}{q\zeta_{nq,p}^2}\Big)\\\big|\eqref{B71c}\big|=&O_p\big((nq)^{-1}\sum\nolimits_{i=1}^n\|\delta\bU_i\|\big)= O_p\Big(\frac{1}{q\zeta_{nq,p}^2}\Big)\\\big|\eqref{B71aa}\big|=& O_p\big(q^{-1}\|\delta\bbf_j\|^2\big)= O_p\Big(\frac{1}{q\zeta_{nq,p}^2}\Big) \\\big|\eqref{B71bb}\big|=&O_p\big((nq)^{-1}{\sqrt p}\|\delta\bbf_j\|\sum\nolimits_{i=1}^n\|\delta\bU_i\|\big)= O_p\Big(\frac{\sqrt p}{q\zeta_{nq,p}^2}\Big)\\\big|\eqref{B71cc}\big|=&O_p\big((nq)^{-1}\sum\nolimits_{i=1}^n\|\delta\bU_i\|\big)= O_p\Big(\frac{1}{q\zeta_{nq,p}^2}\Big).
\end{align*}
Combining these estimates gives $\|\delta{\bphi}^\T  \partial_{f_j}\Hb_L(  {\bphi}^\flat)\delta{\bphi}\|=O_p\big(pq^{-1}\zeta_{nq,p}^{-2}\big)$.
Next we compute $\partial_{ \gamma_{jr}}\Hb_P$ from \eqref{eq:HP}: 
\begin{align*}
\partial_{ \gamma_{jr}}\Hb_P=&\Db_{q}^{-1}\Big[c\sum_{r=1}^K\partial_{ \gamma_{jr}}\big(\bnu_r\bnu_r^\T \big)+c\sum_{l<h}\partial_{\gamma_{jr}}\big(\bu_{hl}\bu_{hl}^\T \big) \Big]\Db_{q}^{-1}\\&\frac{c}{2}\big(q^{-1}\sum_{j=1}^q\partial_{ \gamma_{jr}}( \gamma_{jr}^2))\Db_{q}^{-1}\begin{bmatrix}
    \Ib_q\otimes\Eb_{rr}^{(K+p)}\\&\Ib_n\otimes\Eb_{rr}^{(K)}
\end{bmatrix}\\&c\sum_{l\neq r}\Big[\partial_{\gamma_{jr}}\big(\sum_{j=1}^qq^{-1}\bgamma_{jr}\bgamma_{jl}\big)
\Db_{q}^{-1}\begin{bmatrix}\Ib_q\otimes(\Eb^{(K+p)}_{rl}+\Eb^{(K+p)}_{lr})\\&\zero_{nK\times nK}\end{bmatrix}\Big]\Db_q^{-1}\\
=&c\Db_{q}^{-1}\Big[\partial_{ \gamma_{jr}}\big(\bnu_r\bnu_r^\T \big)+\partial_{ \gamma_{jr}}\big(\sum_{l<r}\bu_{rl}\bu_{rl}^\T +\sum_{r<l}\bu_{lr}\bu_{lr}^\T \big) \Big]\Db_{q}^{-1}\\&\frac{c}{q}\gamma_{jr}\Db_{q}^{-1}\begin{bmatrix}
    \Ib_q\otimes\Eb_{rr}^{(K+p)}\\&\Ib_n\otimes\Eb_{rr}^{(K)}
\end{bmatrix}\Big]\\&cq^{-1}\Db_q^{-1}\begin{bmatrix}\big(\sum_{l\neq r} \gamma_{jl}\big)\Ib_q\otimes (\Eb^{(K+p)}_{rl}+\Eb^{(K+p)}_{lr})\\&\zero_{nK\times nK}\end{bmatrix}\Db_{q}^{-1}.
\end{align*}
Then $\delta\bphi^\T  \partial_{ \gamma_{jr}}H_P(\bphi^\flat)\delta\bphi $ can be bounded as
\begin{align*}
    \big|\delta\bphi^\T  \partial_{ \gamma_{jr}}H_P(\bphi^\flat)\delta\bphi\big|\le&\frac{2c}{q}|\delta \gamma_{jr}|\Big|\frac{1}{q}\sum_{j=1}^q \gamma_{jr}^\flat\delta \gamma_{jr}+\frac{1}{n}\sum_{i=1}^nU_{ir}^\flat\delta U_{ir}\Big|\\&+\frac{2c}{q}\sum_{ l\neq r}|\delta \gamma_{jl}|\Big[\Big|\frac{1}{q}\sum_{j=1}^q( \gamma_{jl}^\flat\delta \gamma_{jr}+ \gamma_{jr}^\flat\delta \gamma_{jl})\Big|+\Big|\frac{1}{n}\sum_{i=1}^n(U_{il}^\flat\delta U_{ir}+U_{ir}^\flat\delta U_{il})\Big|\Big]\\&+\frac{c}{q}| \gamma_{jr}^\flat|\Big|\frac{1}{n}\sum_{i=1}^n\delta U_{ir}^2+\frac{1}{q}\sum_{j=1}^q\delta \gamma_{jr}^2\Big|\\&+2\frac{1}{q^2}\sum_{r\neq l}| \gamma_{jl}^\flat|\Big|\frac{1}{q}\sum_{j=1}^q\delta \gamma_{jh}\delta \gamma_{jl}\Big|,\end{align*}
    which implies that 
    $
        \big|\delta\bphi^\T  \partial_{ \gamma_{jr}}\Hb_P(\bphi^\flat)\delta\bphi\big|=O_p\big({q^{-1}\zeta_{nq,p}^{-2}}\big)$ using \eqref{BB17}-\eqref{BB20} and ${\bphi}^\flat\in\cB(D)\cap\{\bphi:\sqrt{p}\|\Db_q^{-1/2}(\bphi-\bphi^0)\le m\}$. Therefore with $ \big|\delta\bphi^\T  \partial_{ \beta_{js}}\Hb_P(\bphi^\flat)\delta\bphi\big|=0$, we know $\big\|[\Rb_f]_{[P_j]}\big\|=O_p\big(pq^{-1}\zeta_{nq,p}^{-2}\big)$ and consequently $\big\|\Rb_f\big\|=O_p\big(pq^{-1/2}\zeta_{nq,p}^{-2}\big)$.

For $\Rb_U$, Then derivative with respect to $\bU_i$ of $\delta\bphi_v\Hb(\bphi^\flat)\delta\bphi_v$ is given as
\begin{subequations}
\begin{align}
    \delta{\bphi_v}^\T \partial_{ U_{ir}}\Hb(\bphi^\flat)\delta{\bphi_v}=&-(nq)^{-1}\Big\{\sum_{j=1}^q  \gamma_{jr}^\flat l^{\prime\prime\prime}_{ij}(w_{ij}^\flat)\big(\delta\bbf_j^\T \bZ_i^\flat\big)^2+2l_{ij}^{\prime\prime}(w_{ij}^{\flat})\delta\gamma_{jr}\big(\delta\bbf_j^\T\bZ_i^{\flat}\big)\label{B71aaa}\\& +\sum_{j=1}^q\Big[\gamma_{jr}^\flat l^{\prime\prime\prime}_{ij}(w_{ij}^\flat)\big(\delta\bU_i^\T  \bgamma_j^\flat\big)\big(\delta\bbf_j^\T  \bZ_i^\flat\big)+ l^{\prime\prime}_{ij}(w_{ij}^\flat) \big(\delta\bU_i^\T\bgamma_j^\flat\big)\delta \gamma_{jr} \Big]\label{B71bbb}\\&+
    \sum_{j=1}^q\gamma_{jr}^\flat l_{ij}^{\prime\prime\prime}(w_{ij}^\flat)\big(\delta\bU_i^\T\bgamma_j^\flat\big)^2\Big\},\label{B71ccc}
\end{align}\end{subequations}
Then similar with the estimates \eqref{B71a}-\eqref{B71c}, we have
\begin{align*}
    \big|\eqref{B71aaa}\big|=& O_p\big(p(nq)^{-1}\sum_{j=1}^q\|\delta\bbf_j\|^2+\sqrt{p}(nq)^{-1}\sum_{j=1}^q\|\delta \bbf_j\|\|\delta\bgamma_j\|\big)= O_p\Big(\frac{p}{n\zeta_{nq,p}^2}\Big) \\\big|\eqref{B71bbb}\big|=&O_p\big((nq)^{-1}(\sqrt{p}(nq)^{-1}\|\delta\bU_i\|\sum_{j=1}^q\|\delta\bbf_j\|+(nq)^{-1}\|\delta\bU_i\|\sum_{j=1}^q\|\delta\bgamma_j\|\big)= O_p\Big(\frac{p}{n\zeta_{nq,p}^2}\Big)\\\big|\eqref{B71ccc}\big|=&O_p\big(n^{-1}\|\delta\bU_i\|^2\big)= O_p\Big(\frac{p}{n\zeta_{nq,p}^2}\Big).\end{align*}
    And for $\partial_{ U_{ir}}\Hb_P$, similar with $\partial_{ \gamma_{jr}}\Hb_P$, we compute
\begin{align*}
\partial_{ U_{ir}}\Hb_P=&-\Db_{q}^{-1}\Big[c\sum_{r=1}^K\partial_{ U_{ir}}\big(\bnu_r\bnu_r^\T \big)+c\sum_{h<l}\partial_{U_{ir}}\big(\bu_{hl}\bu_{hl}^\T \big) \Big]\Db_{q}^{-1}\\&-\frac{c}{2}\big(n^{-1}\sum_{i=1}^n\partial_{ U_{ir}}( U_{ir}^2))\Db_{q}^{-1}\begin{bmatrix}
    \Ib_q\otimes\Eb_{rr}^{(K+p)}\\&\Ib_n\otimes\Eb_{rr}^{(K)}
\end{bmatrix}\\&-c\sum_{l\neq r}\Big[\partial_{U_{ir}}\big(\sum_{i=1}^nn^{-1}U_{ir}U_{il}\big)
\Db_{q}^{-1}\begin{bmatrix}\Ib_q\otimes(\Eb^{(K+p)}_{rl}+\Eb^{(K+p)}_{lr})\\&\zero_{nK\times nK}\end{bmatrix}\Big]\Db_q^{-1}\\
=&-c\Db_{q}^{-1}\Big[\partial_{ U_{ir}}\big(\bnu_r\bnu_r^\T \big)+\partial_{ U_{ir}}\big(\sum_{r<l}\bu_{rl}\bu_{rl}^\T +\sum_{l<r}\bu_{lr}\bu_{lr}^\T \big) \Big]\Db_{q}^{-1}\\&-\frac{c}{q}U_{ir}\Db_{q}^{-1}\begin{bmatrix}
    \Ib_q\otimes\Eb_{rr}^{(K+p)}\\&\Ib_n\otimes\Eb_{rr}^{(K)}
\end{bmatrix}\Big]\\&-cq^{-1}\Db_q^{-1}\begin{bmatrix}\zero_{q(K+p)\times q(K+p)}\\&\big(\sum_{l\neq r} U_{il}\big)\Ib_n\otimes (\Eb^{(K)}_{rl}+\Eb^{(K)}_{lr})\end{bmatrix}\Db_{q}^{-1}.
\end{align*}
Then similarly $\delta\bphi^\T  \partial_{ U_{ir}}\Hb_P(\bphi^\flat)\delta\bphi$ can be bounded as follows
\begin{align*}
    \big|\delta\bphi^\T  \partial_{ U_{ir}}\Hb_P(\bphi^\flat)\delta\bphi\big|\le&\frac{2c}{n}|\delta U_{ir}|\Big|\frac{1}{q}\sum_{i=1}^n \gamma_{jr}^\flat\delta \gamma_{jr}+\frac{1}{n}\sum_{i=1}^nU_{ir}^\flat\delta \bU_{ir}\Big|\\&+\frac{2c}{n}\sum_{ l\neq r}|\delta U_{il}|\Big[\Big|\frac{1}{q}\sum_{i=1}^n( \gamma_{jl}^\flat\delta U_{ir}+ U_{ir}^\flat\delta \gamma_{jl})\Big|\\
    &+\Big|\frac{1}{n}\sum_{i=1}^n(U_{il}^\flat\delta U_{ir}+U_{ir}^\flat\delta U_{il})\Big|\Big]\\&+\frac{c}{n}| U_{ir}^\flat|\Big|\frac{1}{q}\sum_{j=1}^q\delta \gamma_{jr}^2+\frac{1}{n}\sum_{i=1}^n\delta U_{ir}^2\Big|\\&+2\frac{1}{n^2}\sum_{r\neq l}| \gamma_{jl}^\flat|\Big|\frac{1}{n}\sum_{i=1}^n\delta U_{ih}\delta U_{il}\Big|,\end{align*}
    which implies that $
        \big|\delta\bphi^\T  \partial_{ U_{ir}}\Hb_P(\bphi^\flat)\delta\bphi\big|= O_p\big({n^{-1}\zeta_{nq,p}^{-2}}\big)$ using \eqref{BB17}-\eqref{BB20} with ${\bphi}^\flat\in\cB(D)\cap\{\bphi:\sqrt{p}\|\Db_q^{-1/2}(\bphi-\bphi^0)\le m\}$.
Then 
\begin{equation}
    \big\|\big[\Rb_U\big]_{[K_i]}\big\|=O_p\Big(\frac{p}{n\zeta_{nq,p}^2}\Big),\quad
\big\|\Rb_U\big\|=O_p\Big(\frac{p}{\sqrt{n}\zeta_{nq,p}^2}\Big).\label{estimateRa}
\end{equation}
    \end{proof}
    
\subsection{Proof of Lemma~\ref{lemma:asym_estimate}}
\begin{lemma}\label{lemma:asym_estimate}
    Under Assumption~\ref{assumption: psd covariance}--\ref{assumption:asymptotic normality}, we have 
    
   \noindent(\romannumeral1) \begin{align} \big\|[\cH ^{-1}(\bphi^0)\bS(\bphi^0)]_{[P_j]}&-[ \Hb_{Lff^\prime}^{-1}(\bphi^0)\bS_f(\bphi^0)]_{[P_j]}\big\|=O_p\Big(\frac{p(nq)^{3/\xi}}{\sqrt{nq}}\epsilon_{nq}\Big);\label{eq:prop5_eq11}\\\big\|[\cH ^{-1}(\bphi^0)\Rb]_{[P_j]}\big\|&=O_p\Big(\frac{p(nq)^{3/\xi}}{\zeta_{nq,p}^2}\Big).\label{eq:prop5_eq12}\end{align}
\noindent(\romannumeral2) \begin{align}\big\|[\cH ^{-1}(\bphi^0)\bS(\bphi^0)]_{[q(K+p)+K_i]}&-[\Hb_{Luu^\prime}^{-1}(\bphi^0)\bS_u(\bphi^0)]_{[K_i]}\big\|\label{eq:prop5_eq21}\\
&=O_p\Big(\frac{p^{3/2}(nq)^{3/\xi}}{\sqrt{nq}}\epsilon_{nq}\Big);\nonumber \\\big\|[\cH ^{-1}(\bphi^0)\Rb]_{[q(K+p)+K_i]}\big\|&=O_p\Big(\frac{p^{3/2}(nq)^{3/\xi}}{\zeta_{nq,p}^2}\Big) .\label{eq:prop5_eq22}\end{align}
\end{lemma}
    We write the Hessian matrix on $\bphi^0$ as
     \begin{equation*}
         \cH(\bphi^0)=\begin{bmatrix}
             \cH_{ff^{\prime}}&\cH_{fu^{\prime}}\\\cH_{uf^{\prime}}&\cH_{uu^{\prime}}
         \end{bmatrix},
     \end{equation*}
     with
     \begin{align*}
         \cH_{ff^{\prime}}=\Hb_{Lff^{\prime}}+\Hb_{Rff^{\prime}}+\bLambda_1\bLambda_1^\T,\;\cH_{fu^{\prime}}=\Hb_{Lfu^{\prime}}+\bLambda_1\bLambda_2^\T,\\\cH_{uf^{\prime}}=\Hb_{Luf^{\prime}}+\bLambda_2\bLambda_1^\T,\;\cH_{uu^{\prime}}=\Hb_{Luu^{\prime}}+\Hb_{Ruu^{\prime}}+\bLambda_2\bLambda_2^\T.
     \end{align*}
     We omit the Hessian matrix's dependence on $\bphi^0$ because in this proof we only need estimates of the Hessian matrix on $\bphi^0$.
     By Sherman–Morrison–Woodbury formula, we get the exact expression for the blocks in \begin{equation*}\cH(\bphi^0)^{-1}=\begin{bmatrix}
             \tilde\cH_{ff^{\prime}}&\tilde\cH_{fu^{\prime}}\\\tilde\cH_{uf^{\prime}}&\tilde\cH_{uu^{\prime}}\end{bmatrix},\end{equation*} as $\tilde\cH_{ff^{\prime}}=\cH_{ff^{\prime}}^{-1}+\cH_{ff^{\prime}}^{-1}\cH_{fu^{\prime}}\cH_{-u}\cH_{uf^{\prime}}\cH_{ff^{\prime}}^{-1}$, $\tilde\cH_{fu^{\prime}}=\tilde\cH_{ff^{\prime}}\cH_{fu^{\prime}}\cH_{uu^{\prime}}^{-1}$, $\tilde\cH_{uf^{\prime}}=\tilde\cH_{uu^{\prime}}\cH_{uf^{\prime}}\cH_{ff^{\prime}}^{-1}$ and $\tilde\cH_{uu^{\prime}}=\cH_{uu^{\prime}}^{-1}+\cH_{uu^{\prime}}^{-1}\cH_{uf^{\prime}}\cH_{-f}\cH_{fu^{\prime}}\cH_{uu^{\prime}}^{-1}$ and . Here we define $\cH_{-f}=(\cH_{ff^{\prime}}-\cH_{fu^{\prime}}\cH_{uu^{\prime}}^{-1}\cH_{uf^{\prime}})^{-1}$ and $\cH_{-u}=(\cH_{uu^{\prime}}-\cH_{uf^{\prime}}\cH_{ff^{\prime}}^{-1}\cH_{fu^{\prime}})^{-1}$. By Lemma \ref{prop:min eigenvalue} we have $l_2$ estimates for $\cH_{-f}$ and $\cH_{-u}$ as $$\|\cH_{-f}\|=O_p(q),\;\|\cH_{-u}\|=O_p(n)$$
     Then $j$th block of the first row of $\cH^{-1}(\bphi^0)\bS(\bphi^0)$ is given as 
\begin{subequations}\begin{align}
    \big[\cH ^{-1}(\bphi^0)\bS(\bphi^0)\big]_{[P_j]}=&\big[\cH_{ff^\prime}^{-1}\bS_f\big]_{[P_j]}\label{eq_norma_hs_t1}
    \\&-\big[\cH_{ff^\prime}^{-1}\cH_{fu^{\prime}}\cH_{-u}\cH_{uf^{\prime}}  \cH_{ff^\prime}^{-1}\bS_f\big]_{[P_j]}\label{eq_norma_hs_t2}
    \\&+\Big[\cH_{ff^\prime}^{-1}\cH_{fu^{\prime}} \cH_{uu^{\prime}}^{-1}\bS_u\Big]_{[P_j]}\label{eq_norma_hs_t3}
    \\& -\Big[\cH_{ff^\prime}^{-1}\cH_{fu^{\prime}}\cH_{-u}\cH_{uf^{\prime}}   \cH_{ff^\prime}^{-1}\cH_{fu^{\prime}} \cH_{uu^{\prime}}^{-1}\bS_u\Big]_{[P_j]}.\label{eq_norma_hs_t4}
\end{align}\end{subequations}
and the $j$th block of
the first row of $\cH^{-1}(\bphi^0)\Rb$ is given as 
\begin{subequations}
\begin{align}
    \big[\cH ^{-1}(\bphi^0)\Rb\big]_{[P_j]}&=\big[\cH_{ff^\prime}^{-1}\Rb_f\big]_{[P_j]}\label{Beq_norma_hr_t1}\\&-\big[\cH_{ff^\prime}^{-1}\cH_{fu^{\prime}}\cH_{-u}\cH_{uf^{\prime}}  \cH_{ff^\prime}^{-1}\Rb_f\big]_{[P_j]}\label{Beq_norma_hr_t2}\\
&+\Big[\cH_{ff^\prime}^{-1}\cH_{fu^{\prime}} \cH_{uu^{\prime}}^{-1}\Rb_u\Big]_{[P_j]}\label{Beq_norma_hr_t3}\\&-\Big[\cH_{ff^\prime}^{-1}\cH_{fu^{\prime}}\cH_{-u}\cH_{uf^{\prime}}  \cH_{ff^\prime}^{-1}\cH_{fu^{\prime}} \cH_{uu^{\prime}}^{-1}\Rb_u\Big]_{[P_j]},\label{Beq_norma_hr_t4}\end{align}
\end{subequations}
Here all blocks in $\cH$ are taken on $\bphi^0$.
 For \eqref{eq_norma_hs_t1}, write $$\cH^{-1}_{ff^{\prime}}=\big(\Hb_{Lff^{\prime}}+\bLambda_1\bLambda_1^\T\big)^{-1}=\Hb_{Lff^{\prime}}^{-1}-\Hb_{Lff^{\prime}}^{-1}\bLambda_1\big(\Ib_{K^2}+\bLambda_1^\T\Hb_{Lff^{\prime}}\bLambda_1\big)^{-1}\bLambda_1^\T\Hb_{Lff^{\prime}}^{-1}.$$
By from (i) and (vii) of Lemma~\ref{lemma: l1 estimates} and (i) of Lemma~\ref{lemma: tech_for_normality},
\begin{align*}
\Big\|\big[&\Hb_{Lff^{\prime}}^{-1}\bLambda_1\big(\Ib_{K^2}+\bLambda_1^\T\Hb_{Lff^{\prime}}\bLambda_1\big)^{-1}\bLambda_1^\T\Hb_{Lff^{\prime}}^{-1}\bS_f\big]_{[P_j]}\Big\|\\&=\big\|[\Hb_{Lff^{\prime}}^{-1}]_{[P_j,P_j]}\big\|\big\|[\bLambda_1]_{[P_j,]}\big\|\big\|\big(\Ib_{K^2}+\bLambda_1^\T\Hb_{Lff^{\prime}}^{-1}\bLambda_1\big)^{-1}\big\|\big\|\bLambda_1^\T\Hb_{Lff^{\prime}}^{-1}\bS_f\big\|\le O_p\Big(\sqrt{\frac{p}{nq}}\epsilon_{nq}\Big),
\end{align*}
which implies $\big\|\eqref{eq_norma_hs_t1}-\big[ \Hb_{Lff^\prime}^{-1}\bS_f\big]_{[P_j]}\big\|\le O_p\big(\sqrt{p/(nq)}\big)$.
For the second term~\eqref{eq_norma_hs_t2}, by (i), (iv), (v) of Lemma~\ref{lemma: l1 estimates} and (v) of Lemma \ref{lemma: tech_for_normality}, we have 
\begin{align*}
    \eqref{eq_norma_hs_t2}\le &\big\|\big[\cH_{ff^\prime}^{-1}\cH_{fu^{\prime}}\big]_{[P_j,]}\big\|\big\|\cH_{-u}\big\|\big\|\cH_{uf^{\prime}} \cH_{ff^\prime}^{-1}\bS_f\big\|\\\le &
    O_p\Big(\frac{p(nq)^{2/\xi}}{\sqrt{nq}}\epsilon_{nq}\Big)
\end{align*}
For \eqref{eq_norma_hs_t3}, by (ii), (iii) of Lemma~\ref{lemma: l1 estimates} and (iv), (vi) of Lemma~\ref{lemma: tech_for_normality} we have
\begin{align*}
    \eqref{eq_norma_hs_t3}=&\big\|\big[\Hb_{Lff^{\prime}}^{-1}\big]_{[P_j,P_j]}\big[\cH_{fu^{\prime}} \cH_{uu^{\prime}}^{-1}\bS_u\big]_{[P_j]}\big\|\\&+\big\|\big[\Hb_{Lff^{\prime}}^{-1}\big]_{[P_j,P_j]}\big[\bLambda_1\big]_{[P_j,]}\big(\Ib_{K^2}+\bLambda_1^\T\Hb_{Lff^{\prime}}\bLambda_1\big)^{-1}\bLambda_1^\T\Hb_{Lff^{\prime}}^{-1}\cH_{fu^{\prime}} \cH_{uu^{\prime}}^{-1}\bS_u\big\|\\
    =&O_p\Big(\frac{p}{\sqrt{nq}}(nq)^{1/\xi}\epsilon_{nq}\Big).
\end{align*}
And for \eqref{eq_norma_hs_t4}, by (i), (iii), (v) of Lemma~\ref{lemma: l1 estimates} and (vi) of Lemma~\ref{lemma: tech_for_normality}, we have
\begin{align*}
    \eqref{eq_norma_hs_t4}=&\big\|\big[\cH_{ff^\prime}^{-1}\cH_{fu^{\prime}}\big]_{[P_j,]}\big\|\big\|\cH_{-u}\big\|\big\|\cH_{uf^{\prime}} \big\|\big\| \cH_{ff^\prime}^{-1}\big\|\big\|\cH_{fu^{\prime}} \cH_{uu^{\prime}}^{-1}\bS_u\big\|\\&\le O_p\Big(\frac{\sqrt p(nq)^{3/\xi}}{\sqrt{nq}}\epsilon_{nq}\Big).
\end{align*}
Combining these together, we have 
\begin{equation*}
     \big\|[\cH ^{-1}(\bphi^0)\bS(\bphi^0)]_{[P_j]}-[ \Hb_{Lff^\prime}^{-1}(\bphi^0)\bS_f(\bphi^0)]_{[P_j]}\big\|=O_p\Big(\frac{p(nq)^{3/\xi}}{\sqrt{nq}}\epsilon_{nq}\Big).
\end{equation*}

 For~\eqref{Beq_norma_hr_t1}, the first term can be given as
\begin{align*}
    \Big\|\big[\cH_{ff^{\prime}}^{-1}\Rb_f\big]_{[P_j]}\Big\|\le &\big\|\big[\Hb_{Lff^{\prime}}^{-1}\big]_{[P_j,P_j]}\big\|\big\|\big[\Rb_{f}\big]_{[P_j]}\big\|\\&+\big\|\big[\Hb_{Lff^{\prime}}^{-1}\big]_{[P_j,P_j]}\big\|\big\|[\bLambda_1]_{[P_j,]}\big\|\big\|\big(\Ib_{K^2}+\bLambda_1^\T\Hb_{Lff^{\prime}}\bLambda_1\big)^{-1}\big\|\big\|\bLambda_1^\T\big\|\big\|\Hb_{Lff^{\prime}}^{-1}\big\|\big\|\Rb_f\big\|
\end{align*}
By (i), (vii) of Lemma~\ref{lemma: l1 estimates} and Lemma~\ref{lemma:asym_estimate}, the first term in bounded by $O_p(p\zeta_{nq,p}^{-2})$ and the second term is bounded by $O_p(p\zeta_{nq,p}^{-2})$. 
For the second part~\eqref{Beq_norma_hr_t2}, by (i), (iv), (v) of Lemma~\ref{lemma: l1 estimates} and Lemma~\ref{lemma:asym_estimate} we know \begin{equation*}
    \eqref{Beq_norma_hr_t2}\le \big\|[\cH_{ff^{\prime}}^{-1}\cH_{fu^{\prime}}]_{[P_j,]}\big\|\big\|\cH_{-u}\big\|\big\|\cH_{uf^{\prime}}\big\|\big\|\cH_{ff^{\prime}}^{-1}\big\|\big\|\Rb_f\big\|=O_p\Big(\frac{p(nq)^{2/\xi}}{\zeta_{nq,p}^2}\Big)
\end{equation*}
For third part~\eqref{Beq_norma_hr_t3}, by (ii), (v) of Lemma~\ref{lemma: l1 estimates} and Lemma~\ref{lemma:asym_estimate} we have
\begin{equation*}
\eqref{Beq_norma_hr_t3}=\big\|\big[\cH_{ff^{\prime}}^{-1}\cH_{fu^{\prime}}]_{[P_j,]}\big\|\big\|\cH_{uu^{\prime}}^{-1}\big\|\big\|\Rb_{u}\big\|\le O_p\Big(\frac{p(nq)^{1/\xi}}{\zeta_{nq,p}^2}\Big)
\end{equation*}
For fourth part~\eqref{Beq_norma_hr_t4}, similar by (i), (ii), (iii), (v) of Lemma~\ref{lemma: l1 estimates} and Lemma~\ref{lemma:asym_estimate}, we have
\begin{align*}
\big\|\eqref{Beq_norma_hr_t4}\big\|\le&\big\|\big[\cH_{ff^{\prime}}^{-1}\cH_{fu^{\prime}}\big]_{[P_j,]}\big\|\big\|\cH_{-u}\big\|\big\|\cH_{uf^{\prime}} \big\|\big\| \cH_{ff^{\prime}}^{-1}\big\|\big\|\cH_{fu^{\prime}} \big\|\big\|\cH_{uu^{\prime}}^{-1}\big\|\big\|\Rb_{u}\big\|
    \\&=O_p\Big(\frac{p(nq)^{3/\xi}}{\zeta_{nq,p}^2}\Big).
\end{align*}
Then $\big[\cH ^{-1}(\bphi^0)\Rb\big]_{[P_j]}$ can be bounded by $O_p\big(p(nq)^{3/\xi}\zeta_{nq,p}^{-2}\big)$.

For the $i$th block of $\cH^{-1}(\bphi^0)\bS(\bphi^0)$, it can be written as
\begin{subequations}\begin{align}
    \big[\cH ^{-1}(\bphi^0)\bS(\bphi^0)\big]_{[K_i]}=&\big[\cH_{uu^\prime}^{-1}\bS_u\big]_{[K_i]}\label{eq_norma_hsu_t1}
    \\&-\big[\cH_{uu^\prime}^{-1}\cH_{uf^{\prime}}\cH_{-f}\cH_{fu^{\prime}}  \cH_{uu^\prime}^{-1}\bS_u\big]_{[K_i]}\label{eq_norma_hsu_t2}
    \\&+\Big[\cH_{uu^\prime}^{-1}\cH_{uf^{\prime}} \cH_{ff^{\prime}}^{-1}\bS_f\Big]_{[K_i]}\label{eq_norma_hsu_t3}
    \\& -\Big[\cH_{uu^\prime}^{-1}\cH_{uf^{\prime}}\cH_{-f}\cH_{fu^{\prime}}   \cH_{uu^\prime}^{-1}\cH_{uf^{\prime}} \cH_{ff^{\prime}}^{-1}\bS_f\Big]_{[K_i]}.\label{eq_norma_hsu_t4}
\end{align}\end{subequations}
and similarly
the $i$th block of second row of $\cH^{-1}(\bphi^0)\Rb$ is given as 
\begin{subequations}
\begin{align}
    \big[\cH ^{-1}(\bphi^0)\Rb\big]_{[K_i]}&=\big[\cH_{uu^\prime}^{-1}\Rb_u\big]_{[K_i]}\label{Beq_norma_hru_t1}\\&-\big[\cH_{uu^\prime}^{-1}\cH_{uf^{\prime}}\cH_{-f}\cH_{fu^{\prime}}  \cH_{uu^\prime}^{-1}\Rb_u\big]_{[K_i]}\label{Beq_norma_hru_t2}\\
&+\Big[\cH_{uu^\prime}^{-1}\cH_{uf^{\prime}} \cH_{ff^{\prime}}^{-1}\Rb_f\Big]_{[K_i]}\label{Beq_norma_hru_t3}\\&-\Big[\cH_{uu^\prime}^{-1}\cH_{uf^{\prime}}\cH_{-f}\cH_{fu^{\prime}}  \cH_{uu^\prime}^{-1}\cH_{uf^{\prime}} \cH_{ff^{\prime}}^{-1}\Rb_f\Big]_{[K_i]},\label{Beq_norma_hru_t4}\end{align}
\end{subequations}
Here all blocks in $\cH$ are taken on $\bphi^0$.
For \eqref{eq_norma_hsu_t1}, similar to the procedure of bounding \eqref{eq_norma_hs_t1}, by (ii) and (viii) of Lemma~\ref{lemma: l1 estimates} and (v) of Lemma~\ref{lemma: tech_for_normality} we have
\begin{align*}
\big\|\eqref{eq_norma_hsu_t1}-&\big[\Hb_{Luu^{\prime}}\big]_{[K_i,K_i]}[\bS_u]_{[K_i]}\big\|\le\Big\|\big[\Hb_{Luu^{\prime}}^{-1}\bLambda_2\big(\Ib_{K^2+Kp}+\bLambda_2^\T\Hb_{Luu^{\prime}}\bLambda_2\big)^{-1}\bLambda_2^\T\Hb_{Luu^{\prime}}^{-1}\bS_u\big]_{[K_i]}\Big\|\\&=
\big\|[\Hb_{Luu^{\prime}}^{-1}]_{[K_i,K_i]}\big\|\big\|[\bLambda_2]_{[K_i,]}\big\|\big\|\big(\Ib_{K^2+Kp}+\bLambda_2^\T\Hb_{Luu^{\prime}}^{-1}\bLambda_2\big)^{-1}\big\|\big\|\bLambda_2^\T\Hb_{Luu^{\prime}}^{-1}\bS_u\big\|\\
\\&\le O_p\Big({\frac{p}{\sqrt {nq}}}(nq)^{1/\xi}\epsilon_{nq}\Big).
\end{align*}
 For \eqref{eq_norma_hsu_t2}, by (ii), (iv), (v) of Lemma~\ref{lemma: l1 estimates} and (vi) of Lemma \ref{lemma: tech_for_normality} we have 
\begin{align*}
    \|\eqref{eq_norma_hsu_t2}\|\le\big\|\big[\cH_{uu^\prime}^{-1}\cH_{uf^{\prime}}\big]_{[K_i,]}\cH_{-f}\cH_{fu^{\prime}}  \cH_{uu^\prime}^{-1}\bS_u\big\|\le O_p\Big(\frac{p^{3/2}(nq)^{2/\xi}}{\sqrt{nq}}\epsilon_{nq}\Big)
\end{align*}
For \eqref{eq_norma_hsu_t3}, by Lemma (ii), (vi) of Lemma~\ref{lemma: tech_for_normality} and (iv), (v) of Lemma~\ref{lemma: l1 estimates} we have
\begin{align*}
    \eqref{eq_norma_hsu_t3}\le&\big\|\big[\Hb_{Luu^{\prime}}^{-1}\big]_{[K_i,K_i]}\big\|\big\|\big[\cH_{uf^{\prime}} \cH_{ff^{\prime}}^{-1}\bS_f\big]_{[K_i]}\big\|\\&+\big\|\big[\Hb_{Luu^{\prime}}^{-1}\big]_{[K_i,K_i]}\big\|\big\|\big[\bLambda_2\big]_{[K_i,]}\big\|\big\|\big(\Ib_{K^2+Kp}+\bLambda_2^\T\Hb_{Luu^{\prime}}\bLambda_2\big)^{-1}\big\|\\&\,\times \big\|\bLambda_2^\T\big\|\big\|\Hb_{Luu^{\prime}}^{-1}\big\|\big\|\cH_{uf^{\prime}} \cH_{ff^{\prime}}^{-1}\bS_f\big\|\\
    \le &O_p\Big(\frac{p^{3/2}}{\sqrt{nq}}(nq)^{1/\xi}\epsilon_{nq}\Big).
\end{align*}
Finally for \eqref{eq_norma_hsu_t4}, by (ii), (iv), (vi) of Lemma~\ref{lemma: l1 estimates} and (iii) of Lemma~\ref{lemma: tech_for_normality}, we have
\begin{align*}
    \|\eqref{eq_norma_hsu_t4}\|\le\Big\|\big[\cH_{uu^\prime}^{-1}\cH_{uf^{\prime}}\big]_{[K_i,]}\cH_{-f}\cH_{fu^{\prime}}  \cH_{uu^\prime}^{-1}\cH_{uf^{\prime}} \cH_{ff^{\prime}}^{-1}\bS_f\Big\|
    =O_p\Big(\frac{p^{3/2}(nq)^{3/\xi}}{\sqrt{nq}}\epsilon_{nq}\Big).
\end{align*}
For the first term~\eqref{Beq_norma_hru_t1}, similar with \eqref{Beq_norma_hr_t1}, by (ii), (viii) of Lemma~\ref{lemma: l1 estimates} and Lemma ~\ref{lemma:asym_estimate}, we have
\begin{align*}
\big\|\eqref{Beq_norma_hru_t1}\big\|&\le\big\|[\Hb_{Luu^{\prime}}^{-1}]_{[K_i,K_i]}\big\|\big\|[\bLambda_2]_{[K_i,]}\big\|\big\|\big(\Ib_{K^2+Kp}+\bLambda_2^\T\Hb_{Luu^{\prime}}^{-1}\bLambda_2\big)^{-1}\big\|\big\|\bLambda_2^\T\big\|\big\|\Hb_{Luu^{\prime}}^{-1}\big\|\big\|\Rb_u\big\|\\&=
O_p\Big(\frac{p^{3/2}}{\zeta_{nq,p}^2}\Big).
\end{align*}
For the second term~\eqref{Beq_norma_hru_t2}, similar with \eqref{Beq_norma_hr_t2}, by (ii), (iii), (vi) of Lemma~\ref{lemma: l1 estimates} and Lemma ~\ref{lemma:asym_estimate}, we have
\begin{align*}
\|\eqref{Beq_norma_hru_t2}\|\le&\big\|\big[\cH_{uu^\prime}^{-1}\cH_{uf^{\prime}}\big]_{[K_i,]}\cH_{-f}\cH_{fu^{\prime}}  \cH_{uu^\prime}^{-1}\Rb_u\big\|\\
    \le &\big\|\big[\cH_{uu^\prime}^{-1}\cH_{uf^{\prime}}\big]_{[K_i,]}\big\|\big\|\cH_{-f}\big\|\big\|\cH_{fu^{\prime}}  \big\|\big\|\cH_{uu^\prime}^{-1}\big\|\big\|\Rb_u\big\|\\
    =&
    O_p\Big(\frac{p^{3/2}(nq)^{2/\xi}}{\zeta_{nq,p}^2}\Big).
\end{align*}
For the third term~\eqref{Beq_norma_hru_t3}, similar with \eqref{Beq_norma_hr_t3}, by (i), (vi) of Lemma~\ref{lemma: l1 estimates} and Lemma ~\ref{lemma:asym_estimate}, we have
\begin{equation*}
\eqref{Beq_norma_hru_t3}\le\big\|\big[\cH_{uu^{\prime}}^{-1}\cH_{uf^{\prime}}\big]_{[K_i,]}\big\|\big\| \cH_{ff^{\prime}}^{-1}\big\|\big\|\Rb_f\big\|
    =O_p\Big(\frac{p^{3/2}(nq)^{1/\xi}}{\zeta_{nq,p}^2}\Big).
\end{equation*}
For the fourth term in~\eqref{Beq_norma_hru_t4}, similar with \eqref{Beq_norma_hr_t4}, by (i), (ii), (iii), (vi) of Lemma~\ref{lemma: l1 estimates} and Lemma~\ref{lemma:asym_estimate}, we have
\begin{align*}
\|\eqref{Beq_norma_hru_t4}\|\le&\Big\|\big[\cH_{uu^\prime}^{-1}\cH_{uf^{\prime}}\big]_{[K_i,]}\big\|\big\|\cH_{-f}\big\|\big\|\cH_{fu^{\prime}}  \big\|\big\|\cH_{uu^\prime}^{-1}\big\|\big\|\cH_{uf^{\prime}} \big\|\big\|\cH_{ff^{\prime}}^{-1}\big\|\big\|\Rb_f\big\|
   \\
   =&O_p\Big(\frac{p^{3/2}(nq)^{3/\xi}}{\zeta_{nq,p}^2 }\Big).
\end{align*}

\section{Additional Simulation Results}\label{sec:a_simula}

\subsection{Additional Results at $p_* = 10$}
In Section~5 of the main text, we present simulation results at $p_* = 5$ and $p_* = 30$, across all considered settings. This section provides additional simulation results at $p_* = 10$. Readers are directed to Section~5 for comprehensive details on the considered settings, data generation processes, and computation of averaged powers and type I errors.
The powers and type I errors are presented in Figure~\ref{fig:non-anchor-correlated p10} and Figure~\ref{fig:non-anchor-correlated p10 dense} for sparse and dense settings, respectively.
Comparing Figures~\ref{fig:non-anchor-correlated p10}--\ref{fig:non-anchor-correlated p10 dense} in this section with Figures 3--4 at $p_*=5$ and Figures 5--6 at $p_* = 30$ in Section~5 in the main text, we observe similar patterns in simulation results, indicating the robustness of our proposed method across varying covariate dimensions $p_*$.

\begin{figure}[h!]
\centering    
   \subfigure{
        \includegraphics[width=1.9in]{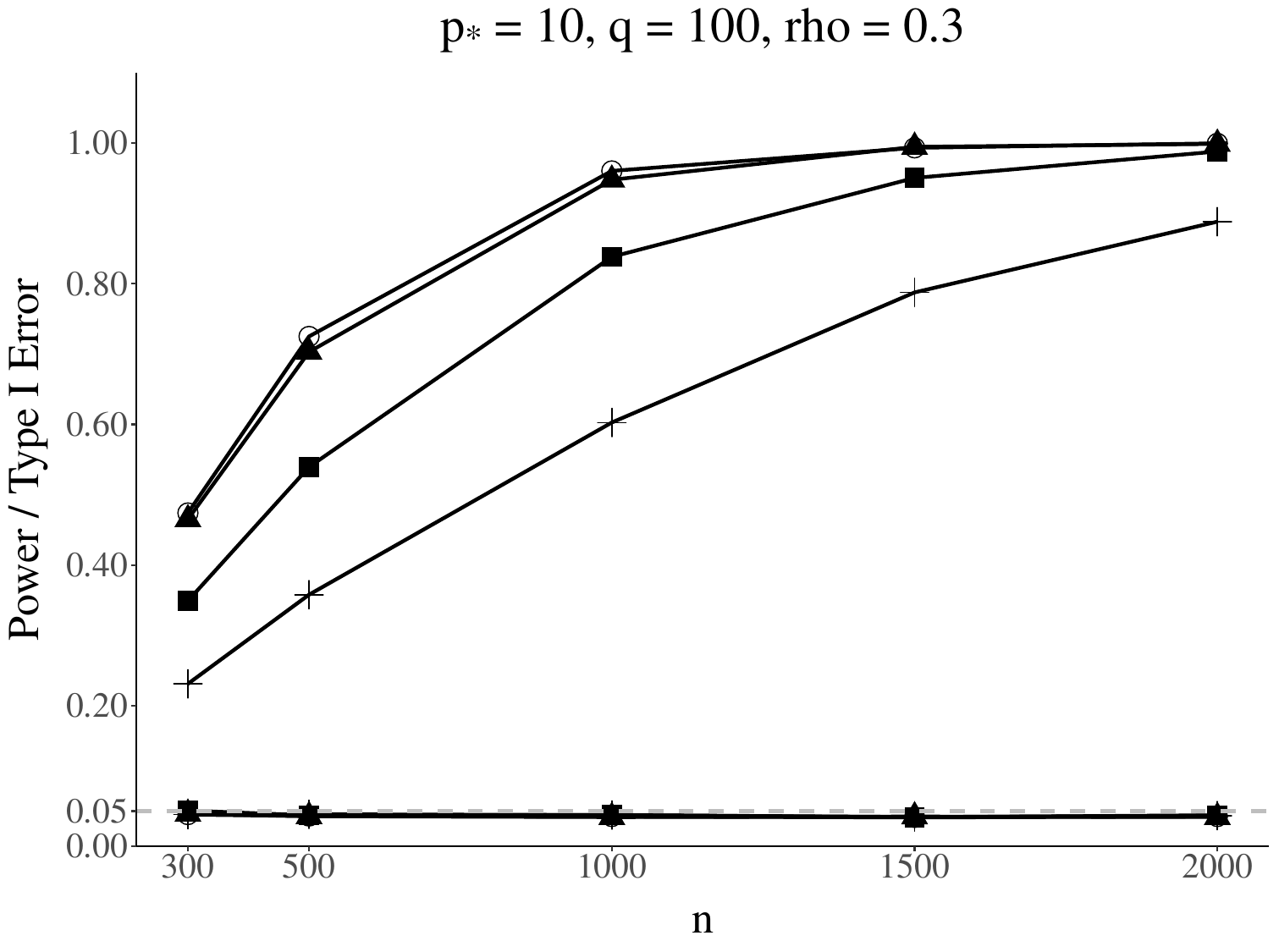}}%
         \subfigure{
        \includegraphics[width=1.9in]{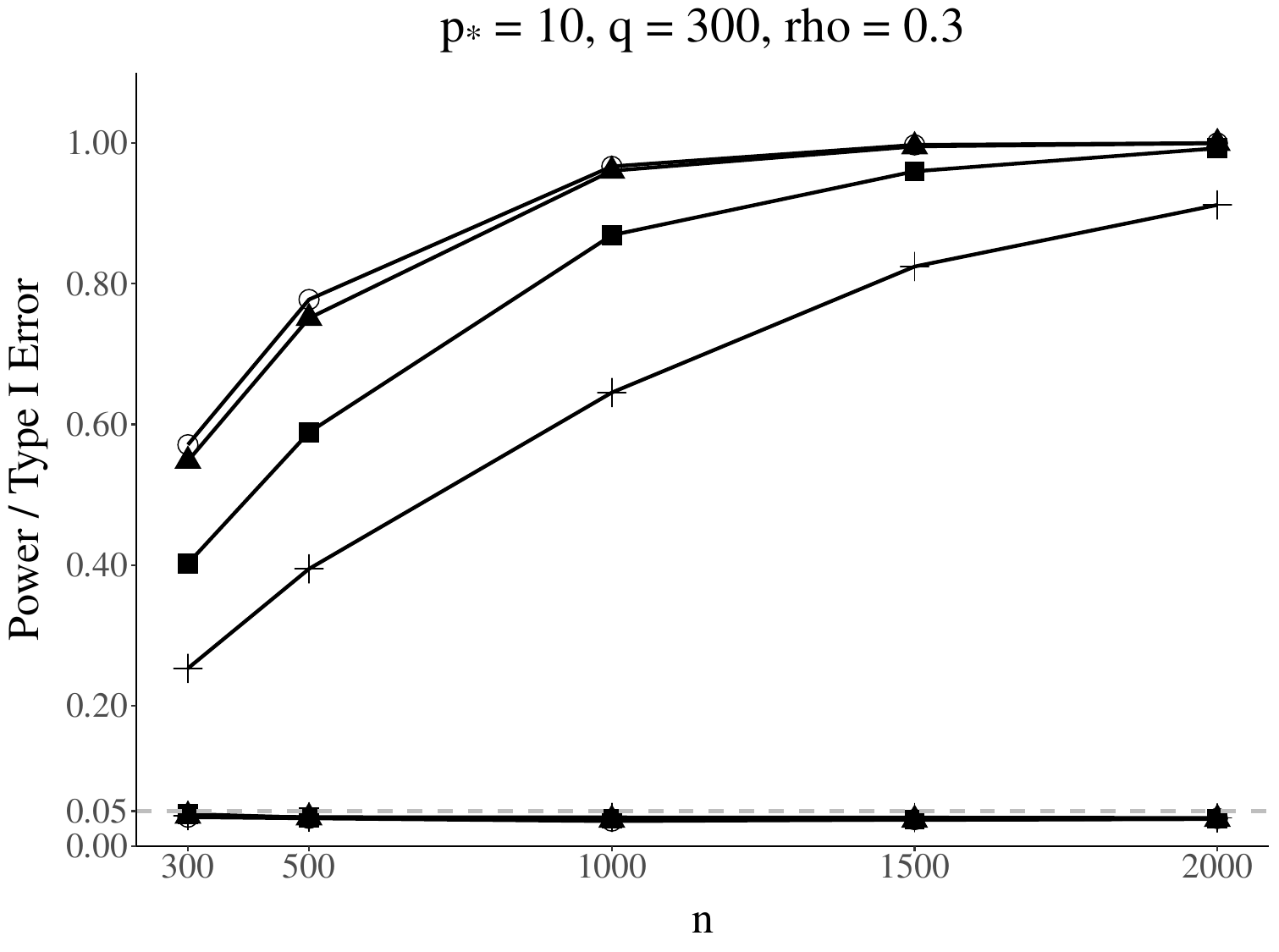}}%
        \subfigure{
        \includegraphics[width=1.9in]{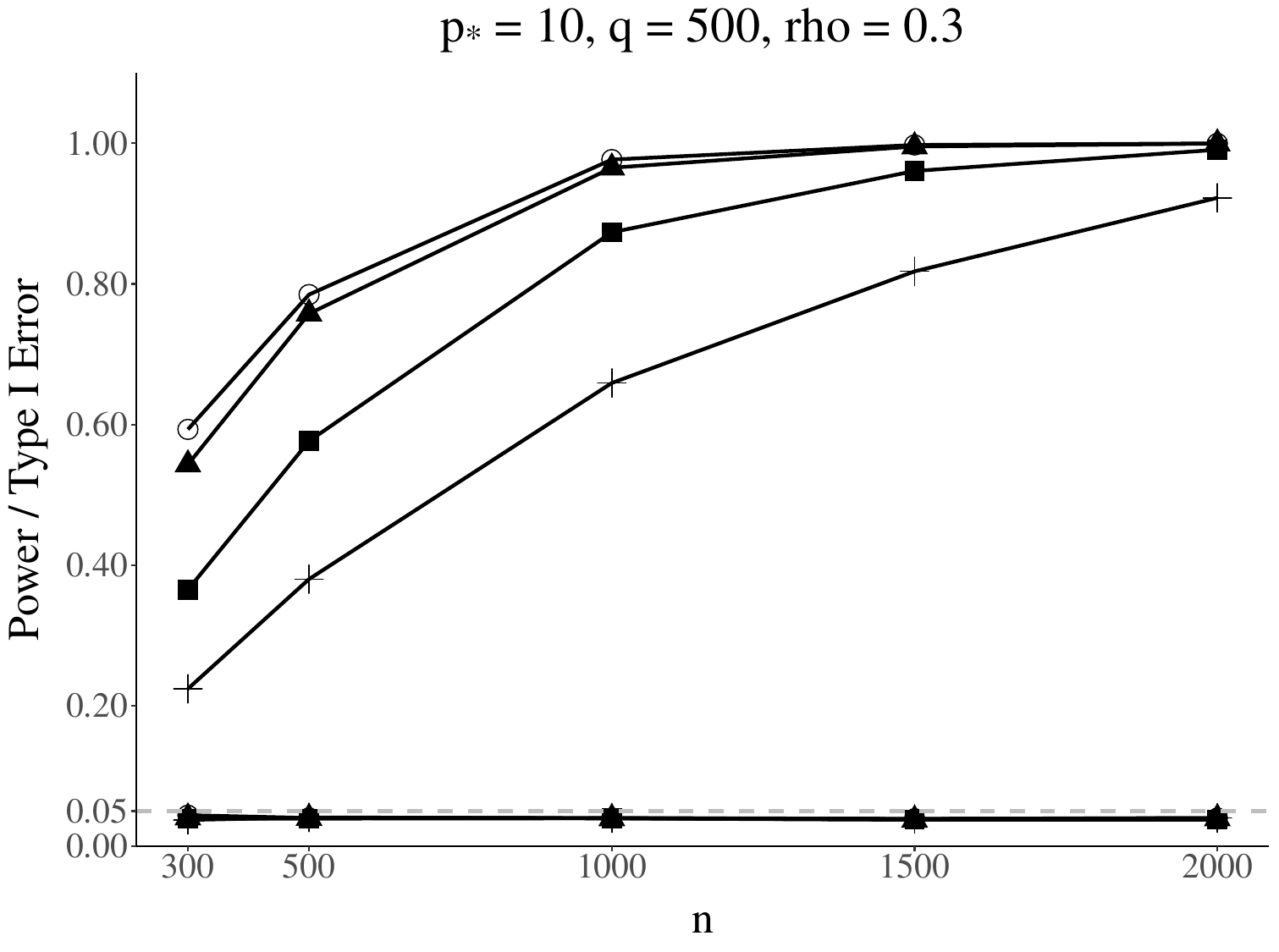}}%
        \\
          \subfigure{
        \includegraphics[width=1.9in]{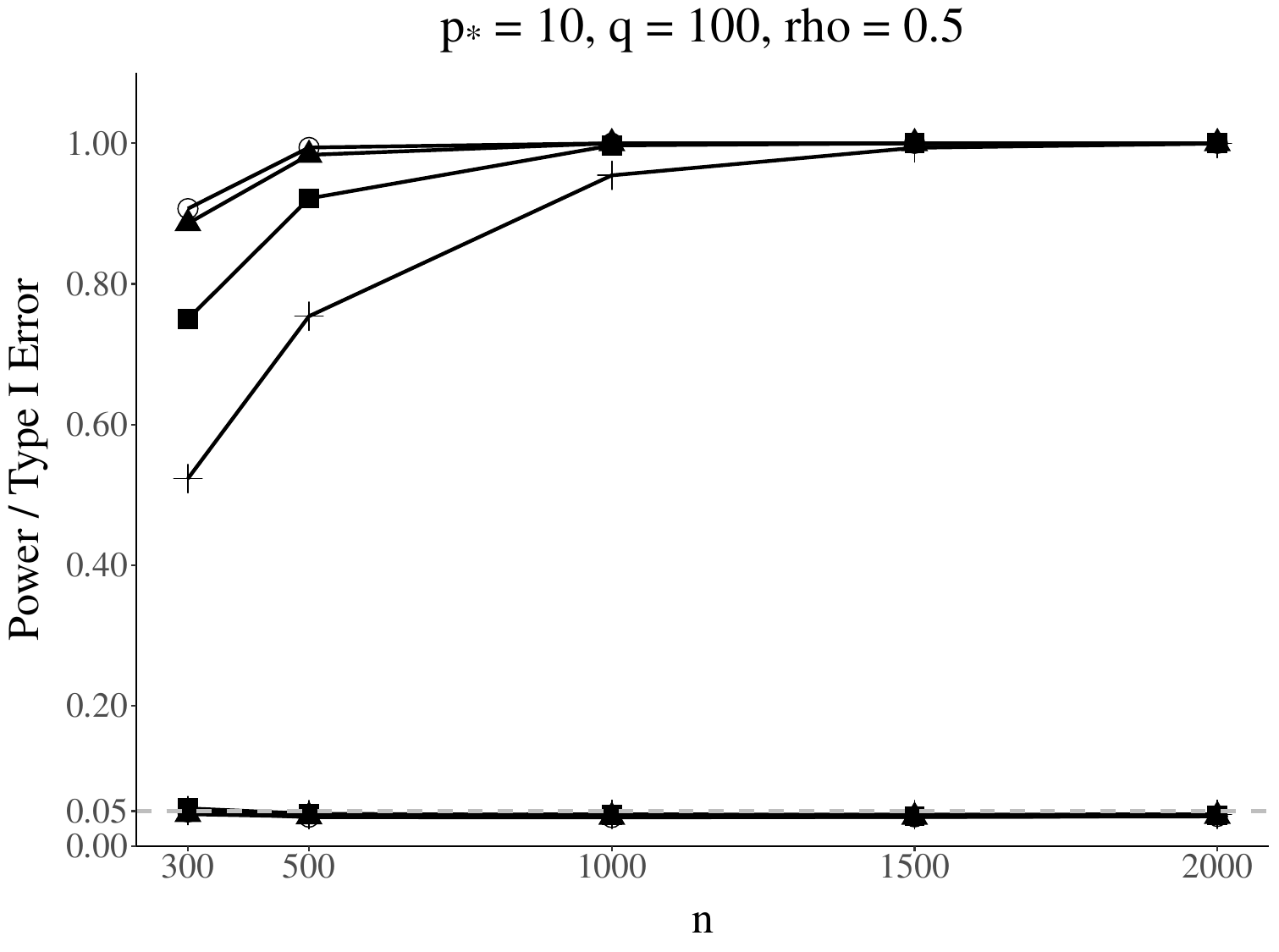}}%
          \subfigure{
        \includegraphics[width=1.9in]{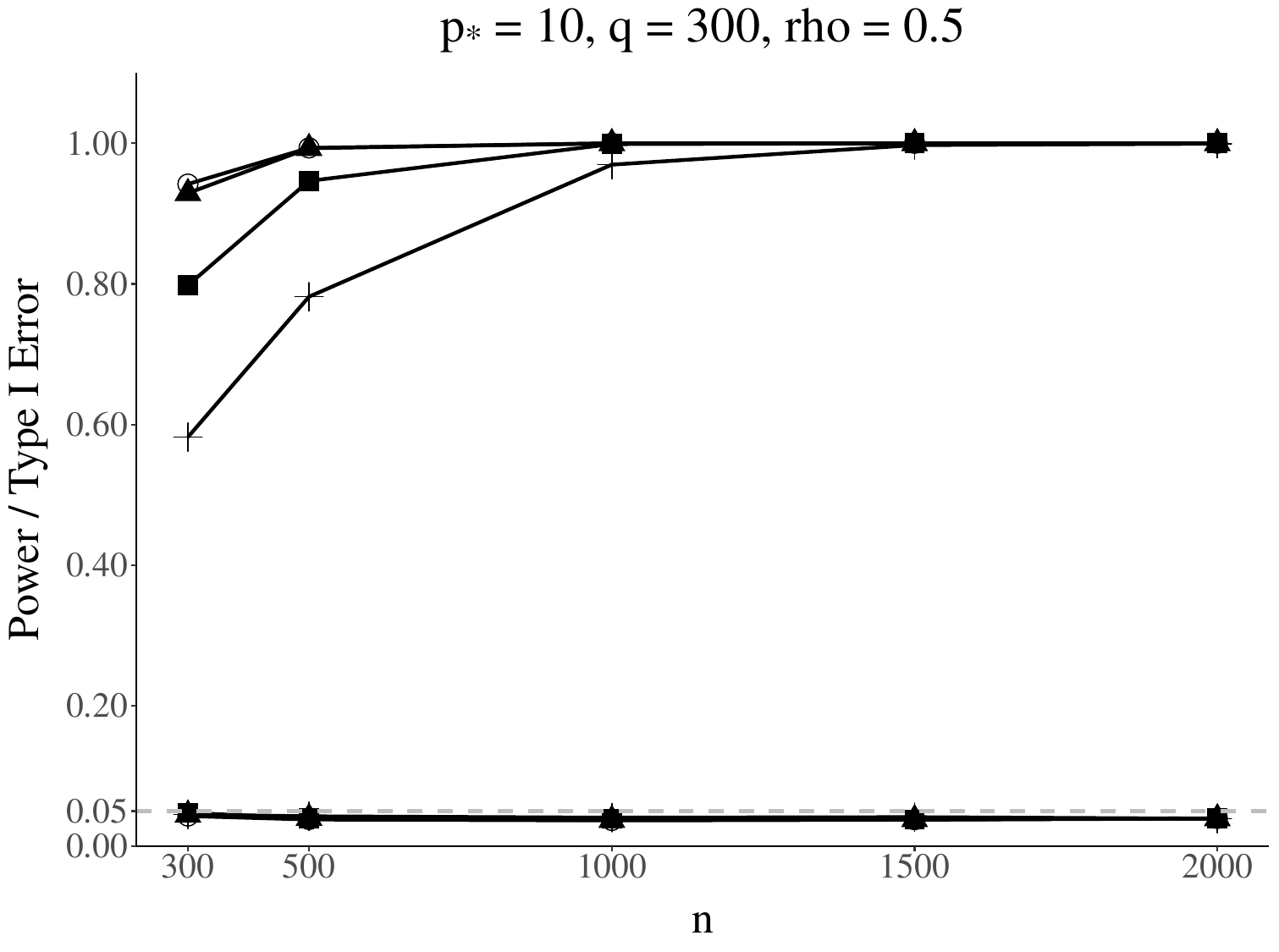}}%
          \subfigure{
        \includegraphics[width=1.9in]{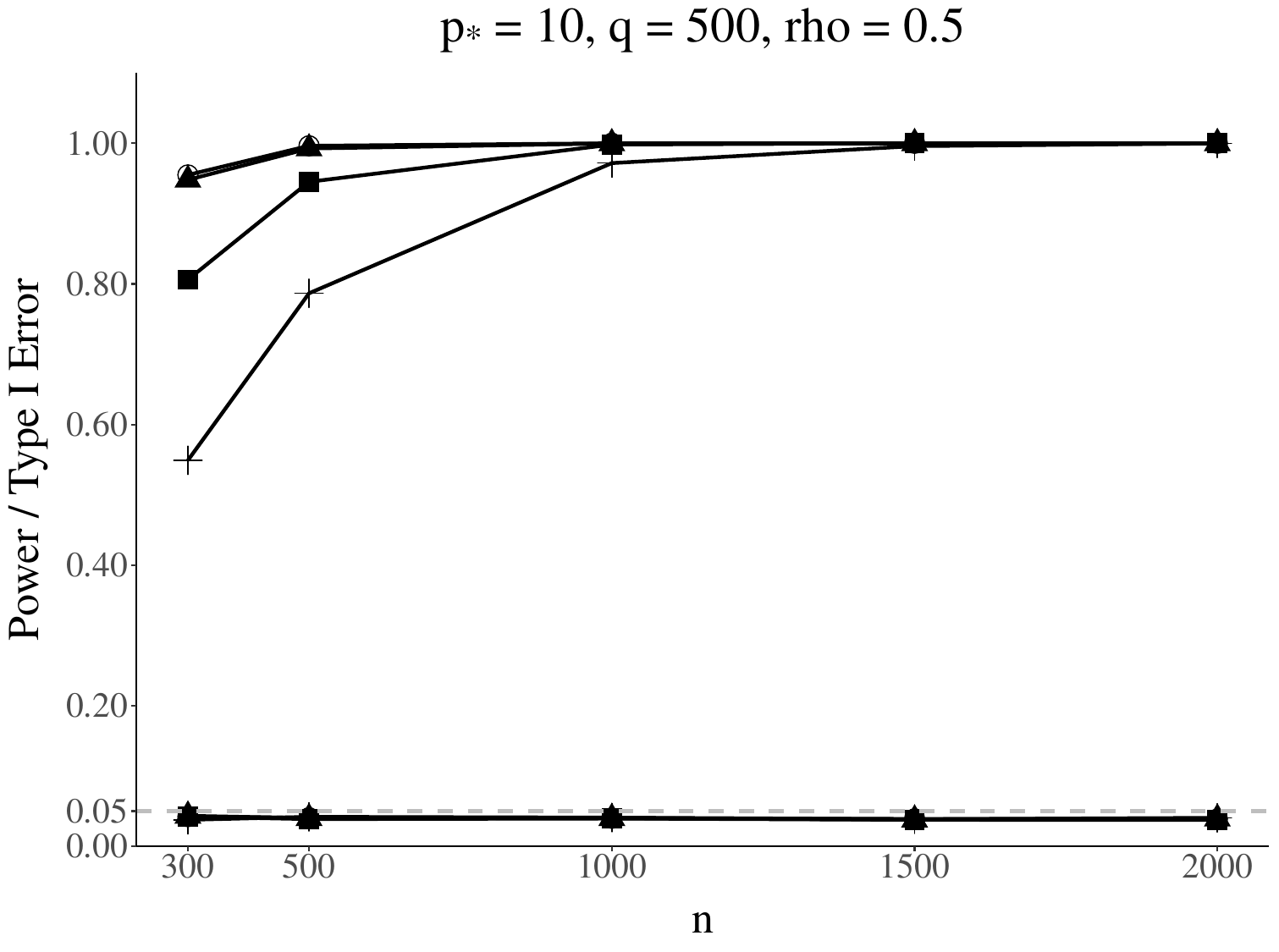}}%
    \caption{Powers and type I errors under sparse setting at $p_*=10$. Circles (\protect\includegraphics[height=0.8em]{legend/new_rho0.png}) denote correlation parameter $\tau = 0$. Triangles (\protect\includegraphics[height=0.8em]{legend/rho0.2.png})  represent the case $\tau = 0.2$. Squares (\protect\includegraphics[height=0.8em]{legend/rho0.5.png}) indicate $\tau = 0.5$ and crosses (\protect\includegraphics[height=1em]{legend/rho0.7.png}) represent the $\tau = 0.7$.}
    \label{fig:non-anchor-correlated p10}
\end{figure}

\begin{figure}[h!]
\centering    
   \subfigure{
        \includegraphics[width=1.9in]{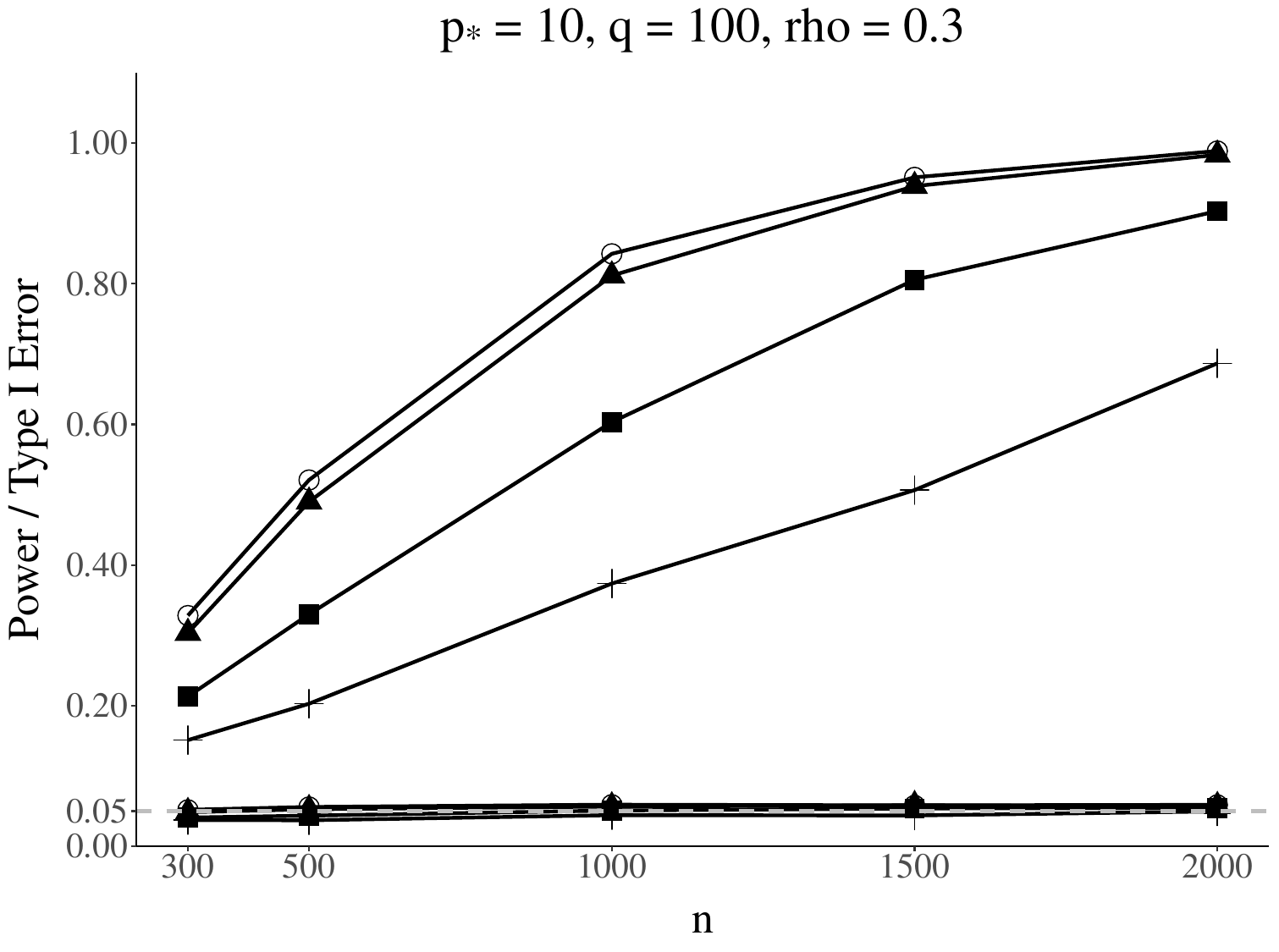}}%
        \subfigure{
        \includegraphics[width=1.9in]{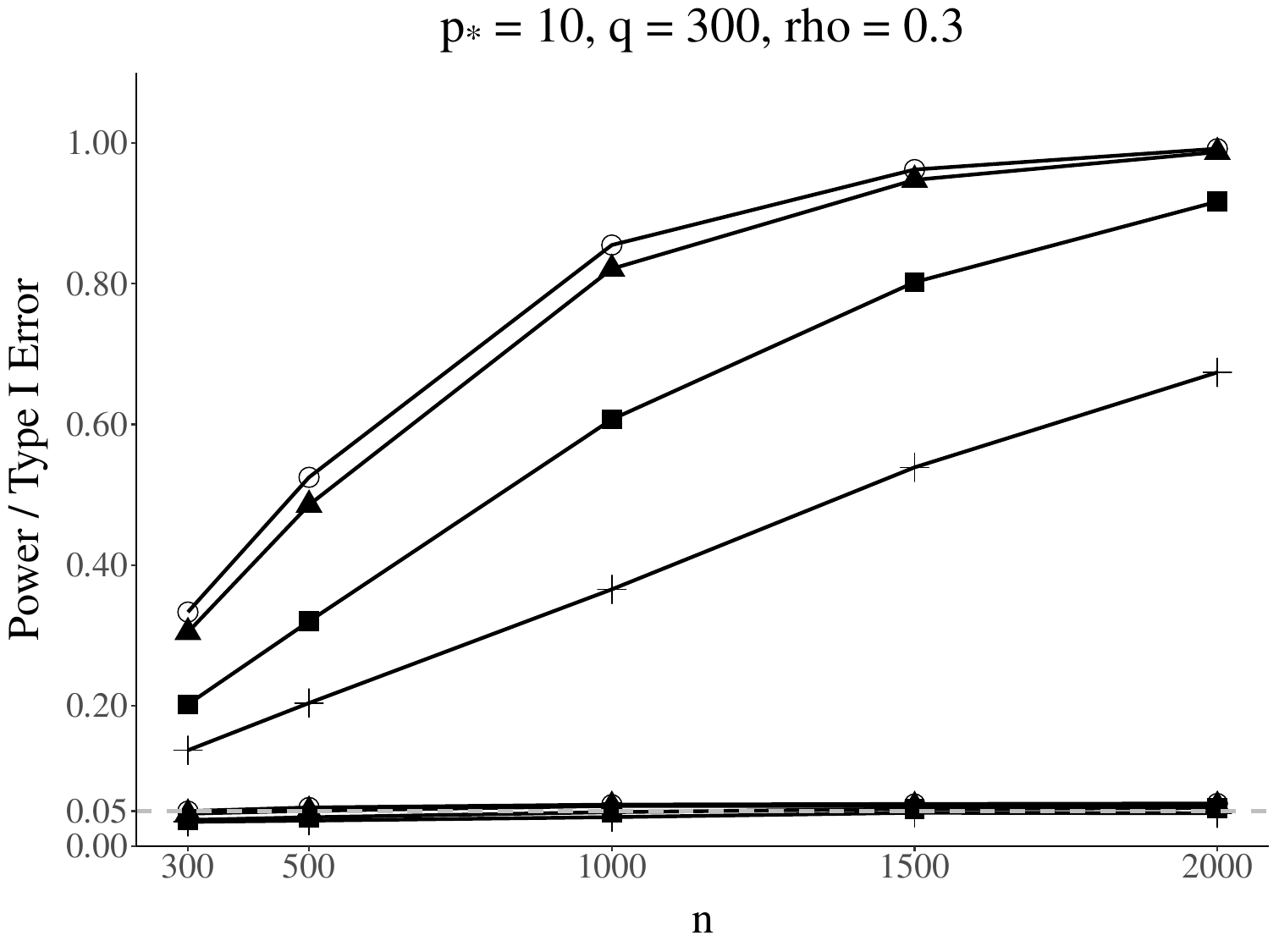}}%
         \subfigure{
        \includegraphics[width=1.9in]{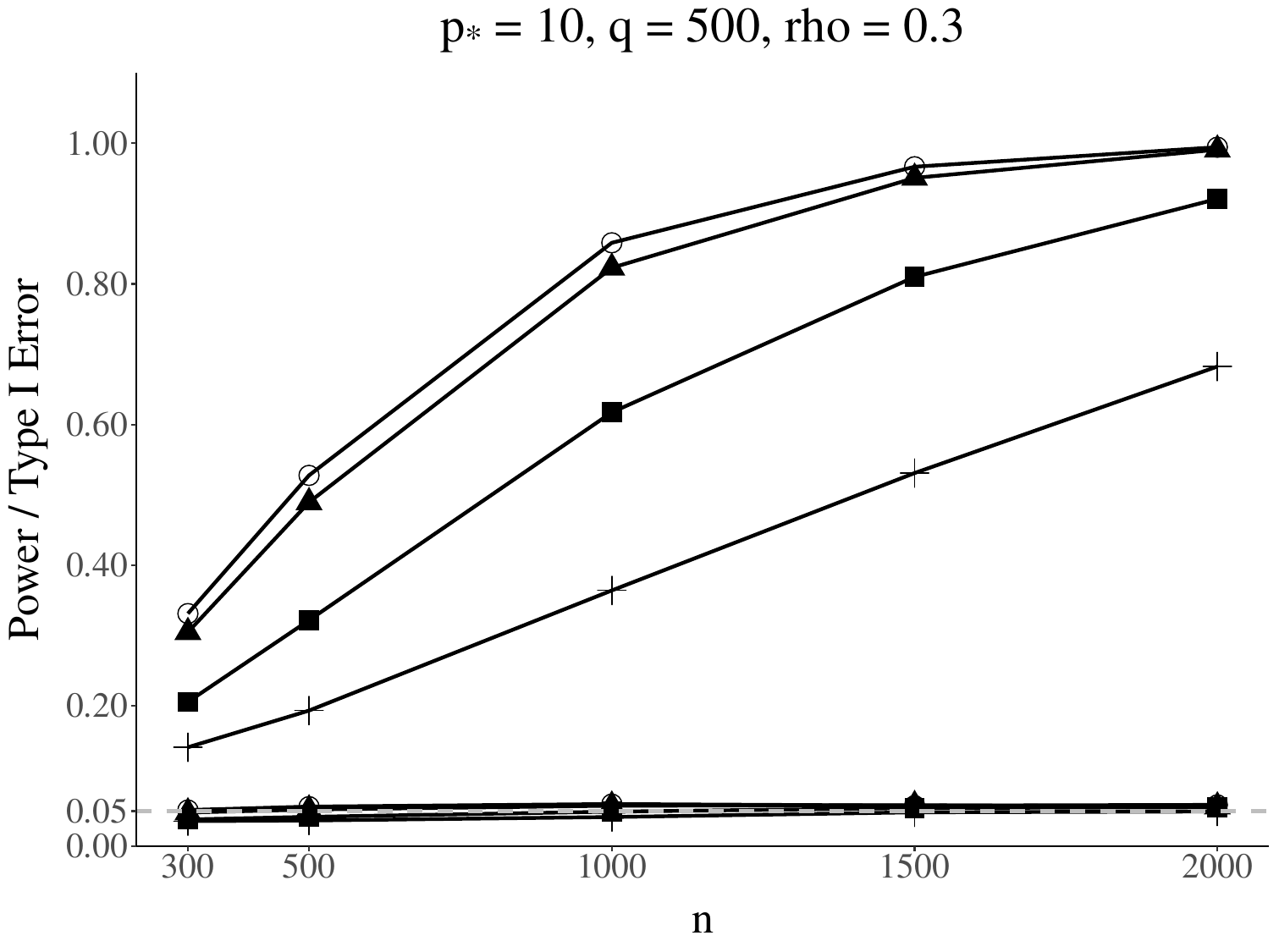}}%
        \\
          \subfigure{
        \includegraphics[width=1.9in]{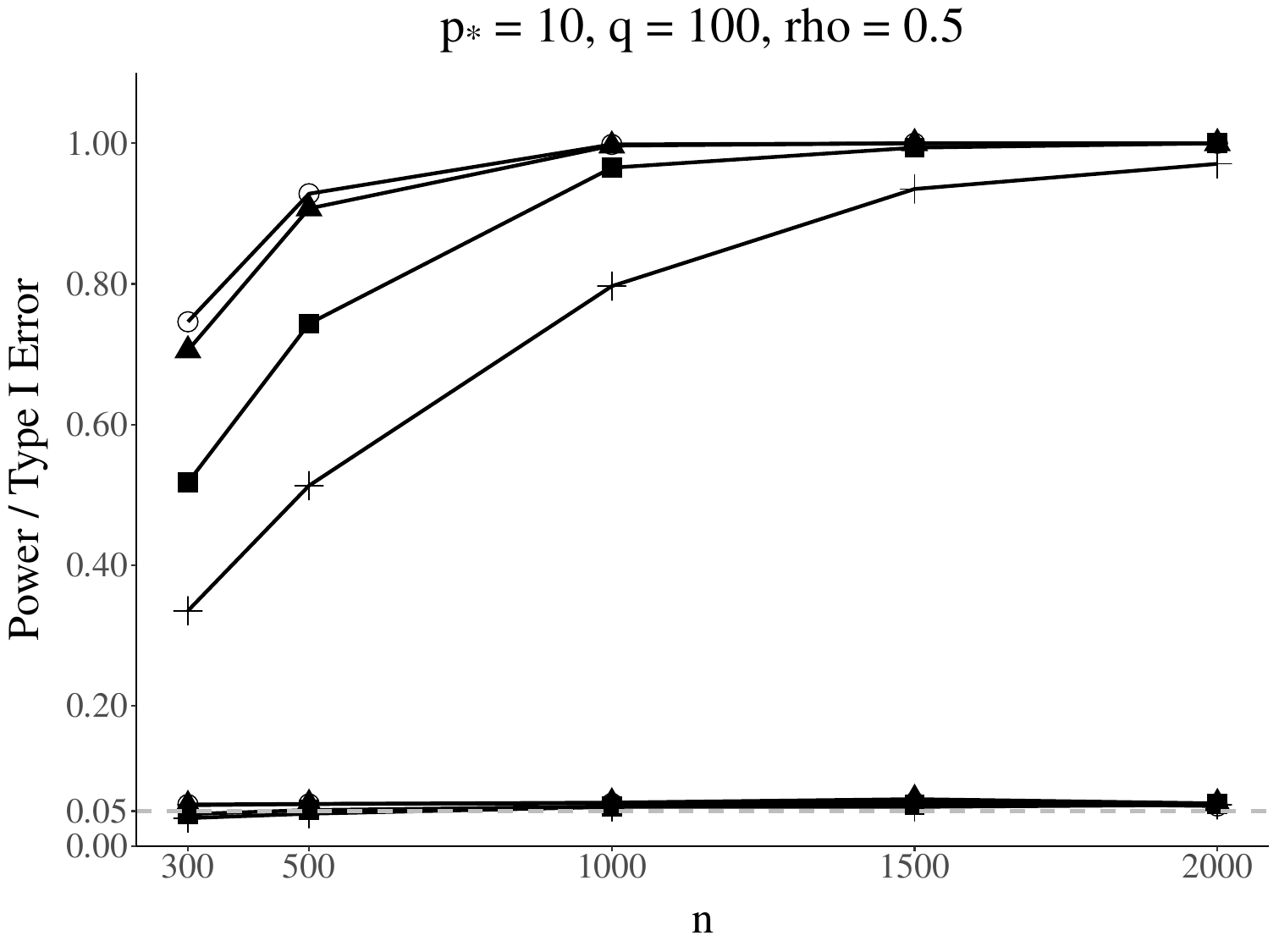}}
          \subfigure{
        \includegraphics[width=1.9in]{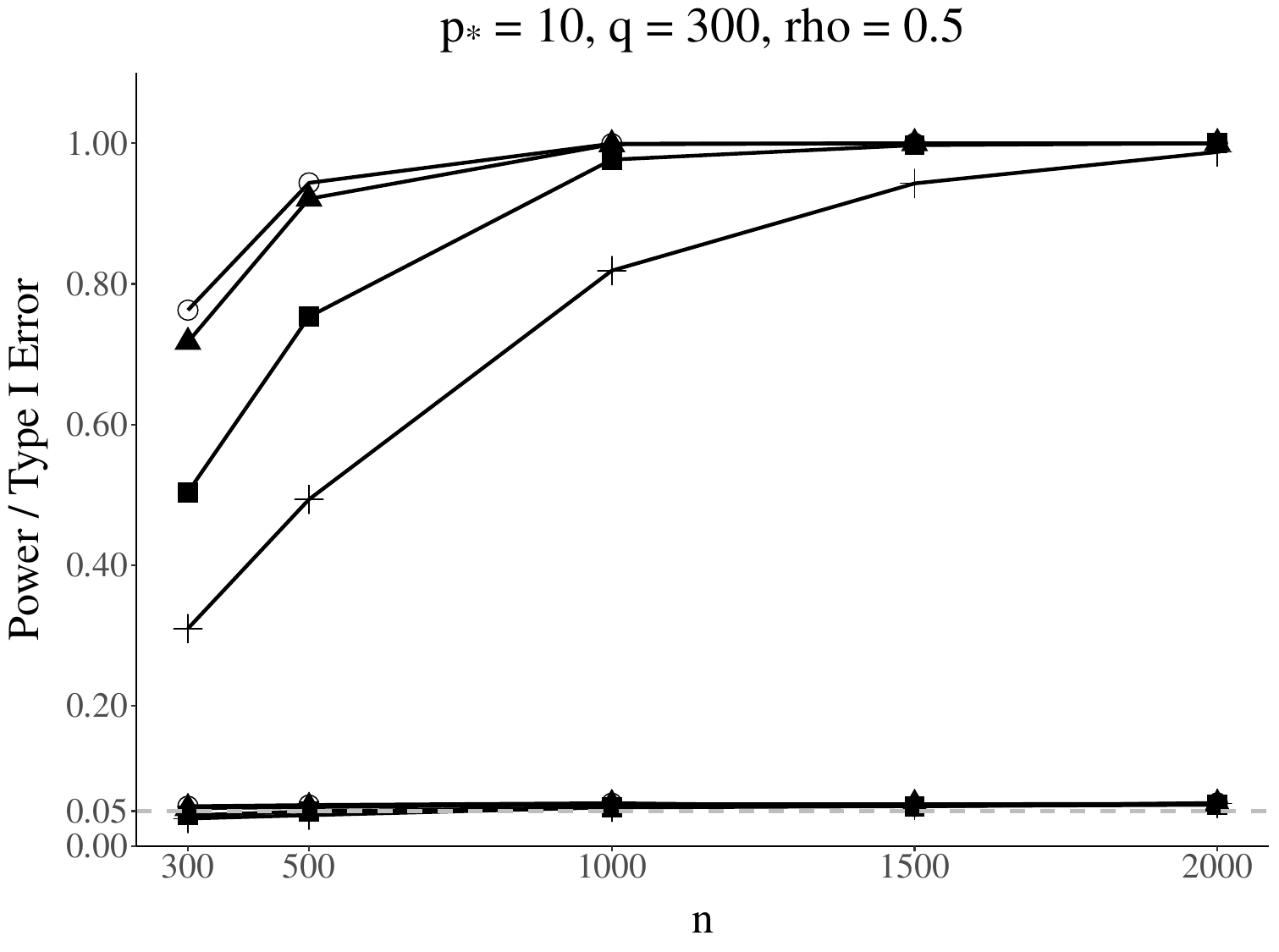}}%
          \subfigure{
        \includegraphics[width=1.9in]{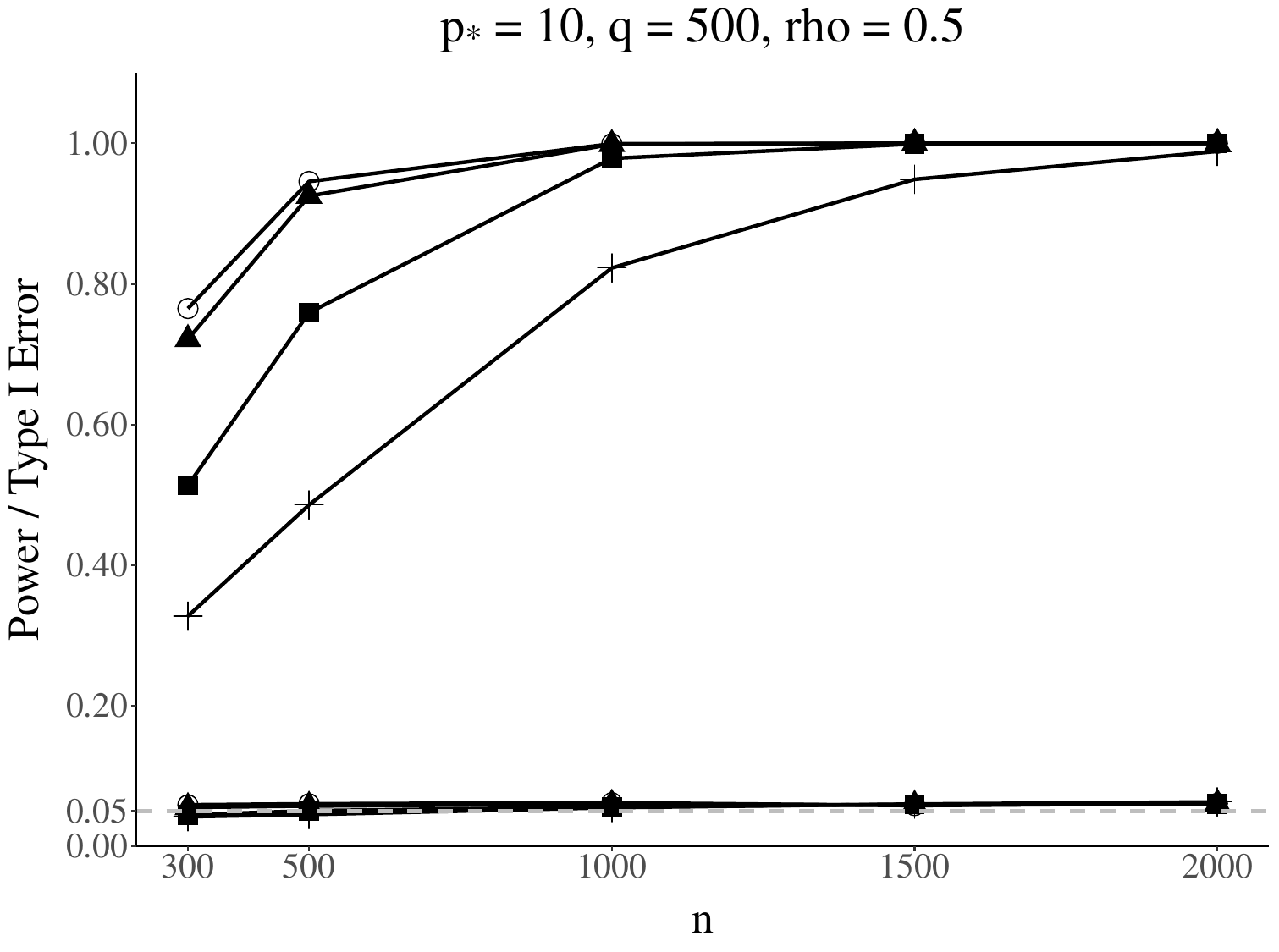}}%
    \caption{Powers and type I errors under dense setting at $p_*=10$. Circles (\protect\includegraphics[height=0.8em]{legend/new_rho0.png}) denote correlation parameter $\tau = 0$. Triangles (\protect\includegraphics[height=0.8em]{legend/rho0.2.png})  represent the case $\tau = 0.2$. Squares (\protect\includegraphics[height=0.8em]{legend/rho0.5.png}) indicate $\tau = 0.5$ and crosses (\protect\includegraphics[height=1em]{legend/rho0.7.png}) represent the $\tau = 0.7$.}
    \label{fig:non-anchor-correlated p10 dense}
\end{figure}

\subsection{Comparing Linear and Nonlinear Models for Binary Responses}

Recall that in regression analysis,  it is well-recognized that when modeling binary or categorical outcomes, generalized linear models (GLMs) with nonlinear link functions (e.g., logit or log links) should be used rather than linear regression models~\citep{nelder1972generalized, mccullagh2019generalized}. In psychometrics, item responses are often discrete, such as binary item response indicating whether a student answered a question correctly. To model such discrete data, it is common to use models with nonlinear link functions. Typically, the item Response Theory (IRT) models with a logit link are widely used, as they relate observed item responses to underlying latent traits (e.g., ability or proficiency) through a nonlinear link function. These models are common in educational and psychological assessments to account for both item characteristics and latent factors~\citep{birnbaum1968some, skrondal2004generalized, reckase2009}. Given that our motivating dataset, i.e., PISA 2018 data, contains mostly binary item responses, it is natural to use generalized (nonlinear) latent variable modeling framework for the estimation and inference in DIF analysis.

To illustrate the advantage of using nonlinear models over linear models for binary outcomes, we conducted a numerical study comparing the two frameworks. Specifically, we generate binary item response following the generating process described in Section~5 of the main text, and then fit the binary data using logistic model framework and linear model framework, separately. For illustration, we conduct simulation studies under (1) sparse setting with $ n \in \{300, 1000, 2000\}, q = 100,  p_*= 5, \rho \in \{0.3, 0.5\}$ and $\tau \in \{0, 0.5, 0.7\}$. The results for power and type I error under logistic model are presented in Figure~\ref{fig:true logistic q1}, while those under linear model are summarized in 
Figure~\ref{fig:3-Misspecified r1q1}. In addition, the empirical coverage probabilities of $\bU_i$ under both logistic model and linear model are summarized in
Table~\ref{tab:9-U_coverage r1}.

Comparing the upper part of Figures~\ref{fig:true logistic q1} with that of~\ref{fig:3-Misspecified r1q1}, the power under the logistic model is generally higher than that under the linear model across all settings. Besides, comparing the bottom part of Figures~\ref{fig:true logistic q1} with that of Figure~\ref{fig:3-Misspecified r1q1}, the type I error rates under logistic model are approximately at significance level 0.05, while those under linear model are much smaller than 0.05, indicating an overly conservative performance. Furthermore, as shown in Table~\ref{tab:9-U_coverage r1}, the empirical coverage probabilities for latent factors under logistic model is close to the nominal level 0.95, while the coverage probabilities under linear model are far below 0.95. These results indicate that linear models underperform in the hypothesis testing on not only covariate (DIF) effect but also latent factors as well.

\begin{figure}[ht]
\centering    
        \includegraphics[width=2.5in]{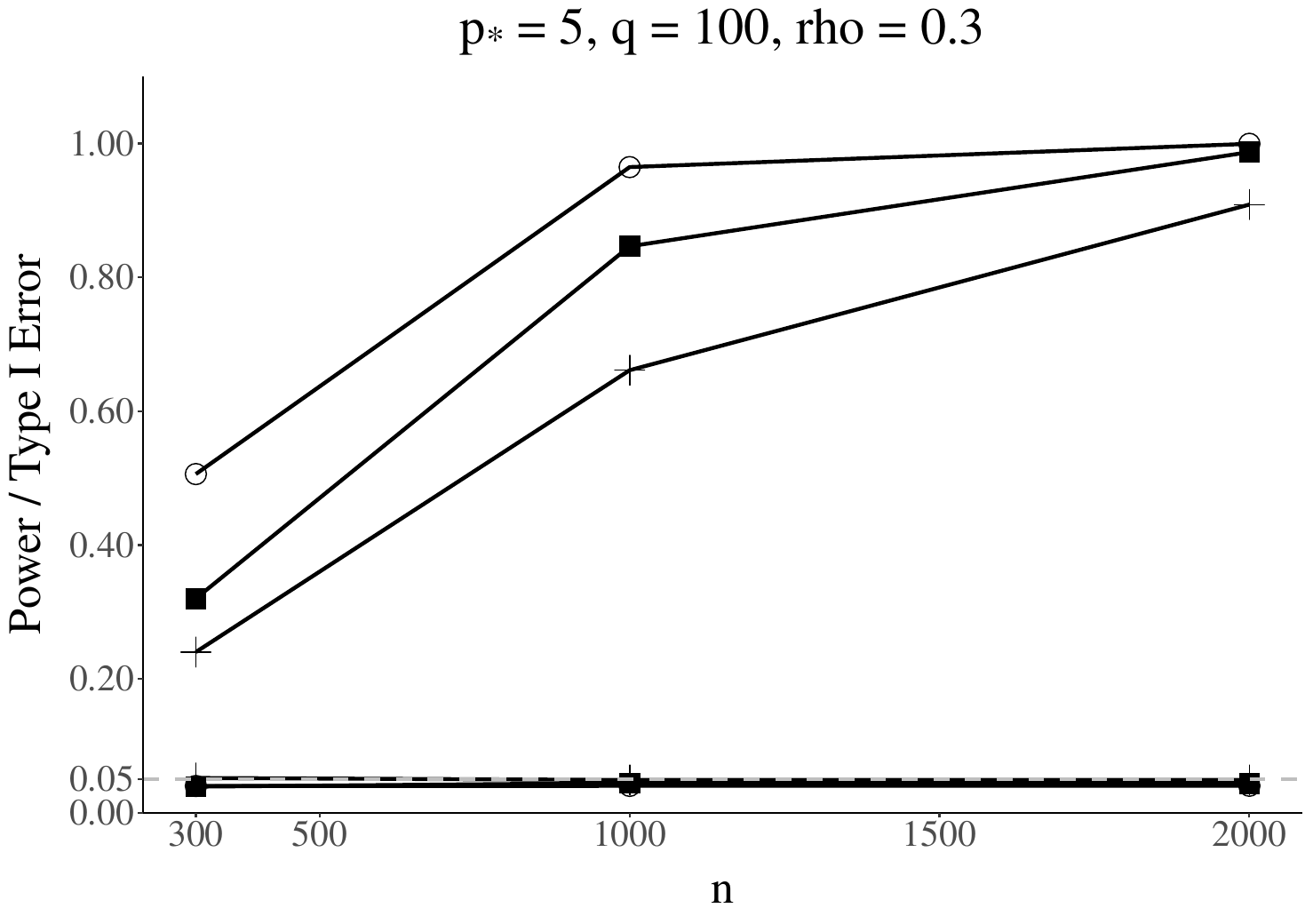}~
        \includegraphics[width=2.5in]{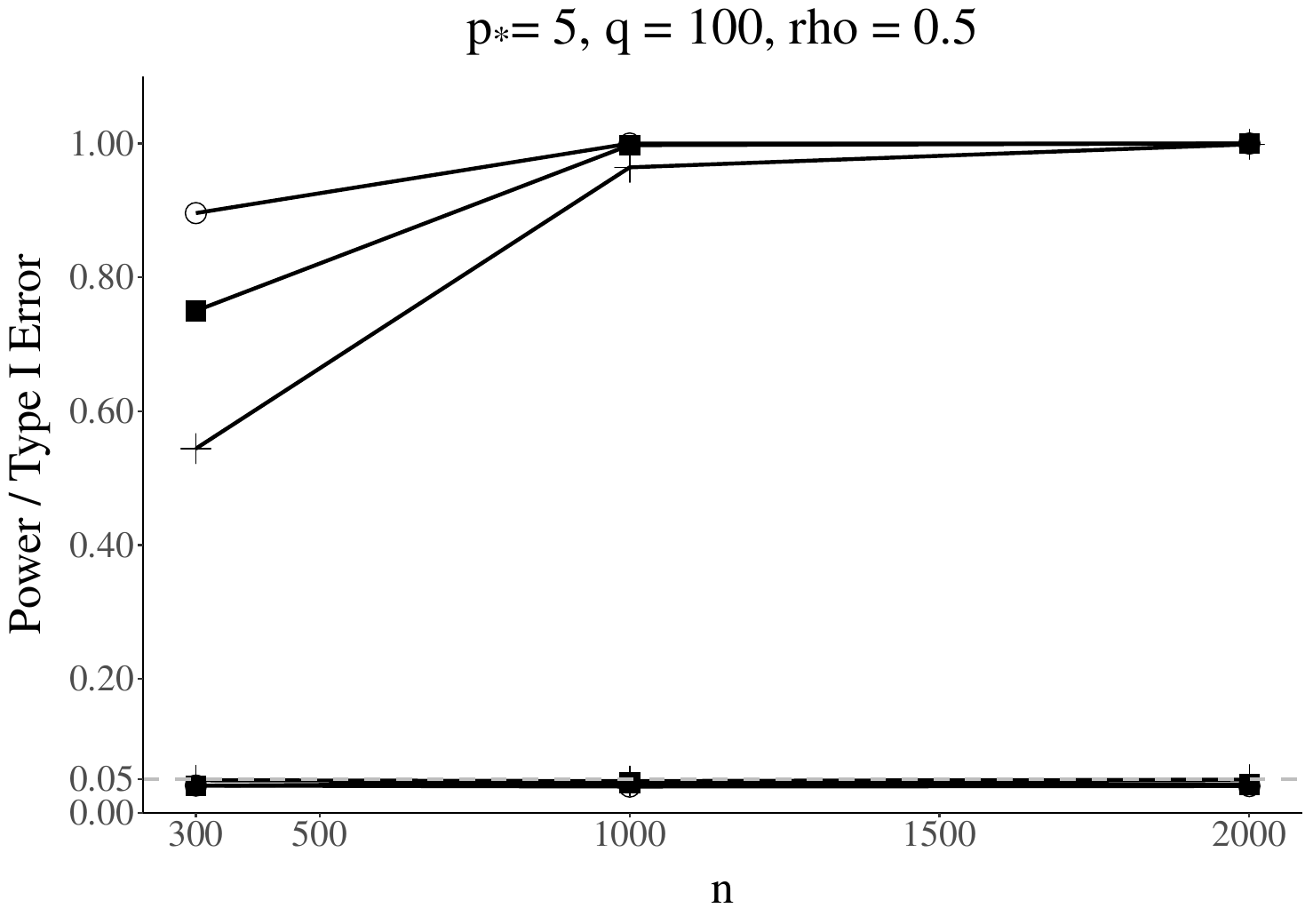}
    \caption{Powers and type I errors under sparse setting at $p_*=5$ and logistic framework. Circles (\protect\includegraphics[height=0.8em]{legend/new_rho0.png}) denote correlation parameter $\tau = 0$. Squares (\protect\includegraphics[height=0.8em]{legend/rho0.5.png}) indicate $\tau = 0.5$. Crosses (\protect\includegraphics[height=1em]{legend/rho0.7.png}) represent the $\tau = 0.7$.}
    \label{fig:true logistic q1}
\end{figure} 

\begin{figure}[ht]
\centering    
        \includegraphics[width=2.5in]{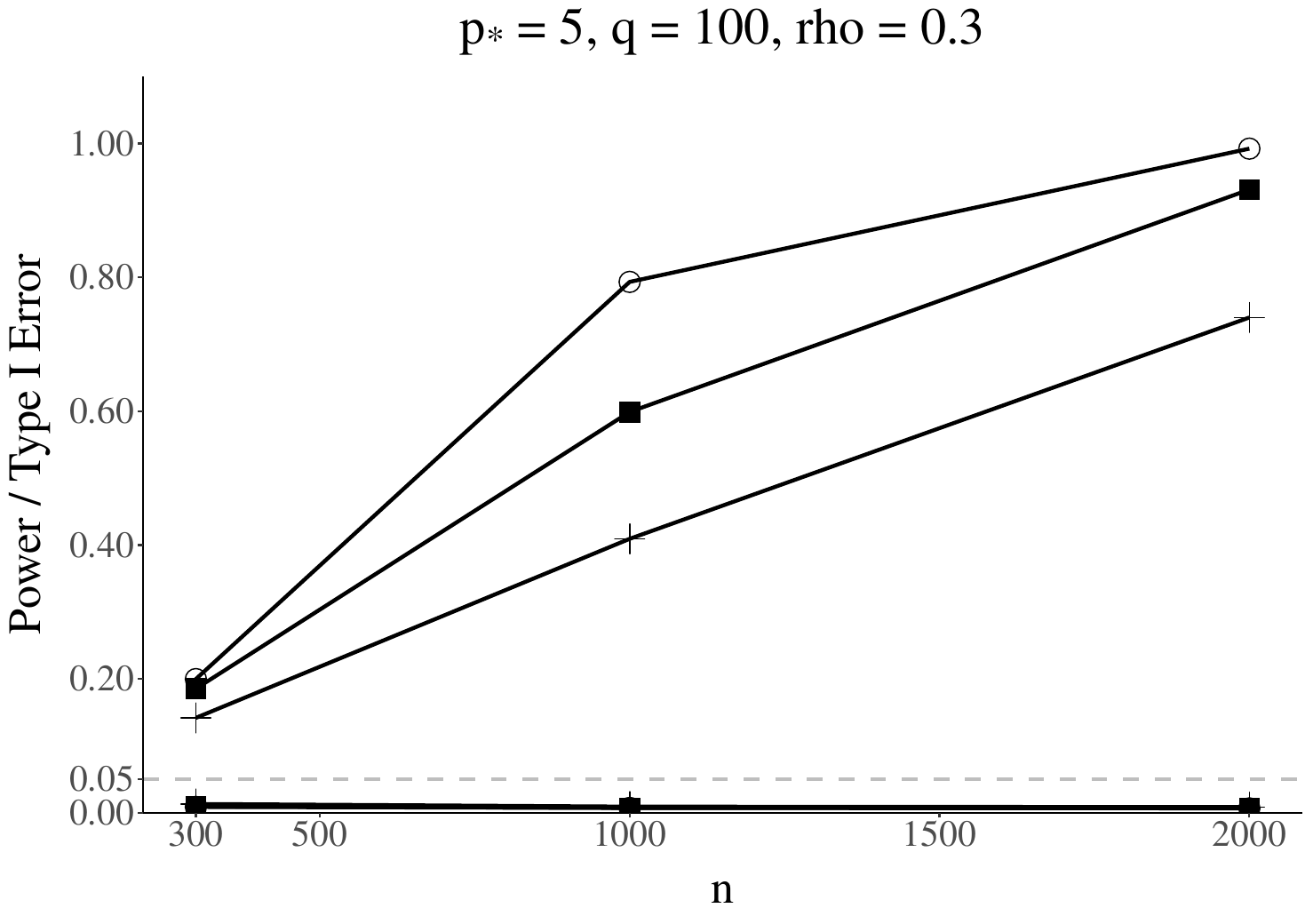}~
        \includegraphics[width=2.5in]{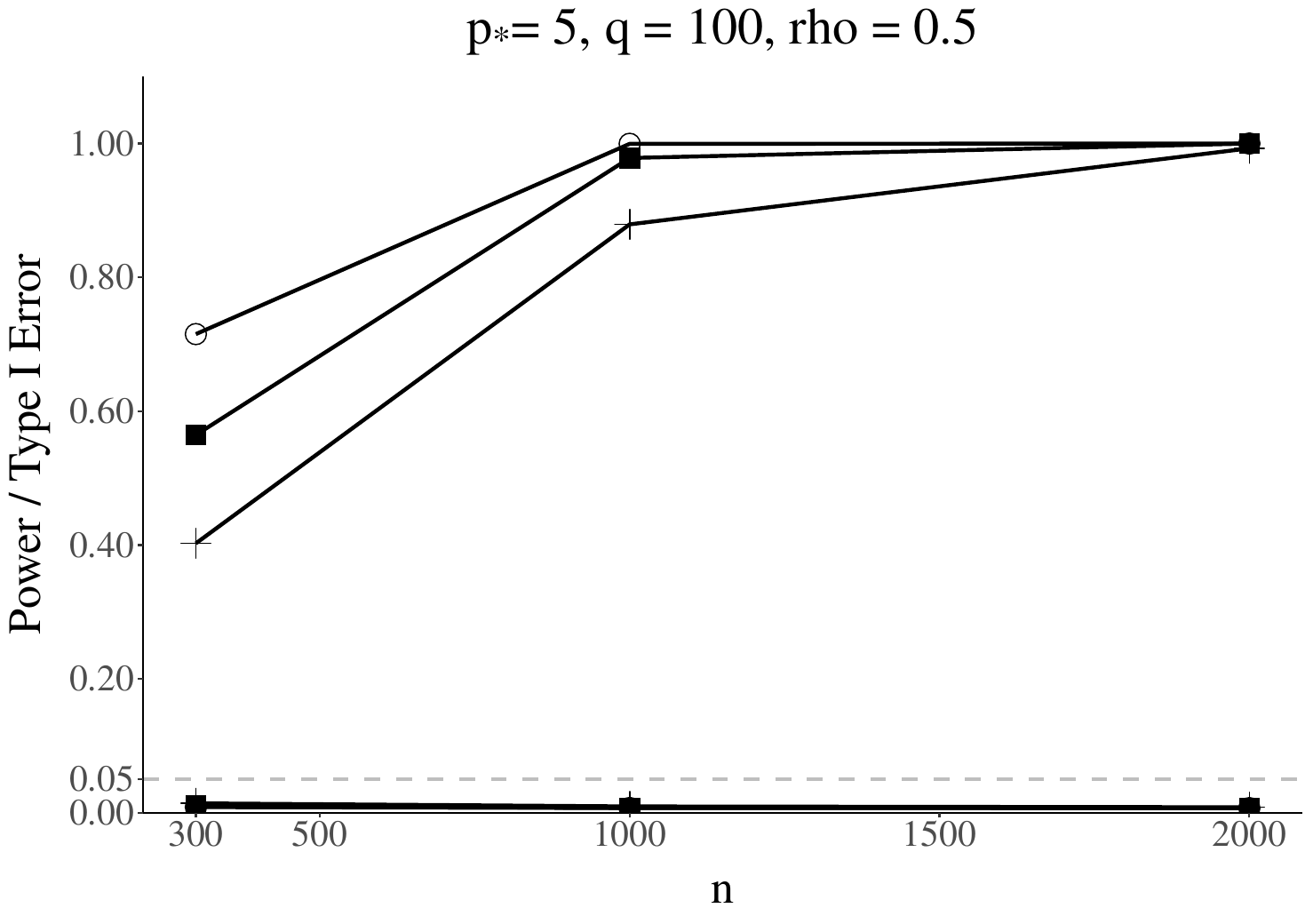}
    \caption{Powers and type I errors under sparse setting at $p_*=5$ and misspecified linear framework. Circles (\protect\includegraphics[height=0.8em]{legend/new_rho0.png}) denote correlation parameter $\tau = 0$. Squares (\protect\includegraphics[height=0.8em]{legend/rho0.5.png}) indicate $\tau = 0.5$. Crosses (\protect\includegraphics[height=1em]{legend/rho0.7.png}) represent the $\tau = 0.7$.}
    \label{fig:3-Misspecified r1q1}
\end{figure} 

\begin{table}[ht]
\centering
\begin{tabular}{clll}
  \hline
\multicolumn{4}{c}{Logistic model} \\
  \hline  
 & n = 300 & n = 1000 & n = 2000 \\ 
  \hline
$\rho = $ 0.5, $\tau = $ 0 & 0.879 (0.0742) & 0.931 (0.0259) & 0.941 (0.0086) \\ 
  $\rho = $ 0.3, $\tau = $ 0 & 0.878 (0.0739) & 0.931 (0.0264) & 0.942 (0.0084) \\ 
  $\rho = $ 0.5, $\tau = $ 0.5 & 0.898 (0.0408) & 0.932 (0.0183) & 0.942 (0.0067) \\ 
  $\rho = $ 0.3, $\tau = $ 0.5 & 0.896 (0.0483) & 0.932 (0.0180) & 0.941 (0.0087) \\ 
  $\rho = $ 0.5, $\tau = $ 0.7 & 0.888 (0.0570) & 0.929 (0.0234) & 0.939 (0.0102) \\ 
  $\rho = $ 0.3, $\tau = $ 0.7 & 0.889 (0.0582) & 0.929 (0.0228) & 0.940 (0.0098) \\ 
   \hline
\end{tabular}

\begin{tabular}{clll}
  \hline
\multicolumn{4}{c}{Linear model} \\
  \hline
 & n = 300 & n = 1000 & n = 2000 \\ 
  \hline
$\rho = $ 0.5, $\tau = $ 0 & 0.586 (0.0473) & 0.612 (0.0212) & 0.622 (0.0110) \\ 
  $\rho = $ 0.3, $\tau = $ 0 & 0.597 (0.0382) & 0.619 (0.0143) & 0.625 (0.0074) \\ 
  $\rho = $ 0.5, $\tau = $ 0.5 & 0.598 (0.0365) & 0.612 (0.0245) & 0.621 (0.0124) \\ 
  $\rho = $ 0.3, $\tau = $ 0.5 & 0.605 (0.0350) & 0.620 (0.0150) & 0.625 (0.0085) \\ 
  $\rho = $ 0.5, $\tau = $ 0.7 & 0.606 (0.0408) & 0.611 (0.0248) & 0.614 (0.0178) \\ 
  $\rho = $ 0.3, $\tau = $ 0.7 & 0.609 (0.0349) & 0.616 (0.0224) & 0.620 (0.0112) \\ 
   \hline
\end{tabular}
\caption{Empirical coverage probability for latent factors $U_{ik}^*$ under (1) sparse setting at $p_*=5$ and $q = 100$. Each entry reports the mean coverage (standard deviation) across 100 simulation replicates. }
\label{tab:9-U_coverage r1}
\end{table}

\subsection{Inferential Results for Latent Factors}
 Valid inference for the latent factors is important in both theoretical and practical aspects. From a theoretical perspective, developing inferential results leads to a deeper understanding of the behaviour of the MLE in nonlinear latent factor models, especially under complex identifiability constraints. It addresses a notable gap in the existing literature.
From a practical perspective, inference on latent factors enables a range of downstream analyses under DIF settings, such as comparing the latent abilities of individuals across different groups. These applications are of great importance in educational assessment and related fields. 

Motivated by the importance of, we include an additional evaluation metrics for simulation studies, specifically, we construct confidence intervals for each $U_{ik}$ for $i \in [n]$ and $k \in [K]$ and calculate the empirical coverage probabilities of these intervals on true parameter values $U_{ik}^*$ over 100 replications. For illustration, we conduct simulation studies under (1) sparse setting with $ n \in \{300, 1000, 2000\}, q = 100,  p_*= 5, \rho \in \{0.3, 0.5\}$ and $\tau \in \{0, 0.5, 0.7\}$. The empirical coverage results are summarized in the top panel of Table~\ref{tab:9-U_coverage r1}. As shown in the top panel of Table~\ref{tab:9-U_coverage r1}, the coverage probabilities for latent factors are close to the nominal level of 0.95 across all settings, showing the validity and reliability of the proposed inference results.

\subsection{Comparison with Anchor-Based Methods}
We have added anchor-based method as a comparative method and examined its performances under two cases: (a) the correctly specified anchor case: true DIF items are used as anchors; and (b) the misspecified anchor case: true non-DIF items are used as anchors.
    For illustrative purposes, we consider (1) sparse setting with combinations of sample size and covariate/item dimensions: $ n \in \{300, 1000, 2000\}, q = 100,  p=5, \rho \in \{0.3, 0.5\}$ and $\tau \in \{0, 0.5, 0.7\}$. We vary the number of anchor items to be 5 and 10.
    For case (a), the anchor items are set as the items 1 to 5 and items 1 to 10. For case (b), the anchor items are set as items $q-4$ to $q$ and items $q-9$ to $q$. The power and type I error results of anchor-based method are summarized in Figure~\ref{fig:7-anchor_item_method r2}.

    Comparing top two panels of Figure~\ref{fig:7-anchor_item_method r2} with its bottom two panels, we observe that the power will be lower in case (a) compared to case (b) since the latter case correctly specifies DIF-free items as anchors, while the former relies on misspecified DIF items. Moreover, even for the better performed case with non-DIF items as anchors, the type I error rates cannot be controlled under significance level 0.05 for $\tau = 0.5$ and $\tau =0.7$. While for our proposed method (see Figure~\ref{fig:true logistic r2}), the power are close to 1 when $n$ is large and the type I error is controlled under significance level 0.05 across all settings.

\subsection{Relative Scales of $n$ and $q$}

We considered a simulation setting with a fixed relative scale of $n$ and $q$. Specifically, we set $q = 5n$, where $n \in \{300, 500, 1000\}$, and the remaining data-generating process follows the same in Section~5 of the main text. The results are summarized in Figure~\ref{fig:8-relative_scale_nq_r3}. We notice that under the regime of $q = 5n$, the type I error rates remain contolled under the significance level 0.05 and the powers grow when sample sizes $n$ grow and are close to 1.

\subsection{Robustness to Discrete Covariates}
We include an additional simulation study in which the covariates $\bX_i^c$ are discrete to demonstrate the robustness of the proposed method under various settings. Specifically, we jointly generate $\bX_i^c$ and $\bU_i^*$ from $\cN(\bm{0},\bSigma)$ and then dichotomize the generated $\bX_i^c$ into binary random variable by thresholding each entry at zero. Specifically, each entry of $\bX_i^c$ is set to 1 if its original value is greater than 0, and 0 otherwise. We apply our proposed method to the resulting binary covariates. From Figure~\ref{fig:5-binary_X r3}, we observe that the type I error remains under the significance level and power remains close to 1, indicating that the method performs great when covariates are discrete.

\subsection{Sensitivity Analysis with Misspecified Number of Factors}

In our work, the covariate-adjusted generalized factor model is used within a confirmatory factor analysis framework. In contrast to exploratory factor analysis, where the number of latent factors is empirically determined, confirmatory factor analysis specifies the number of factors a priori based on theoretical considerations or prior research~\citep{brown2012confirmatory, chen2025item}.

In practice, the number of latent factors is typically pre-specified based on domain expertise and prior analysis. For example, in the data application to PISA 2018, the assessment evaluates student performances through three core domains—mathematics, science, and reading. Accordingly, we pre-specify the number of latent factors as $K = 3$ to align with the three underlying abilities targeted in these domains~\citep{schleicher2019pisa, pisatechnicalreport2018}. Other popularly used methods for selecting the number of factors can be found in exploratory factor analysis (EFA), including scree plots~\citep{cattell1966scree}, parallel analysis~\citep{horn1965}, the eigenvalue ratio method~\citep{kaiser1960application, lam2012factor}, and information-based criteria such as AIC~\citep{akaike1987factor}, BIC~\citep{schwarz1978estimating}, and other information criterion based on joint likelihood~\citep{chen2022determining}.


To evaluate the impact of the factor dimension $K$ on the estimation and inference of DIF effects, we conduct a sensitivity analysis by comparing model performance under different values of $K$. In our simulation studies, we assess the robustness of our method by examining cases where the factor dimension used in estimation differs from the true value $K = 4$. Specifically, we generate data following the same process in Section~5 of main text with factor dimension set to be $K = 4$. Then we fit the model with ${K}_{mis}=3, 4, 5, 8$ to assess the method performances of underestimating, correctly estimating, and overestimating the number of factors. We use the proposed method and follow the same procedure of hypothesis testing on covariate effect as in the main text, and summarize the results in Figure~\ref{fig:6-sensitivity analysis}.
    As shown in the second to fourth panels of Figure~\ref{fig:6-sensitivity analysis}, the type I error rates are below the significance level 0.05 for ${K}_{mis} = 4, 5, 8$, implying that our method is robust to overestimation of factor dimension. However, when the factor dimension is underestimated (${K}_{mis} = 3$), the type I error rates exceed the significance level 0.05 as shown in the top panel of Figure~\ref{fig:6-sensitivity analysis}, showing an underperformance of the proposed method when the factor dimension is lower than the true value of factor dimension.

\subsection{Robustness to the Number of Random Initializations}
In terms of finding the global optimum of the proposed alternating minimization algorithm, we follow a common practice of generating multiple random initial values and returning the solution that yields the largest likelihood~\citep{collins2001generalization}.
 We conduct simulation studies to show that the algorithm is robust to the number of random initial values. Specifically, we follow the same data-generating process as in Section~5 of the main text, but vary the number of initial values in the algorithm to be $3, 5$, and $8$ in the algorithm. The results are presented in Figure~\ref{fig:11-num_initial}. Across different numbers of initial values, the pattern of type I error and power remain very similar. In particular, the method performs consistently great with type I errors under 0.05 and empirical powers close to 1.

\subsection{Full DIF Effect Model}
Consider a full DIF effect model incorporating the effect of $\bX_i^c$ that runs via $\bU_i$ with:
\begin{align*}
w_{i j}=\beta_{j 0}+\boldsymbol{\gamma}_j^T \boldsymbol{U}_i+\boldsymbol{\beta}_{j c}^T \boldsymbol{X}_i^c+\boldsymbol{\gamma}_j^T\left(\ba_0+\boldsymbol{a}_1^T \boldsymbol{X}_i^c\right)
\end{align*}
where $\ba_0$ is the intercept in the regression of $\boldsymbol{U}_i$ on $\boldsymbol{X}_i^c$ (which is commonly fixed to 0 for identification purposes) and $\ba_1 \in \RR^{K\times p}$ is the slope parameter.

Under the present joint maximum likelihood framework, fixing $\ba_1$ to $\bm 0_p$ will not affect the estimation or inference for $\bbeta_{jc}$, as the effect of $\ba_1$ is fully absorbed into the estimates of $\bU_i$. 
First, we view the model (1) in the main text as a reparameterization of the full DIF effect model.
Specifically, we can rearrange the suggested model into:
\begin{subequations}\begin{align}
    w_{ij} & = \beta_{j0} + \bgamma_j^{\T} \bU_i + \bbeta_{jc}^{\T} \bX_i^c + \bgamma_j^{\T} (\ba_0 + \ba_1^{\T} \bX_i^c) \nonumber \\
    & = (\beta_{j0} + \bgamma_j^{\T}\ba_0 ) + \bgamma_j^{\T} (\bU_i + \ba_1 ^{\T} \bX_i^c)  + \bbeta_{jc}^{\T} \bX_i^c \label{eq:rearrange r3} \\
 & = \beta_{j0}  + \bgamma_j^{\T} (\bU_i + \ba_1^{\T} \bX_i^c)  + \bbeta_{jc} ^{\T} \bX_i^c \label{eq:decomposition wij} \\
 & =: \tilde{\beta}_{j0}+\tilde{\bgamma}_j^{\T} \tilde{\bU}_i + \tilde{\bbeta}_{jc}^{\T} \bX_i^c, \label{eq:our model}
\end{align}\end{subequations}
where the simplification~\eqref{eq:decomposition wij} follows from setting $\ba_0 = \bm{0}$, a common practice and also suggested by the reviewer. As a result, our model~\eqref{eq:our model} is equivalent to suggested full DIF model after parameter transformation: 
\begin{align}
\label{eq:r2 transformation}
    \tilde{\beta}_{j0} = \beta_{j0},\; \tilde{\bgamma}_j = \bgamma_j,\; \tilde{\bU}_i =\bU_i + \ba_1^{\T} \bX_i^c,\; \tilde{\bbeta}_{jc} = \bbeta_{jc}.
\end{align}
Applying our proposed estimation method in the Section 3.2 of the main text, we obtain the estimators $\hat{\beta}_{j0}^*, \hat{\bgamma}_j^*, \hat{\bU}_i^*, \hat{\bbeta}_{jc}^*$ for the parameters $ \tilde{\beta}_{j0}, \tilde{\bgamma}_j , \tilde{\bU}_i , \tilde{\bbeta}_{jc}$, respectively. Since the true parameters $\bbeta_{jc}^*$ satisfy the Condition 1(ii) $\sum_{j=1}^q \| \bbeta_{jc} \|_1 < \sum_{j=1}^q \| \bbeta_{jc} - \Ab_{c}^{\T} \bgamma_j\|_1$ for any $\Ab_c\neq \bm{0}$, all the effect of $\ba_1$ is absorbed into estimates of $\bU_i$ instead of entering the parameter $\tilde{\bbeta}_{jc}$.

In addition, the slope parameter $\ba_1$ in the regression of $\bU_i$ on $\bX_i^c$ captures the dependence between $\bU_i$ and $\bX_i^c$. As the existence of $\ba_1$ will not affect the estimation and inference of $\bbeta_{jc}$, such model reparameterization implies that our considered framework allow arbitrary dependence between $\bU_i$ and $\bX_i^c$. 

Next, we perform a numerical study to demonstrate that the nonzero $\ba_1$ is fully absorbed in the $\bU_i$ estimates and do not affect the estimation and inference results for $\bbeta_{jc}$. We adopt the same data-generating process for the covariate $\bX_i$ and parameters $\bU_i$, $\bgamma_j$, $\bbeta_{jc}$, $\beta_{j0}$ as in Section~5 of main text. In addition, following the suggested full DIF model, we incorporate $\ba_1 = (a\Ib_K, \bm{0})$, with $a\in \{0.3, 0.5\}$ when generating the item response data. We then apply our proposed method to the resulting data and covariates. The results for power and type I error of $\bbeta_{jc}$ are presented in Figure~\ref{fig:13-marginal_full_DIF}. As shown, the type I errors of $\bbeta_{jc}$ remain well controlled below 0.05 across all the settings with nonzero $\bm{a}_1$, confirming that nonzero $\ba_1$ does not affect the inference results for $\bbeta_{jc}$. 

\clearpage

\begin{figure}[htbp]
\centering   
\includegraphics[width=2.5in]{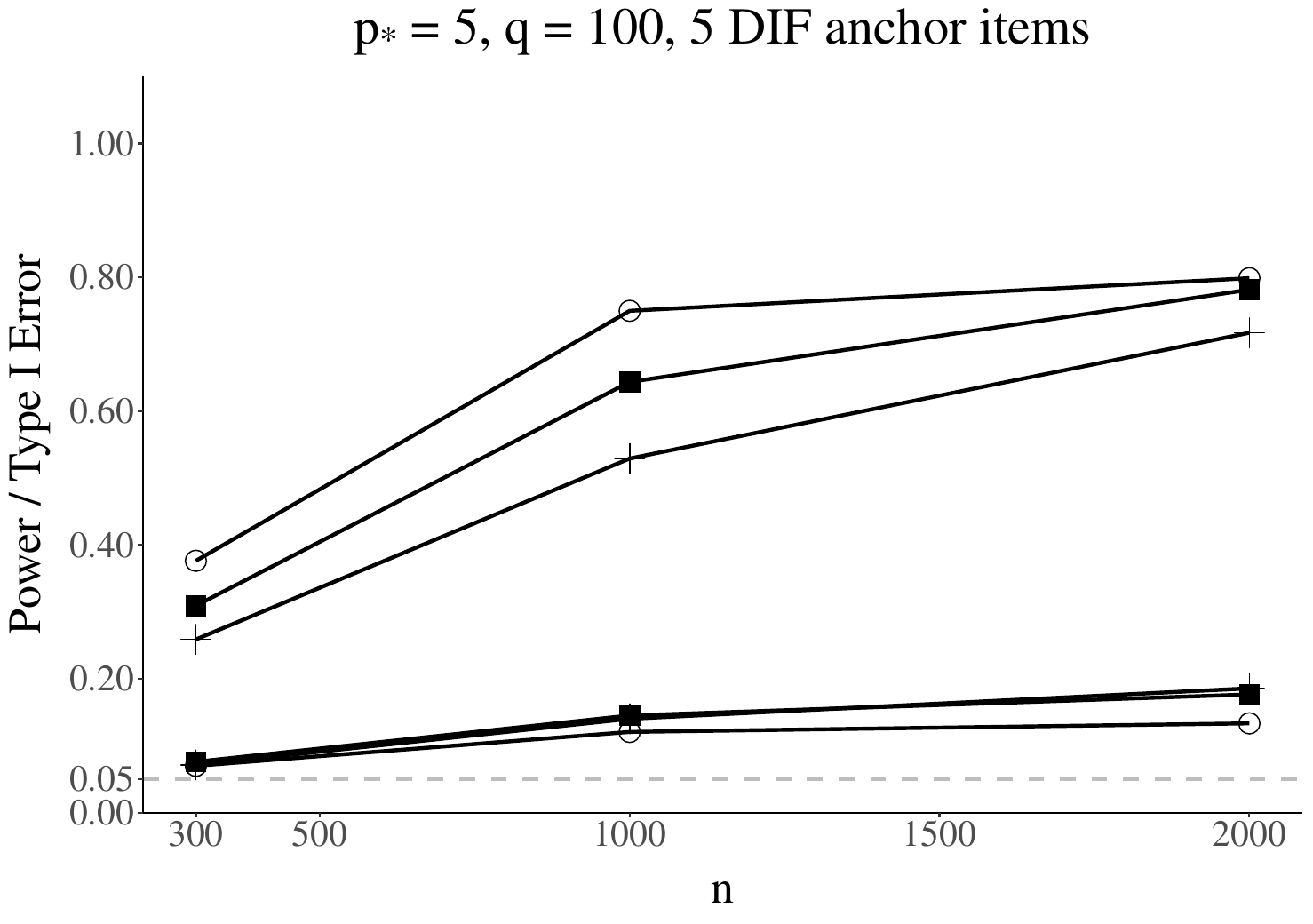}~
        \includegraphics[width=2.5in]{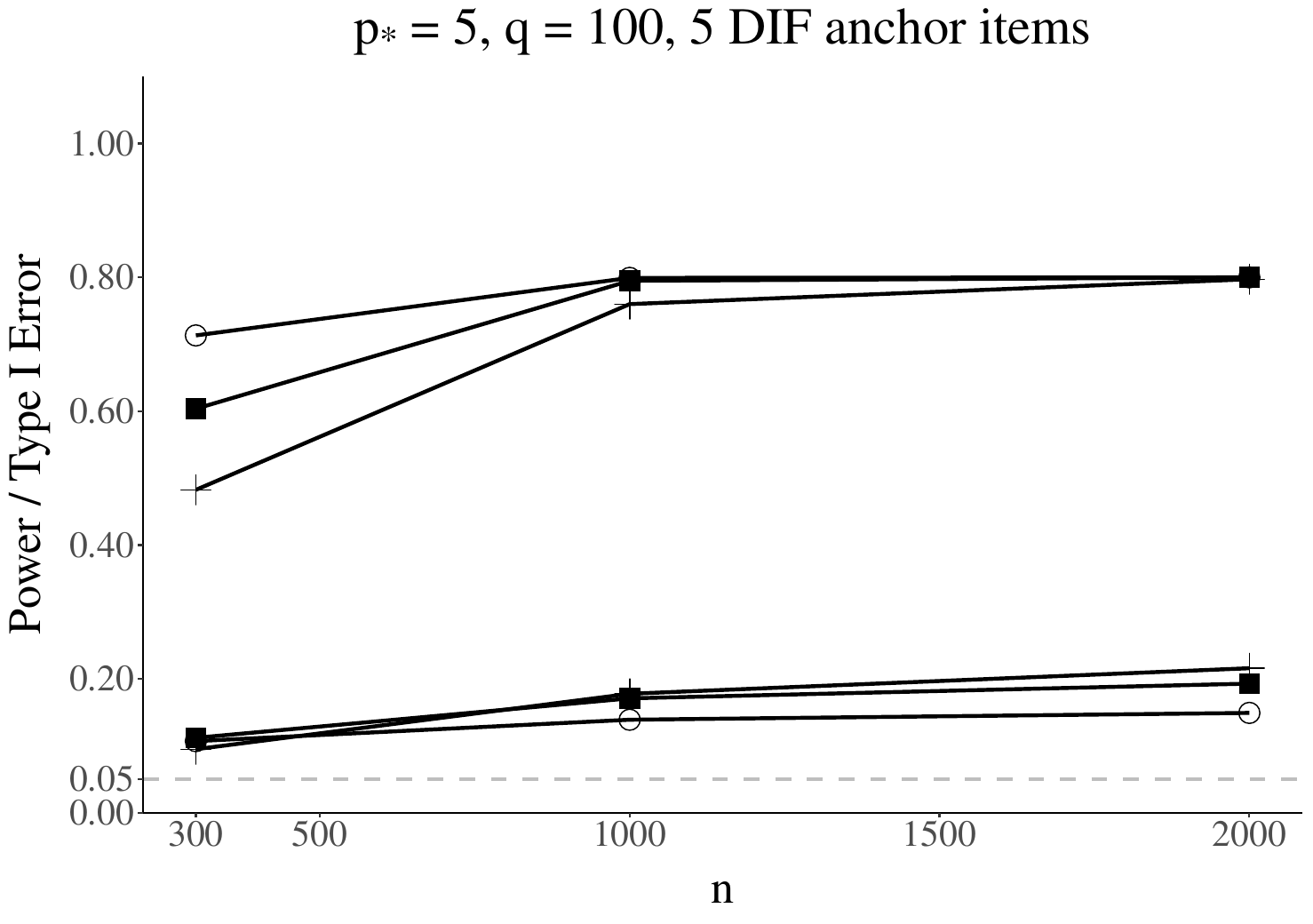}\\
         \includegraphics[width=2.5in]{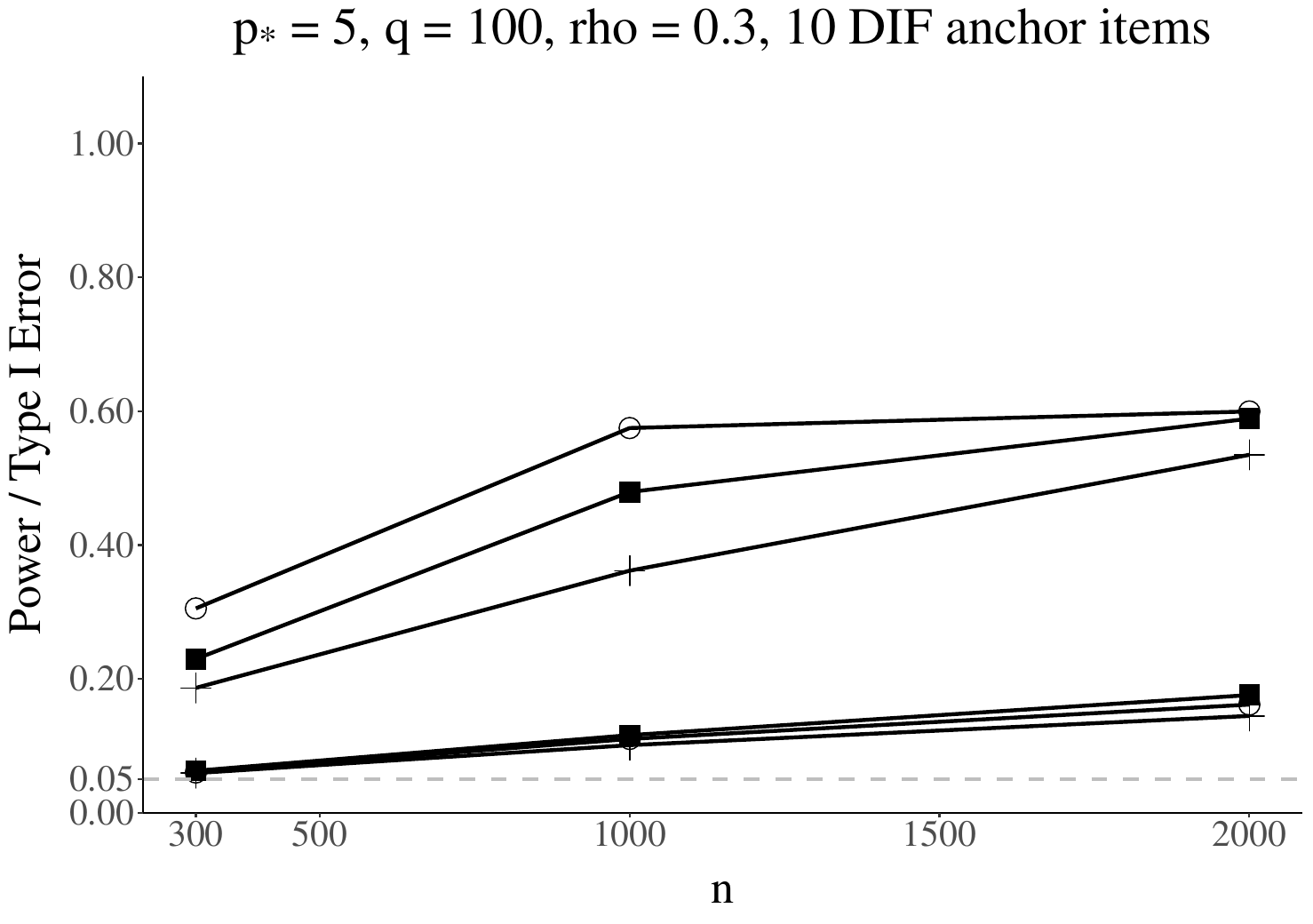}~
        \includegraphics[width=2.5in]{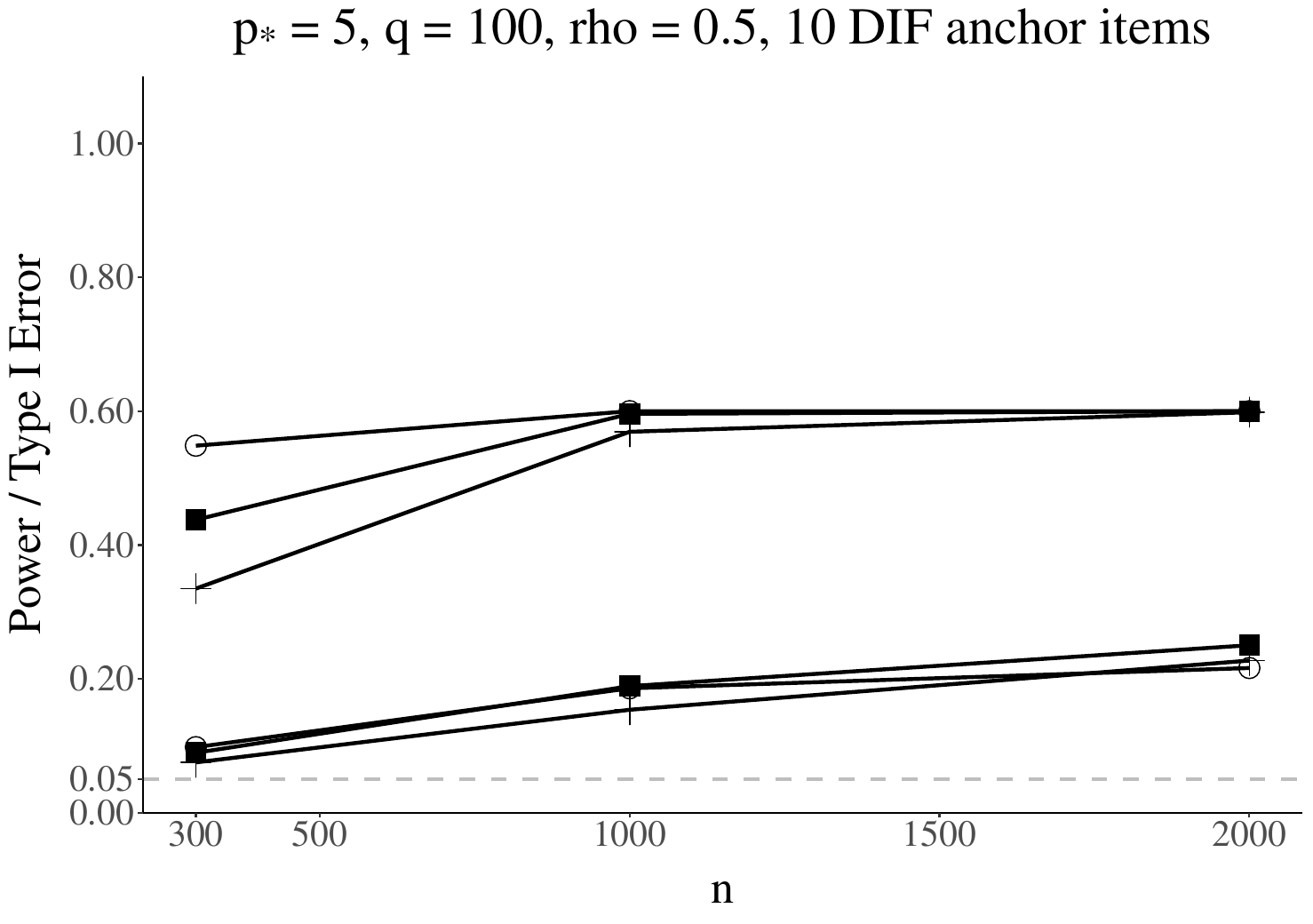}\\
         \includegraphics[width=2.5in]{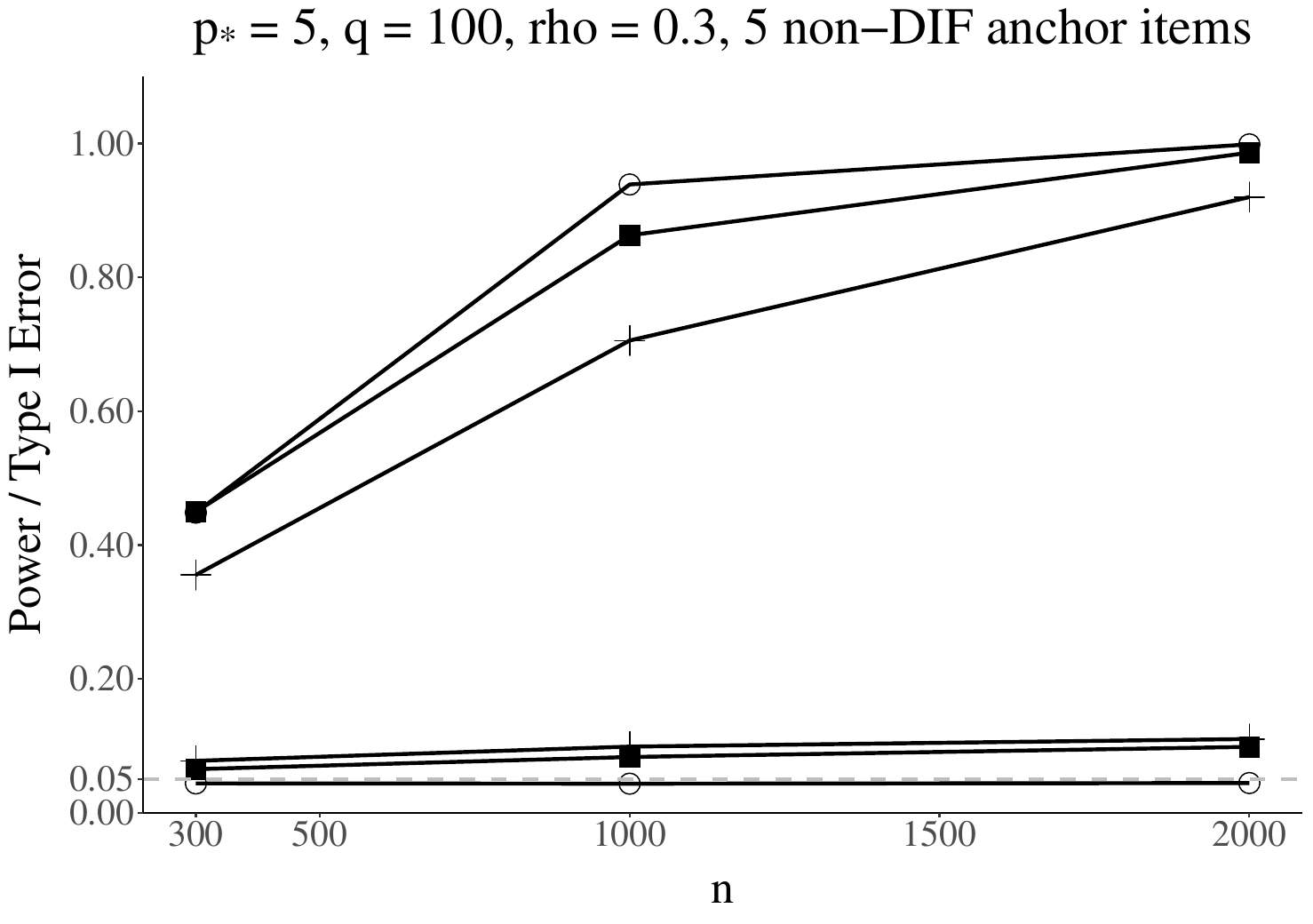}~
        \includegraphics[width=2.5in]{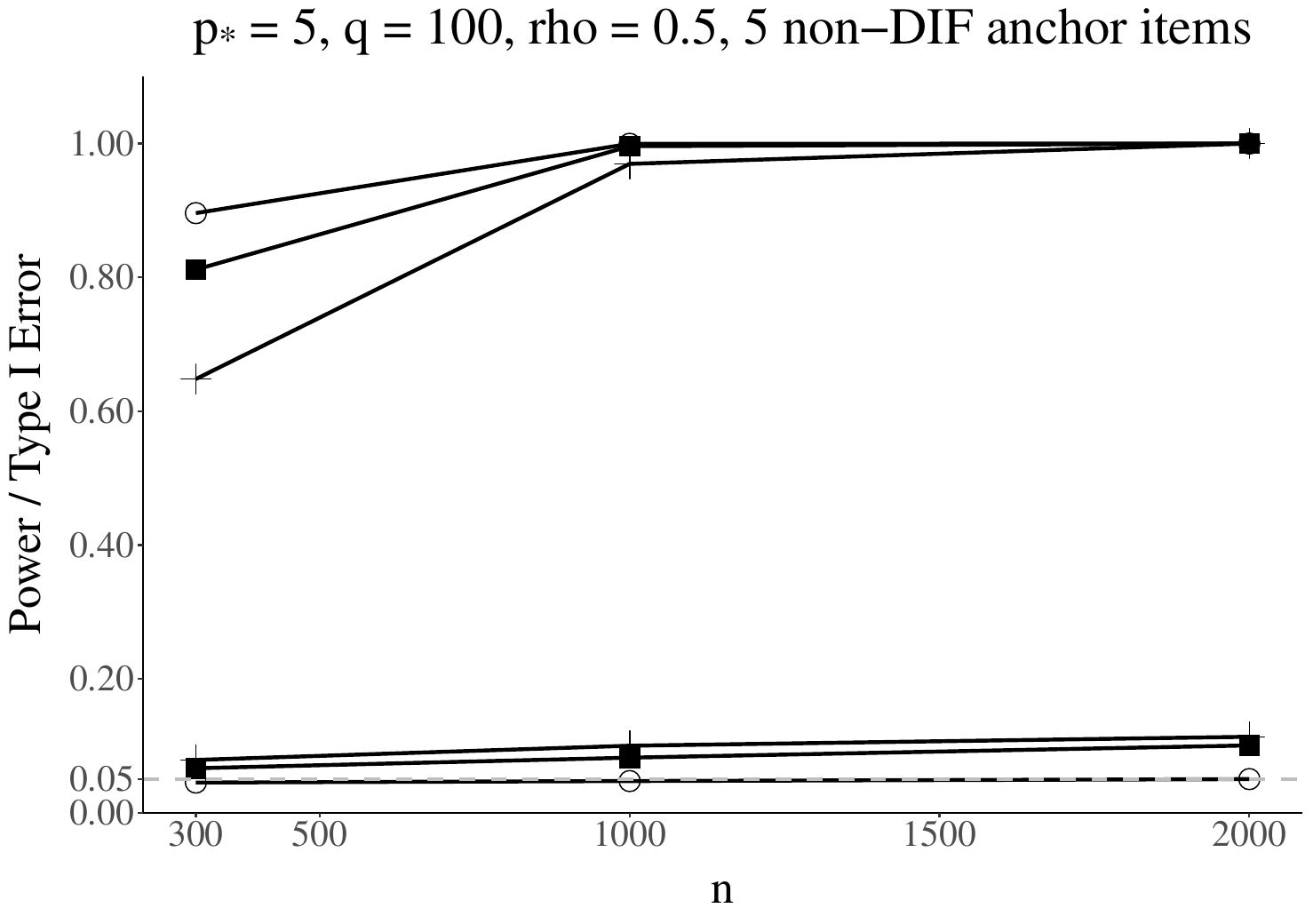}\\
         \includegraphics[width=2.5in]{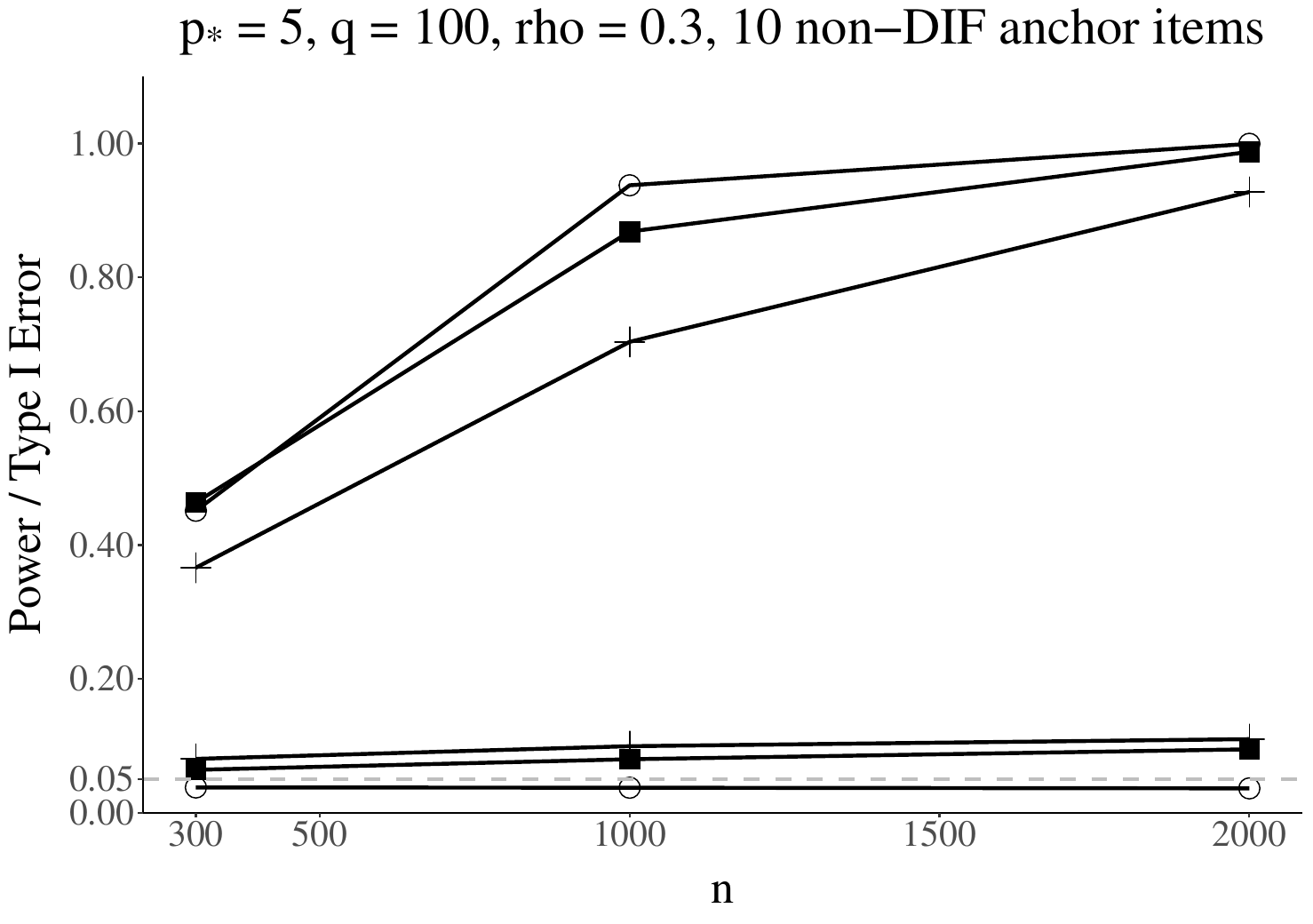}~
        \includegraphics[width=2.5in]{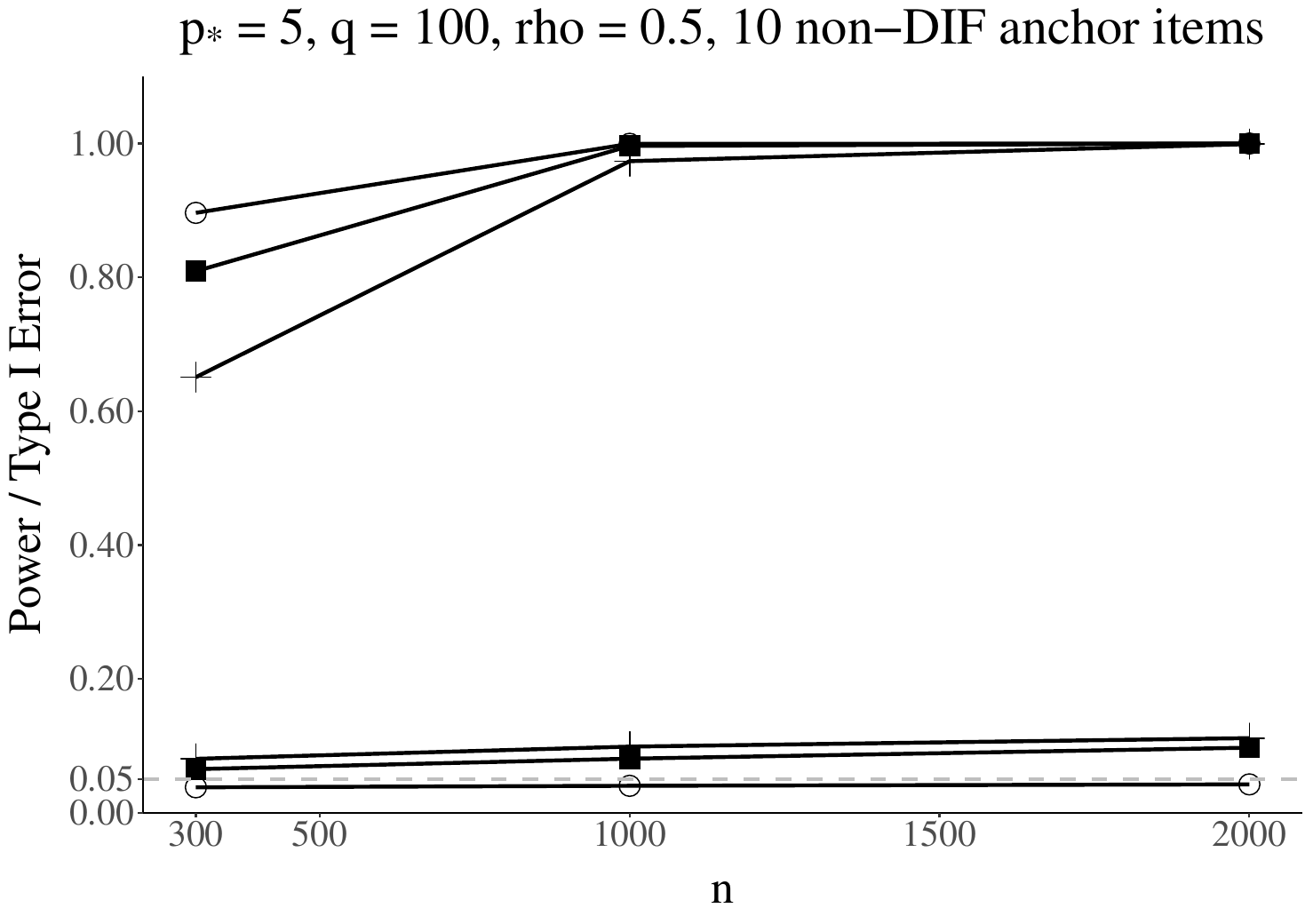}\\
    \caption{Powers and type I errors by anchor-item based estimation method and sparse setting at $p_*=5$. Circles (\protect\includegraphics[height=0.8em]{legend/new_rho0.png}) denote correlation parameter $\tau = 0$. Squares (\protect\includegraphics[height=0.8em]{legend/rho0.5.png}) indicate $\tau = 0.5$. Crosses (\protect\includegraphics[height=1em]{legend/rho0.7.png}) represent the $\tau = 0.7$.}
    \label{fig:7-anchor_item_method r2}
\end{figure}

\begin{figure}[ht]
\centering    
        \includegraphics[width=2.5in]{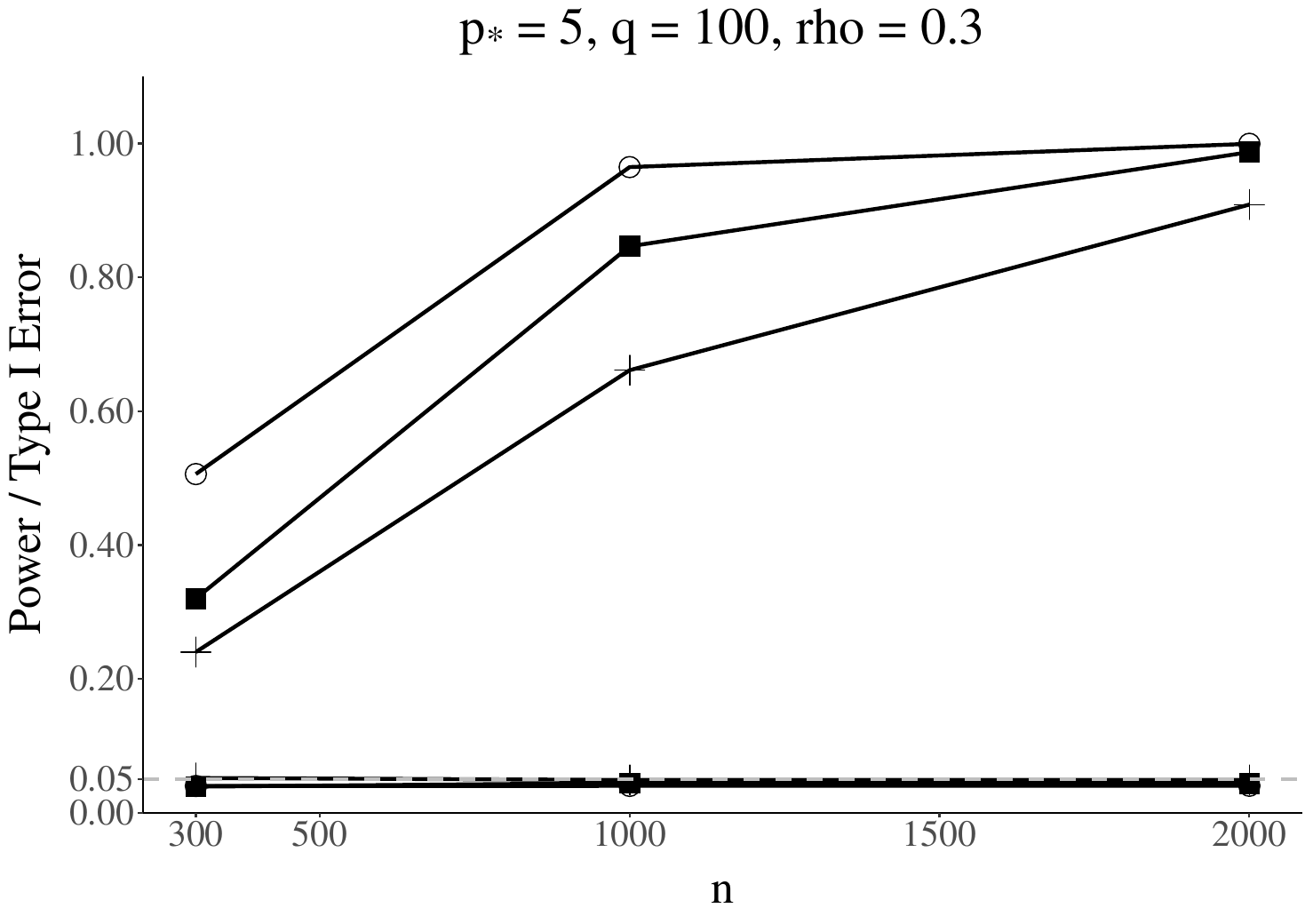}~
        \includegraphics[width=2.5in]{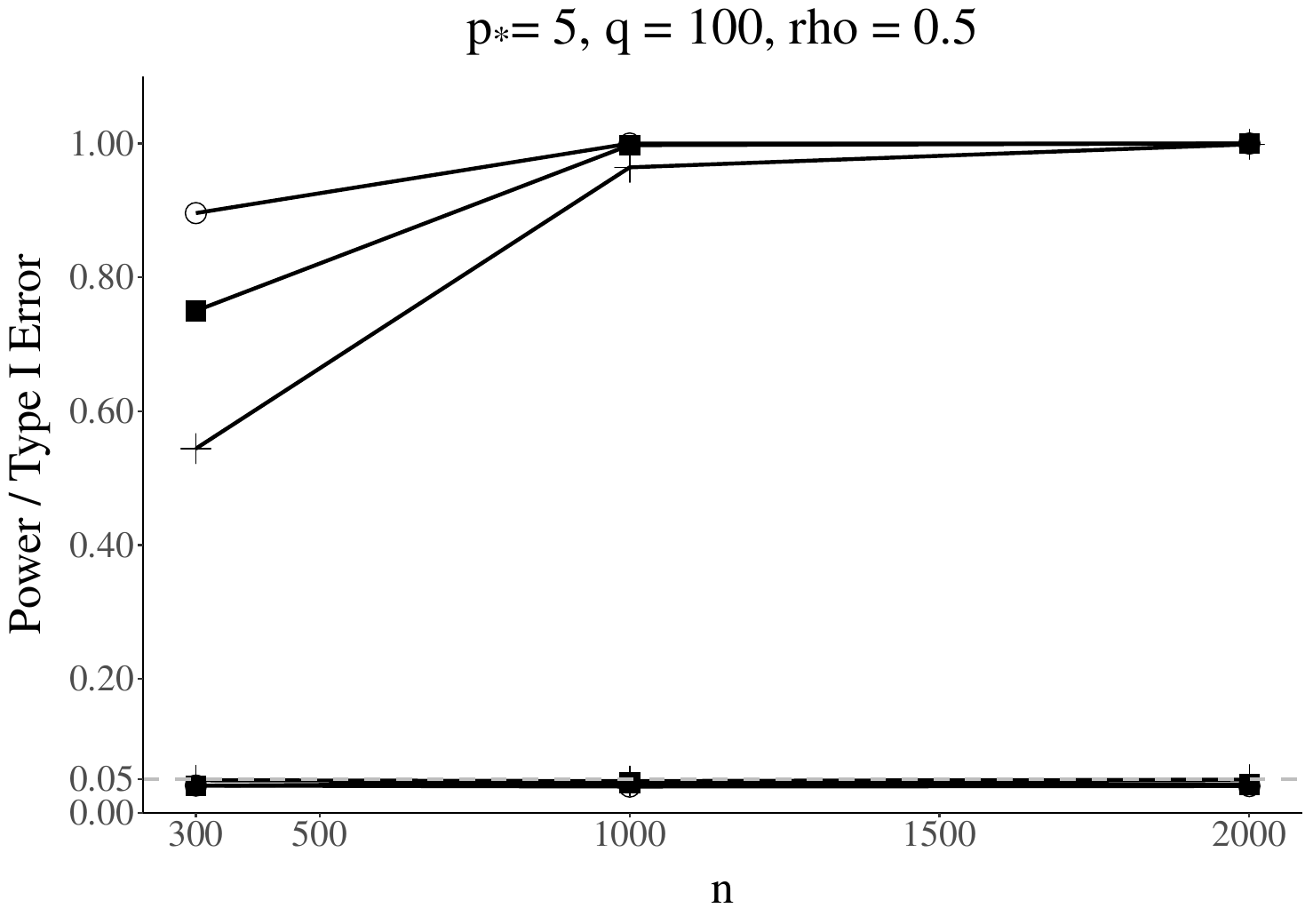}
    \caption{Powers and type I errors under sparse setting at $p_*=5$ and logistic framework. Circles (\protect\includegraphics[height=0.8em]{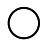}) denote correlation parameter $\tau = 0$. Squares (\protect\includegraphics[height=0.8em]{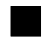}) indicate $\tau = 0.5$. Crosses (\protect\includegraphics[height=1em]{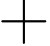}) represent the $\tau = 0.7$.}
    \label{fig:true logistic r2}
\end{figure} 

\begin{figure}[h!]
\centering    
        \includegraphics[width=2.5in]{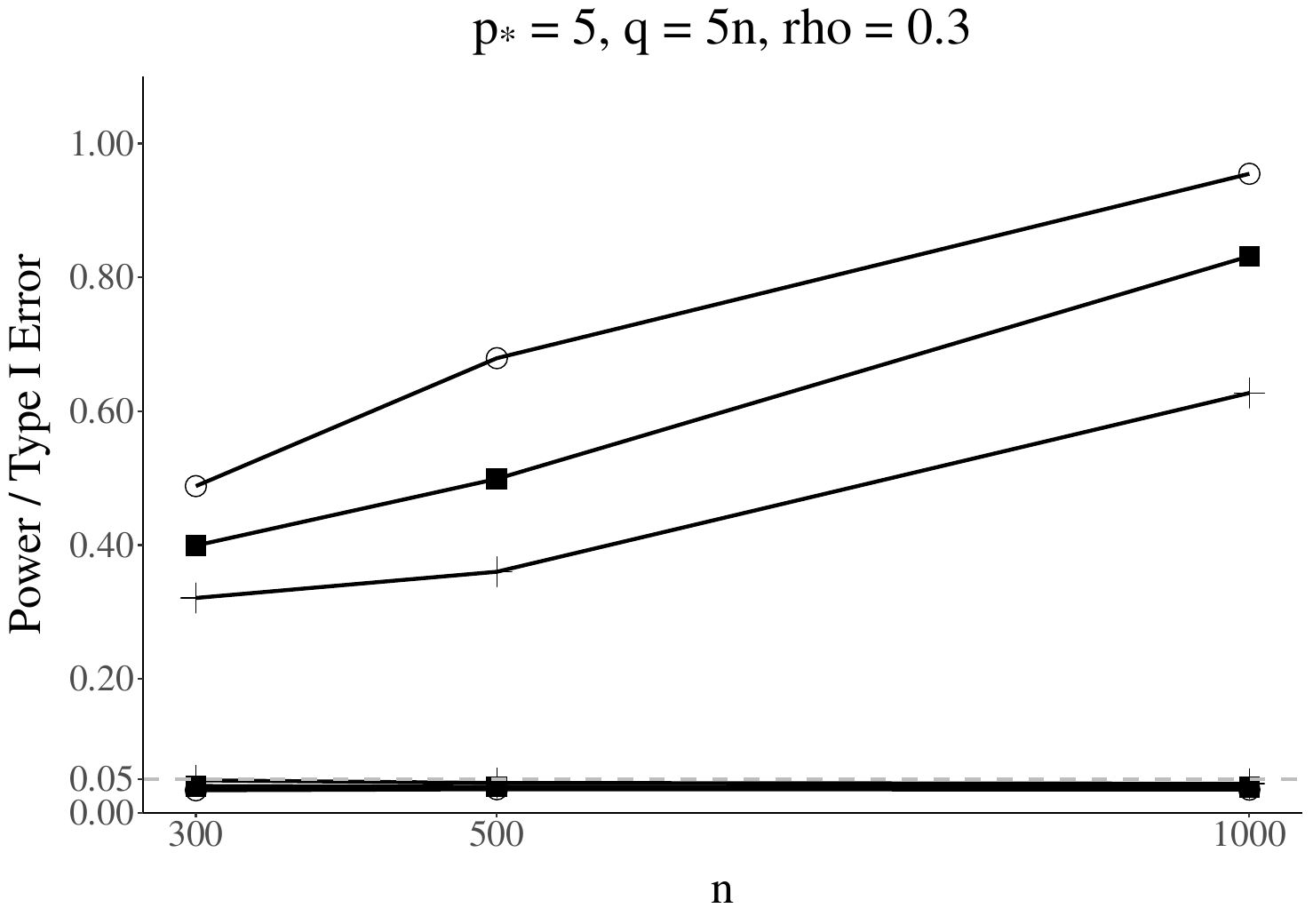}~
        \includegraphics[width=2.5in]{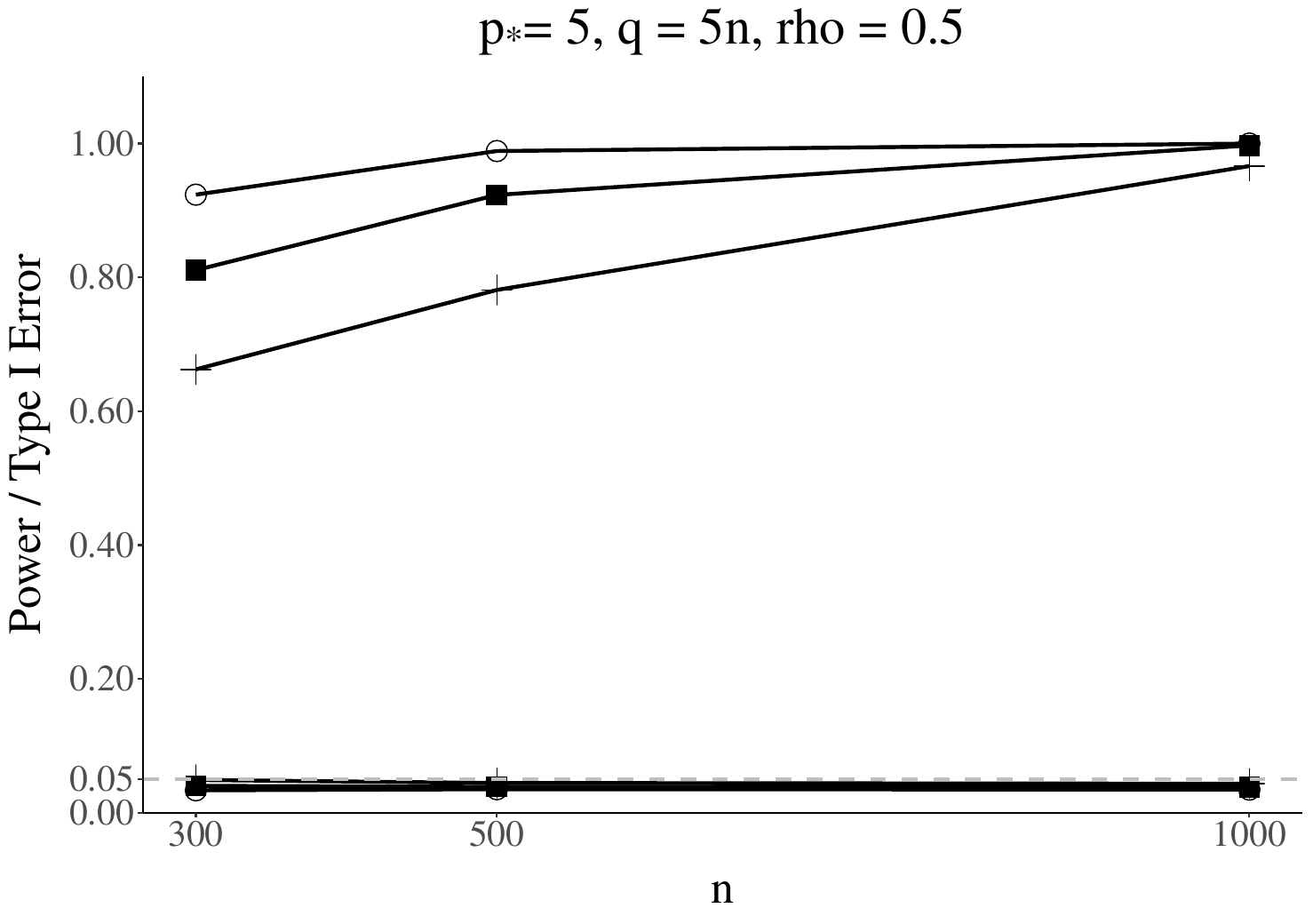}
    \caption{Powers and type I errors for sparse setting at $p_*=5$ and $q = 5n$. Circles (\protect\includegraphics[height=0.8em]{legend/new_rho0.png}) denote correlation parameter $\tau = 0$. Squares (\protect\includegraphics[height=0.8em]{legend/rho0.5.png}) indicate $\tau = 0.5$. Crosses (\protect\includegraphics[height=1em]{legend/rho0.7.png}) represent the $\tau = 0.7$.}
    \label{fig:8-relative_scale_nq_r3}
\end{figure} 

\begin{figure}[h]
\centering    
        \includegraphics[width=2.5in]{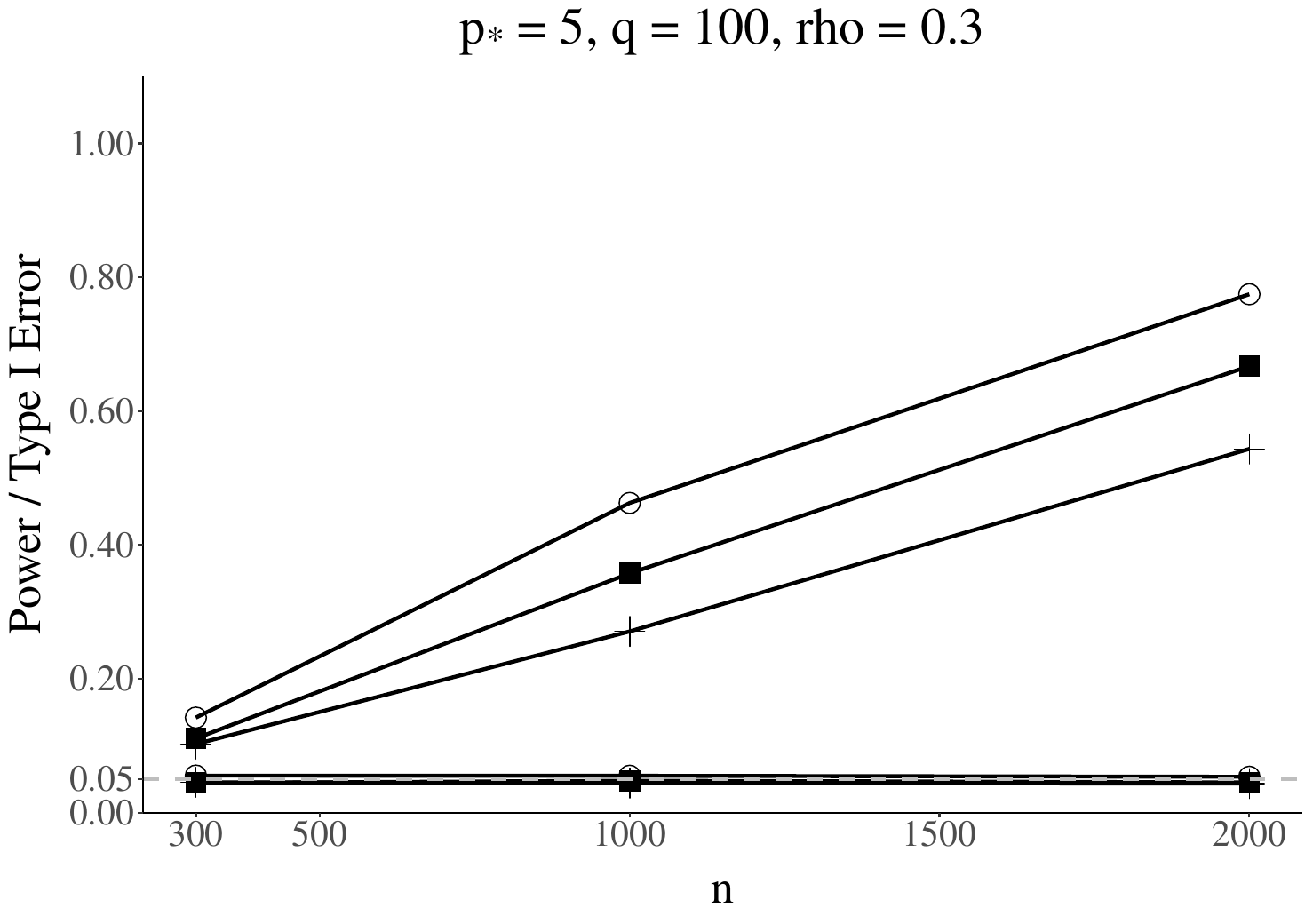}~
        \includegraphics[width=2.5in]{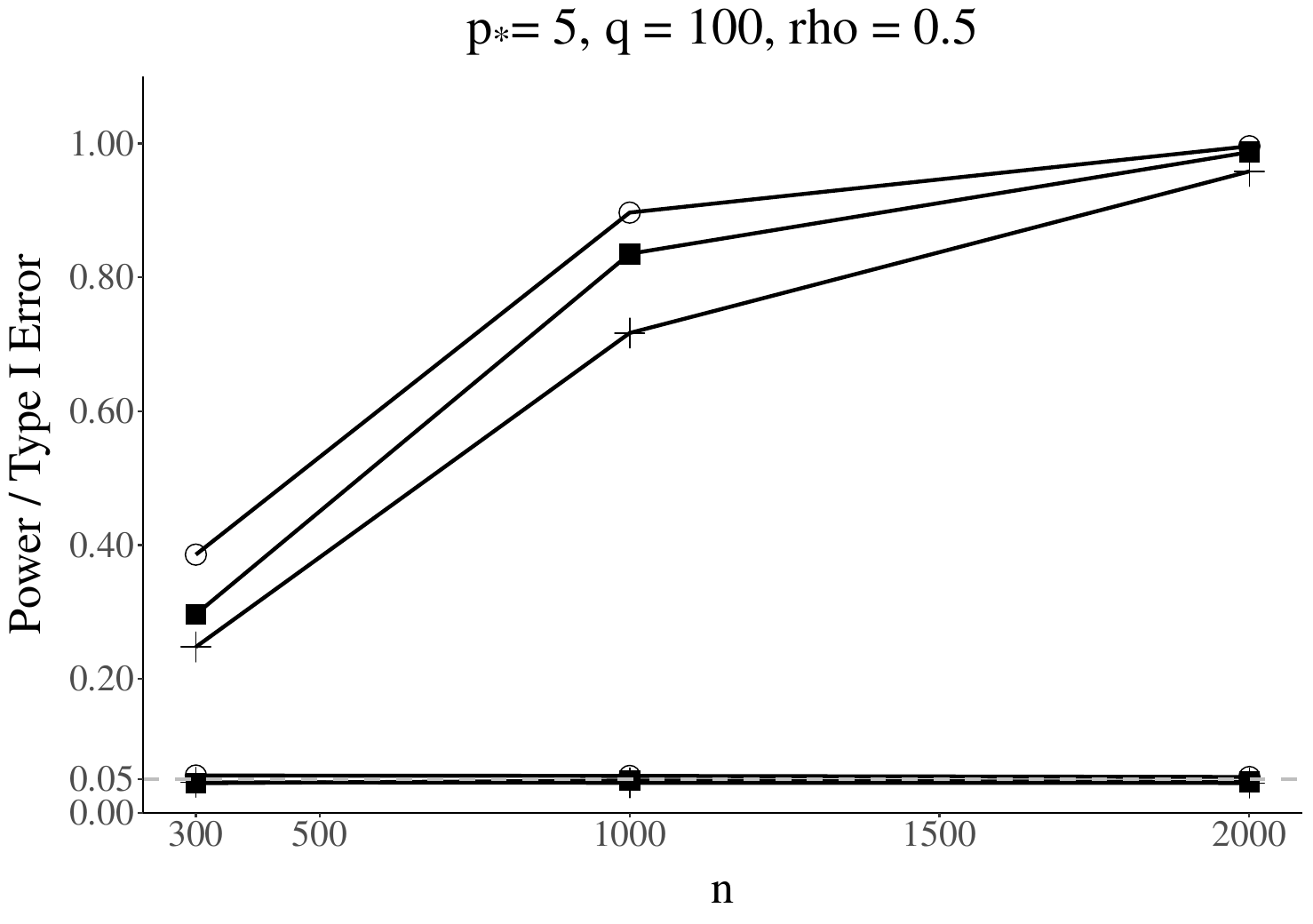}
    \caption{Powers and type I errors for sparse setting at $p_*=5$ and binary $\bX_i^c$. Circles (\protect\includegraphics[height=0.8em]{legend/new_rho0.png}) denote correlation parameter $\tau = 0$. Squares (\protect\includegraphics[height=0.8em]{legend/rho0.5.png}) indicate $\tau = 0.5$. Crosses (\protect\includegraphics[height=1em]{legend/rho0.7.png}) represent the $\tau = 0.7$.}
    \label{fig:5-binary_X r3}
\end{figure}

\begin{figure}[htbp]
\centering    
        \includegraphics[width=2.5in]{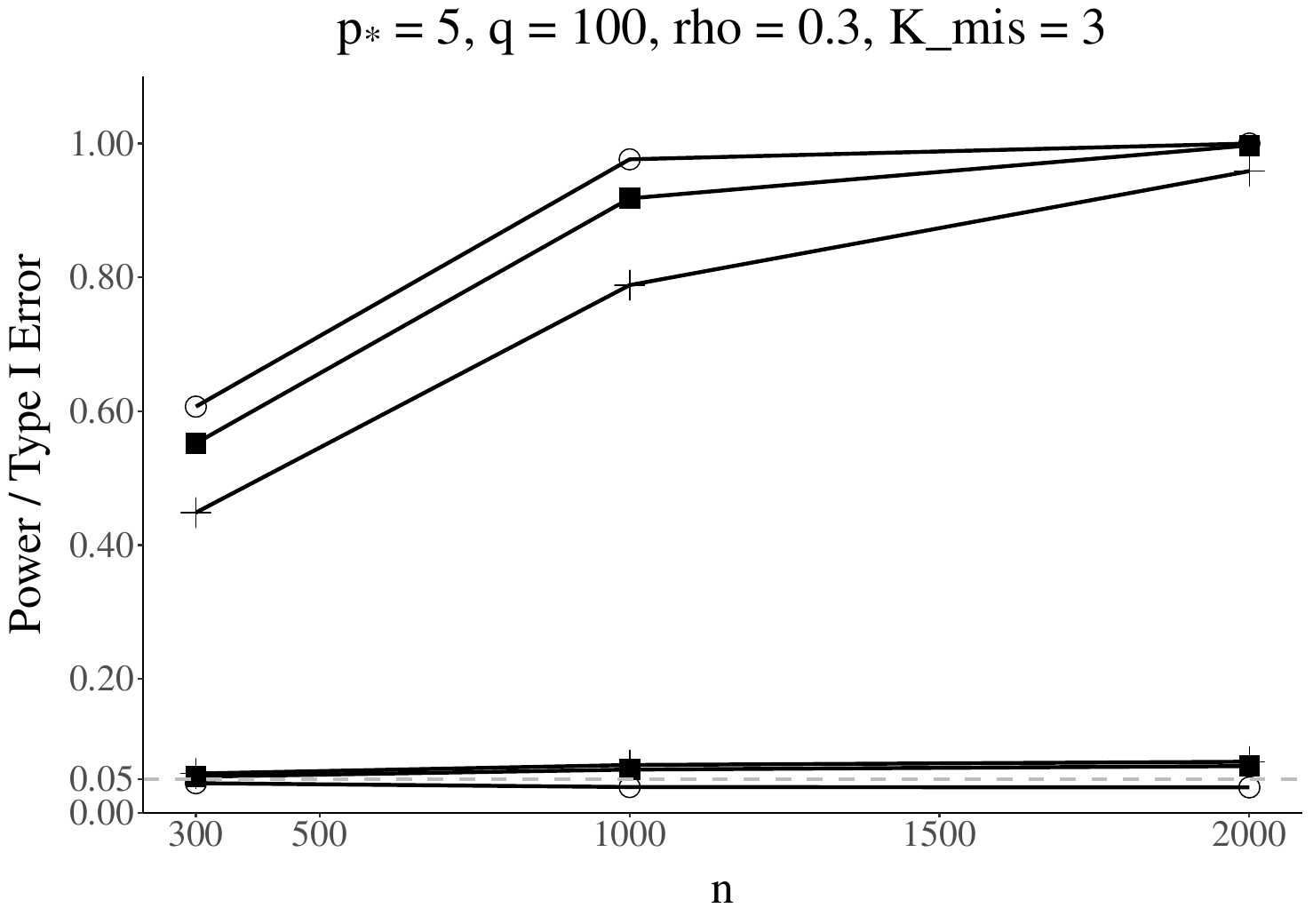}~
        \includegraphics[width=2.5in]{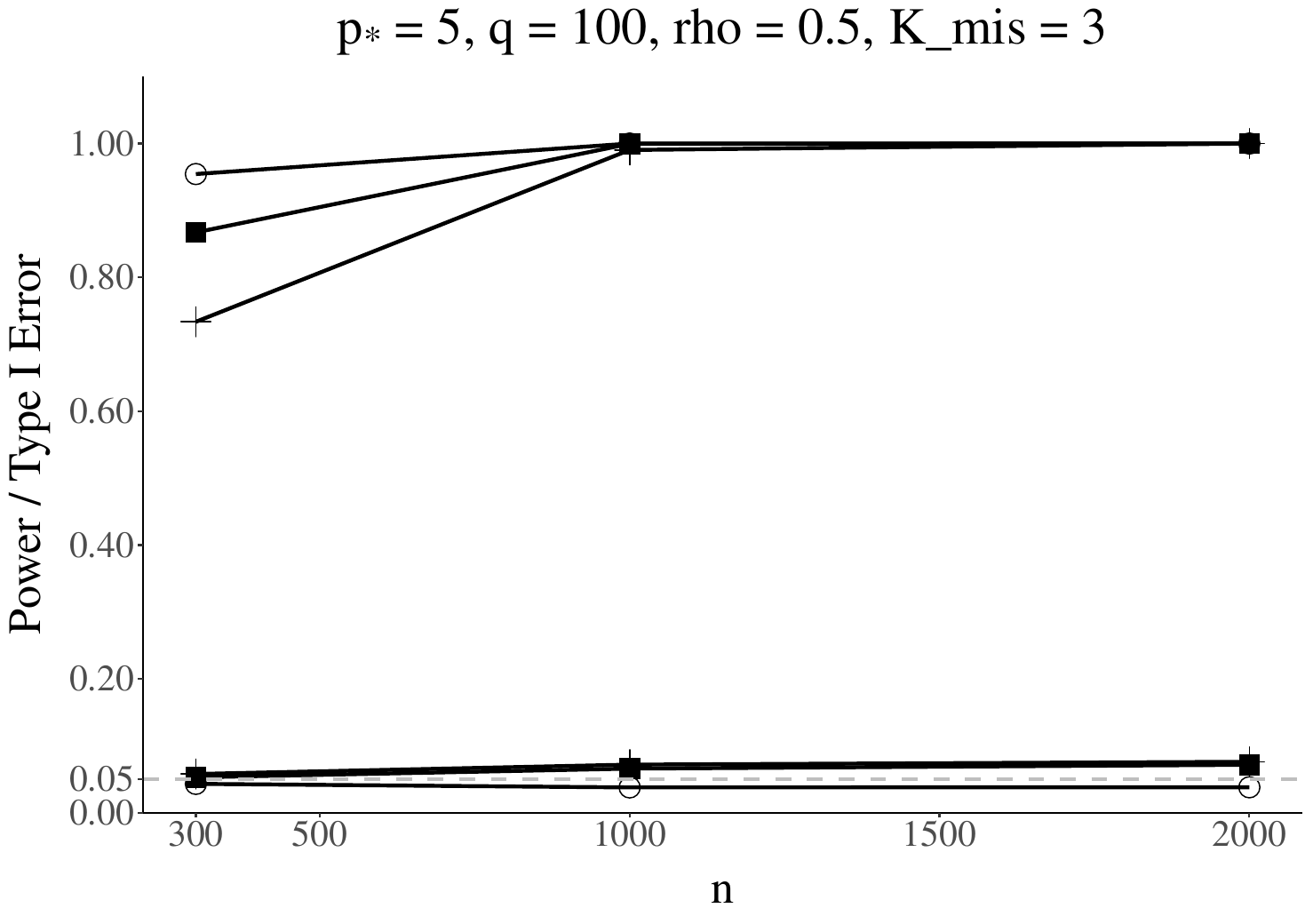}\\
         \includegraphics[width=2.5in]{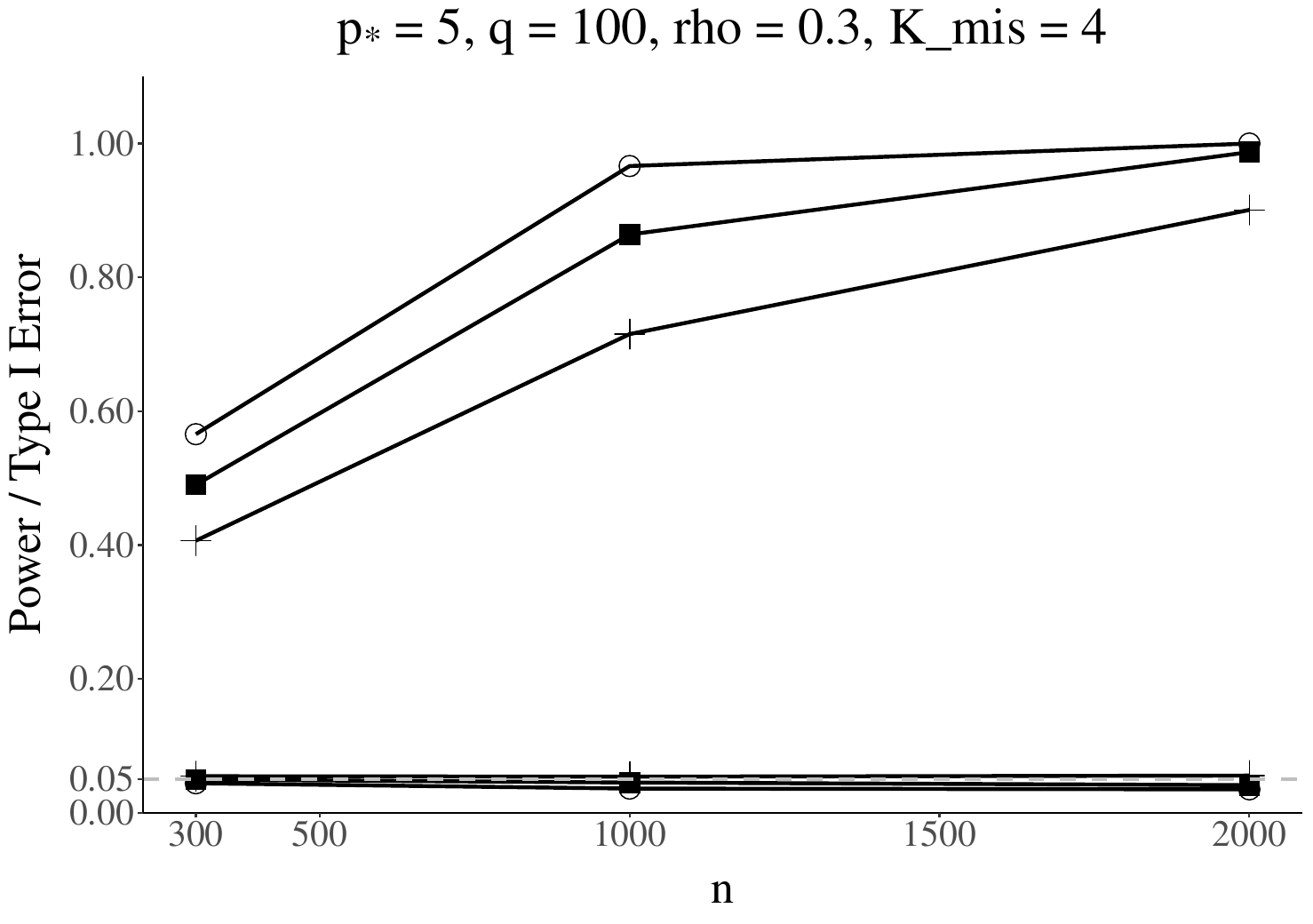}~
        \includegraphics[width=2.5in]{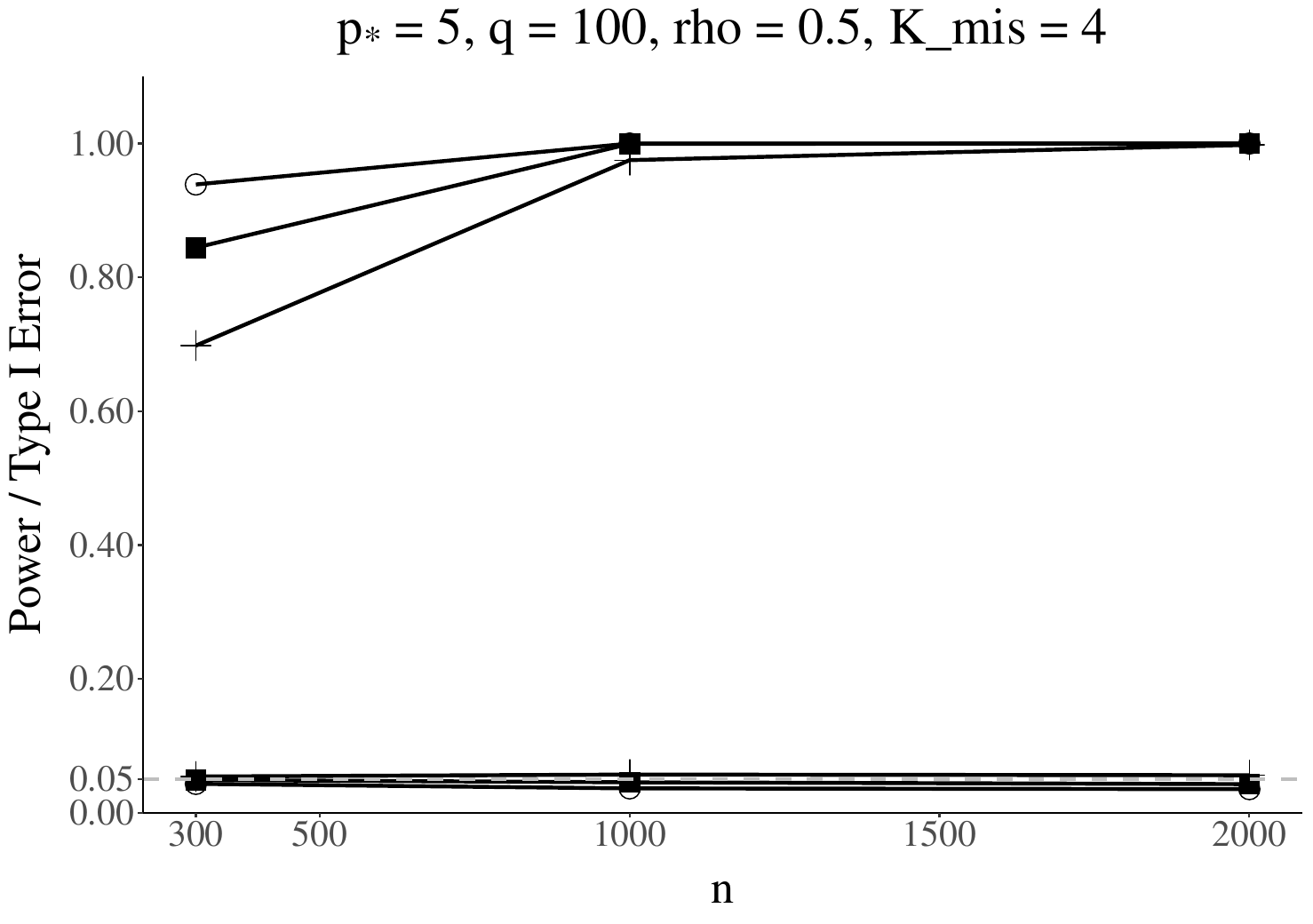}\\
         \includegraphics[width=2.5in]{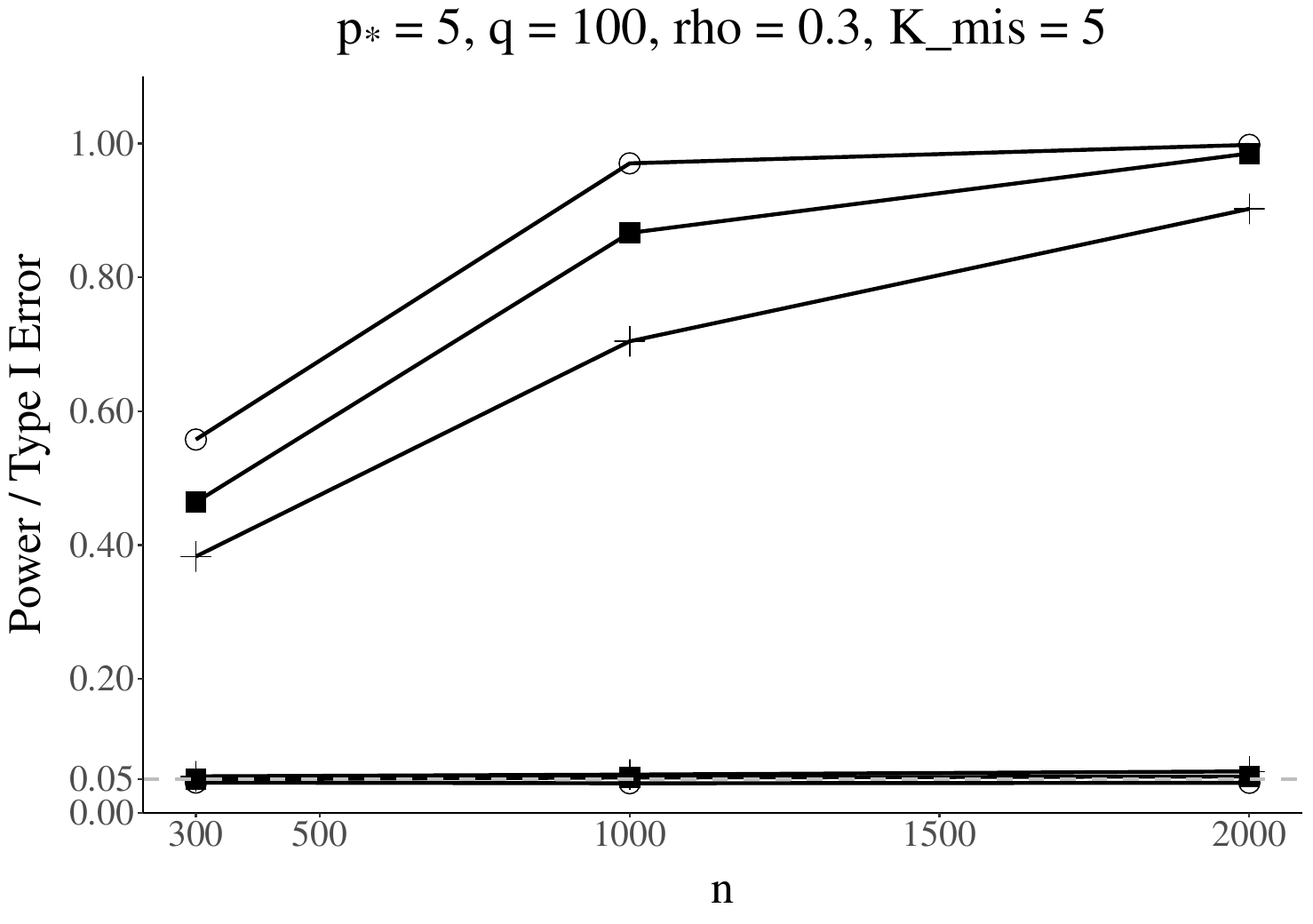}~
        \includegraphics[width=2.5in]{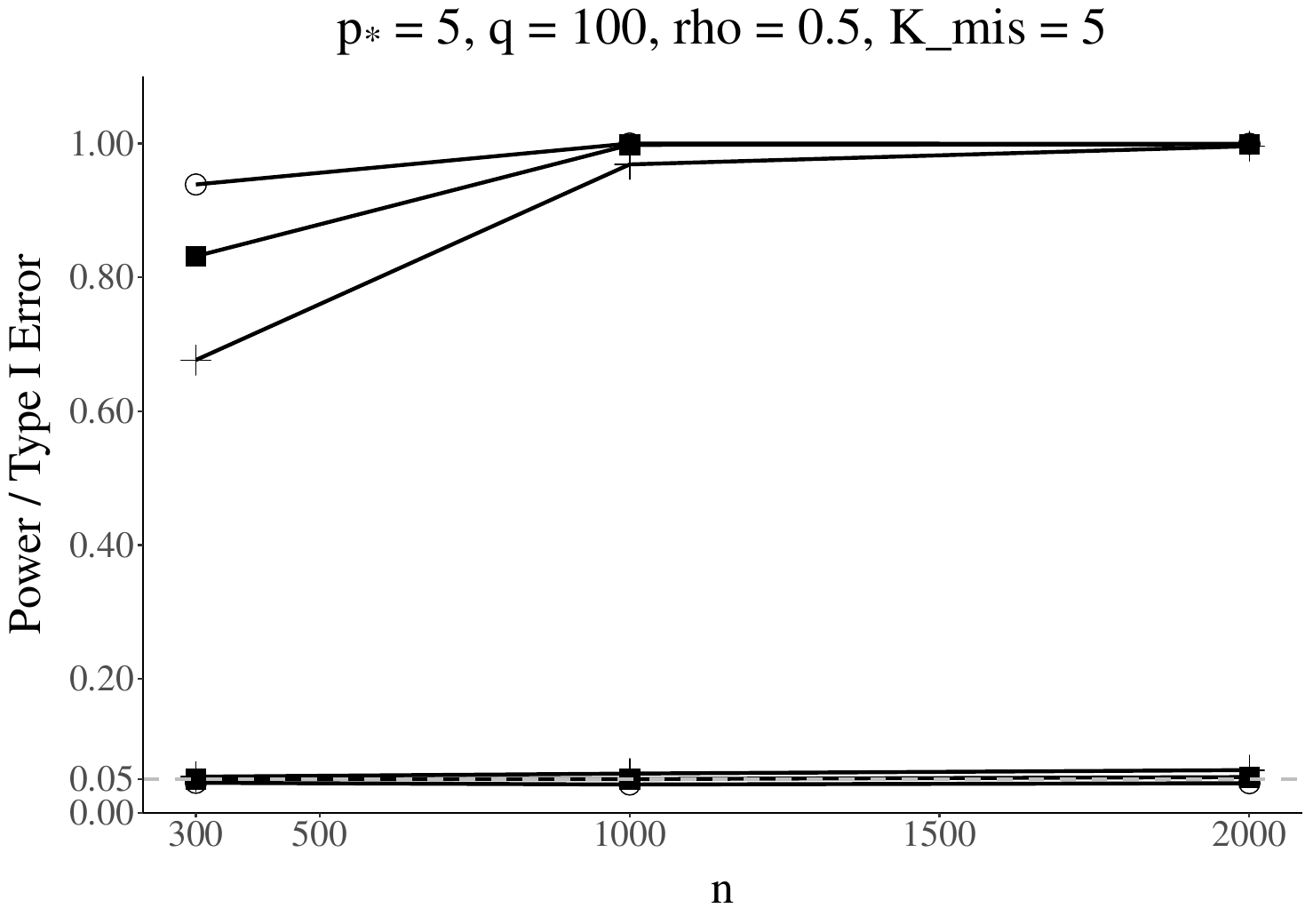}\\
         \includegraphics[width=2.5in]{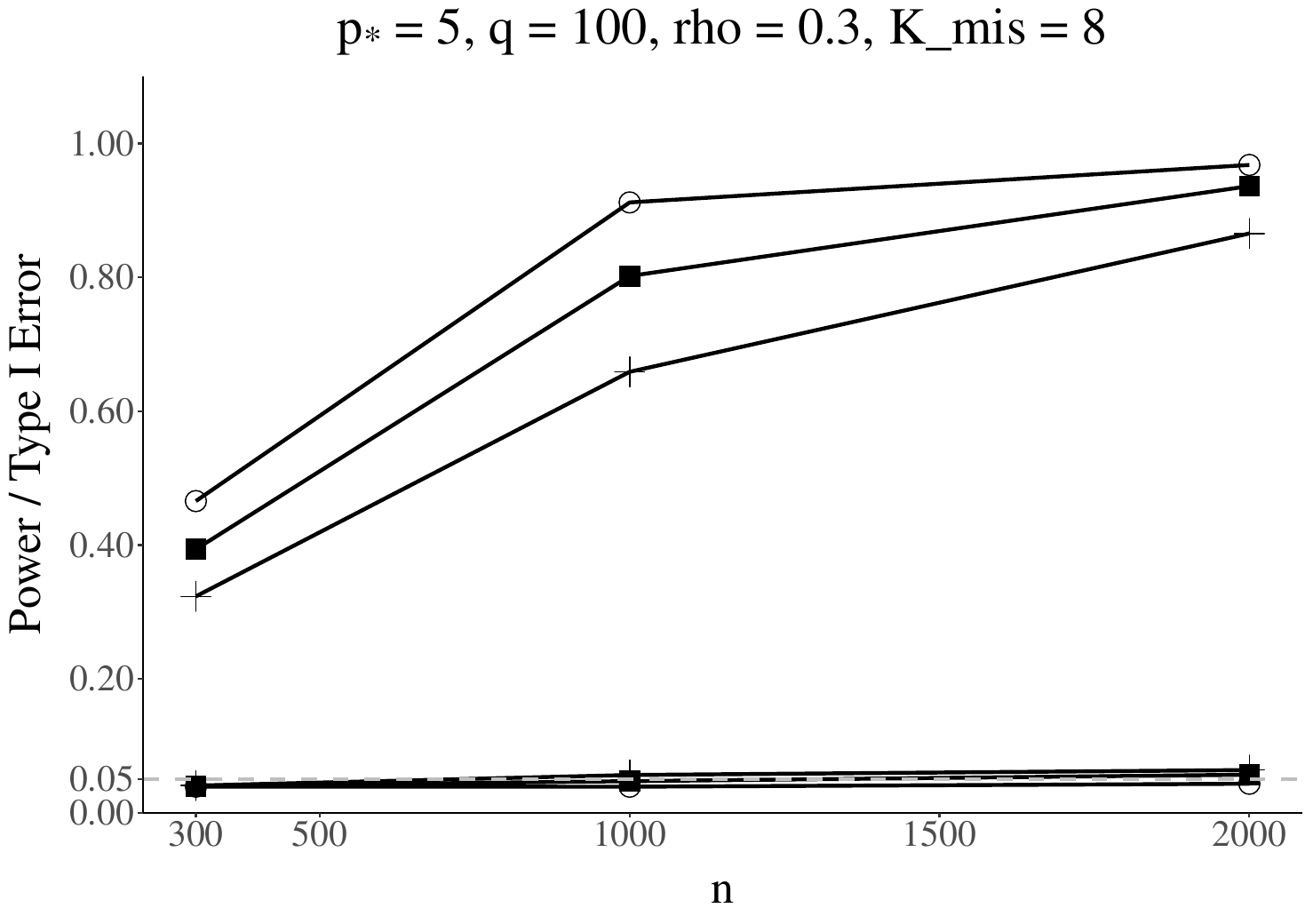}~
        \includegraphics[width=2.5in]{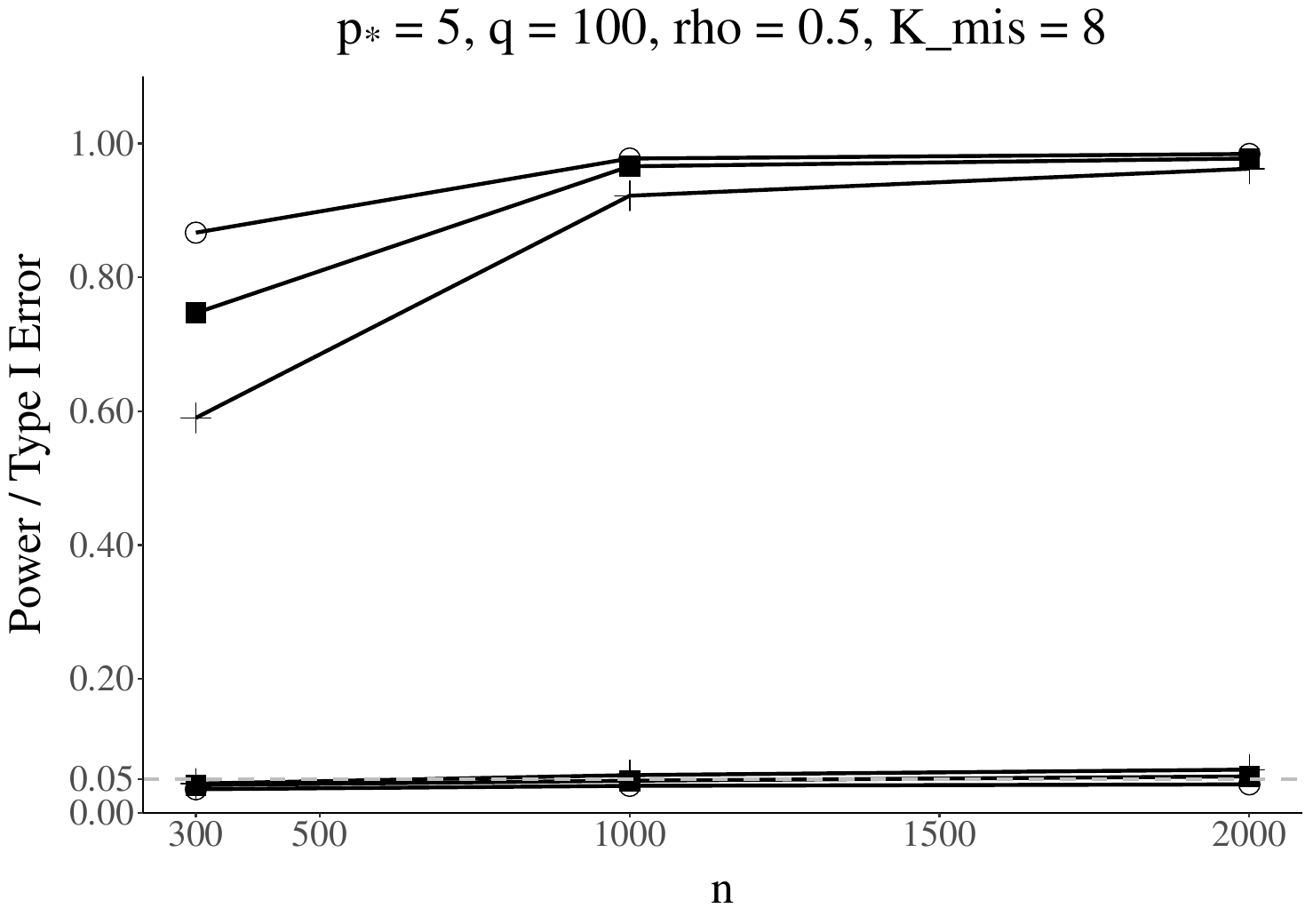}\\
    \caption{Powers and type I errors under misspecified number of latent factors ${K}_{mis} = 3, 4, 5,8$ and sparse setting at $p_*=5$. Circles (\protect\includegraphics[height=0.8em]{legend/new_rho0.png}) denote correlation parameter $\tau = 0$. Squares (\protect\includegraphics[height=0.8em]{legend/rho0.5.png}) indicate $\tau = 0.5$. Crosses (\protect\includegraphics[height=1em]{legend/rho0.7.png}) represent the $\tau = 0.7$.}
    \label{fig:6-sensitivity analysis}
\end{figure}

\begin{figure}[htbp]
\centering    
        \includegraphics[width=2.5in]{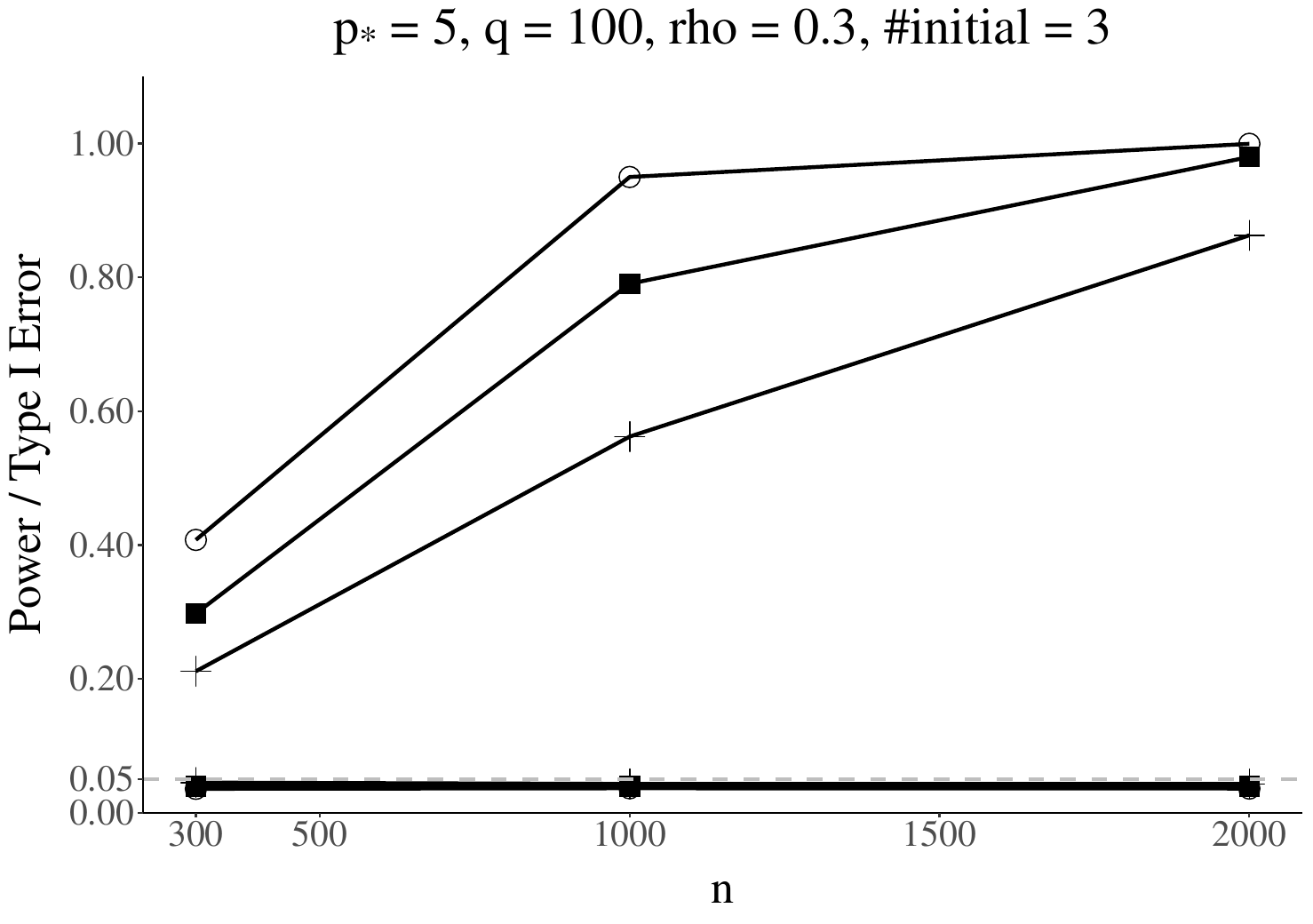}~
        \includegraphics[width=2.5in]{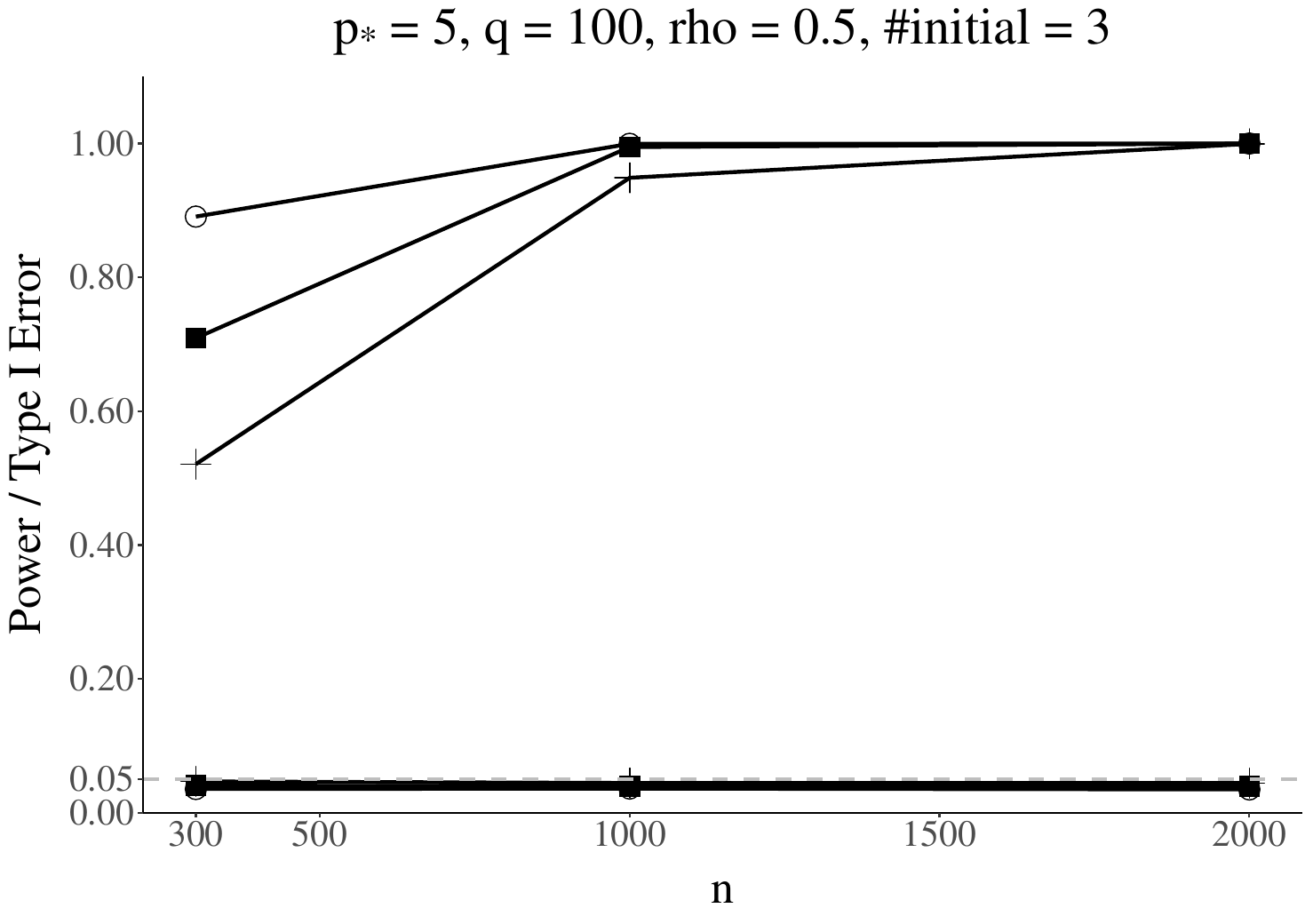}\\
        \includegraphics[width=2.5in]{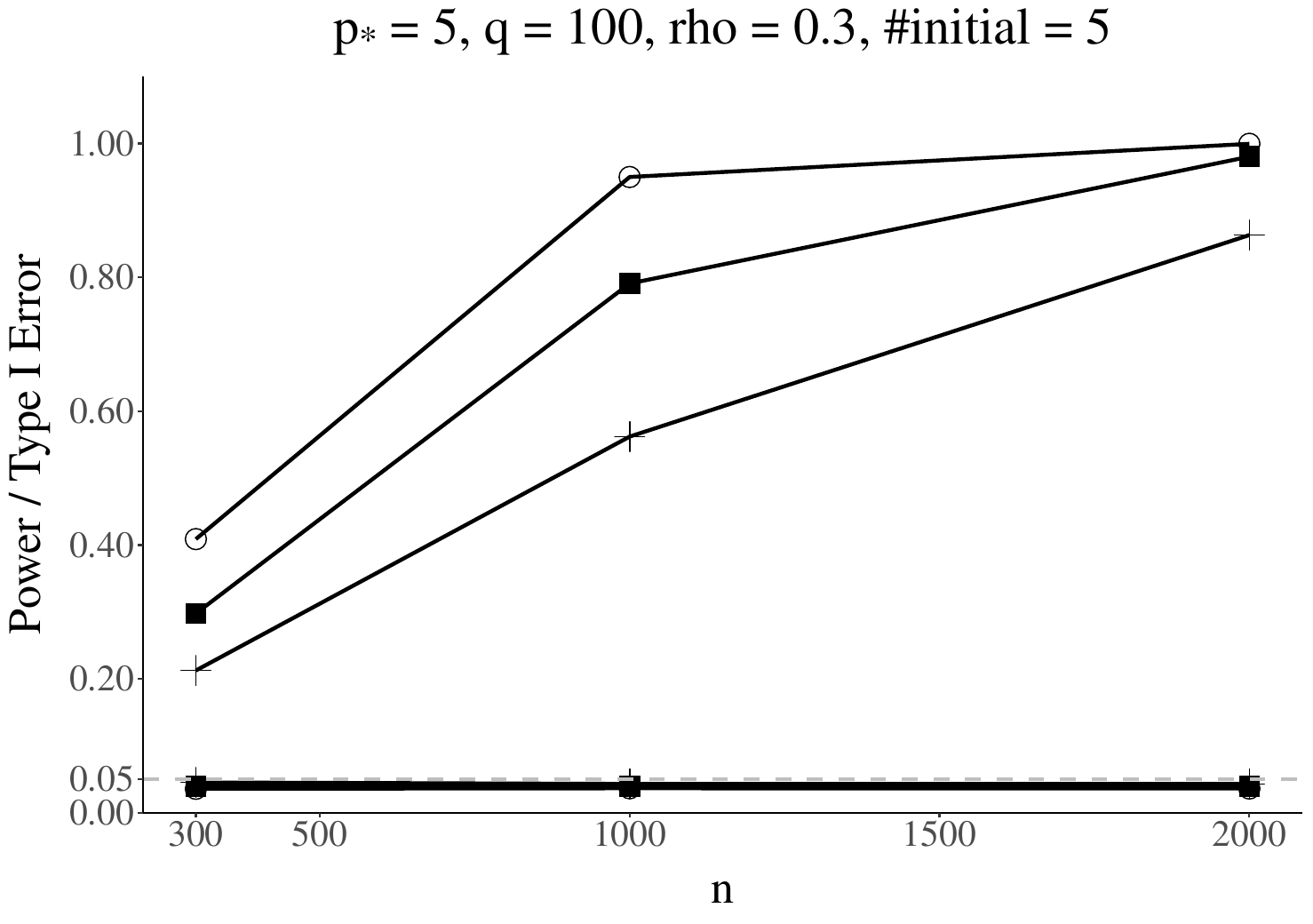}~
        \includegraphics[width=2.5in]{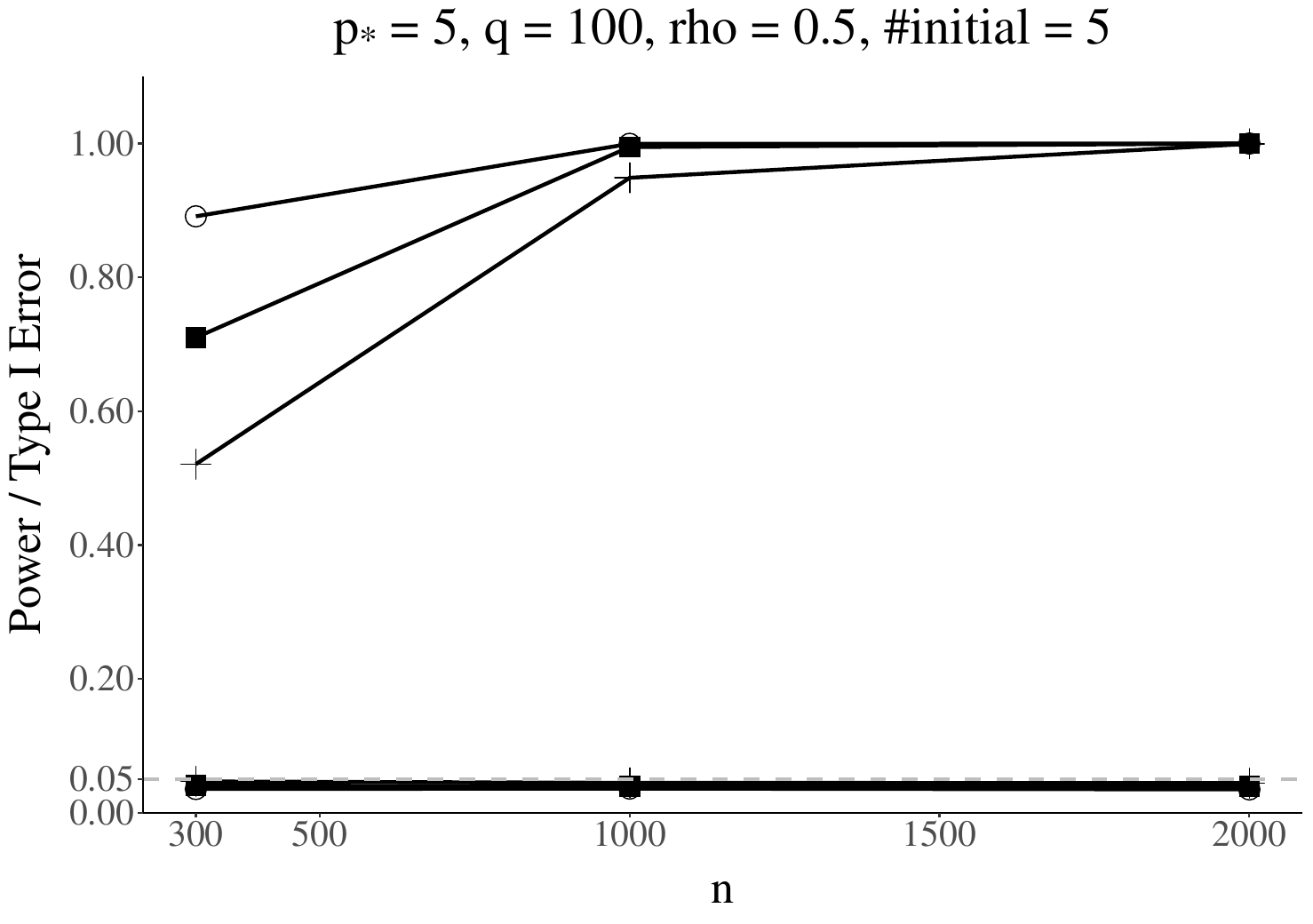}\\
        \includegraphics[width=2.5in]{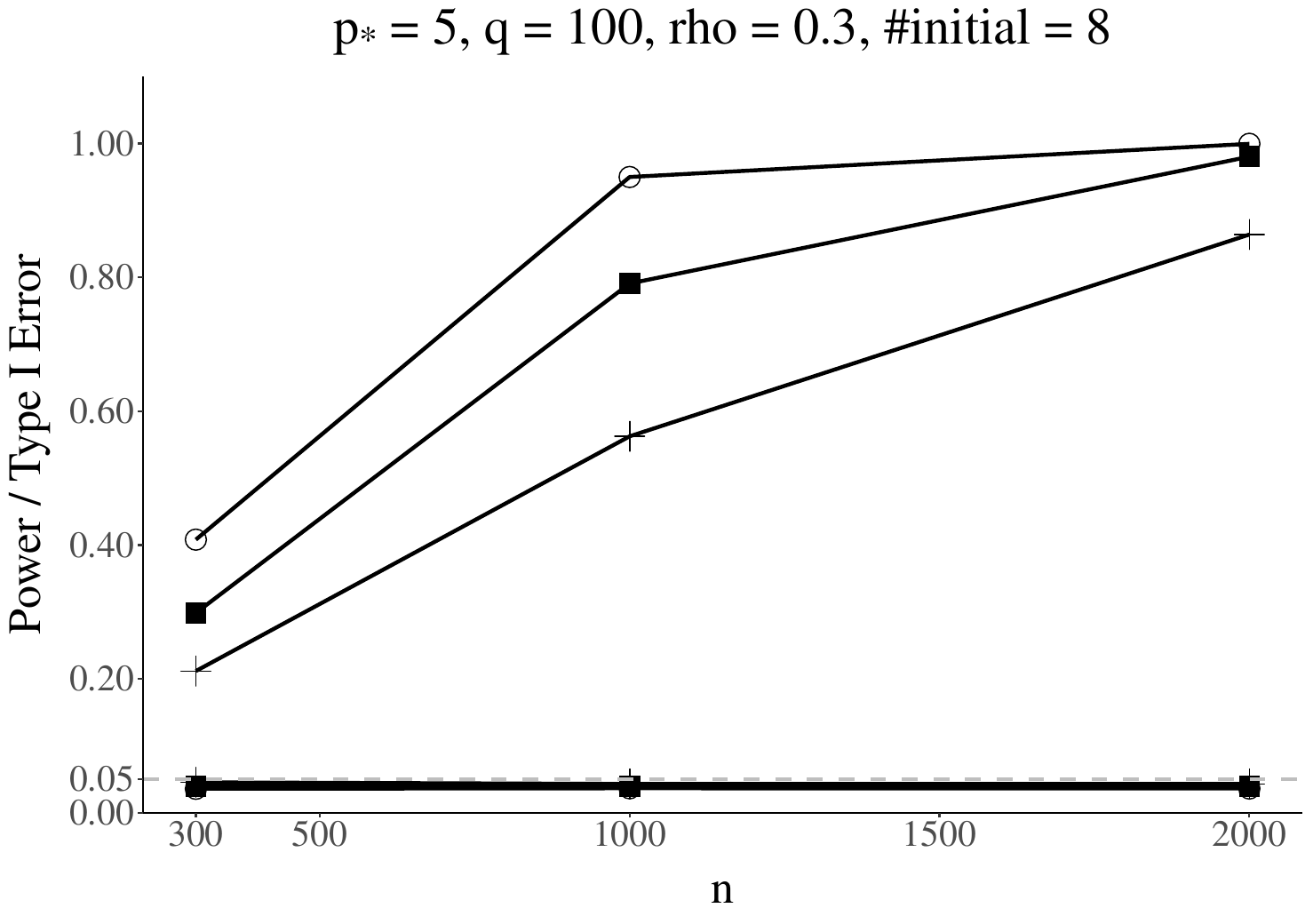}~
        \includegraphics[width=2.5in]{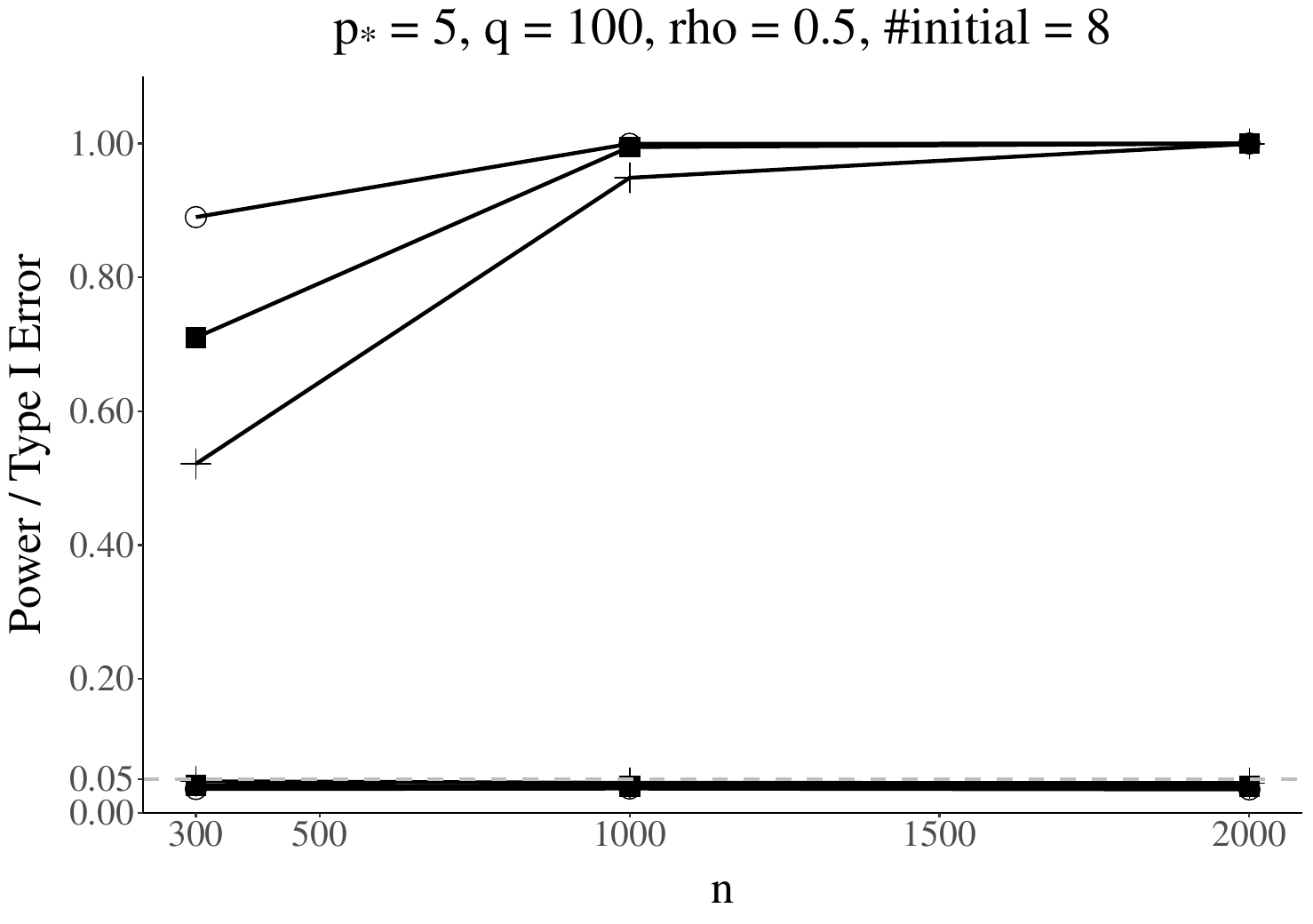}\\
    \caption{Powers and type I errors under different number of initial points $3, 5, 8$ and sparse setting at $p_*=5$. Circles (\protect\includegraphics[height=0.8em]{legend/new_rho0.png}) denote correlation parameter $\tau = 0$. Squares (\protect\includegraphics[height=0.8em]{legend/rho0.5.png}) indicate $\tau = 0.5$. Crosses (\protect\includegraphics[height=1em]{legend/rho0.7.png}) represent the $\tau = 0.7$.}
    \label{fig:11-num_initial}
\end{figure}

\begin{figure}[htbp]
\centering    
        \includegraphics[width=2.5in]{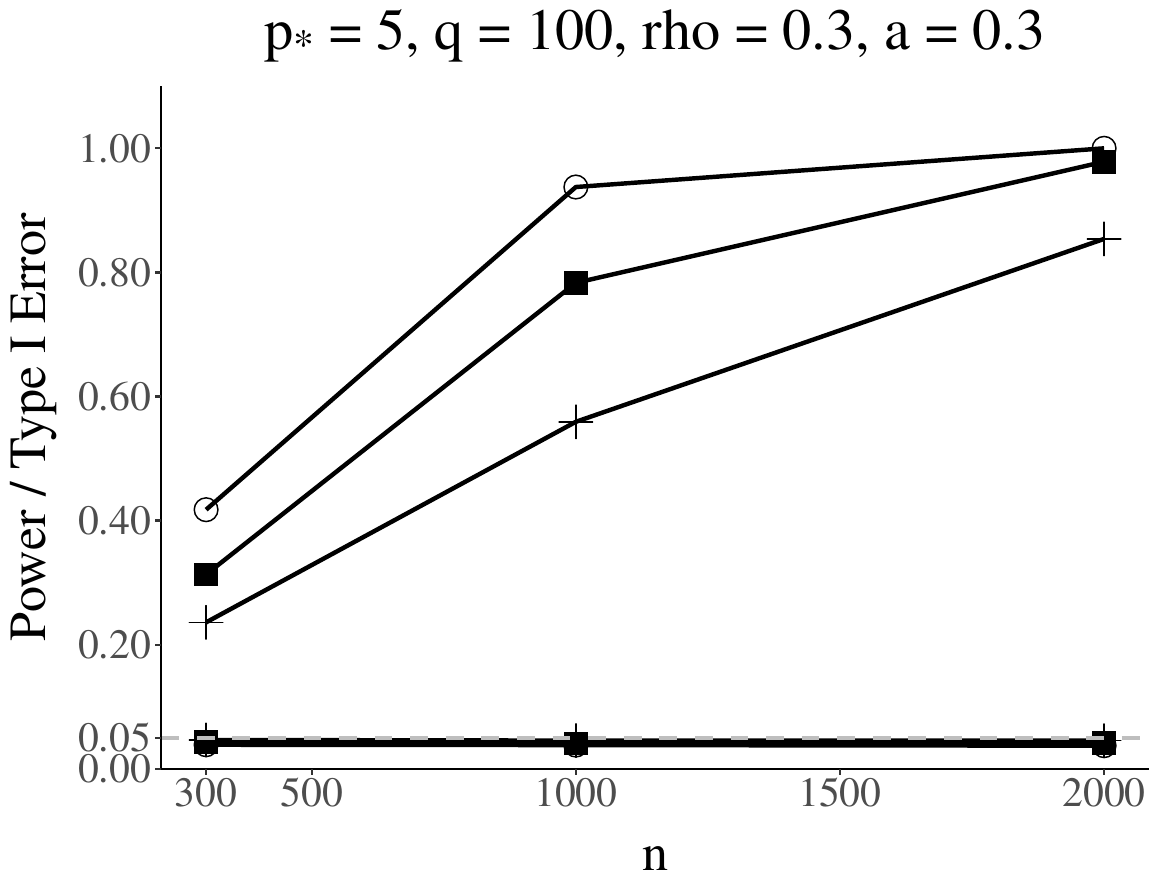}~
        \includegraphics[width=2.5in]{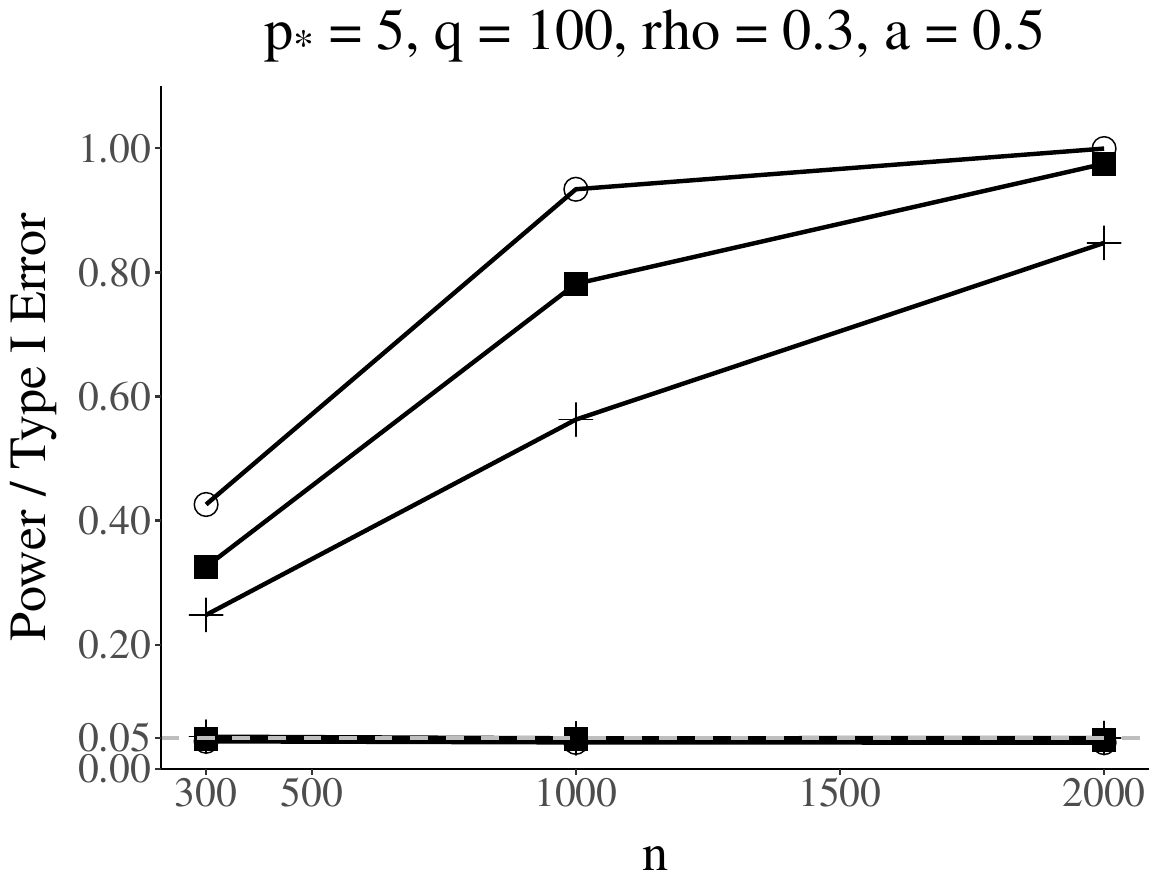}\\
        \includegraphics[width=2.5in]{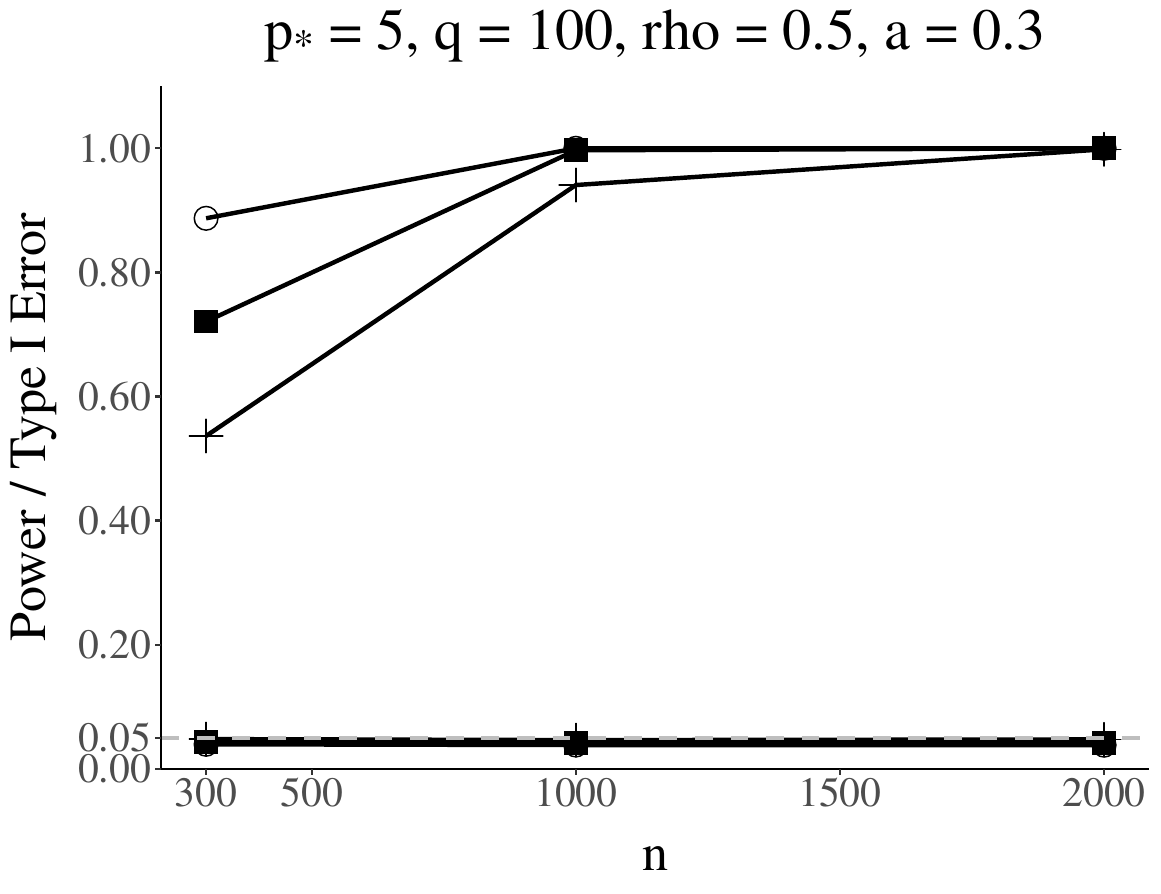}~
        \includegraphics[width=2.5in]{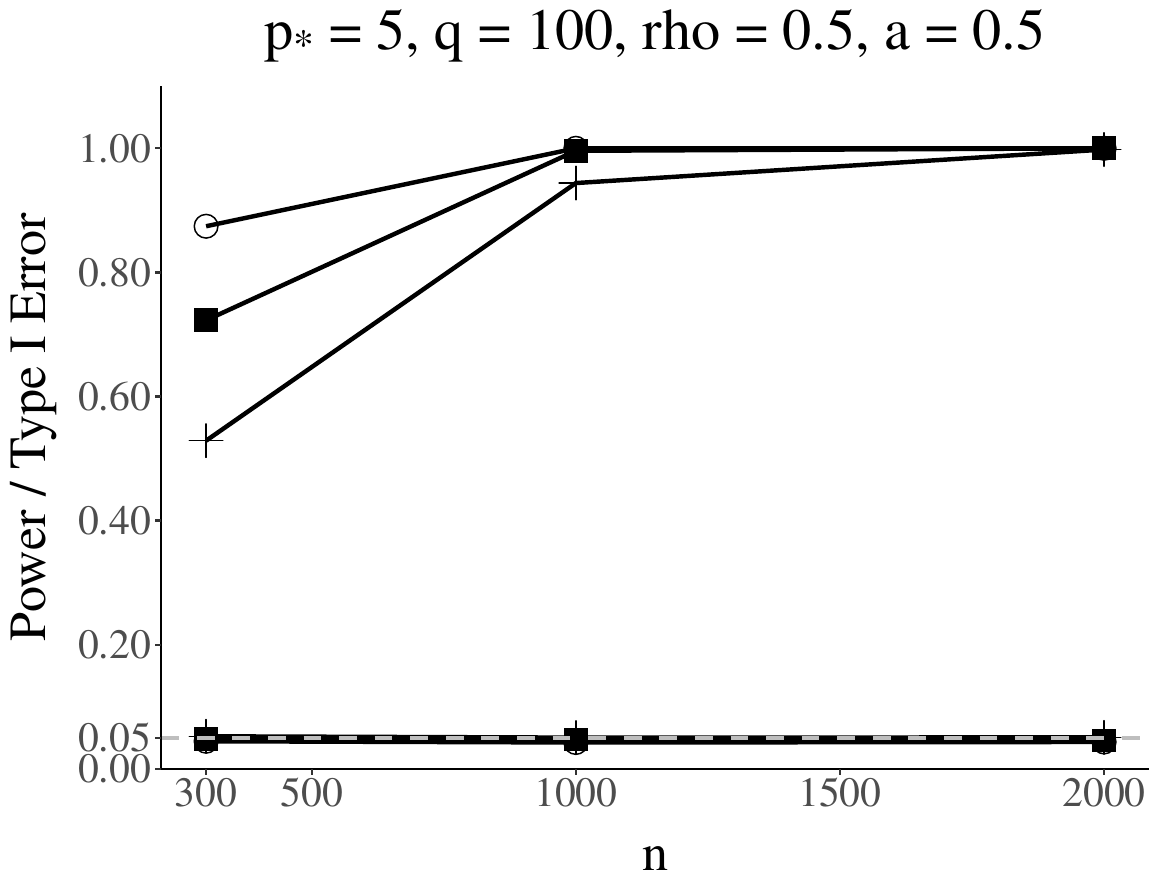}
    \caption{Powers and type I errors under full DIF model and sparse setting at $p_*=5$. Circles (\protect\includegraphics[height=0.8em]{legend/new_rho0.png}) denote correlation parameter $\tau = 0$. Squares (\protect\includegraphics[height=0.8em]{legend/rho0.5.png}) indicate $\tau = 0.5$. Crosses (\protect\includegraphics[height=1em]{legend/rho0.7.png}) represent the $\tau = 0.7$.}
    \label{fig:13-marginal_full_DIF}
\end{figure}

\clearpage

\section{Additional Data Application Results}

\subsection{Remark on Accommodating Missing Response Data}
In the data application to PISA, we extend our proposed methods and theoretical results to accommodate missing data. We note that under commonly studied missing pattern, such as when the missingness status indicators are independently and identically distributed Bernoulli random variables~\citep{davenport20141}, follow non-uniform distributions~\citep{cai2013max}, or follow flexible missing-entry scheme that generalizes beyond random sampling scheme~\citep{chen2023statistical},
then our results can be easily extended by maximizing the joint likelihood function that accounts for the missingness status random variables as they are independent of the response distribution given the latent factors and other model parameters.

Specifically, we first modify the joint log-likelihood function defined in Equation (2) of Section 2 in the main text to $$L^{obs}(\Yb \mid \bGamma, \Ub, \Bb,\Xb) = \sum_{i=1}^n \sum_{j \in \cQ_i} l_{ij}(\bgamma_j^{\intercal}\bU_i  + \bbeta_j^{\intercal}\bX_i),$$
where $\cQ_i$ denotes the set of items for which responses from student $i$ are observed. Following the same approach as in Equation (5) in Section 3.2 of the main text, we define the estimator based on the observed data by maximizing the modified log-likelihood as
\begin{equation*}
    \hat\bphi^{obs} = \mathop{\arg\max}_{\phi \in \cB(D)} L^{obs}(\bY\mid \bphi,\Xb)
\end{equation*}
with $\hat\bphi^{obs} = (\hat\bGamma^{obs}, \hat\bU^{obs},\hat\bB^{obs})$. By applying the same transformation steps as in Equations (6)--(8) (with $\hat\bphi$ replaced by $\hat\bphi^{obs}$), we obtain the final estimator under the missing data setting, denoted by $\hat\bphi^{obs}_* = (\hat\bGamma^{obs}_*, \hat\bU^{obs}_*,\hat\bB^{obs}_*)$.

Under many common missing data patterns~\citep{davenport20141, cai2013max, chen2023statistical}, it can then be verified that our consistency results in Theorem 1 still hold for $\hat\bphi^{obs}_*$ under some regularity conditions.  Under this setup, the asymptotic normality for $\hat\bphi^{obs}_*$ presented in Theorem 2 and can also be established, with appropriate modifications to the asymptotic variance-covariance matrices. Specifically, in the expressions for $\bSigma_{\beta,j}^*$ and $\bSigma_{\gamma,j}^*$, the summations $\sum_{i=1}^n$ should be replaced by $\sum_{i\in\cN_j}$, where $\cN_j$ is the set of subjects whose responses to item $j$ are observed.
Similarly, in the expression for $\bSigma_{u,i}^*$, the summation $\sum_{j=1}^q$ should be replaced by $\sum_{j \in \cQ_i}$, where $\cQ_i$ is defined above.
Finally, Corollary 1, which establishes the estimators for the asymptotic covariance matrices $\bSigma_{\beta,j}^*$, $\bSigma_{\gamma,j}^*$ and $\bSigma_{u,i}^*$, remains valid under the same modifications. Our real data analysis is carried out with these adjustment to accommodate missing data.

\subsection{Complete Point and Interval Estimation for School Strata Effects}
In real data analysis, we apply the proposed method to item response and covariate data with the covariate dimension $p_* = 9$. In particular, we include one gender variable and eight school strata variables, which are indicator variables capturing the type and location of the school attended by each student. Their detailed definitions are as follows:

\begin{enumerate}
        \item Whether the attended school is a public school;
         \item Whether the attended school is in rural area (as opposed to urban or suburban);
         \item Whether the attended school is in urban area (as opposed to suburban or rural);
    \item Whether the attended school is a junior high school;
    \item Whether the attended school is a regular senior secondary school;
        \item Whether the attended school is a skill-based senior secondary school;
        \item Whether the attended school is a comprehensive senior secondary school;
        \item Whether the attended school is a five-year junior college;
\end{enumerate}

These variables are included as covariates and the corresponding point and interval estimation results are provided in Figure~\ref{fig:School Strata}.

\begin{figure}[htbp]
\centering    
        \includegraphics[width=2in]{application_revision/TAP_public_improved_plot.pdf}~
        \includegraphics[width=2in]{application_revision/TAP_rural_improved_plot.pdf}\\
        \includegraphics[width=2in]{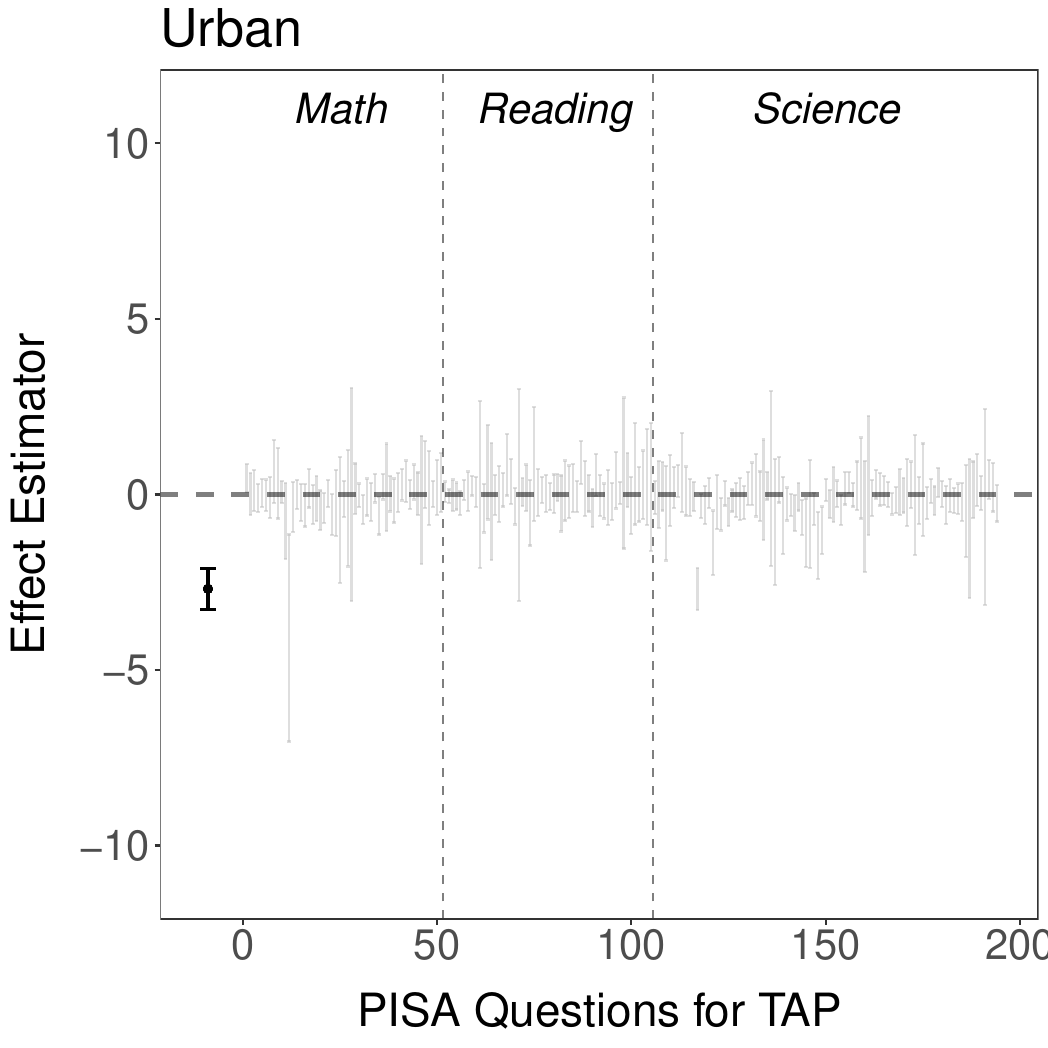}~
        \includegraphics[width=2in]{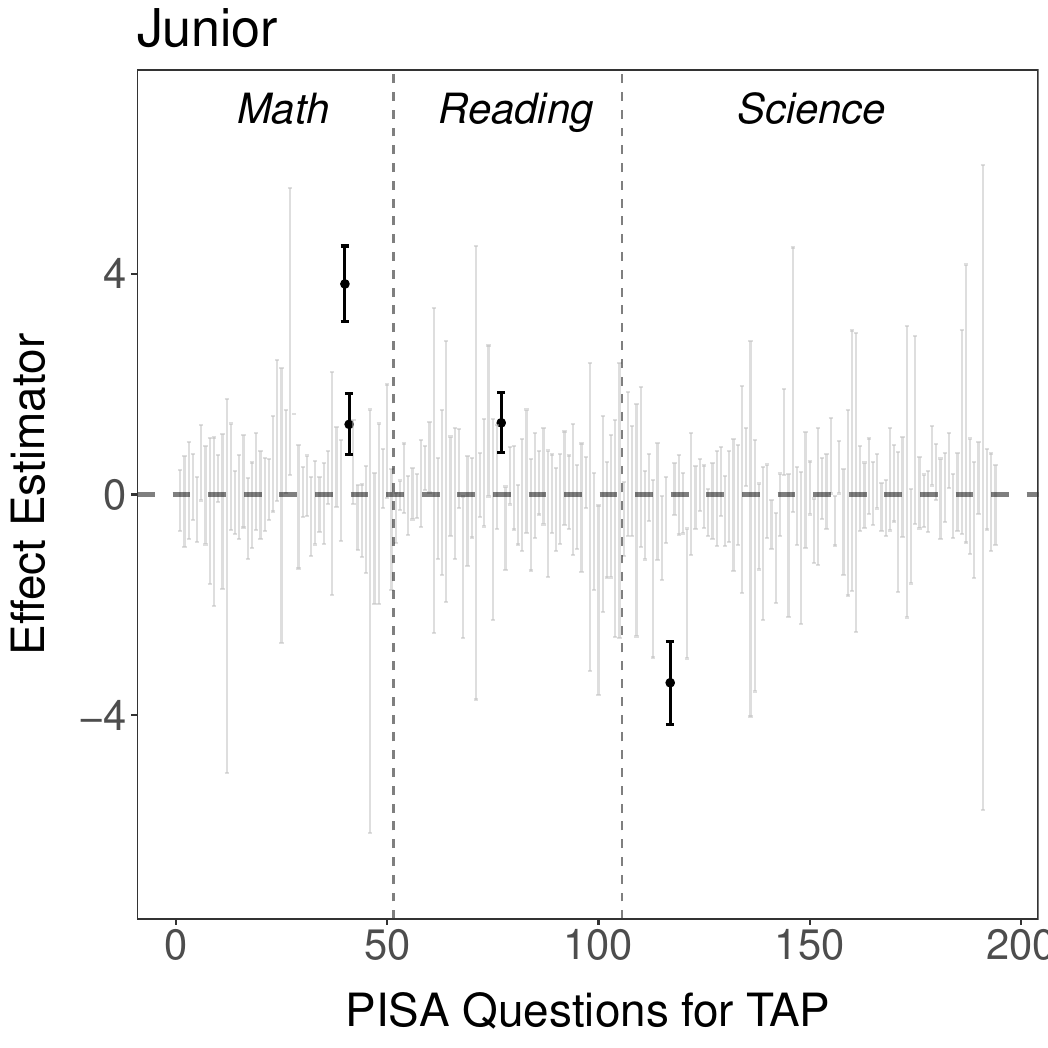}\\
        \includegraphics[width=2in]{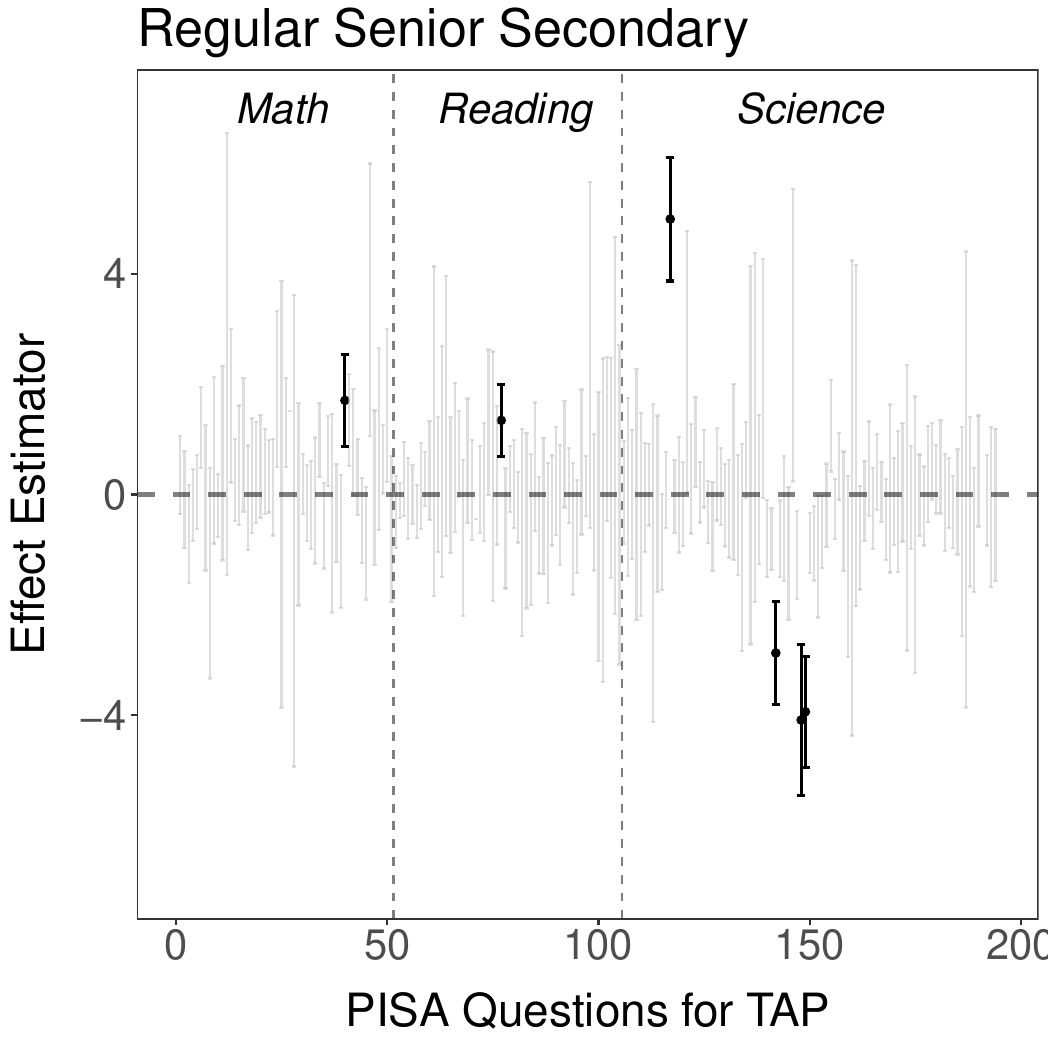}~
        \includegraphics[width=2in]{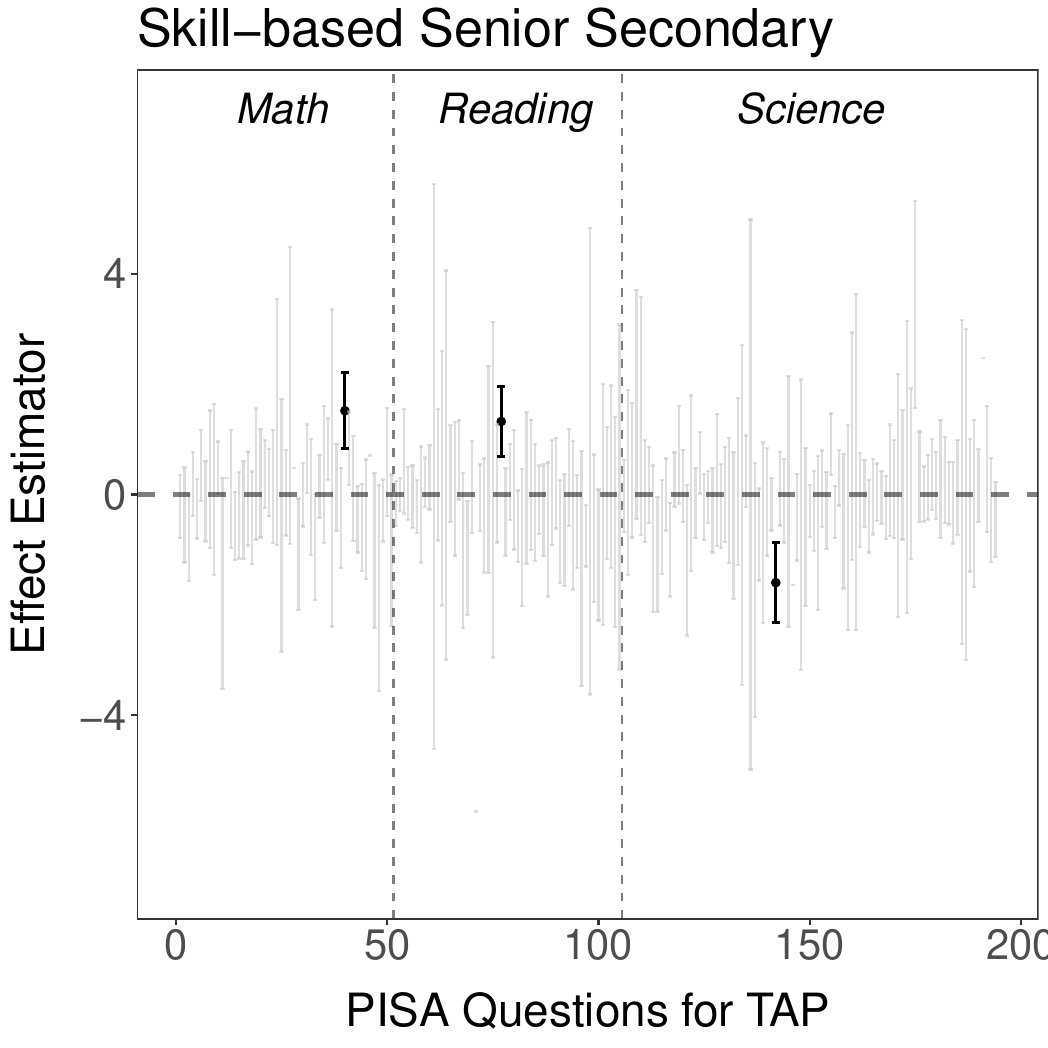}\\
        \includegraphics[width=2in]{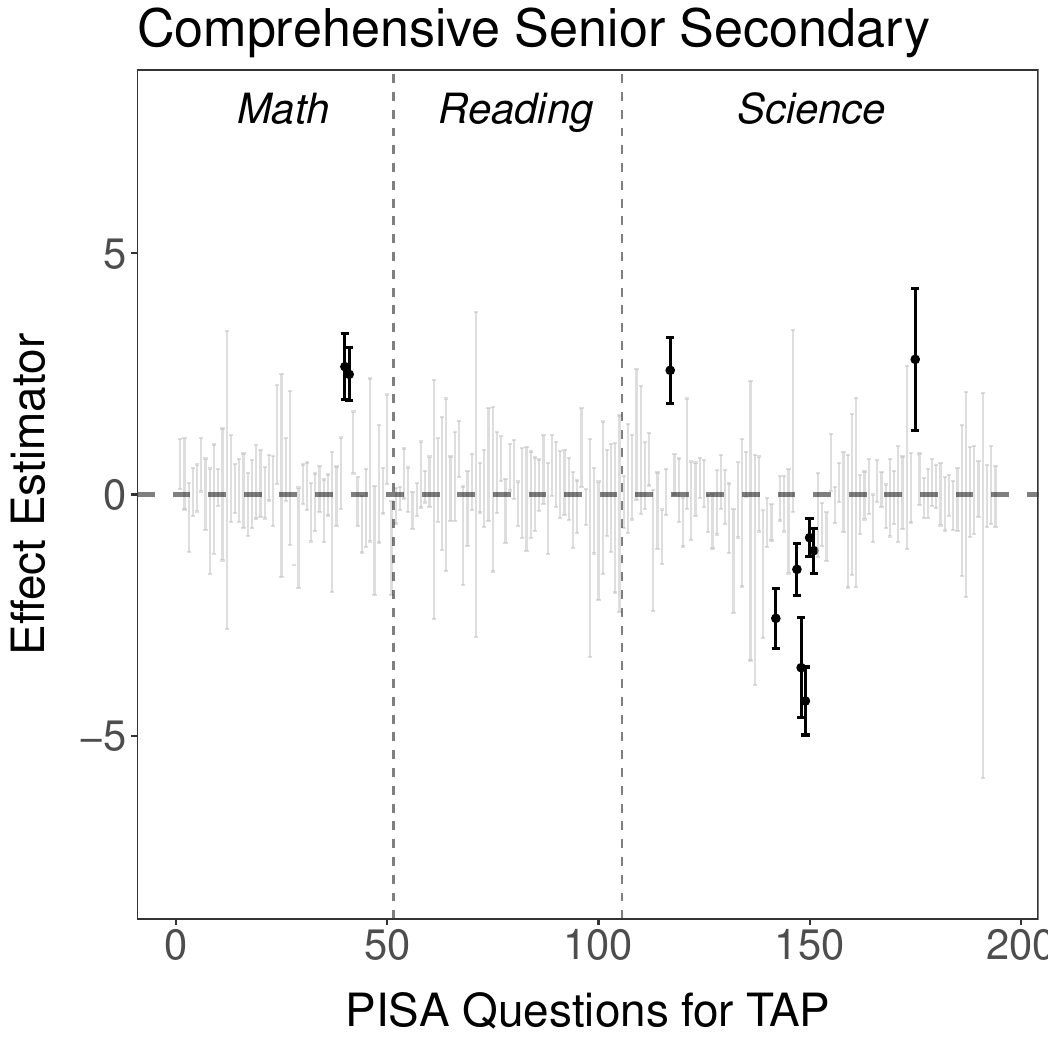}~
        \includegraphics[width=2in]{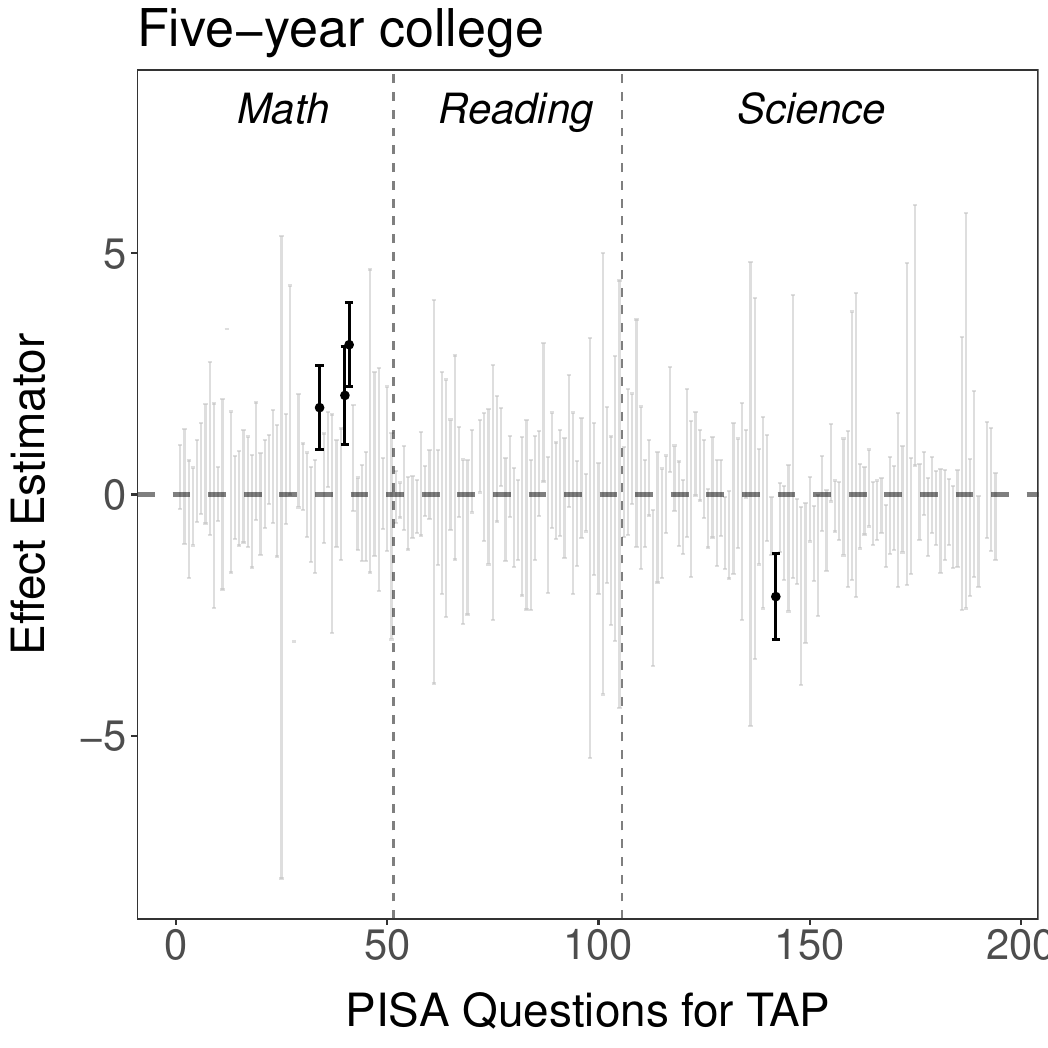}\\
    \caption{Confidence intervals for the effect of all eight school stratum covariate on each PISA question. Dark intervals correspond to confidence intervals for questions with significant school stratum bias after Bonferroni correction. }
    \label{fig:School Strata}
\end{figure}

\end{appendix}

\newpage

\end{document}